\documentclass[prx,amsmath,amssymb, notitlepage, onecolumn,
nofootinbib,
superscriptaddress,
longbibliography
]{revtex4-1}

\usepackage{amsmath}
\usepackage{tabularx,graphicx}
\usepackage{epstopdf}
\usepackage{makecell}
\usepackage{graphicx}
\usepackage{latexsym}
\usepackage{amssymb}
\usepackage{amsthm}
\usepackage{color, colortbl}
\usepackage{psfrag}
\usepackage{bbm}
\usepackage{bm}
\usepackage{titlesec}
\usepackage{dsfont}
\usepackage{feynmp}
\usepackage{slashed}
\usepackage[utf8]{inputenc}
\usepackage[symbols,nogroupskip,sort=none]{glossaries-extra}
\usepackage{nomencl}

\usepackage{multirow}
\usepackage{url}
\usepackage{xr}
\usepackage{xcite}
\usepackage{tikz}
\usepackage{microtype}
\usetikzlibrary{shapes.geometric, arrows}
\tikzstyle{arrow} = [thick,->,>=stealth]
\tikzstyle{startstop} = [rectangle, rounded corners, minimum width=3cm, minimum height=1cm,text centered, draw=black, fill=white!30]
\tikzstyle{process} = [rectangle, minimum width=3cm, minimum height=1cm, text centered, draw=black, fill=orange!30]
\tikzstyle{decision} = [diamond, minimum width=3cm, minimum height=1cm, text centered, draw=black, fill=green!30]

\makeatletter
\newcommand*{\addFileDependency}[1]{
  \typeout{(#1)}
  \@addtofilelist{#1}
  \IfFileExists{#1}{}{\typeout{No file #1.}}
}
\makeatother

\newcommand*{\myexternaldocument}[1]{%
    \externaldocument{#1}%
    \addFileDependency{#1.tex}%
    \addFileDependency{#1.aux}%
}

\myexternaldocument{main_05_12}

\newcommand{\beq}{\begin{eqnarray}}
	\newcommand{\eeq}{\end{eqnarray}}

\newcommand{\la}{\langle}
\newcommand{\ra}{\rangle}

\newcommand{\tr}{{\rm tr}}

\newcommand{\bsp}{\begin{aligned}}
	\newcommand{\esp}{\end{aligned}}
\newcommand{\const}{{\rm const}}

\newcommand{\ie}{{i.e., }}
\newcommand{\eg}{{e.g., }}

\newcommand{\supp}{\mathrm{supp}}
\newcommand{\diff}{\mathrm{diff}}
\newcommand{\rH}{\mathrm{H}}
\newcommand{\R}{\mathbb{R}}
\newcommand{\Aut}{\mathrm{Aut}}
\newcommand{\GNS}{\mathrm{GNS}}
\newcommand{\dist}{\mathrm{dist}}
\newcommand{\fa}{\mathfrak{a}}
\newcommand{\ad}{\mathrm{ad}}
\newcommand{\cS}{\mathcal{S}}
\newcommand{\cK}{\mathcal{K}}
\newcommand{\im}{\mathrm{im}}

\usepackage{xcolor}
\definecolor{darkblue}{rgb}{0.,0.,0.4}
\definecolor{darkred}{rgb}{0.5,0.,0.}
\definecolor{BlueViolet}{RGB}{138,43,226}
\definecolor{SkyBlue}{RGB}{30,144,255}
\definecolor{DarkGreen}{RGB}{0,100,0}

\usepackage[colorlinks=true,linkcolor=darkblue,citecolor=blue,urlcolor=darkred]{hyperref}
\usepackage[normalem]{ulem}

\usepackage{mathrsfs}

\usepackage{comment}

\newcommand{\z}{\mathbb{Z}}
\newcommand{\T}{\mathcal{T}}
\newcommand{\A}{\mathscr{A}}
\newcommand{\B}{\mathcal{B}}
\newcommand{\G}{\mathcal{G}}
\newcommand{\cH}{\mathcal{H}}
\newcommand{\cU}{\mathcal{U}}
\newcommand{\Ad}{\mathrm{Ad}}
\newcommand{\cC}{\mathcal{C}}
\newcommand{\bbC}{\mathbb{C}}
\newcommand{\rspan}{\mathrm{span}}
\newcommand{\dd}{\mathrm{d}}
\newcommand{\rHom}{\mathrm{Hom}}
\newcommand{\QCA}{\mathrm{QCA}}
\newcommand{\ind}{\mathrm{ind}}

\newcommand{\rk}{\mathrm{rank}}
\newcommand{\SR}{\mathrm{SR}}
\newcommand{\tphi}{\tilde{\phi}}
\newcommand{\tH}{\tilde{H}_{t}'}
\newcommand{\fg}{\mathfrak{g}}
\newcommand{\id}{\mathrm{id}}
\newcommand{\cM}{\mathscr{M}}
\newcommand{\cF}{\mathscr{F}}
\DeclareMathOperator*\lowlim{\underline{lim}}

\newcommand{\diam}{\mathrm{diam}}
\def\U{\mathrm{U}(1)}

\newtheorem{corollary}{Corollary}
\newtheorem{theorem}{Theorem}
\newtheorem{lemma}{Lemma}
\newtheorem{definition}{Definition}
\newtheorem{example}{Example}
\newtheorem{proposition}{Proposition}
\newtheorem{remark}{Remark}
\makenomenclature

\numberwithin{equation}{section}
\numberwithin{corollary}{section}
\numberwithin{theorem}{section}
\numberwithin{lemma}{section}
\numberwithin{definition}{section}
\numberwithin{example}{section}
\numberwithin{proposition}{section}
\numberwithin{remark}{section}

\usepackage{pst-all}

\usepackage{leftidx}

\usepackage[off]{auto-pst-pdf} 
\usepackage{booktabs}
\usepackage{yfonts}
\makeatletter

\usepackage[caption=false]{subfig}
\usepackage{enumitem}

\glsxtrnewsymbol[description={Absolute value of a complex number or norm of a vector, Eq.~\eqref{eq:vec_norm}}]{abs}{\ensuremath{|\cdot|}}
\glsxtrnewsymbol[description={norm of an operator or a functional, Eq.~\eqref{eq:operator_norm}, Eq.~\eqref{eq: def_functional_norm}}]{norm}{\ensuremath{||\cdot||}}

\glsxtrnewsymbol[description={algebra of local operators, definition~\ref{def:local_operators}}]{local}{\ensuremath{\A^{\ell}}}

\glsxtrnewsymbol[description={algebra of local operators supported on $\Gamma$}]{local_g}{\ensuremath{\A^{\ell}_{\Gamma}}}

\glsxtrnewsymbol[description={algebra of quasi local operators, definition~\ref{def:quasi-local}}]{quasilocal}{\ensuremath{\A^{ql}}}
\glsxtrnewsymbol[description={algebra of quasi local operators supported on $\Gamma$}]{quasilocal_g}{\ensuremath{\A^{ql}_{\Gamma}}}
\glsxtrnewsymbol[description={algebra of bounded operators on Hilbert space $\cH$}]{bounded}{\ensuremath{\B(\cH)}}
\glsxtrnewsymbol[description={cardinality of set $\Gamma$}]{card}{\ensuremath{|\Gamma|}}
\glsxtrnewsymbol[description={group of circuits, example~\ref{example:circuit}}]{circuits}{\ensuremath{\G^{cir}}}
\glsxtrnewsymbol[description={group of quantum cellular automata (QCA), definition~\ref{def:QCA}}]{QCA}{\ensuremath{\G^{\QCA}}}
\glsxtrnewsymbol[description={group of locality-preserving automorphisms (LPA), definition~\ref{def:LPA}}]{LPA}{\ensuremath{\G^{lp}}}
\glsxtrnewsymbol[description={$a$ is nearly included in $B$, definition~\ref{def:near_inclusion}}]{nrinc}{\ensuremath{a\stackrel{\epsilon}{\in}B}}
\glsxtrnewsymbol[description={$A$ is nearly included in $B$, definition~\ref{def:near_inclusion}}]{nrinc2}{\ensuremath{A\stackrel{\epsilon}{\subset}B}}

\glsxtrnewsymbol[description={GNVW index map, Eq.~\eqref{eq:GNVW_index}}]{index}{\ensuremath{\ind}}
\glsxtrnewsymbol[description={an abstract quantum state, definition~\ref{def:State}}]{state}{\ensuremath{\psi}}

\glsxtrnewsymbol[description={The GNS ideal associated to state $\psi$, definition~\ref{eq:GNS_ideal}}]{GNSideal}{\ensuremath{N_{\psi}}}

\glsxtrnewsymbol[description={The GNS Hilbert space associated to state $\psi$, definition~\ref{def:GNS_triple}}]{GNSspace}{\ensuremath{\cH_{\psi}}}

\glsxtrnewsymbol[description={The representative of a state $\psi$ in GNS Hilbert space $\cH_{\psi}$, definition~\ref{def:GNS_triple}}]{rep}{\ensuremath{|\psi\ra}}

\glsxtrnewsymbol[description={The GNS homomorphism of a state $\psi$, definition~\ref{def:GNS_triple}}]{map}{\ensuremath{\pi_{\psi}}}

\glsxtrnewsymbol[description={The restriction of a state $\psi$ on $\Gamma$}]{res}{\ensuremath{\psi|_\Gamma}}

\glsxtrnewsymbol[description={pure product state, definition~\ref{def:product_state}; anomaly index, Eq.~\eqref{eq:anomaly_index}}]{anomaly}{\ensuremath{\omega}}

\glsxtrnewsymbol[description={derivation a.k.a Hamiltonian $H$}]{der}{\ensuremath{\delta_{H}}}

\glsxtrnewsymbol[description={time evolution, Eq.~\eqref{eq:limit_dynamics}}]{evo}{\ensuremath{\alpha^{t}}}

\glsxtrnewsymbol[description={Norm on Hamiltonians associated to a reproducing function $F$, Eq.~\eqref{eq:F_norm}}]{norm_H}{\ensuremath{||H||_{F}}}

\glsxtrnewsymbol[description={degree $n$ group cohomology of $G$ with coefficients in $A$, Eq.~\eqref{eq:group_coho}}]{grp}{\ensuremath{\rH^{n}(G;A)}}

\glsxtrnewsymbol[description={subgroups of $\G^{lp}$, definition~\ref{def:subgroups}}]{subgrp}{\ensuremath{\G^{lp}_{0,<0,\geqslant0,\pm}}}

\glsxtrnewsymbol[description={finite time evolution of a local Hamiltonian (possibly time-dependent)}]{evo2}{\ensuremath{\tau}}

\begin{document}

\title{Entanglement area law and Lieb-Schultz-Mattis theorem in long-range interacting systems, and symmetry-enforced long-range entanglement}

\author{Ruizhi Liu}
\affiliation{Department of Mathematics and Statistics, Dalhousie University, Halifax, Nova Scotia, Canada, B3H 4R2}
\affiliation{Perimeter Institute for Theoretical Physics, Waterloo, Ontario, Canada N2L 2Y5}

\author{Jinmin Yi}
\affiliation{Perimeter Institute for Theoretical Physics, Waterloo, Ontario, Canada N2L 2Y5}
\affiliation{Department of Physics and Astronomy, University of Waterloo, Waterloo, Ontario, Canada N2L 3G1}

\author{Shiyu Zhou}
\affiliation{Perimeter Institute for Theoretical Physics, Waterloo, Ontario, Canada N2L 2Y5}

\author{Liujun Zou}
\affiliation{Perimeter Institute for Theoretical Physics, Waterloo, Ontario, Canada N2L 2Y5}
\affiliation{Department of Physics, National University of Singapore, Singapore 117542}

\begin{abstract}

We establish multiple interrelated, fundamental results in quantum many-body systems that can have long-range interactions. For a sufficiently long quantum spin chain, we first show that if the multi-spin interactions in the Hamiltonian decay fast enough as their ranges increase and the Hamiltonian is gapped, then the ground states satisfy the entanglement area law, even if there is a ground state degeneracy due to a spontaneously broken discrete symmetry. This area law also holds for certain excited states. Second, if such a long-range interacting Hamiltonian has an anomalous symmetry, then the Lieb-Schultz-Mattis theorem applies, \ie the Hamiltonian cannot have a unique gapped symmetric ground state. If the Hamiltonian contains only 2-spin interactions, these results hold when the interactions decay faster than $1/r^2$, with $r$ the distance between the two interacting spins. Third, we show that pure states with an anomalous symmetry, which may not be a ground state of any natural Hamiltonian, must be long-range entangled. The symmetries we consider include on-site internal symmetries combined with lattice translation symmetries, and they can also extend to purely internal but non-on-site symmetries. Moreover, these internal symmetries can be discrete or continuous. We explore the applications of these results through various examples.

\end{abstract}

\maketitle
\tableofcontents\

\section{Introduction}

The entanglement structure is an important characterization of quantum matter \cite{Zeng2015}. In this regard, a fundamental result is the {\it entanglement area law}, which asserts that the entanglement entropy of a large subregion in the ground states of many physically natural Hamiltonians increases like the boundary area of this subregion \cite{Eisert2008}. The area law and its analogs have been verified in numerous concrete examples and rigorously proved under certain general conditions \cite{Hastings2007, Masanes2009, matsui2011boundedness, Arad2011, Brandao2015, Arad2013, Brandao2013, Cho2014, Brandao2014, Arad2016, Cho2017a, Anshu2019, Kuwahara2019, Ukai2024}. Besides being a common characteristic of many physically relevant states, the area law also provides the basis of some powerful numerical algorithms, such as density-matrix renormalization group \cite{Schollwock2010}.

The area law was initially shown to hold in the ground state of a gapped local one dimensional (1D) Hamiltonian that has a unique ground state \cite{Hastings2007}. Recently, it was extended to the ground state of certain gapped long-range interacting systems, assuming no ground state degeneracy \cite{Kuwahara2019}. However, the following important question has been left open in Ref. \cite{Kuwahara2019}:

\begin{itemize}

\item In the presence of ground state degeneracy, which is a common scenario whenever there is a spontaneously broken symmetry,  does the area law hold in such gapped long-range interacting systems?{\footnote{Ref. \cite{Ukai2024} proves an area law for gapped long-range interacting systems that may have multiple ground states. However, Ref. \cite{Ukai2024} studies systems with infinite size, whereas real physical systems are all finite. So additional analysis is needed to address this question for physically relevant systems.}}

\end{itemize}

While the area law captures the {\it amount} of entanglement in a quantum many-body state, another crucial aspect of quantum many-body entanglement is whether the entanglement is short-ranged or long-ranged. Roughly speaking, long-range (short-range) entangled states cannot (can) be prepared from a product state by a finite-time evolution generated by a Hamiltonian with good locality properties \cite{Zeng2015}. States satisfying the entanglement area law may be either long-range entangled or short-range entangled. There are many interesting long-range entangled states, such as cat states with spontaneous symmetry breaking, topologically ordered states with emergent anyons, critical states described by conformal field theories, etc, and some of them are potentially useful for quantum computation \cite{Kitaev1997, Nayak2007}.

Because of the interesting properties of the long-range entangled states, it is desirable to realize them in quantum materials and quantum simulators. So a natural question follows: What is the general condition under which long-range entangled states can emerge? This question is related to another fundamental result in quantum many-body physics, the Lieb-Schultz-Mattis (LSM) theorem, which states that a system cannot have a unique gapped ground state if its Hamiltonian satisfies certain symmetry conditions \cite{Lieb1961, Oshikawa1999, Hastings2003}.  Recently, the LSM constraints have been interpreted from various perspectives and generalized to different contexts~\cite{Cheng2015, Po2017, Jian2017,  Cho2017,Watanabe2018LSM,Metlitski2018,Cheng2018a, Kobayashi2018, Ogata2019LSM,Else2020, Jiang2019,Yao2021twisted, Ogata_2021, Aksoy2021, Ye2021a,Ma2022a, Cheng2022, Kawabata2023, Aksoy2023, Seifnashri2023, Zhou2023, kapustin2024anomalous,Garre_Rubio2024anomalous, Pace2024}. 
Furthermore, these constraints are identified as a key ingredient to study the classification of quantum phases of matter in a lattice system~\cite{Zou2021, Ye2021a, Ye2023, Liu2024}. All examples of ground states of a system satisfying the symmetry conditions of the LSM theorem turn out to be long-range entangled, which suggests that one should look for long-range entangled states in such systems. Therefore, a basic question is: 

\begin{itemize}
    
    \item Do the symmetry conditions in the LSM theorem guarantee all symmetric states to be long-range entangled?
    
\end{itemize}

This question is widely believed to have an affirmative answer, but it has only been addressed in Refs. \cite{Sanz2009, Chen2010a, Else_2014, Gioia2021} for some special cases, and the argument in Ref. \cite{Else_2014} is only heuristic. Here we aim for a general and rigorous proof.

Previous studies of LSM constraints often focus on systems with local interactions. However, many systems feature long-range interactions, which usually take the form of a 2-body interaction that decays as $1/r^\mathfrak{a}$, with $r$ the distance between the two interacting objects and $\mathfrak{a}$ an exponent. As examples, electronic systems have Coulomb interaction with $\mathfrak{a}=1$, Rydberg atoms have dipolar or van der Waals interactions with $\mathfrak{a}=3$ or $\mathfrak{a}=6$, and for trapped ions $\mathfrak{a}$ can be tuned between 0 and 3~\cite{Defenu2022}. So a pertinent question is: 

\begin{itemize}

\item Is the LSM theorem applicable to long-range interacting systems?

\end{itemize}

Motivated by these questions, in this paper, we show that the entanglement area law holds in certain gapped long-range interacting systems that may have ground state degeneracy, extending the results in Refs. \cite{Hastings2007, Kuwahara2019}. Interestingly, this result immediately becomes useful and enables us to prove the LSM theorem in long-range interacting systems with certain symmetry conditions. Moreover, we show that the symmetry conditions in our LSM theorem force all pure states to be long-range entangled, no matter whether the state is a ground state of any natural Hamiltonian, rigorously generalizing the results in Refs. \cite{Sanz2009, Chen2010a, Else_2014, Gioia2021} to a wider class of states and symmetry settings.

We remark that the Hamiltonians we consider can contain generic $k$-body interactions with any finite $k$. For 2-body interactions decaying as $1/r^\mathfrak{a}$, our area law and LSM theorem hold when $\mathfrak{a}>2$ (for $k$-body interactions with $k>2$, the condition under which our theorems hold is stated in Eq.~\eqref{eq:admissible}).  The type of symmetries under consideration is also very broad, including an on-site symmetry combined with the lattice translation symmetry, as featured in the original LSM theorem. Additionally, the symmetry can be purely internal but non-on-site. Furthermore, the internal symmetries can be either discrete or continuous. Our results have wide applicability, and we will discuss some examples below.

\section{Organization of this paper}

This paper contains a concise main text which summarizes our main results, and extensive appendices which present the technical details.

In Sec.~\ref{sec:locality}, we review the definition of quantum cellular automata (QCA), which describes the relevant symmetries in the present paper. The most important feature of these symmetries is that they transform local operators into local operators. We then define the anomaly index following Ref.~\cite{kapustin2024anomalous}.

In Sec.~\ref{Sec:main_area_law}, we introduce a particular class of long-range interacting Hamiltonians we wish to work with, called {\it admissible Hamiltonians} (see Eq.~\eqref{eq:admissible}). We also state our first main results (Theorem \ref{thm:main_area_law} and Proposition \ref{prop:degenerate_area_law_main}) on the entanglement area law of gapped ground states of admissible Hamiltonians.

In Sec.~\ref{Sec:main_LSM}, we come to the next main theorems of this paper (Theorem \ref{thm:main} and Corollary \ref{corollary: LSM finite}), \ie we apply the result of the entanglement area law to derive the LSM theorem of admissible Hamiltonians with anomalous symmetries. 

Sec.~\ref{sec: symmetry-enforced long-range entanglement} is devoted to discussing the entanglement properties of symmetric states under anomalous symmetries, and we show that all pure states in sufficiently long quantum spin chains with an anomalous symmetry must be long-range entangled (Theorem \ref{thm:finite_SRE_main}).

We then present several examples and applications of our results  in Sec.~\ref{Sec:main_examples}, which are potentially relevant to the recent progress in long-range interacting Rydberg systems.

In the last Sec.~\ref{Sec:main_discussion}, we discuss several directions to generalize the current work.

The organization of the appendices is given at their beginning.

\section{Locality, symmetries and anomalies in quantum spin chains}\label{sec:locality}

Many important concepts in quantum many-body physics, such as locality, long-range entanglement and quantum phases of matter, are most cleanly defined in systems with an infinite size. In the next few sections, we always first present results about infinite systems, which can be viewed as the thermodynamic limits where a sequence of finite systems converge to. However, unlike many previous works that exclusively focus on infinite systems \cite{bratteli2013operator1, bratteli2013operator2, Landsman:2017hpa, Ogata2021}, we note that all real physical lattice systems are of finite sizes, no matter how large they are. Therefore, from the results on infinite systems, we take another step to extract their implications on finite systems. This step is not only physically relevant, but also, as we see, sometimes very interesting and nontrivial.

In the rest of this section, we introduce various notations relevant to infinite-size spin chains, which are needed to formulate our results. The most important concept here is the anomaly index defined for a symmetry action described by a quantum cellular automaton (QCA). A more detailed review of these notions is in Appendix \ref{sec: formalism review}.

One of our goals is to consider quantum spin chains with a very general class of symmetries. Nevertheless, we assume that the symmetry operations are unitary and transform a local operator to another local operator in the nearby region. QCA is exactly this type of operations by definition. More precisely, we denote the action of a QCA by $\alpha$ (\ie this QCA transforms a local operator $A$ to $\alpha(A)$), then it satisfies

\begin{enumerate}

    \item $\alpha(AB)=\alpha(A)\alpha(B)$ and $\alpha(A^\dag)=\alpha(A)^\dag$ for any local operators $A$ and $B$.
    
    \item $\alpha$ is invertible.

    \item For any local operator $A$, $\alpha(A)$ is again a local operator, and the distance between the supports of $A$ and $\alpha(A)$ cannot exceed a certain value that is independent of $A$.
    
\end{enumerate}

It is known that QCA form a group, denoted by $\G^{\QCA}$. The structure of $\G^{\QCA}$ is well-understood in 1D \cite{Gross_2012}. In essence, 1D QCA are combinations of finite-depth quantum circuits and translations (see Refs.~\cite{arrighi2019overview,Farrelly_2020} for review). Given a symmetry group $G$, by slightly abusing the notations{\footnote{Previously the notation $\alpha$ is used to represent an operation acting on operators, but here we use it to represent a map from the symmetry group $G$ to all possible QCA opetations $\G^{\QCA}$.}}, the symmetry action can be represented by a group homomorphism $\alpha: G\to \G^{\QCA}$. This symmetry may contain internal and/or translation symmetry, and the internal symmetry, which acts as a finite-depth quantum circuit, may be discrete or continuous, on-site or non-on-site. This general type of symmetry actions covers many physically relevant cases.

\begin{figure}[!t]
    \centering  \includegraphics[width=0.47\textwidth]{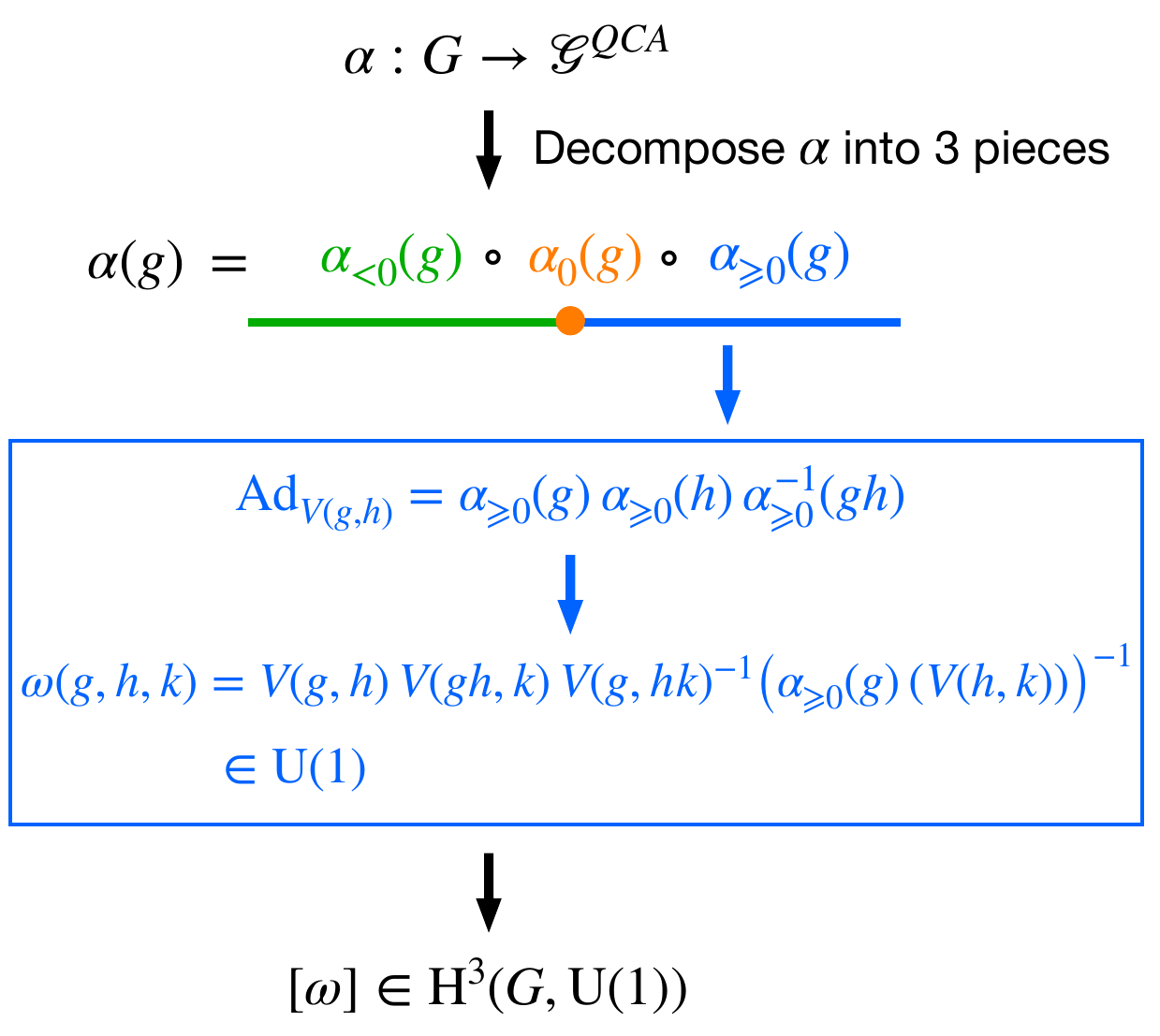}
    \caption{An illustration on how to obtain anomaly index $\omega\in\rH^{3}(G;\U)$ from the symmetry action $\alpha$.}
    \label{fig:anomaly_indx}
\end{figure}

Given such a symmetry action $\alpha: G\to\G^{\QCA}$, an important concept is the anomaly index, which takes values in $\rH^3(G, \U)$~\cite{kapustin2024anomalous} (see Appendix \ref{sec:group_cohomology} for a review of group cohomology). The construction of this anomaly index is similar to the previous work~\cite{Else_2014}, and the innovation of this new anomaly index is that it applies to translation symmetries and continuous internal symmetries. Below we sketch the definition of the anomaly index, and more details can be found in Ref.~\cite{kapustin2024anomalous} and Appendix \ref{sec:anomaly_index}.

First, suppose $\alpha$ is an internal symmetry action (\ie it is a finite-depth quantum circuit and contains no translation). For an arbitrary site, say, the origin, it can be shown that $\alpha$ can be decomposed as
\beq \label{eq: decomposition}
    \alpha=\alpha_{<0} \, \alpha_0 \, \alpha_{\geqslant 0}
    \; ,
\eeq
where $\alpha_{\geqslant 0}$ (resp. $\alpha_{< 0}$) is a finite-depth quantum circuit supported on $[0, \infty)$ (resp. $(-\infty, 0)$), and $\alpha_0$ is a finite-depth quantum circuit with a finite support near the origin (see Fig.~\ref{fig:anomaly_indx}). Although $\alpha$ is a group homomorphism, in general $\alpha_{\geqslant 0}$ is not. In fact, for any $g, h\in G$,

\beq
\alpha_{\geqslant 0}(g) \, \alpha_{\geqslant 0}(h)=\Ad_{V(g,h)} \, \alpha_{\geqslant 0}(gh)
\; ,
\eeq
where $V:G\times G\to \cU^{\ell}$ with $\cU^\ell$ the group of local unitaries is not necessarily a homomorphism, and $\Ad_V(A):=VAV^\dag$ for any local operator $A$. The associativity of $\alpha_{\geqslant 0}$, \ie $\left( \alpha_{\geqslant 0}(g) \, \alpha_{\geqslant 0}(h) \right) \, \alpha_{\geqslant 0}(k) = \alpha_{\geqslant 0}(g) \, \left( \alpha_{\geqslant 0}(h) \, \alpha_{\geqslant 0}(k) \right)$ with $g, h, k\in G$, puts further constraints on $V$: $\Ad_{\omega(g,h,k)}=1$, where 
\begin{widetext}
\beq
    \omega(g,h,k)
    =V(g,h)V(gh,k)V(g,hk)^{-1}(\alpha_{\geqslant 0}(g)(V(h,k)))^{-1}
    \; .
\eeq
\end{widetext}
This means the above $\omega$ is actually a phase since it commutes with all local operators. It can be checked that $\omega$ satisfies the 3-cocycle condition, and multiplying $V(g,h)$ by a phase $\rho(g,h)\in \U$ shifts $\omega$ by a 3-coboundary. Therefore, $\omega$ specifies an element in $\rH^3(G, \U)$, and this element is defined as the anomaly index associated with the symmetry action $\alpha$.

If $\alpha$ contains translation, because the information moves unidirectionally in this case, which cannot be achieved by finite-depth quantum circuits, the decomposition as in Eq. \eqref{eq: decomposition} is not possible. However, one can stack the system with another copy on which the translation acts oppositely. The symmetry action on this composite system (denoted by $\alpha_{\otimes}$) contains no translation, and the anomaly index of $\alpha$ is defined to be the anomaly index of $\alpha_{\otimes}$.

In Appendix \ref{sec:anomaly_index}, we prove that this anomaly index is independent of the choice of the site to decompose $\alpha$ in Eq. \eqref{eq: decomposition}, which was not explicitly proved in Refs.~\cite{Else_2014,kapustin2024anomalous}. 

With the above definition of anomaly index, we say that the $G$-symmetry is anomalous if $\omega\not =1\in \mathrm{H}^{3}(G;\U)$. Otherwise, we say it is anomaly-free or non-anomalous.

To connect the above discussion with the more familiar notions, let us discuss an example. Consider a quantum spin chain with a symmetry $G=\z\times G_{\text{int}}$, where $\z$ represents translation and  $G_{\text{int}}$ is an internal symmetry (taken as either a discrete group or a finite dimensional Lie group). Then $\rH^3(G, \U)\simeq \rH^2(G_{\text{int}}, \U)\oplus \rH^3(G_{\text{int}}, \U)$ \cite{Cheng2015}. The part $\rH^2(G_{\text{int}}, \U)$ means if the degrees of freedom in a unit cell form a projective representation under $G_{\text{int}}$, which is precisely the condition of the original LSM theorem, then the $G$-symmetry is anomalous. The part $\rH^3(G_{\text{int}}, \U)$ means that even for a purely internal symmetry $G_{\text{int}}$, the $G$-symmetry can be anomalous if its anomaly index corresponds to a nontrivial element in $\rH^3(G_{\text{int}}, \U)$. We will present an example of such internal symmetries in Sec.~\ref{Sec:main_examples}.

Before ending this section, we remark that although we focus on symmetries implemented by QCA in the main text, as QCA preserve locality in the most strict sense, in the appendices our considerations are extended to locality-preserving automorphisms, an even more general class of symmetry actions that preserve locality only in an approximate form (\ie a local operator can acquire a tail after it is acted by a locality-preserving automorphism), and our main theorems still hold.

\section{Admissible Hamiltonians and entanglement area law}\label{Sec:main_area_law}

Now we proceed to our main theorems, which accommodate long-range and many-body interactions. Denote the lattice where our quantum spin chain lives by $\Lambda$ ($\Lambda=\z$ for an infinite chain). Consider a 1D Hamiltonian with at most $k$-body interactions, $H=\sum_{|Z|\leqslant k}h_{Z}$, that satisfies
\beq\label{eq:admissible}
\begin{split}
    \max_{i\in\Lambda}  \left( \sum_{Z:Z\owns i,~\diam(Z)=r} ||h_{Z}|| \right) &< \frac{J}{r^{\fa}},\quad \mathfrak{a}>2,\,\forall\, r>0
    \; ,
    \\ 
    \text{and} ~ \max_{i\in\Lambda} \| h_{i} \| &< B
    \; ,
\end{split}
\eeq
where $h_Z$ is an interaction supported on region $Z\subset\Lambda$, $\diam(Z)=\sup_{x,y\in Z}|x-y|$, $h_{i}$ is an on-site potential at site $i$, and $J,B>0$ are constants. Such a Hamiltonian $H$ is deemed as {\it{admissible}} \cite{Kuwahara2019}. Specifically, if the Hamiltonian includes at most 2-body long-range interactions, Eq.~\eqref{eq:admissible} indicates that the interactions decay faster than $r^{-2}$, with $r$ the distance between the two interacting spins. In general, Eq.~\eqref{eq:admissible} ensures that for any disjoint intervals $X$ and $Y$ separated by $d$ (see Fig.~\ref{fig:addmissable}), their interaction
$V_{X,Y}=\sum_{Z:Z\cap X\neq\emptyset, Z\cap Y\neq\emptyset} h_{Z}$
goes to 0 as $d\to \infty$. This property is crucial in the proof of the area law.

\begin{figure}[!t]
    \centering
    \includegraphics[width=0.4\textwidth]{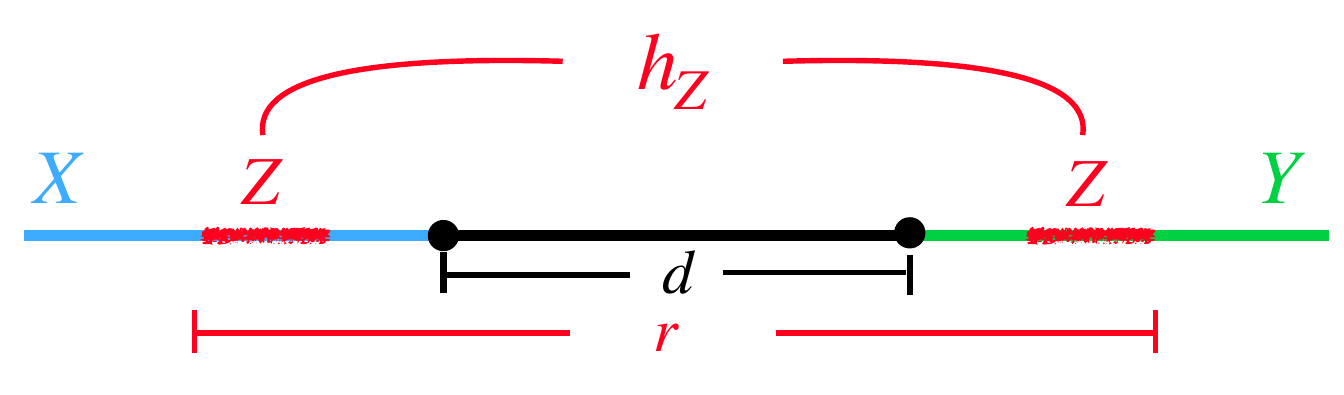}
    \caption{The interaction of two disjoint intervals $X,Y$ seprated by distance $d$. The range of interaction is denoted by $Z$ with $\diam(Z)=r$.}
    \label{fig:addmissable}
\end{figure}

It is shown in Ref.~\cite{Kuwahara2019} that for a \textit{finite-size} system with a gapped admissible Hamiltonian and a unique ground state, the ground state satisfies the entanglement area law, Eq. \eqref{eq: entanglement area law}. This result was extended to infinite-size systems that may have multiple ground states in Ref. \cite{Ukai2024}. Using an approach different from Ref. \cite{Ukai2024}, we also prove the following infinite-size version of area law (see Appendix \ref{sec:area_law} for the detailed proof and comparisons between our approach and the ones in Refs. \cite{Kuwahara2019} and \cite{Ukai2024}).
\begin{theorem}[\textbf{Theorem \ref{thm:area_law} in Appendix \ref{subapp: proving area law}}]\label{thm:main_area_law}
    Consider a locally unique gapped ground state of an infinite quantum spin chain with an admissible Hamiltonian $H$ (see Eq.~\eqref{eq:admissible}). The entanglement entropy associated with the interval $[0, n]$
\beq\label{eq: entanglement area law}
S([0,n])<S_{0}
\eeq
where $S_{0}$ is a constant which depends on $\{\fa,J,B,k,d,\Delta\}$ only, but does not depend on $n$. Here $\Delta$ is the energy gap, $d$ is the dimension of the local Hilbert space and we have at most $k$-body interactions in $H$.

\end{theorem}

Here a locally unique gapped ground state is a terminology coined in Ref. \cite{Tasaki2022topological} and it is defined in infinite-size systems, which essentially means a gapped ground state of the Hamiltonian within a superselection sector. Namely, different locally unique gapped ground states cannot be converted into each other by any local operator. Notice a locally unique gapped ground state may not be a ground state of the Hamiltonian in the standard sense (see Appendix \ref{sec: formalism review} for more details of this notion). For example, for the Ising model in a finite lattice with $H=-\sum_{i}Z_iZ_{i+1}$, there are two degenerate gapped ground states, \ie the all-up state and the all-down state. Both of these states are locally unique gapped ground states after taking the thermodynamic limit, because they cannot be coupled by any local operator and they belong to two superselection sectors. When this Hamiltonian is perturbed and becomes $H=-\sum_{i}Z_{i}Z_{i+1}-hZ_{i_0}$ with $0<h<1$ and $i_0$ an arbitrary site, although the all-up state becomes the unique gapped ground state and the all-down state is an excited state, both of them are still locally unique gapped ground states in the thermodynamic limit, as they are still gapped ground states in their respective superselection sectors.

As the infinite spin chain can be viewed as the thermodynamic limit of a sequence of finite spin chains with increasing sizes, Theorem \ref{thm:main_area_law} has important implications on finite-size systems whose (almost degenerate) ground states satisfy the superselection rules. In this case, we have
\begin{proposition}[Informal]\label{prop:degenerate_area_law_main}
    Let $\Gamma\subseteq\Lambda$ be a fixed finite interval in integers, then when the system is sufficiently large and the almost degenerate ground states satisfy a superselection rule, all states in the almost degenerate ground state subspace satisfy the entanglement area law, \ie the von Neumann entropy associated with the region $\Gamma$ is upper bounded when the size of $\Gamma$ increases.
\end{proposition}
The precise version of this statement is in Proposition~\ref{prop:degenerate_area_law}. 

One important application of this result is to systems with a spontaneously broken discrete symmetry, where the almost degenerate ground states can indeed be shown to satisfy the superselection rule, given the presence of the local order parameters associated with the broken symmetry. Moreover, when such a system is perturbed by adding a local order parameter operator with a small coefficient to the Hamiltonian, some of the almost degenerate ground states will become excited states. Since in the thermodynamic limit these excited states are still locally unique gapped ground states (similar to the above example of Ising model), they also satisfy the entanglement area law (see Lemma \ref{lemma:purity} for the details).

Therefore, our Theorem \ref{thm:main_area_law} and Proposition \ref{prop:degenerate_area_law_main} greatly extend the entanglement area law in Refs. \cite{Hastings2007, Kuwahara2019} to ground states of gapped admissible Hamiltonians that may have ground state degeneracy, and also to certain excited states. Moreover, besides being interesting on its own, we will see that Theorem \ref{thm:main_area_law} is very useful for the proof of the LSM theorem in the next section.

\section{LSM theorem in long-range interacting systems}\label{Sec:main_LSM}

To derive our Theorem \ref{thm:main}, which generalizes the standard LSM theorem to long-range interacting systems with a general anomalous symmetry described by QCA, we present a powerful lemma that connects the anomaly of a symmetry with the correlation and entanglement properties of a state.

\begin{lemma}\label{lemma:generalized_KS}
    Let $\alpha:G\to\G^{\QCA}$ be a symmetry action. If there exists a $G$-symmetric state satisfying the clustering property\footnote{The clustering property means that all connected 2-point correlation functions of local operators decay to 0 as the distance between the two operators increases to infinity. No assumption on how fast the correlation functions decay is made here.} and entanglement area law, then $\alpha$ has a vanishing anomaly index $\omega=1$.
\end{lemma}

This lemma is a simplified version of Lemma \ref{lemma:typeI_KS} in Appendix \ref{subapp: consequence of anomaly}, a generalization of the main theorem of Ref.~\cite{kapustin2024anomalous} (see Remark 4.1 therein). Its proof is nontrivial and relies on advanced mathematical tools such as von Neumann algebras. The advantage of Lemma \ref{lemma:generalized_KS} is that its statement only involves more familiar notions like the clustering property and entanglement area law, while the version in Ref.~\cite{kapustin2024anomalous} and Lemma \ref{lemma:typeI_KS} involve some relatively less familiar notions in the theory of operator algebras.

Our Theorem \ref{thm:main_area_law} shows that the locally unique gapped ground states of an admissible Hamiltonian satisfies the entanglement area law. On the other hand, these states are known to satisfy the clustering property (see Theorem 2.8 and Assumption 2.2 of Ref.~\cite{Hastings2006decay}).
Combining these results with Lemma \ref{lemma:generalized_KS}, we deduce

\begin{theorem}\label{thm:main}
If $\alpha:G\to \G^{\QCA}$ is a symmetry action on an infinite quantum spin chain with an anomaly index $\omega\not=1$, then a $G$-symmetric admissible Hamiltonian $H$ cannot have a $G$-symmetric locally unique gapped ground state.
\end{theorem}

To prove this theorem, we employ an argument by contradiction (see also Appendix \ref{sec: theorem 1} for another proof). Suppose there is a locally unique gapped ground state of a $G$-symmetric admissible Hamiltonian with $G$ an anomalous symmetry. As discussed above, this state satisfies the clustering property and entanglement area law. Also, it is $G$-symmetric by assumption. However, Lemma~\ref{lemma:generalized_KS} shows that the anomaly index $\omega=1$, which contradicts our initial assumption.

Theorem~\ref{thm:main} concerns infinite systems, but real systems are of finite size. To extract useful implications on finite systems, we utilize another theorem.

\begin{theorem}[\textbf{Theorem \ref{thm:limit_ad_gs} in Appendix \ref{sec: thermodynamic limit}}]\label{theorem: finite and infinite}
Suppose a sequence of $G$-symmetric admissible Hamiltonians, $\{H_L\}$, converges to an admissible Hamiltonian $H$ as the system size $L\rightarrow\infty$. If each $H_{L}$ has a unique $G$-symmetric gapped ground state, then this sequence of ground states converges to a $G$-symmetric locally-unique gapped ground state of $H$ as $L\rightarrow\infty$. 
\end{theorem}

Combining Theorems~\ref{thm:main} and \ref{theorem: finite and infinite}, we deduce the following corollary.

\begin{corollary} \label{corollary: LSM finite}
    If a sequence of large but finite systems described by admissible Hamiltonians with an anomalous symmetry have a well-defined thermodynamic limit, then they cannot have a unique $G$-symmetric gapped ground state.
\end{corollary}

Theorem \ref{thm:main} and Corollary \ref{corollary: LSM finite} generalize the standard LSM theorem in multiple aspects. First, they apply to systems with long-range interactions. Second, these results apply to a wide range of symmetry settings, \ie general anomalous symmetries described by QCA. In fact, the proofs of these results in the appendices apply to an even more general class of symmetry settings, known as locality-preserving automorphisms, which include QCA as special examples.

\section{Symmetry-enforced long-range entanglement} \label{sec: symmetry-enforced long-range entanglement}

Next, we move to symmetry-enforced long-range entanglement, \ie any state invariant under an anomalous symmetry must be long-range entangled.

To sharpen this result, let us first make the notions of short-range entangled (SRE) and long-range entangled (LRE) states more precise. In an infinite-size system, a state is considered as an SRE state if it can be disentangled (\ie converted to a product state) by a finite-time evolution generated by an almost local Hamiltonian, otherwise it is an LRE state. Here an almost local Hamiltonian is a Hamiltonian that can contain non-local interactions, but as the spatial sizes of these interactions increase their magnitudes decay faster than any polynomial function. We will elaborate on the motivation to invoke almost local Hamiltonians (as opposed to local Hamiltonians) in defining SRE and LRE states at the end of this section. In any case, according to this definition an LRE state naturally has a higher quantum complexity than an SRE state.

Borrowing some results from Ref. \cite{kapustin2024anomalous}, in Appendix \ref{subsec: thermodynamic limit of SRE} we show

\begin{theorem}[Corollary \ref{coro:SRE} in Appendix \ref{subsec: thermodynamic limit of SRE}]\label{thm:infinite_SRE_main}
    Given an anomalous symmetry $\alpha:G\to\G^{\QCA}$, no short-range entangled state can be invariant under this symmetry.
\end{theorem}

This result is the infinite-size version of symmetry-enforced long-range entanglement. Although it is not explicitly stated in Ref. \cite{kapustin2024anomalous}, we believe it is known to the authors of this reference.

Below we present the finite-size version of symmetry-enforced long-range entanglement, which is much more nontrivial and stronger than Theorem \ref{thm:infinite_SRE_main}. In this context, we say that a sequence of finite but large systems is SRE if each of them can be disentangled by a time evolution generated by an almost local Hamiltonian over a duration that does not diverge as the system size goes to infinity, otherwise it is LRE. Then we have

\begin{theorem}[Theorem \ref{thm:finite_SRE} in Appendix \ref{subsec: thermodynamic limit of SRE}]\label{thm:finite_SRE_main}
    Let $\alpha_L$ be an operation on a system with size $L$, and suppose $\alpha_L$ converges to an anomalous symmetry $\alpha: G\to\G^{\QCA}$ as $L\rightarrow\infty$. Let $\{|\psi_L\ra\}$ be a sequence of SRE states defined on a system of size $L$. Then $|\psi_{L}\ra$ is symmetric under $\alpha_L$ for at most finitely many $L$.
\end{theorem}
To prove this theorem (see Appendix \ref{subsec: thermodynamic limit of SRE} for details), we first show that such a sequence of states must have a subsequence that has a well-defined thermodynamic limit, and the thermodynamic limit must satisfy the clustering property and entanglement area law. Hence a symmetric SRE can be ruled out by the powerful Lemma \ref{lemma:generalized_KS}.

Therefore, we see that an anomalous symmetry forces the system to be long-range entangled{\footnote{Our analysis actually implies a slightly stronger result, namely, not only SRE states are incompatible with an anomalous symmetry, all invertible states are also incompatible with such a symmetry. Here a sequence of states is invertible if it can become a sequence of SRE states after being stacked with another sequence of states.}}, and a natural platform to find LRE states is systems with anomalous symmetries.

Before ending this section, we elaborate on why our Theorem \ref{thm:finite_SRE_main} is a rather strong result. First, we remark that we do not need to assume that the sequence of SRE states $\{|\psi_L\ra\}$ converges to a well-defined thermodynamic limit to prove Theorem \ref{thm:finite_SRE_main}, because it always has a convergent subsequence. Second, for the convergent subsequence, we do not need to assume that its thermodynamic limit is an SRE state defined above for an infinite system. More precisely, although this sequence of states converges, we do not have to assume that their disentangling unitaries converge to an infinite-size version of finite-time evolution generated by an almost local Hamiltonian{\footnote{In fact, given a convergent sequence of finite-size SRE states, we are unable to prove or disprove the convergence of their disentangling unitaries.}}. Third, there are multiple commonly used notions of SRE and LRE states that differ by their requirements of the disentangling unitaries. Besides our definitions, one may define an SRE state to be one that can be disentangled by a finite-depth quantum circuit or by a finite-time evolution generated by a strictly local Hamiltonian. Since these two types of disentangling unitaries are special examples of the ones in our definition, we get the strongest form of symmetry-enforced long-range entanglement, which implies that an anomalous symmetry is incompatible with SRE states based on the two other definitions.

In addition to obtaining the strongest version of symmetry-enforced long-range entanglement, there are further motivations to invoke disentangling unitaries generated by almost local Hamiltonians in defining SRE states, compared to the two other definitions. First of all, although finite-depth quantum circuits are natural from the perspective of digital quantum simulation, states that can be disentangled by these circuits have strictly zero two-point correlation functions when the two operators are far enough, which are fine tuned from a condensed matter viewpoint. But why do we consider disentangling unitaries generated by almost local Hamiltonians, rather than strictly local Hamiltonians? The first reason is that effective Hamiltonians in nature often do have non-local tails and are not strictly local, and the second reason is that it is desirable that LRE states are in different quantum phases from product states, and LRE states defined using almost local Hamiltonians indeed satisfy this condition \cite{Hastings2005, Bachmann2012automorphic}. However, for states that cannot be disentangled by unitaries generated by local Hamiltonians, currently it is unknown whether they must be in a quantum phase different from a product state. Therefore, from the standpoint of quantum matter, it is more natural to define SRE and LRE states in terms of disentangling unitaries generated by almost local Hamiltonians.

\section{Examples and applications}\label{Sec:main_examples}

After presenting our main theorems, below we apply these results to various examples.

Our first example, which demonstrates the entanglement area law, is the following Ising-like model:
\beq\label{eq:LRTFI}
H=\sin\theta\sum_{i>j}Z_{i}Z_{j}\frac{1}{|i-j|^{\fa}}+\cos\theta\sum_{i}X_{i}
\eeq
where $\theta$ labels the coupling constant. This model is symmetric under the usual non-anomalous $\z_{2}$ symmetry that takes $Z_{i}\to-Z_{i},X_{i}\to X_{i}$.

This model was numerically studied in Ref.~\cite{Koffel_2012}, where it is found that when $0<\theta\ll1$, there is a unique gapped ground state, while the connected 2-point function $\la Z_{i} Z_{i+l}\ra_{c}$ decays as a power law of $l$. In the other extreme with $\theta\simeq\frac{\pi}{2}$, the system is in the Néel phase with 2 almost degenerate gapped ground states which spontaneously break the $\z_{2}$ symmetry. Ref. \cite{Koffel_2012} suggests violations of the entanglement area law for $0<\fa\leqslant 3$, although it is also commented that these results may not be conclusive due to the small system sizes in the numerical simulations. Our Proposition~\ref{prop:degenerate_area_law_main} (or more precisely, Proposition~\ref{prop:degenerate_area_law}) implies that any ground state in these two gapped phases satisfies the entanglement area law if $\fa>2$, so this model should be revisited numerically.

The second example is the spin-$1/2$ XXZ chain with long-range interactions~\cite{Geier_2021, Scholl_2022}. The Hamiltonian is 
\beq
    H = \sum_{i>j} \frac{1}{|i-j|^{\mathfrak{a}}}
        \left( J^{z}_{ij} \, S^{z}_{i} S^{z}_{j} - S^{x}_{i} S^{x}_{j} - S^{y}_{i} S^{y}_{j} 
        \right) 
    \; ,
    \label{eq:ham_2}
\eeq
where $J^z_{ij}\in[-J, J]$ can be positive or negative with $J>0$ a constant. 
Notice that $H$ satisfies the admissible condition in Eq.~\eqref{eq:admissible} when $\mathfrak{a} > 2$. This model has a $O(2)\times\z$ symmetry, where $\z$ is lattice translation and $O(2)$ includes any spin rotation around the $z$-axis and $\pi$-rotation around the $x$-axis. We have $\rH^{3}(O(2)\times\z;\U)\simeq \z_{2}$, which measures the on-site spin quantum number $S$. If $S\in\z+\frac{1}{2}$, the anomaly index is nontrivial.

From Ref.~\cite{Maghrebi_2017}, for $\mathfrak{a} > 2$ the phase diagram of Eq.~\eqref{eq:ham_2} contains a ferromagnetic phase, an anti-ferromagnetic phase and a continuous symmetry breaking phase, which all spontaneously break some symmetries, and an XY phase, which is symmetric but gapless. Indeed, none of these phases has a unique symmetric gapped ground state, agreeing with our Theorem \ref{thm:main} and Corollary \ref{corollary: LSM finite}. Moreover, the only symmetric phase (XY) is a gapless conformal field theory with long-range entanglement, agreeing with our Theorems \ref{thm:infinite_SRE_main} and \ref{thm:finite_SRE_main}.

In our last example, the only relevant symmetry is an anomalous $\z_{2}$ internal symmetry. Concretely, we place a qubit at each lattice site, and this $\z_2$ symmetry acts as \cite{Levin_2012}
\beq
\begin{split}
    \alpha(Z_{i})&=-Z_{i}
    \; , \\
    \alpha(X_{i})&=-Z_{i-1}X_{i}Z_{i+1} 
    \; ,
\end{split}
\eeq
where $X_{j},Z_{j}$ are usual Pauli matrices, and we also denote $Y_{j}=iX_{j}Z_{j}$.
Formally, this symmetry is generated by conjugation with the following infinite product
\beq
    \prod_{j\in\z}e^{\frac{i\pi}{4}Z_{j}Z_{j+1}}\prod_{k\in\z}X_{k}
    \; .
\eeq
This choice of symmetry action corresponds the nontrivial anomaly class in $\mathrm{H}^{3}(\z_{2};\U)$~\cite{kapustin2024anomalous}. We consider the following Hamiltonian, which is invariant under this $\z_2$ symmetry and may be realizable in experimental setups similar to those in Refs.~\cite{Geier_2021, Scholl_2022}:
\begin{widetext}
\beq
    H=-\sum_{i,j}J_{ij}Z_{i}Z_{j}-\sum_{i}g_{i}(X_{i}+Z_{i-1}X_{i}Z_{i+1})-\sum_{j}h_{j}Y_{j}(1-Z_{j}Z_{j+1})
    \; ,
\eeq
\end{widetext}
where $J_{ij}=O(|i-j|^{-\mathfrak{a}})$ when $|i-j|\to \infty$ with $\mathfrak{a}>2$. For example, $J_{ij}$ can be chosen as 
\beq\label{eq:coupling}
    J_{ij}=\frac{C_{ij}}{|i-j|^{\mathfrak{a}}}
    \; ,
\eeq
where $\mathfrak{a}>2$, and $C_{ij}$ depend on $i, j$ and are bounded by constant $D$ for all $i,j$. This Hamiltonian can break the translation symmetry explicitly.

In the special case where $h_{i}=0$, $g_{i}=0$  and $C_{ij}=1$, this system is the classical long-range Ising model. There are two gapped ground states (\ie the all-up state and the all-down state) and the anomalous $\z_{2}$ symmetry is spontaneously broken, agreeing with Theorem \ref{thm:main} and Corollary \ref{corollary: LSM finite}. Also, these ground states indeed satisfies the entanglement area law, consistent with Theorem \ref{thm:main_area_law} and Proposition \ref{prop:degenerate_area_law_main}. Taking the two ground states to be the linear superpositions of the all-up and all-down states that are symmetric under the anomalous $\z_2$ symmetry, we see that they do not satisfy the clustering property and are indeed long-range entangled, compatible with Theorem \ref{thm:finite_SRE_main}. In the regime where $g>0$ and $J_{ij}=h_i=0$, this model realizes a gapless symmetric long-range entangled Luttinger liquid~\cite{Levin_2012}, agreeing with Theorem \ref{thm:main}, Corollary \ref{corollary: LSM finite}, Theorem \ref{thm:infinite_SRE_main} and Theorem \ref{thm:finite_SRE_main}. For more general couplings such as Eq.~\eqref{eq:coupling}, other phases are also possible, which are left to future works.

\section{Discussions}\label{Sec:main_discussion}

In this work, we have proved three interrelated fundamental results in quantum many-body physics. Concretely, we first show that the entanglement area law holds for the ground states of a quantum spin chain described by certain gapped long-range interacting Hamiltonian, even if there is a ground state degeneracy associated with a spontaneously broken discrete symmetry. This result is then used to prove the Lieb-Schultz-Mattis theorem in long-range interacting systems with an anomalous symmetry. Moreover, we show that short-range entangled states are not compatible with an anomalous symmetry. 

We highlight the generality of our results:

\begin{itemize}

    \item The long-range interacting Hamiltonians under consideration can have multi-spin interactions, and they are only required to satisfy Eq. \eqref{eq:admissible}.

    \item The symmetries under consideration are general quantum cellular automata. They include internal symmetry and translation symmetry, and the internal symmetry can be discrete or continuous, on-site or non-on-site. In the appendices, all our main results are shown to hold for an even wider class of symmetry actions, known as locality preserving automorphism.

    \item We present results on both infinite systems and sequences of finite systems with increasing sizes. The former may be viewed as the thermodynamic limits of the latter, while all lattice systems in experimental and numerical studies belong to the latter.

    \item Short-range and long-range entanglement here are defined in terms of finite-time evolutions generated by almost local Hamiltonians, rather than finite-depth quantum circuits or finite-time evolutions generated by strictly local Hamiltonians. Therefore, our version of symmetry-enforced long-range entanglement is the most general one.

\end{itemize}

Our results are expected to be relevant to a wide class of physical systems, including, in particular, trapped ions and Rydberg atom arrays where long-range interactions are important. These results are already applied and demonstrated in various examples in Sec. \ref{Sec:main_examples}. For example, our rigorously proven entanglement area law suggests that the model in Ref. \cite{Koffel_2012} should be revisited.

Below we list some interesting and important directions for future studies.

\begin{itemize}

    \item Currently it is not rigorously proved whether the entanglement area law and Lieb-Schultz-Mattis theorem must fail in systems with Hamiltonians not obeying Eq. \eqref{eq:admissible}. It is useful to better understand this.

    \item In this paper we only consider quantum spin chains with unitary symmetries that preserve the spacetime orientation. It will be interesting and important to generalize these considerations to systems with time reversal, reflection and generalized symmetries, and to higher dimensional systems and fermionic systems. Studies in this direction, which focus on specific forms of symmetries or Hamiltonians and possess various levels of rigor, have been carried out \cite{Oshikawa1999, Hastings2003, Cheng2015, Po2017, Jian2017,  Cho2017,Watanabe2018LSM,Metlitski2018,Cheng2018a, Kobayashi2018, Ogata2019LSM,Else2020, Jiang2019,Yao2021twisted, Ogata_2021, Aksoy2021, Ye2021a,Ma2022a, Cheng2022, Kawabata2023, Aksoy2023, Seifnashri2023, Zhou2023, kapustin2024anomalous,Garre_Rubio2024anomalous, Ma2024, Pace2024, Kawagoe2025, Tu2025}, but a fully comprehensive and rigorous understanding is still lacking.
    
\end{itemize}

{\it Note added: } Two related generalizations of the Lieb-Schultz-Mattis theorem in long-range interacting systems are also discussed in Refs.~\cite{Ma2024,Zhou2024LSM}, where Ref. \cite{Ma2024} appeared on the same arXiv post as the first arXiv version of the present paper, after which Ref. \cite{Zhou2024LSM} appeared. Both of Refs.~\cite{Ma2024,Zhou2024LSM} focus on the specific case where the system has an on-site $SO(3)$ symmetry and lattice translation symmetry, and they study whether the Hamiltonian can have a unique gapped ground state. Also, in Ref.~\cite{Ma2024} results in one, two and three dimensions are reported. The proof there follows the original proof of the Lieb-Schultz-Mattis theorem \cite{Lieb1961}, \ie constructing low-energy excited states by twist operators. One byproduct of this method is a bound of the finite-size gap for such systems (see, \eg Eq. (5) of Ref.~\cite{Ma2024}).
In our work, we consider general one dimensional Hamiltonians with symmetries described by general quantum cellular automata (which are extended to more general locality preserving automorphisms in the appendices), and our proof is based on the entanglement point of view of the Lieb-Schulz-Mattis theorem and the entanglement area law. However, since we work with infinite system from the beginning, we cannot estimate the finite-size gap. Besides the Lieb-Schultz-Mattis theorem in long-range interacting systems, we also establish the entanglement area law and symmetry-enforced long-range entanglement, and these contents are absent in Refs.~\cite{Ma2024,Zhou2024LSM}.

\begin{acknowledgements}

We thank Theo Johnson-Freyd, Tomotaka Kuwahara, Michael Levin and Zhi Li for helpful discussions. LZ thanks the Kavli Institute
for Theoretical Physics for their hospitality during the
course of this research. Research at Perimeter Institute is supported in part by the Government of Canada
through the Department of Innovation, Science and Industry Canada and by the Province of Ontario through
the Ministry of Colleges and Universities. This research was supported in part by grant NSF PHY-2309135 to the Kavli Institute for Theoretical Physics (KITP). LZ is supported in part by the National University of Singapore start-up
grants A-0009991-00-00 and A-0009991-01-00. RL is also supported by the Simons Collaboration on Global Categorical Symmetries through Simons Foundation grant 888996.
    
\end{acknowledgements}

\newpage
\onecolumngrid
\appendix

\section*{Appendices}

In these appendices, we provide more details relevant to the main text. Specifically, we review the operator algebra formalism in Appendix \ref{sec: formalism review}, which aims to provide a basic introduction to this formalism to readers unfamiliar with it. This formalism is the theoretical tool we utilize to study lattice systems with an infinite size. Although most of this appendix is devoted to reviewing the known results, this appendix also contains an important new result and its proof, \ie Proposition \ref{prop:factor_entropy}, which is crucial for proving Lemma \ref{lemma:typeI_KS} and Lemma \ref{lemma:generalized_KS}. In Appendix \ref{sec:area_law}, we present the detailed proof of the entanglement area law in infinite systems (Theorem \ref{thm:main_area_law} in the main text). Theorem \ref{thm:cut_off} is an ingredient in this proof, and its own proof is quite technical, which is presented in Appendix \ref{sec:proof_cut_off}. In Appendix \ref{sec:group_cohomology}, we review the mathematical definitions of group cohomology and differentiable group cohomology, which are necessary to define the anomaly index in Appendix \ref{sec:anomaly_index}. In Appendix \ref{sec:anomaly_index}, we review the construction of the anomaly index. In particular, we prove that this anomaly index is independent of the choice of the cut, which was not explicitly proved in Ref. \cite{kapustin2024anomalous}. We also derive some consequences of an anomalous symmetry. In Appendix \ref{sec: theorem 1}, we provide the proof of Theorem \ref{thm:main} in the main text. In Appendix \ref{sec: thermodynamic limit}, we discuss the connections between infinite systems and sequences of finite systems, and also present the proofs of Proposition \ref{prop:degenerate_area_law_main}, Theorem \ref{theorem: finite and infinite}, Theorem \ref{thm:infinite_SRE_main} and Theorem \ref{thm:finite_SRE_main} in the main text. Because we will employ a large number of notations in these appendices, for the convenience of the readers, we list the frequently used symbols at the end of this document.

\section{Review of the operator algebra formalism} \label{sec: formalism review}

As explained in the main text, many important concepts in quantum many-body physics, such as locality, long-range entanglement and quantum phases of matter, are most cleanly defined in infinite-size systems. Therefore, in our analysis we start with infinite systems, and then deduce implications on sequences of finite systems with increasing sizes. 

One of the major differences between finite and infinite lattice systems is that the familiar notion of the Hilbert space of the entire system is not applicable to the latter (see Appendix \ref{subsec:quasi_local} below). The basic tool to study infinite lattice systems is the operator algebra formalism, which is a generalization of the Heisenberg picture of the usual quantum mechanics. In this Appendix, we review the operator algebra formalism. Our goal here is to introduce the essential aspects of this formalism to readers unfamiliar with it. We will not only give definitions and prove theorems, but also provide illuminating examples of various concepts. We also remark that this appendix contains a new and useful result and its proof, \ie Proposition \ref{prop:factor_entropy}.

\subsection{Algebras of local and quasi-local operators}\label{subsec:quasi_local}

In this appendix we introduce some general background of the algebras of local operators and quasi-local operators. Readers are referred to Refs. \cite{bratteli2013operator1,bratteli2013operator2,Naaijkens_2017,Landsman:2017hpa} for a more thorough treatment.

Throughout this appendix, we work on lattices of general spatial dimension $d$, \ie our lattice is $\Lambda\simeq \z^{d}$ unless otherwise specified. On the other hand, most of the rest of the appendices specifically concern quantum spin chains.

We assume that our on-site Hilbert space $\cH_{k}$ satisfies $\dim\cH_{k}\leqslant d$ for some $d>0$, where $k$ labels a site in $\Lambda$ (we do {\it not} assume that $\dim\cH_{k}$ is the same for different $k$'s). However, as we explain now, the  ``total'' Hilbert space $\cH:=\bigotimes_{k\in\Lambda}\cH_{k}$ is not well defined for infinite lattices. Especially, this naive infinite tensor space lacks a well-defined inner product, which is essential to the usual quantum mechanics. In more detail, suppose we have a ``quantum state'' in this total Hilbert space $\cH$:
\beq
|\psi\ra=\bigotimes_{k\in\Lambda}|\psi_{k}\ra
\eeq
where $|\psi_{k}\ra\in\cH_{k}$ is a normalized vector for each $k\in\Lambda$. Given an arbitrary sequence $\{a_{k}\}_{k\in \Lambda},a_{k}\in \R$, one can construct another ``quantum state''
\beq
|\psi'\ra:=\bigotimes e^{ia_{k}}|\psi_{k}\ra
\eeq
The inner product between $|\psi\ra$ and $|\psi'\ra$ is
\beq
\la\psi|\psi'\ra=\exp(i\sum_{k\in\Lambda}a_{k})
\eeq
Since the sequence $\{a_{k}\}_{k\in\Lambda}$ is arbitrary and there is no obvious regularization scheme, this inner product has no definite answer. Thus $\cH$ is not a well-defined object for quantum mechanics on infinite lattices. So one must be careful about the meaning of a quantum state on infinite lattices.

Nevertheless, the Heisenberg picture of quantum mechanics, which focuses on operators rather than states, is still applicable even in the presence of infinitely many degrees of freedom, such as spin systems on infinite lattices. To understand this formalism, let us start with the notion of {\it{local operators}}. Given a {\it finite subset} $\Gamma\subset\Lambda$, one can talk about the Hilbert space $\cH_{\Gamma}:=\bigotimes_{k\in \Gamma}\cH_{k}$ on $\Gamma$. The operators supported on $\Gamma$ are defined to be all operators on this finite dimensional Hilbert space $\cH_{\Gamma}$, including $c$-numbers (\ie scalar multiples of the identity operator). Note that any (finite) additions and multiplications of operators on $\cH_{\Gamma}$ give another operator on $\cH_{\Gamma}$, and hence these operators form an algebra, denoted by $\A_{\Gamma}$. We call $\A_{\Gamma}$ the algebra of local operators support on $\Gamma$. It is obvious that if $A\in\A_{\Gamma}$ then its Hermitian conjugate $A^{\dagger}\in\A_{\Gamma}$ as well. 

Now we introduce a norm on this algebra.
First, for a state $|\phi\ra\in\cH_\Gamma$, we define its norm as 
\beq\label{eq:vec_norm}
||\phi\ra|=\sqrt{\la\phi|\phi\ra}\geqslant 0
\eeq
This norm satisfies the usual triangle inequality
\beq\label{eq:tri_ineq_vec}
||\phi_{1}\ra+|\phi_{2}\ra|\leqslant ||\phi_{1}\ra|+||\phi_{2}\ra|
\eeq
Next, for a local operator $A\in\A^{\ell}_{\Gamma}$, we define its operator norm by 
\beq\label{eq:operator_norm}
||A||:=\sup_{\substack{|\psi\ra\in\cH_{\Gamma}\\ \la\psi|\psi\ra=1}}|A|\psi\ra|
\eeq
Namely, the norm of the operator $A$ is the square root of the largest eigenvalue of $A^\dag A$ in the case of a finite system.

By the definition in Eq. \eqref{eq:operator_norm}, for any vector $|\phi\ra$ we have
\beq\label{eq:norm_def_ineq}
|A|\phi\ra|^{2}\leqslant ||A||^{2}||\phi\ra|^{2}
\eeq
Another useful property of the operator norms is the triangle inequality,
\beq\label{eq:tri_ineq_op}
||A_{1}+A_{2}||\leqslant ||A_{1}||+||A_{2}||,\quad\forall\,A_{1},A_{2}\in\A_{\Gamma}
\eeq
for some finite subset $\Gamma$. To wit, note that for any $|\psi\ra\in\cH_{\Gamma}$, by Eq.~\eqref{eq:tri_ineq_vec},
\beq
| A_{1}|\psi\ra+A_{2}|\psi\ra|\leqslant |A_{1}|\psi\ra|+|A_{2}|\psi\ra|
\eeq
The desired inequality Eq.~\eqref{eq:tri_ineq_op} follows by taking supremes on both sides with respect to $\la\psi|\psi\ra=1$.

Given any local operator $A\in \A_{\Gamma}$, if $\Gamma'$ is another finite subset containing $\Gamma$ (\ie $\Gamma\subset \Gamma'$), then there is a natural way to extend $A$ to a local operator supported on $\Gamma'$,
\beq
A'=A\bigotimes_{k\in\Gamma'\setminus \Gamma}I_{k}
\eeq
where $A'$ is the extension of $A$ on $\Gamma'$ and $I_{k}$ is the identity operator on $\cH_{k}$. In this case, we say that $A$ acts as identity outside $\Gamma$.
It is often convenient to identify these two operators, \ie we will not distinguish $A'$ and $A$ in the following. One can easily check that $||A'||=||A||$, so this identification is unambiguous in terms of the operator norm.
We then make the following definition:

\begin{definition}\label{def:local_operators}
The algebra of local operators
    \beq\label{eq:local_operator_algebra}
\A^{l}:=\bigcup_{\Gamma\subset\Lambda,|\Gamma|<\infty}\A_{\Gamma}
\eeq
with the above identification $A'\sim A$. We also define $\A_{\emptyset}=0$.
\end{definition}

Here $|\Gamma|$ means the cardinality of $\Gamma$ and we write $|\Gamma|<\infty$ if $\Gamma$ is a finite set. More explicitly, $A\in \A^{l}$ if there is a finite subset $\Gamma$ such that $A\in \A_{\Gamma}$ (this $\Gamma$ is not unique due to the freedom to extend $A$). If $A_{i}\in\A^{l},i=1,2$, then one can find a finite subset $\Gamma$ such that $A_{i}\in\A_{\Gamma}$. Because $\lambda_{1}A_{1}+\lambda_{2}A_{2}\in\A_{\Gamma},\,\lambda_{i}\in\mathbb{C},i=1,2$ and $A_{1}A_{2}\in\A_{\Gamma}$, $\A^{l}$ is indeed an algebra. This algebra contains a special element, called the unit and denoted by $I$, which satisfies $IA=AI=A$ for any $A\in\A^{l}$. In $\A^{l}$, this unit is given by the identity operator which acts trivially at all sites.

One crucial property of $\A^{l}$ is the locality. Given $A_{i}\in \A_{\Gamma_{i}},i=1,2$, if $\Gamma_{1}\cap \Gamma_{2}=\emptyset$, then
\beq
[A_{1},A_{2}]=0
\eeq
We remark that $A_{1},A_{2}$ in above equation really mean their extensions on some $\Gamma$ which contains $\Gamma_{1}\cup\Gamma_{2}$.

However, it is often insufficient to only work with $\A^{l}$. Consider a time evolution that is generated by a local Hamiltonian
\beq\label{eq:charge_densities}
H=\sum_{|\Gamma|<\infty}h_{\Gamma}
\eeq
where $h_{\Gamma}$ is some local term supported on $\Gamma$ in the sense that $h_{\Gamma}=0$ if $\diam(\Gamma)>R$ for some fixed $R>0$, with $\diam(\Gamma):=\max_{x,y\in\Gamma}d(x,y)$ where $d$ is the distance. Given such an $H$, the evolution is denoted by $\alpha^{t}$ for $t\in\R$, which transforms a local operator $A\in \A_{\Gamma}$ into
\beq\label{eq:symLPA}
\alpha^{t}(A)=e^{it H}A e^{-it H}
\eeq

As noted above, for a given $A\in\A_\Gamma$, if $B\in \A_{\Gamma'}$ is another local operator such that $\Gamma\cap \Gamma'=\emptyset$, then $[A,B]=0$ by locality. However, $[\alpha^{t}(A),B]\not=0$ in general and this commutator is controlled by a Lieb-Robinson bound\footnote{There are many different versions of Lieb-Robinson bounds.} \cite{Lieb:1972wy,Matsuta2017LR,Else2018LR,_Anthony_Chen_2023}
\beq\label{eq:LR_bound}
||[\alpha^{t}(A),B]||\leqslant |\Gamma| C e^{-a(L-vt)}
\eeq
where $C,a,v$ are non-universal positive constants\footnote{It will be important later that all of these constants are independent of $|\Gamma'|$.} and $L:=d \, (\Gamma,\Gamma')$  is the distance between $\Gamma$ and $\Gamma'$. The constant $v$ is usually called the Lieb-Robinson velocity. Actually, one version of the Lieb-Robinson bounds implies that $\alpha^{t}$ maps a quasi-local operator to be defined below to a quasi-local operator (see Theorem 3.2 and Lemma 3.3 of Ref. \cite{Ranard_2022}).

Therefore, to ensure that time evolutions can properly act on our operator algebra, one has to consider the {\it{quasi-local}} operator algebra $\A^{ql}$ (see below for its definition), rather than $\A^{l}$ only. Intuitively, $\A^{ql}$ is obtained by taking sequential limits in $\A^{l}$. Concretely, an operator $A\in \A^{ql}$ if and only if there is a Cauchy sequence $A_{j}\in\A^{l}$ such that 
\beq\label{eq:limits}
A=\lim_{j\to \infty}A_{j}
\eeq
By a Cauchy sequence, we mean that $\forall\,\epsilon>0$, there exists $N\in\z^{>0}$ such that $||A_{j}-A_{j'}||<\epsilon$ for all $j,j'>N$, where $||\cdot||$ is the operator norm defined in Eq. \eqref{eq:operator_norm}. More precisely, one says that $\A^{ql}$ is the completion of $\A^{l}$ with respect to the operator norm $||\cdot||$. An analogue for the relation between $\A^{ql}$ and $\A^{l}$ is the relation between $\R$ and $\mathbb{Q}$. In the latter case, $\R$ can be defined to be the limits of Cauchy sequences in $\mathbb{Q}$ (or the completion of $\mathbb{Q}$ with respect to usual absolute value of $\mathbb{Q}$) \cite{Tao2022}. In essence, if $A\in\A^{ql}$, $A$ may not act as identity outside any finite region. However, it can be approximated by local operators with any desired accuracy. Besides, $\A^{ql}$ has a natural norm inheriting from the operator norm in Eq. \eqref{eq:operator_norm} on $\A^{l}$. Explicitly, let $A\in\A^{ql}$ and $A_{j}\in\A^{l}$ be a sequence convergent to $A$, we define
\beq\label{eq:limit_norm}
||A||:=\lim_{j\to\infty}||A_{j}||
\eeq
It is easy to check that $||A||$ does not depend on the choice of the sequence $\{A_{j}\}_{j=1,2,...}$. By this definition, the unit in $\A^{ql}$, \ie the identity operator $I$, has norm 1. 

To summarize,

\begin{definition}\label{def:quasi-local}
    The quasi-local operator algebra $\A^{ql}$ is defined to be the completion of $\A^{l}$ with respect to the operator norm in Eq. \eqref{eq:operator_norm}. We also denote the group of quasi-local unitary operators as $\cU^{ql}$.
\end{definition}

Given a subset $\Gamma$ (which may be either finite or infinite), we now use $\A_\Gamma$ to represent the algebra of {\it quasi-local} operators supported on $\Gamma$, \ie it acts as identity outside of $\Gamma$.

Our $\A^{ql}$ is a special example of the so-called $C^{*}$-algebra in the mathematical literature \cite{davidson1996c,arveson1998invitation,murphy2014c}. 
\begin{definition}\label{def:C*_alg}
    A $C^{*}$-algebra $\cC$ is an algebra equipped with an anti-linear involution $*$ (which models Hermitian conjugation in quantum mechanics) and a norm $||\cdot||$ (which may or may not be the operator norm defined in Eq. \eqref{eq:operator_norm}), such that
\begin{enumerate}
    \item $(A^{*})^*=A$ for all $A\in \cC$.
    \item $(AB)^{*}=B^{*}A^{*}$ for all $A,B\in \cC$.
    \item $(\lambda A+B)^{*}=\bar{\lambda}A^{*}+B^{*}$ for all $A,B\in \cC$, $\lambda\in \bbC$ and $\bar{\lambda}$ is the complex conjugate of $\lambda$.
    \item (Banach property) $||AB||\leqslant ||A||\cdot||B||$ and $||A^{*}||=||A||$ for all $A,B\in\cC$.
    \item ($C^{*}$ property) $||A^*A||=||A||^{2}$.
\end{enumerate}
Moreover, $\cC$ must be complete\footnote{The completeness of $\cC$ means that every Cauchy sequence of $\cC$ (see the paragraph under Eq.~\eqref{eq:limits}) has a limit in $\cC$. For example, $\mathbb{Q}$ is incomplete with respect to usual absolute value while $\R$ is complete.} with respect to $||\cdot||$.
\end{definition}

\begin{example}\label{example:bounded_operators}
    If $\cH$ is a Hilbert space (not necessarily finite dimensional), then its algebra of bounded operators (\ie operators with finite norms) $\B(\cH)$ is a $C^*$-algebra with $*=\dagger$. The first 3 properties are obvious. The Banach property is true according to Eq.~\eqref{eq:norm_def_ineq}
    \beq
    |AB|\psi\ra|\leqslant ||A||\cdot|B|\psi\ra|\leqslant||A||\cdot||B||\cdot ||\psi\ra|
    \eeq
    Hence $||AB||\leqslant ||A||\cdot||B||$. For the $C^*$-property, again note that for any normalized $|\psi\ra$
    \beq
    |A|\psi\ra|^{2}=\la\psi|A^{\dagger}A|\psi\ra\leqslant ||A^{\dagger}A||
    \eeq
    Thus $||A||^{2}\leqslant ||A^{\dagger}A||\leqslant ||A^{\dagger}||\cdot ||A||$ by the Banach property. So we have $||A||\leqslant ||A^{\dagger}||$. Note that $(A^{\dagger})^{\dagger}=A$. Hence $||A||=||A^{\dagger}||$. The $C^*$-property then follows. 
     
\end{example}

\begin{example}
    Our $\A^{ql}$ with $*=\dagger$ and the norm Eq. \eqref{eq:limit_norm} is a $C^{*}$-algebra (see Sec. 3.2.3 of Ref. \cite{Naaijkens_2017}). As a consequence, for any $U\in\cU^{ql}$, we have $||U||=1$ since $||U||^{2}=||U^{\dagger}U||=1$. However, the algebra of local operators $\A^{l}$ is not $C^{*}$ because it is incomplete with respect to $||\cdot||$ (\ie some Cauchy sequences do not have a limit within $\A^{l}$).
\end{example}

In fact, finite dimensional $C^*$-algebras can be completely classified.
\begin{proposition}[Theorem C.163 of Ref.~\cite{Landsman:2017hpa}]\label{prop:fd_C*}
    Let $\A$ be a finite-dimensional $C^{*}$-algebra, then $\A$ must be of the following form
    \beq
    \A=\bigoplus_{i=1}^{n}M_{d_{i}}(\bbC)
    \eeq
    where $d_{i}$ are positive integers and $M_{d_{i}}(\bbC)$ is the complex matrix algebra of size $d_{i}\times d_{i}$.
\end{proposition}

In this paper, we will focus on three examples of $C^*$-algebra, \ie the algebra of quasi-local operators $\A^{ql}$, the algebra of bounded operators $\B(\cH)$ in a Hilbert space $\cH$ and von Neumann algebras that will be introduced in Appendix \ref{sec:vN}, but the discussions presented in this section apply to general $C^{*}$-algebras.

\subsection{Quantum cellular automata and locality-preserving automorphisms}\label{subsec:QCA_LPA}

After introducing the basic notion of operator algebras, our next goal is to define a proper notion of symmetry action on local operators. This symmetry action is often required to preserve certain notion of locality. Given a symmetry group $G$, it is natural to define the symmetry action as a homomorphism $G\to \Aut(\A^{l})$, where $\Aut(\A^{l})$ is the automorphism group of $\A^{l}$. Recall that 
\begin{definition}
    We say $\alpha:\A^{l}\to\A^{l}$ is an automorphism of $\A^{l}$ if
\begin{enumerate}
    \item $\alpha(A+B)=\alpha(A)+\alpha(B)$
    \item $\alpha(AB)=\alpha(A)\alpha(B)$
    \item $\alpha(A^{\dagger})=\alpha(A)^{\dagger}$
    \item $\alpha(\lambda A)=\lambda\alpha(A),\forall\,\lambda\in\bbC$
    \item $\alpha$ is invertible.
\end{enumerate}
All automorphisms of $\A^{l}$ form a group under \textbf{finite} compositions, denoted by $\Aut(\A^{l})$
\end{definition}
There are similar automorphism groups for other algebras, \eg $\A^{ql}$ and $\B(\cH)$.

There is a special subgroup of $\Aut(\A^{l})$ called quantum cellular automata \cite{arrighi2019overview,Farrelly_2020}.

\begin{definition}\label{def:QCA}
    A quantum cellular automaton (QCA) is an automorphism $\alpha$ of $\A^{l}$ such that for a local operator $A\in\A_{X}$ (where $X$ is a finite subset), $\alpha(A)\in\A_{B(X,r_{\alpha})}$ where $r_{\alpha}>0$ does not depend on $A$. Here $B(\Gamma,r_{\alpha}):=\{x\in\Lambda|d(x,\Gamma)\leqslant r_{\alpha}\}$, where $d$ is the distance on the lattice and $0\leqslant r_{\alpha}<\infty$ does not depend on $A$ (but can depend on $\alpha$).
\end{definition}

According to the definition, QCA preserve locality in the strongest sense and QCA form a group under finite compositions, which is denoted by $\G^{\QCA}$.
\nomenclature{$\G^{QCA}$}{The group of quantum cellular automata, def. \ref{def:QCA}}

\begin{example}\label{example:circuit}

    One particular example of QCA is a (finite-depth unitary) circuit. For simplicity, we describe it in 1D. Let $\{P_{k}\}_{k\in\z}$ be a set of disjoint intervals with $|P_{k}|<l$ for a constant length $l$ ($P_{k}$ can be empty for some $k$'s). Then, we define a block-partitioned unitary (BPU) as 
    \beq
    \alpha=\prod_{k=-\infty}^{\infty}\Ad_{U_{k}}
    \eeq
    where $U_{k}$ is a unitary operator supported on $P_{k}$ and $\Ad_{U_{k}}(A)=U_{k}AU_{k}^{\dagger}$ for local operator $A$. To see that a BPU is a QCA, let us assume $\mathrm{supp}(A)\subset P_{k}$ for some $k$. Thus $\alpha(A)\in P_{k}$ again has finite support. It is easy to generalize this argument to other local operators. Hence it is a QCA by definition.

    A circuit is a finite composition of BPU's (these BPU's may be defined for different partitions). The group of all circuits are denoted by $\G^{cir}$. Later we will see that not every QCA is a circuit.
    
\end{example}
\nomenclature{$\G^{cir}$}{The group of circuits under finite composition, example \ref{example:circuit}}

As explained before in the last subsection, one needs to consider not only QCA's, but the automorphism group $\Aut(\A^{ql})$ of the quasi-local operator algebra $\A^{ql}$ is also of fundamental importance.

\begin{definition}\label{def:LPA}
    An automorphism $\alpha$ of $\A^{ql}$ is called a locality-preserving automorphism (LPA) if for each local operator $A\in\A_{X}$ with $X$ a finite subset of $\Lambda$ and any $r>0$, there exists a local operator $B$ (which depends on $A$ and $\alpha$), such that
    \beq
    ||\alpha(A)-B||<f_{\alpha}(r)||A||,
    \eeq
    where $f_{\alpha}(r)$ is a positive decreasing function independent of the choice of $A$ and $\lim_{r\to \infty}f_{\alpha}(r)=0$.

    The group of LPA's under finite composition is denoted by $\G^{lp}$.
\end{definition}
\nomenclature{$\G^{lp}$}{Group of locality-preserving automorphism}
More explicitly, if $\alpha\in \G^{lp}$ and $A\in\A_{\Gamma}$ is a local operator with $|\Gamma|<\infty$, then
\beq\label{eq:LPA}
\frac{||\alpha(A)-B||}{||A||}<f_{\alpha}(r)
\eeq
In this situation, we say that $\alpha$ has an $f_{\alpha}$-tail\footnote{The choice of $f_{\alpha}$ is not unique. Any other non-negative decreasing function $h(r)\geqslant f(r)$ with $\lim_{r\to\infty}h(r)=0$ also does the job. So for such $h_{\alpha}$, one can say $\alpha$ has $h_{\alpha}$-tail as well.}.
Intuitively, $\alpha(A)$ is {\it not} a local operator any more, but it can be approximated by another local operator $B$ defined on a larger support $B(\Gamma,r)$ with an error controlled by $f_{\alpha}(r)$.

We introduce another useful notation to deal with LPA's.
\begin{definition}\label{def:near_inclusion}
    Given an operator $A\in\A^{ql}$ and a subalgebra $\B$ of $\A^{ql}$, fixing a constant $\epsilon>0$, we write $A\stackrel{\epsilon}{\in}\B$ if there is $B\in\B$ such that $||A-B||\leqslant\epsilon||A||$. Similarly, for two subalgebras of $\A^{ql}$, we write $\A\stackrel{\epsilon}{\subset}\B$ if for $\forall\,A\in\A$ we have $A\stackrel{\epsilon}{\in}\B$.
\end{definition}
\nomenclature{$A\stackrel{\epsilon}{\subset}B$}{$A$ is nearly included in $B$, def. \ref{def:near_inclusion}}
Thus, the condition Eq. \eqref{eq:LPA} can be written as $\alpha(\A^{\ell}_{\Gamma})\stackrel{f(r)}{\in}\A^{\ell}_{B(\Gamma,r)}$.
Theorem 3.2 and Lemma 3.3 of Ref.~\cite{Ranard_2022} give a simple criterion for an automorphism to be an LPA in 1D. In particular, a finite-time evolution generated by a local (time dependent) Hamiltonian is an LPA.

In order to study symmetry actions on lattice systems in more detail, one first needs to understand the structure of $\G^{lp}$, which will be used to describe symmetry actions. In fact, the structures of $\G^{\QCA}$ and $\G^{lp}$ are rather clear in 1D,\footnote{Some classifications of higher dimensional QCA's are also proposed recently, see \eg Ref. \cite{Freedman_2020}.} thanks to the Gross-Nesme-Vogts-Werner (GNVW) index \cite{Gross_2012,Ranard_2022}. Below we first give a brief introduction of the GNVW index of QCA's, and then review the GNVW index of LPA's. We do not provide a detailed construction of the GNVW index here since it is rather technical and we will not use it in any essential way. Interested readers are referred to Refs. \cite{Gross_2012,arrighi2019overview,Farrelly_2020,Ranard_2022}.

We assume that all on-site local Hilbert spaces $V$ have the same dimension $D$. This does not lose any generality, since it can always be achieved by tensoring the original degrees of freedom with some other degrees of freedom at each site, such that our QCA acts trivially on the additional degrees of freedom.

Roughly speaking, the GNVW index is a group homomorphism, denoted by
\beq\label{eq:GNVW_index}
\ind:\G^{\QCA}\to \z[\{\log(p_{j}) \}_{j\in J}]
\eeq
\nomenclature{$\ind$}{The GNVW index map, \eqref{eq:GNVW_index}}
where $\{p_{j}\}_{j\in J}$ is the set of all prime divisors of $D$.

The index map satisfies the following properties:
\begin{enumerate}
    \item $\ind(\alpha\beta)=\ind(\alpha)+\ind(\beta),\forall\,\alpha,\beta\in\G^{\QCA}$,
    \item $\ker(\ind)=\G^{cir}$,
    \item $\ind$ is a surjection.
\end{enumerate}
It can be verified that a circuit is a QCA with a vanishing GNVW index. On the other hand, one can verify that if $\tau$ is a shift on the lattice by $+1$
\beq\label{eq:Generalized_translations}
\ind(\tau)=\log D\in \z[\{\log p_{j}\}_{j\in J}]
\eeq
\nomenclature{$\G^{T}$}{Group of generalized translations, \eqref{eq:Generalized_translations}}

To show that $\ind$ is a surjection, note that $\z[\{\log p_{j}\}_{j\in J}]$ is generated by $\log p_i$, so it suffices to identify an element in $\G^{\QCA}$ whose index is $\log p_i$ for each $p_i$. To this end, one defines a generalized translation (or partial translation) which only shifts part of the degrees of freedom at each site and which has an index $\log p_i$. To be more precise, one can fix a $p_{i}$ (where $p_{i}$ is a prime divisor of $D$) dimensional subspace $V_{i}$ of $V$. One can then factorize $V$ into $V\simeq V'\otimes V_{i}$ at each site, and a generalized translation $\tau_{i}$ only shifts the $V_{i}$-part while fixing $V'$. The GNVW index of $\tau_i$ is then verified to be
\beq
\ind(\tau_{i})=\log p_{i}
\eeq

Generalized translations form an Abelian group under finite compositions, and we denote it by $\G^{T}$. Thus, the GNVW index actually shows that $\G^{\QCA}$ is a semi-direct product
\beq\label{eq:QCA_semi_direct}
\G^{\QCA}=\G^{cir}\rtimes \G^{T}
\eeq

Now we turn to the structure theory of LPA's. Recall that $\A^{ql}$ is obtained by taking limits of $\A^{l}$, or more formally, any quasi-local operator can be approximated by local operators. Any LPA can also be approximated by QCA's. The basic strategy to study LPA's is to approximate it by a sequence of QCA's. For example, if $\alpha\in\G^{lp}$, there exists a sequence $\{\beta_{j}\}_{j=1,2...}$ of QCA's such that
\beq
\lim_{j\to\infty}\beta_{j}=\alpha
\eeq
Then one defines the GNVW index of $\alpha$ as
\beq
\ind(\alpha)=\lim_{j\to\infty}\ind(\beta_{j})
\eeq
It can be checked that this index is well-defined (finite and independent of the choice of the sequence) \cite{Ranard_2022}. Given the existence of this index map, many results of QCA's can be carried over to LPA's. For example, 
\beq
\G^{lp}=\G^{loc}\rtimes\G^{T}
\eeq
where $\G^{loc}$ is the subgroup of time evolution generated by local Hamiltonians ({\it which can be time-dependent}) over a finite duration.

\subsection{States in the operator algebra formalism}\label{Sec:states}

In the above sections, we have introduced the algebra of operators, and the concepts of QCA and LPA that can be used to describe symmetry actions on local operators. In this section, we discuss the notion of states in the operator algebra formalism, in the context of infinite systems.

In quantum mechanics, a state $|\psi\ra$ is a vector in some Hilbert space $\cH$. Equivalently, this state can be represented by a density matrix $\rho_{\psi}=|\psi\ra\la\psi|$. Another slightly unusual point of view is to think of $\rho_{\psi}$ as a linear functional $\psi$ on the algebra of operators,
\beq
\psi(A):=\tr(\rho_{\psi}A)
\eeq
for any operator $A$ in $\cH$. Once $\psi$ is determined, the density matrix $\rho_{\psi}$ (hence the vector state $|\psi\ra$) is also determined\footnote{To see this, one chooses $A$ to be projectors to basis vectors, which is enough to reconstruct $|\psi\ra$ up to an overall phase.}. The linear functional corresponding to a quantum state must satisfy further properties, such as positivity
\beq
\psi(A^{\dagger}A)=\la\psi|A^{\dagger}A|\psi\ra\geqslant 0
\eeq
for any operator $A$ in $\cH$. Besides, if $|\psi\ra\not= 0$, one requires it be normalized, that is 
\beq
\psi(I)=\la\psi|\psi\ra=1
\eeq

This motivates the following definition of states in infinite systems.
\begin{definition}\label{def:State}
    A quantum state of the quasi-local operator algebra $\A^{ql}$ is a nonzero linear functional $\psi:\A^{ql}\to \bbC$ with the following properties:
    \begin{enumerate}
        \item Positivity: $\psi(A^{\dagger}A)\geqslant 0$ for all $A\in\A^{ql}$.
        \item Normalization\footnote{For non-unital $C^{*}$-algebra, one can embed it into another unital $C^{*}$-algebra, see Sec. 2.2 of Ref. \cite{Naaijkens_2017}. Hence this definition still works.}: $\psi(I)=1$.
    \end{enumerate}
\end{definition}
\nomenclature{$\psi$}{A state of operator algebra, def. \ref{def:State}}
\begin{example}\label{example:denstiy_matrix}
    Consider the algebra $\B(\cH)$ (see example \ref{example:bounded_operators} for its definition) with $\dim\cH<\infty$. Let $\psi$ be a state of $\B(\cH)$. Pick up an orthonormal basis $|e_{i}\ra$ of $\cH$, we define 
    \beq
    \begin{split}
        P_{ij}&:=|e_{i}\ra\la e_{j}|\in\B(\cH)\\
        \lambda_{ij}&:=\psi(P_{ij})
    \end{split}
    \eeq
    Note that $P_{ij}^{\dagger}=P_{ji}$ and $P_{ij}P_{kl}=\delta_{jk}P_{il}$, where $\delta_{jk}$ is usual Kronecker delta. One can define a density matrix
    \beq
    \rho_{\psi}=\sum_{i,j}\lambda_{ij}P_{ji}.
    \eeq
    For any $A\in\B(\cH)$, we have
    \beq\label{eq:state_density_matrix}
    \psi(A)=\tr(\rho_{\psi}A).
    \eeq
    Hence the abstract state $\psi$ is represented by the  $\rho_{\psi}$ and one can check that $\rho_{\psi}$ is indeed a density matrix. 
\end{example} 

The above example shows how to connect a quantum state in the language of operator algebra with a quantum state in terms of a usual density matrix.
Below we discuss some basic properties of states, such as the norm of states, Cauchy-Schwarz inequality and orthogonality of states. This part can be safely skipped for a first reading. Readers only interested in the proof of our main theorems can move forward to definition \ref{def:mixed_pure}.

The positivity actually implies that $\psi(A)\geqslant 0$ for any positive quasi-local operator $A$ (\ie a self-adjoint operator whose eigenvalues are all non-negative). To see it, note that if $A$ is a positive local operator, then by linear algebra, there is a unique positive square root $\sqrt{A}$ of $A$. Thus $\psi(A)=\psi((\sqrt{A})^{\dagger}\sqrt{A})\geqslant 0$. If $A$ is a positive quasi-local operator, one can find a sequence of positive operators $A_{j}\in\A^{l}$ and we define $\sqrt{A}:=\lim_{j\to\infty}\sqrt{A_{j}}$. In this case, we again have $\psi(A)=\psi((\sqrt{A})^{\dagger}\sqrt{A})\geqslant 0$. More generally, any linear functional $f:\A^{ql}\to\bbC$ is \textit{positive} if $f(A)\geqslant 0$ for any positive operator $A\in\A^{ql}$.

It is also useful to define a norm on the space of linear functionals on $\A^{ql}$. Given a linear functional $f:\A^{ql}\to\bbC$, its norm is defined as
\beq \label{eq: def_functional_norm}
||f||:=\sup_{\substack{A\in\A^{ql},\\||A||=1}}|f(A)|.
\eeq
As a norm, it satisfies the triangle inequality
\beq\label{eq:tri_ineq_functional}
||f_{1}+f_{2}||\leqslant||f_{1}||+||f_{2}||
\eeq
To wit, note that $|f_{1}(A)+f_{2}(A)|\leqslant |f_{1}(A)|+|f_{2}(A)|$ for any $A\in\A^{ql}$. The desired inequality follows by taking supremes on both sides with respect to $||A||=1$. Notice this is the triangle inequality of the functional norms, and previously we have proved the triangle inequality for operator norms in Eq. \eqref{eq:tri_ineq_op}.

The corollary below will be useful.
\begin{corollary}\label{corollary:properties_state}
    The following properties about the state $\psi$ and positive linear functionals are true.
    \begin{enumerate}
        \item $\psi(A^{\dagger})=\psi(A)^{*}$ for all $A\in\A^{ql}$, where $z^{*}$ means the complex conjugation of the complex number $z$.
        \item Cauchy-Schwarz inequality: $|\psi(B^{\dagger}A)|^{2}\leqslant \psi(B^{\dagger}B)\psi(A^{\dagger}A)$.
        \item The norm $||\psi||=1$ if $\psi$ is normalized, \ie $\psi(I)=1$. More generally, if $f$ is a positive linear functional, then $||f||=f(I)$ where $I$ is the identity.
        \item If $f_{1},f_{2}$ are positive linear functionals, then $||f_{1}+f_{2}||=||f_{1}||+||f_{2}||$.
    \end{enumerate}
\end{corollary}
\begin{proof}
    For the first property, note that each operator $A$ has the following decomposition
    \beq
    A=\frac{A+A^{\dagger}}{2}+i\frac{A-A^{\dagger}}{2i}
    \eeq
    So we only have to show that $\psi(A)\in\R$ if $A$ is Hermitian. Assume $A$ is Hermitian. Note 
    \beq
    \psi(A+\lambda I)=\psi(A)+\lambda,\,\forall\lambda\in\bbC
    \eeq
    So without loss of generality, we can assume $A$ is positive by shifting $A\to A+\lambda I$ with $\lambda\geqslant 0$. In this case, $A$ has a square root $\sqrt{A}$ which is again positive. Then
    \beq
    \psi(A)=\psi((\sqrt{A})^{\dagger}\sqrt{A})\geqslant 0
    \eeq
    In particular, $\psi(A)\in\R$ and the first property follows.

    The second property (Cauchy-Schwarz inequality) can be proved as in the usual quantum mechanics. Consider
    \beq
    F(\lambda):=\psi((A+\lambda B)^{\dagger}(A+\lambda B))=\psi(A^{\dagger}A)+\lambda\psi(A^{\dagger}B)+\lambda^{*}\psi(B^{\dagger}A)+|\lambda|^{2}\psi(B^{\dagger}B)\geqslant 0
    \eeq
    for any $\lambda\in \bbC$. This quadratic function $F(\lambda)$ has to be positive definite, which results in the desired inequality.
    
    To see the third property, we will use the definition in Eq. \eqref{eq: def_functional_norm}, which requires us to maximize $|\psi(A)|$ over all $A\in\A^{ql}$ with $||A||=1$. To this end, by the Cauchy-Schwarz inequality,
    \beq
    |\psi(A)|^{2}=|\psi(I^{\dagger}A)|^{2}\leqslant \psi(A^{\dagger}A)
    \eeq
    so we have $\sup_{||A||=1}|\psi(A)|\leqslant\sup_{||A||=1}\sqrt{\psi(A^{\dagger}A)}$.
    Below we show that $\sup_{||A||=1}\sqrt{\psi(A^{\dagger}A)}=1$. Since $1=\psi(I)\leqslant ||\psi||\leqslant 1$, we then deduce that $||\psi||=1$.

    To show that $\sup_{||A||=1}\sqrt{\psi(A^{\dagger}A)}=1$, note that $A^{\dagger}A$ is self-adjoint and $||A^{\dagger}A||=||A||^{2}=1$ by the $C^*$-property. So to maximize $\psi(A^{\dagger}A)$ with $||A||=1$ amounts to maximizing $|\psi(T)|$ for \textit{positive} $T$ with $||T||=1$, where $T=A^{\dagger}A$ is self-adjoint.
    For later convenience, we define a partial order on self-adjoint operators. Let $T$ and $B$ be two self-adjoint operators, we write $T\geqslant B$ if $T-B$ is positive. Since $\psi$ is a positive linear functional, this implies that $\psi(T-B)\geqslant 0$. Now note that for all self-adjoint $T\in\A^{ql}$, we have
    \beq
    -||T||I\leqslant T\leqslant ||T||I
    \eeq
    Applying $\psi$, we deduce
    \beq
    -||T||\leqslant \psi(T)\leqslant ||T||
    \eeq
    This means that $\sup_{||T||=1}|\psi(T)|\leqslant 1$ for positive $T$. However, $|\psi(I)|=1$, hence $\sup_{||T||=1}|\psi(T)|=1$. Therefore, we conclude that $||\psi||=1$. For a more general positive linear functional $f$, we normalize it as $\psi:=f/f(I)$ which is a state. Hence we have $||f||=f(I)||\psi||=f(I)$.

    For the last property, note that
    \beq
    f_{1}(I)+f_{2}(I)\leqslant||f_{1}+f_{2}||\leqslant||f_{1}||+||f_{2}||=f_{1}(I)+f_{2}(I)
    \eeq
    The second “$\leqslant$” above involves Eq.~\eqref{eq:tri_ineq_functional}. Hence $||f_{1}+f_{2}||=f_1(I)+f_2(I)=||f_{1}||+||f_{2}||$.
\end{proof}

Note that for any two quantum states $\psi_{1},\psi_{2}$, according to Eq.~\eqref{eq:tri_ineq_functional}, we have
\beq
||\psi_{1}-\psi_{2}||\leqslant ||\psi_{1}||+||\psi_{2}||=2
\eeq
If this bound is saturated, which means these states are separated as far as possible, then
we say that $\psi_{1}$ is orthogonal to $\psi_{2}$.

\begin{definition}\label{def:orthogonal}
    Two states $\psi_{1},\psi_{2}$ are orthogonal if
    \beq
    ||\psi_{1}-\psi_{2}||=2
    \eeq
\end{definition}

\begin{remark}
    In Ref.~\cite{bratteli2013operator1}, a different notion of orthogonality is used, which we call independence of states (see definition ~\ref{def:independence}). We avoid that usage of orthogonality because in that definition, any two different pure states are orthogonal even for $\B(\cH)$ with $\dim(\cH)<\infty$. That means, their definition cannot reduce to the usual orthogonality in quantum mechanics.
\end{remark}

\begin{example}
    Now we show that definition \ref{def:orthogonal} reduces to the usual notion of orthogonality in quantum mechanics for $\cC=\B(\cH)$ (where $\cH$ is a finite dimensional Hilbert space). In this case, we represent these states by vectors in $\cH$, such that
    \beq
    \psi_{1}(A)-\psi_{2}(A)=\la\psi_{1}|A|\psi_{1}\ra-\la\psi_{2}|A|\psi_{2}\ra.
    \eeq
    If $\la\psi_{1}|\psi_{2}\ra=0$, \ie they are orthogonal to each other in the usual sense, we take
    \beq
    A=|\psi_{1}\ra\la\psi_{1}|-|\psi_{2}\ra\la\psi_{2}|.
    \eeq
    Note that $||A||=1$ and $\psi_{1}(A)-\psi_{2}(A)=2$. So we conclude that $||\psi_{1}-\psi_{2}||=2$ if these states are orthogonal in the usual sense. More generally if $\la\psi_{1}|\psi_{2}\ra\not=0$, with some linear algebra, one can show that
    \beq
    ||\psi_{1}-\psi_{2}||=2(1-|\la\psi_{1}|\psi_{2}\ra|^{2}).
    \eeq
    So conversely, $||\psi_{1}-\psi_{2}||=2$ also implies that $\la\psi_{1}|\psi_{2}\ra=0$.
\end{example}

Given two quantum states $\psi_{0},\psi_{1}$ and a real number $t\in [0,1]$, one can construct another state
\beq\label{eq:convex_combination}
\psi_{t}=t\psi_{1}+(1-t)\psi_{0}
\eeq
It is easy to check $\psi_{t}$ is positive as well as normalized. 
\begin{definition}\label{def:mixed_pure}
    A state $\psi$ is called a mixed state if there exists $\psi_{0},\psi_{1}$ and $0<t<1$ such that $\psi=\psi_{t}$. A state is pure if it is not mixed.
\end{definition}

In fact, these definitions are equivalent to the definitions of mixed and pure states in quantum mechanics if there are only finitely many degrees of freedom. Namely, one can show that for $C^*$-algebra $\B(\cH)$ where $\cH$ is a finite dimensional Hilbert space, a pure state defined in definition \ref{def:mixed_pure} coincides with the usual pure states (see Sec. 2.3 of Ref. \cite{Landsman:2017hpa}).

There are many equivalent definitions of pure states, and the following one is often useful.
\begin{definition}[Alternative definition for pure states]\label{def:alternative_pure}
    A state $\psi$ of a $C^*$-algebra $\cC$ is called a pure state if for any positive linear functional $\rho:\cC\to \bbC$ (may not be normalized) majorized by $\psi$ (\ie $\psi-\rho$ is again positive), we have $\rho=t\psi$ for some $t\in [0,1]$.
\end{definition}

\begin{lemma}\label{lemma:equivalence_pure}
    Pure states defined by definitions \ref{def:mixed_pure} and \ref{def:alternative_pure} are equivalent.
\end{lemma}

One can easily check this for $\B(\cH)$ with finite dimensional $\cH$, where states are represented by density matrices. Thus lemma \ref{lemma:equivalence_pure} becomes obvious by diagonalizing this density matrix. For a more general proof applicable to infinite systems, see theorem 2.3.15 of Ref. \cite{bratteli2013operator1}.

\subsection{Gelfand-Naimark-Segal construction} \label{subsec: GNS construction}

Now we do have the notion of states in the context of operator algebra, then it is tempting to talk about the Hilbert space. Indeed, working with Hilbert spaces comes with benefits. For example, abstract states in definition~\ref{def:State} do not form a linear space, which means we lack one of the most important ingredients in quantum mechanics, \ie the coherent superposition of states (see remark \ref{remark:coherent_superposition} for details). Moreover, it is often easier to work with matrices rather than abstract algebra of operators. So we want to represent our algebra $\A^{ql}$ on some Hilbert space. Furthermore, some usual notions of representation theory, such as irreducible representations and Schur's lemma can help us further decompose these matrices into simpler pieces, \ie making these matrices block-diagonalized.

Recall what we have in hand is states and the algebra of (quasi-)local operators. In quantum many-body physics, a Hilbert space can be built up by applying local operators to a fixed reference state. There is a similar way to build up a Hilbert space in a generic $C^{*}$-algebra, known as the Gelfand-Naimark-Segal (GNS) construction. We now present this construction for a general $C^*$-algebra $\cC$, but readers can keep only the example of $\A^{ql}$ in mind.

Let us fix a state $\psi$ which can be pure or mixed. We start with a {\it pure} state for simplicity. We define the GNS ideal $N_{\psi}$ as
\beq\label{eq:GNS_ideal}
N_{\psi}:=\{A\in\cC|\psi(A^{*}A)=0\}
\eeq
Remember that $\psi$ is now our input state. Informally one can think of it as $|\psi\ra$ since it is pure, and $N_{\psi}$ includes operators which annihilate $|\psi\ra$. If $A\in N_{\psi}$ (\ie $A$ annihilates $|\psi\ra$), then $BA$ also annihilates $|\psi\ra$. This means $BA$ is again in $N_{\psi}$.  Similarly if $A,B\in N_{\psi}$ then $A+B\in N_{\psi}$. More rigorously, note that by the Cauchy-Schwarz inequality (see Corollary \ref{corollary:properties_state})
\beq
|\psi(A^*B)|^{2}=|\psi(B^{*}A)|^{2}\leqslant \psi(A^*A)\psi(B^*B)=0.
\eeq
This shows if $A,B\in N_{\psi}$, then $\psi((A+B)^*(A+B))=0$, \ie $A+B\in N_\psi$. On the other hand, if $A\in N_{\psi}$ and $B\in\cC$, we have
\beq
|\psi(A^{*}B^*BA)|^{2}\leqslant \psi((A^*B^*B)^*(A^*B^*B))\psi(A^*A)=0
\eeq
So again we see $BA\in N_{\psi}$.
Thus mathematically we say that $N_{\psi}$ is a left ideal.  We will make the relation between the abstract state $\psi$ and the vector state $|\psi\ra$ precise in Eq. \eqref{eq: GNS reference state}.

The GNS Hilbert space $\cH_{\psi}$ is defined to be\footnote{Actually, this quotient space is often incomplete (\ie it is not well-behaved when taking limits) for infinite dimensional $\cC$, so one needs to do further completion. We omit this detail.}
\beq\label{eq:GNS_Hilbert_space}
\cH_{\psi}:=\cC/N_{\psi}
\eeq 
where the equivalence relation is defined as $A\sim B$ if $A-B\in N_{\psi}$ (that is, elements in $N_{\psi}$ are identified as 0 in $\cH_\psi$). We denote the equivalence class of operator $A$ by $[A]$. The $\cH_\psi$ defined above is indeed a Hilbert space, where each vector in this space is an equivalence class of operators, the addition of vectors inherits from the addition of these operators, and the inner product is given by $\la[A],[B]\ra:=\psi(A^{*}B)$. One can check that these addition and inner product are well-defined, \ie their results do not depend on the choice of the representative.
More formally, we have defined a representation of $\cC$ on $\cH_{\psi}$ as follows
\beq\label{eq:GNS_homomorphism}
\pi_{\psi}(A)[B]:=[AB]
\eeq
The last ingredient in the GNS construction is a reference state $|\psi\ra$. Given an abstract state $\psi$, we say a state $|\psi\ra\in\cH_{\psi}$ is a reference state  representing $\psi$ if
\beq \label{eq: GNS reference state}
\psi(A)=\la\psi|\pi_{\psi}(A)|\psi\ra,\,\forall\,A\in\cC.
\eeq
We note that $|\psi\ra=[I]$ provides a possible choice satisfying this equation, where $[I]$ means the equivalence class of the identity operator in above quotient $\cH_{\psi}=\cC/N_{\psi}$. So such a representative does exist.

\begin{definition}\label{def:GNS_triple}
    The above constructed $\cH_{\psi}$ (Eq.~\eqref{eq:GNS_Hilbert_space}), $\pi_{\psi}$ (Eq.~\eqref{eq:GNS_homomorphism}) together with the state $|\psi\ra$ is called a GNS triple, denoted by $(\pi_{\psi},\cH_{\psi},|\psi\ra)$. In mathematics, this $|\psi\ra$ is called a cyclic vector.
\end{definition}
\nomenclature{$\pi_{\psi}$}{The representation homomorphism, def. \ref{def:GNS_triple}}
\nomenclature{$\cH_{\psi}$}{The GNS Hilbert space associated $\psi$, def. \ref{def:GNS_triple}}

However, it is important to note that the representative $|\psi\ra$ in the Hilbert space $\cH_\psi$ of an abstract state $\psi$ is not unique. For example, we can replace it with $|\Psi\ra=U|\psi\ra$ and work with
\beq
\pi(A):=U\pi_{\psi}(A)U^{-1}
\eeq
for any unitary operator $U\in\B(\cH_{\psi})$. We emphasize that this unitary operator $U$ is only defined on the particular GNS Hilbert space $\cH_{\psi}$. This also defines another GNS triple $(\pi,\cH_{\psi},|\Psi\ra)$. It turns out that given a state $\psi$, the GNS representation is unique up to unitary equivalence.

\begin{corollary}[Theorem 2.5.3 of Ref.~\cite{Naaijkens_2017}]\label{corollary:unique:GNS}
    GNS triple is unique in the following sense. Suppose that $\psi$ is a state of $\A^{ql}$, and $(\pi_{\psi},\cH_{\psi},|\psi\ra)$ and $(\pi,\cH_{\psi},|\Psi\ra)$ are two different GNS triples associated with $\psi$, then there exists a unitary operator $U\in\B(\cH_{\psi})$ such that
    \beq
    \begin{split}
        U|\psi\ra&=|\Psi\ra\\
        U\pi_{\psi}(A)U^{-1}&=\pi(A)
    \end{split}
    \eeq
\end{corollary}

One can think of the GNS construction in this case as a certain kind of {\it{state-operator correspondence}}, since all states in $\cH_{\psi}$ can be obtained by applying $\pi_{\psi}(A)$ to the reference state $|\psi\ra$ for some $A\in\cC$. Thus, each state corresponds to certain equivalence class of states in $\cC/N_{\psi}$.

Notice the above definition of the inner product in $\cH_\psi$ indicates what it means for two vectors in $\cH_\psi$ to be orthogonal. However, this orthogonality is in general unrelated to the orthogonality defined in definition \ref{def:orthogonal}.

\begin{remark}
    One may wonder, if we start with a pure state of $\A^{ql}$, as explained in the beginning of Sec. \ref{subsec:quasi_local}, there is no Hilbert space for an infinite spin chain, then what is the GNS Hilbert space on earth? As we will see in Appendix \ref{subsec:Superselection_sectors} (in particular, example \ref{example:Superselection_Ising}), the GNS Hilbert space only describes a superselection sector of our model.
\end{remark}

\begin{remark}\label{remark:coherent_superposition}
    So far we have only defined the convex combinations of abstract states (\ie normalized positive linear functionals), see Eq.~\eqref{eq:convex_combination}, which model classical mixtures, rather than quantum coherent superpositions in quantum mechanics. In the operator algebra formalism applied to infinite systems, we can only talk about coherent superpositions within a GNS Hilbert space, which is given by the usual vector addition.
\end{remark}

The definition of GNS triple (definition ~\ref{def:GNS_triple}) also applies to mixed states. To see the structure of the GNS Hilbert space $\cH_{\psi}$ for a mixed state $\psi$ intuitively, note that one can always decompose a mixed state into a convex linear combinations of pure states (see Eq.~\eqref{eq:decomposition_states} below) and then we build up the GNS Hilbert space for each pure state separately. The space $\cH_{\psi}$ is the direct sum of these Hilbert spaces. Below we discuss the detailed structure of the GNS Hilbert space for mixed states. Readers can skip this discussion for a first reading and move forward to example \ref{example:QM_GNS}.

To study the GNS Hilbert space of $\cH_{\psi}$ for a mixed state $\psi$ in more detail, we need a notion called independent decomposition. To begin, we need the notion of independence of positive linear functionals (see the paragraph below definition~\ref{def:State} for the definition of positive linear functionals).
\begin{definition}\label{def:independence}
    Given two positive linear functionals $f_{1},f_{2}:\A^{ql}\to\bbC$, if there exists no nonzero positive linear functional $f$ such that $f_{1}-f$ and $f_{2}-f$ are again positive, then we say $f_{1}$ and $f_{2}$ are independent.
\end{definition}
\begin{example}\label{example:indenpendent_states}
    By definition~\ref{def:alternative_pure}, any two different nonzero pure states $\psi_{1},\psi_{2}$ are independent. To wit, suppose they are not independent, then there is a nonzero positive linear functional $f$ in definition \ref{def:independence}. But if $\psi_{1}-f$ is positive, then $f=\lambda_{1}\psi_{1},0<\lambda_{1}\leqslant 1$ since $\psi_{1}$ is pure. Similarly, $f=\lambda_{2}\psi_{2},0<\lambda_{2}\leqslant 1$. Taking norms of these equations and noting that $||\psi_{1}||=||\psi_{2}||=1$ (see the third property of corollary \ref{corollary:properties_state}), we find $\lambda_{1}=\lambda_{2}$ and hence $\psi_{1}=\psi_{2}$, contradicting to our assumption.
    
    Moreover, orthogonal states (in the sense of definition~\ref{def:orthogonal}) are independent. To see it, let $\psi_{1}$ and $\psi_{2}$ be orthogonal states (not necessarily pure) and $f$ be a positive linear functional such that $\psi_{1}-f$ and $\psi_{2}-f$ are positive. By the third property of Corollary \ref{corollary:properties_state},
    \beq
    ||\psi_{1}-f||=\psi_{1}(I)-f(I)\leqslant ||\psi_{1}||=1
    \eeq
    Similarly, $||\psi_{2}-f||\leqslant ||\psi_{2}||=1$. On the other hand,
    \beq
    2=||\psi_{1}-\psi_{2}||\leqslant ||\psi_{1}-f||+||f-\psi_{2}||\leqslant ||\psi_{1}||+||\psi_{2}||=2
    \eeq
    Note the equality holds only if $||\psi_{1}-f||=||\psi_{2}-f||=1$. However, by the fourth property in Corollary \ref{corollary:properties_state},
    \beq
    1=||\psi_{1}||=||\psi_{1}-f||+||f||
    \eeq
    which implies $||f||=0$ and thus $f=0$. Therefore, $\psi_{1}$ and $\psi_{2}$ are independent.
    
\end{example}

\begin{example}\label{example:dependent_states}
    We also give an example where two linear functionals are \textbf{not} independent. By abusing terms, we say that two density matrices $\rho_{1},\rho_{2}$ are independent if the associated linear functionals $\tr(\rho_{1}\cdot),\tr(\rho_{2}\cdot)$ are independent. Consider the following two density matrices in a qubit system
    \beq
    \begin{split}
        \rho_{1}&=\frac{1}{2}(|\uparrow\ra\la\uparrow|+|\downarrow\ra\la\downarrow|),\\
        \rho_{2}&=|+\ra\la +|,
    \end{split}
    \eeq
    where $|\uparrow\ra$, $|\downarrow\ra$ are orthonormal basis on $\cH=\bbC^{2}$ and $|+\ra:=\frac{1}{\sqrt{2}}(|\uparrow\ra+|\downarrow\ra)$. Note
    \beq
    \rho_{1}-\frac{1}{2}|+\ra\la+|,\,\rho_{2}-\frac{1}{2}|+\ra\la+|
    \eeq
    are positive and $\frac{1}{2}|+\ra\la+|$ corresponds to a positive linear function. Hence by definition \ref{def:independence} $\rho_{1}$ and $\rho_{2}$ are not independent.
\end{example}

The importance of independent states lies in the following theorem.
\begin{theorem}[Lemma 4.1.19 of Ref.~\cite{bratteli2013operator1}]
    If $\psi_{0},\psi_{1}$ are independent states, and 
    \beq
    \psi:=t\psi_{1}+(1-t)\psi_{0},0<t<1
    \eeq
    then the GNS representation decomposes as
    \beq
    \pi_{\psi}=\pi_{\psi_{0}}\oplus\pi_{\psi_{1}}
    \eeq
    The converse is also true.
\end{theorem}

So the structure of the GNS representation of a mixed state $\psi$ is clear if we can decompose it into independent states. Such a decomposition is indeed possible due to the following theorem.

\begin{theorem}[Theorem 4.4.9 of Ref.~\cite{bratteli2013operator1}]\label{thm:independent_decomposition}
    Any state $\psi$ admits an independent decomposition\footnote{It can happen that $\psi$ is decomposed into uncountably many pure states. In that case, the direct sum on the right-hand side of Eq.~\eqref{eq:decomposition_states} is replaced by an integral on space of states while the direct sum in Eq.~\eqref{eq:mixed_state_GNS_representation} is replaced by direct integrals (see Ref.~\cite{bratteli2013operator1}). The discrete part like Eq.~\eqref{eq:decomposition_states} and Eq.~\eqref{eq:mixed_state_GNS_representation} are called atomic in Ref.~\cite{kadison1997fundamentals2}.}
    \beq\label{eq:decomposition_states}
    \psi=\sum_{i\in I}\lambda_{i}\psi_{i}
    \eeq
    where $I$ is an index set, $\psi_{i}$ is pure for each $i\in I$ and $0<\lambda_{i}\leqslant1$ with $\sum_{i\in I}\lambda_{i}=1$. This decomposition is independent in the sense that, for any subset $J\subset I$, $\psi_{J}:=\sum_{i\in J}\lambda_{i}\psi_{i}$ is independent of $\psi_{J^{c}}=\sum_{i\in I\setminus J}\lambda_{i}\psi_{i}$. Besides, the GNS representation $\pi_{\psi}$ is also decomposed as
    \beq\label{eq:mixed_state_GNS_representation}
    \pi_{\psi}=\bigoplus_{i\in I}\pi_{\psi_{i}}
    \eeq
    That means, the representation matrices $\pi_{\psi}(A)$ can be simultaneously block-diagonalized for all $A\in\A^{ql}$.
\end{theorem}

\begin{proposition}\label{prop:irrep_pure}
    One says that a GNS representation $\pi_{\psi}$ is irreducible if Eq. \eqref{eq:mixed_state_GNS_representation} has only 1 direct summand, \ie the state $\psi$ is pure.
\end{proposition}

The above proposition gives another equivalence condition for a state to be pure, \ie its GNS representation is irreducible.

Given the notion of irreducibility, the usual Schur's lemma follows.
\begin{lemma}[Schur]\label{lemma:Schur}
    For a GNS representation $(\pi_{\psi},\cH_{\psi},|\psi\ra)$ of a $C^*$-algebra $\cC$, it is irreducible if and only if, for each $T\in\B(\cH_{\psi})$ commuting with all $\pi_{\psi}(\cC)$, we have $T=\lambda I$ for some $\lambda\in \bbC$.
\end{lemma}

Now we are ready to give some examples of the GNS construction.

\begin{example}\label{example:QM_GNS}
    Consider a qubit which lives in $\cH=\bbC^{2}$. The operator algebra is $M_{2}(\bbC)$, the 2-by-2 matrix algebra over $\bbC$. We begin with the following state,
    \beq
    \psi(A)=A_{11}
    \eeq
    where $A_{11}$ is the first matrix element of $A$. We denote a basis of $M_{2}(\bbC)$ as
    \beq
    \begin{split}
        E_{11}=\begin{pmatrix}
        1&&0\\
        0&&0
        \end{pmatrix}, E_{12}=\begin{pmatrix}
        0&&1\\
        0&&0
        \end{pmatrix}\\
        E_{21}=\begin{pmatrix}
        0&&0\\
        1&&0
        \end{pmatrix},E_{22}=\begin{pmatrix}
        0&&0\\
        0&&1
        \end{pmatrix}
    \end{split}
    \eeq
    Then it is easy to check
    \beq
    \begin{split}
        \psi(E_{11}^{\dagger}E_{11})=1,\psi(E_{12}^{\dagger}E_{12})=0\\
        \psi(E_{21}^{\dagger}E_{21})=1,\psi(E_{22}^{\dagger}E_{22})=0
    \end{split}
    \eeq
    Hence $N_{\psi}=\rspan\{E_{12},E_{22}\}$, or equivalently $E_{12},E_{22}$ annihilate $\psi$. Thus, reassuringly, $\cH_{\psi}=M_{2}(\bbC)/N_{\psi}\simeq \bbC^{2}$ is exactly the space we start with! Besides, under the basis $[E_{11}],[E_{21}]$ (or $|E_{11}\ra,|E_{21}\ra$ if one prefers) the representation $\pi_{\psi}$ is given by
    \beq
    \begin{split}
        \pi_{\psi}(E_{11})=\begin{pmatrix}
        1&&0\\
        0&&0
        \end{pmatrix}, \pi_{\psi}(E_{12})=\begin{pmatrix}
        0&&1\\
        0&&0
        \end{pmatrix}\\
        \pi_{\psi}(E_{21})=\begin{pmatrix}
        0&&0\\
        1&&0
        \end{pmatrix},\pi_{\psi}(E_{22})=\begin{pmatrix}
        0&&0\\
        0&&1
        \end{pmatrix}
    \end{split}
    \eeq
    That is, these operators are represented in the {\it{standard}} way.
    
    On the other hand, if we consider 
    \beq
    \rho(A)=tA_{11}+(1-t)A_{22}
    \eeq
    for some $0<t<1$. Then one can calculate
    \beq
    \begin{split}
        \rho(E_{11}^{\dagger}E_{11})=t,\rho(E_{12}^{\dagger}E_{12})=1-t\\
        \rho(E_{21}^{\dagger}E_{21})=t,\rho(E_{22}^{\dagger}E_{22})=1-t
    \end{split}
    \eeq
    Thus $N_{\rho}=0$ and $\cH_{\rho}\simeq M_{2}(\bbC)/N_\rho\simeq \bbC^{4} \simeq \bbC^{2}\oplus \bbC^{2}.$

    However, consider the following state
    \beq
    \omega=\tr(\Omega\cdot)
    \eeq
    where $\Omega=\frac{1}{3}(|\uparrow\ra\la\uparrow|+|\downarrow\ra\la\downarrow|+|+\ra\la +|)$ is the density matrix representing $\omega$ (see example \ref{example:dependent_states}). This state contains 3 pure states in the decomposition.  Naively we will get 3 copies of $\bbC^{2}$, so do we have $\cH_{\omega}=\bbC^{6}$? However, we cannot get a 6-dimensional Hilbert space since $\B(\cH)$ is 4-dimensional after all. What happens is that this decomposition of $\Omega$ above is not independent, as is checked in example \ref{example:dependent_states}, so we cannot conclude that $\cH_\omega=\bbC^6$.
     
\end{example}

Another example will be the Ising chain.
\begin{example}\label{example:Ising}
    In this example, we have an on-site Hilbert space $\cH_{k}\simeq \bbC^{2}$ for each site $k\in\z$, and $X,Y,Z$ will be the usual Pauli operators.
    We consider the all-spin-up state $\psi_{+}$. More suggestively, we write it as $|\uparrow\ra$.

    Note that $\psi_{+}((Z_{k}-1)^{\dagger}(Z_{k}-1))=0$, and it is easy to see the left ideal $N_{\psi_{+}}$ is generated by $(Z_{k}-1)$ for all $k\in\z$, which means if $A\in N_{\psi_{+}}$, there exists $B\in \A^{ql}$ such that
    \beq
    A=B(Z_{k}-1)
    \eeq
    for some $k\in\z$. We claim that
    \beq
    \cH_{\psi_{+}}\simeq\A^{ql}/N_{\psi_{+}}\simeq \mathrm{span}\{\prod_{k\in I}[X_{k}]|I\subset \z,|I|<\infty\}
    \eeq
    This is because any local operator can be written as a linear combination of Pauli basis
    \beq
    \prod_{i\in I}X_{i}\prod_{j\in J}Z_{j},\,|I|,|J|<\infty
    \eeq
    (there is no Pauli $Y_{k}$ operator above since it is not independent, \ie $Y_{k}=-i X_{k}Z_{k}$). The quotient procedure amounts to regarding $Z_{j}=1$ for all $j\in\z$. So the claim above follows. Having $\cH_{\psi_+}$ in hand, it is straightforward to obtain the representation of operators, $\pi_{+}$.
    
    We give a more intuitive explanation of the this example below, by explaining how applying operators to $|\uparrow\ra$ results in other states. More intuitively,
    \beq
    \pi_{+}(Z_{k})|\uparrow\ra=|\uparrow\ra
    \eeq
    for any $k\in \z$. On the other hand,
    \beq
    \pi_{+}(X_{k})|\uparrow\ra\not=|\uparrow\ra
    \eeq
    which can be seen by applying $Z_{k}$ at site $k$. The action of $Y_{k}$ is not independent since $Y_{k}=iZ_{k}X_{k}$, so we ignore it.
    Thus states in GNS Hilbert space $\cH_{+}$ of $|\uparrow\ra$ can be identified as
    \begin{equation}
        \prod_{j\in J}\pi_{+}(X_{j})|\uparrow\ra
    \end{equation}
    where $J$ is a \textbf{finite} subset of $\z$. That is, the states in $\cH_{+}$ are configurations where almost all spins are up but only finitely many spins are flipped\footnote{As noted before, one needs completion to make it really a Hilbert space.}. Also, it is easy to check that this representation is irreducible, hence $\psi_{+}$ is pure as expected.

    Similarly, for $\psi_{-}$, the all-spin-down state (also denoted as $|\downarrow\ra$), the GNS Hilbert space $\cH_{-}$ in this case is similar to $\cH_{+}$, but most of its spins are down while only finitely many of them are up. The associated representation is denoted by $\pi_{-}$.
     
\end{example}

\subsection{Superselection sectors}\label{subsec:Superselection_sectors}

In this subsection, we discuss an important application of the concept of GNS construction discussed in the previous subsection.

One drastic difference between infinite systems and finite systems is the notion of superselection sectors in infinite systems. Physically speaking, two states in an infinite system fall into different superselection sectors if and only if they cannot be connected by local operators (see Sec. 7.1 of Ref.~\cite{Halvorson2006AQFT} for a mathematical definition of superselection sectors).
 
For example, different ground states related to spontaneous symmetry breaking fall into different superselection sectors, and topologically degenerate ground states in the toric code are also in different superselection sectors.

A related notion in the GNS representations is as follows.

\begin{definition}\label{def:inequivalence_GNS}
    For $\A^{ql}$, the GNS representations of two different states $\psi$ and $\sigma$ are said to be inequivalent if there is no unitary map $U:\cH_{\rho}\to\cH_{\psi}$ such that $U^{\dagger}\pi_\psi(A)U=\pi_{\rho}(A)$. Otherwise, they are said to be equivalent.
\end{definition}
 
As we will see later, the (in)equivalence of irreducible GNS representations is closely related to superselection sectors. For demonstration, we consider the following example.

\begin{example}[Sec. 3.1 of Ref. \cite{Naaijkens_2017}]\label{example:Superselection_Ising}
    Let us continue our Ising model example \ref{example:Ising}. Now we show that $\pi_{+}$ and $\pi_{-}$ defined earlier are inequivalent in the sense of definition \ref{def:inequivalence_GNS}.
    
    Suppose they are equivalent, \ie there is a unitary operator $U:\cH_{+}\to \cH_{-}$ such that $\pi_{+}(A)=U^{\dagger}\pi_{-}(A)U$ for all $A\in\A^{ql}$. We define the polarization operator
    \beq
    m_{N}:=\frac{1}{2N+1}\sum_{k=-N}^{N}Z_{k}
    \eeq
    where $N$ is arbitrarily big but stays finite. We \textbf{do not} talk about the limit of $m_{N}$ as $N\to \infty$. Note that
    \beq
    \la\uparrow|\pi_{+}(m_{N})|\uparrow\ra=1
    \eeq
    On the other hand,
    \beq
    \la\uparrow|\pi_{+}(m_{N})|\uparrow\ra=\la\uparrow|U^{\dagger}\pi_{-}(m_{N})U|\uparrow\ra
    \eeq
    Note $U|\uparrow\ra\in\cH_{-}$, hence it has almost all spins being down, and only finitely many of them are up. Therefore, if $N$ is sufficiently large, one has
    \beq
    \la\uparrow|U^{\dagger}\pi_{-}(m_{N})U|\uparrow\ra\stackrel{N\gg 1}{\longrightarrow} -1
    \eeq
    Thus one obtains a contradiction, which means there cannot be such a unitary operator $U$.
     
    \end{example}
    
    The lesson from the above example is that $\pi_{+}$ is inequivalent to $\pi_{-}$ because $|\uparrow\ra$ cannot be transformed into $|\downarrow\ra$ by local operators only. Generalizing this example, one can obtain the following theorem.
    \begin{theorem}[Theorem 2.6.1 of Ref. \cite{Naaijkens_2017}]
        Let $\cC$ be a $C^*$-algebra and $\psi_{1},\psi_{2}$ be two pure states, then their GNS representations $\pi_{\psi_{1}}$ and $\pi_{\psi_{2}}$ are equivalent if and only if $\psi_{1}$ and $\psi_{2}$ fall into the same superselection sector.
    \end{theorem}

    To justify our earlier assertion that a {\it single} finite-size system has only one superselection sector, we have
    \begin{proposition}\label{prop:fd_C*_rep}
        Let $\A$ be a finite dimensional $C^*$-algebra, then $\A$ has only one irreducible representation (up to unitary equivalence).
    \end{proposition}
    \begin{proof}
        By proposition \ref{prop:fd_C*} $\A$ is the direct sum of complex matrix algebras. It is well known that the irreducible representation of a complex matrix algebra is unique, \ie its defining representation (Theorem 4.5.10 of Ref.~\cite{hassani2013mathematical}).
    \end{proof}

    We emphasize that the above discussion applies to a {\it single} finite-size system, and the above result should not be confused with the superselection rules discussed in definition \ref{def: superselection rule} in Appendix \ref{subapp: SSB area law}, which concern with sequences of finite systems with increasing sizes.
    
   As a further consequence, if two pure states are in the same superselection sector, \ie they can be related to each other by (quasi-)local operators, then they are almost the same outside a small region. Thus one concludes the following.
   \begin{proposition}[Proposition 3.2.8 of Ref. \cite{Naaijkens_2017}, Theorem 10.2.6 of Ref.~\cite{kadison1997fundamentals2}]\label{prop:equivalence_of_states}
       Given two {\bf pure} states $\psi_{1}$ and $\psi_{2}$ of $\A^{ql}$, the following statements are equivalent:
        \begin{enumerate}
            \item The GNS representation $\pi_{\psi_{1}}$ is equivalent to $\pi_{\psi_{2}}$.
            \item There exists a quasi-local unitary operator $U$ such that $\psi_{1}=\psi_{2}\circ\Ad_{U}$.
            \item For any $\epsilon>0$, there is a finite subset $\Gamma_{\epsilon}$ such that
            \beq
            |\psi_{1}(A)-\psi_{2}(A)|<\epsilon||A||
            \eeq
            for any $A\in\A^{ql}_{\Gamma^{c}_{\epsilon}}$, where $\Gamma_{\epsilon}^{c}$ means the complement of $\Gamma_{\epsilon}$ in $\Lambda$.
        \end{enumerate}
        If one of these conditions is satisfied, we say that $\psi_{1}$ is equivalent to $\psi_{2}$ and write $\psi_{1}\simeq\psi_{2}$.
   \end{proposition}
       
       \begin{remark}
        As a remark on the above example, one can similarly show that ground states exhibiting spontaneous symmetry breaking (SSB) fall into different superselection sectors if there is a local order parameter $\mathcal{O}_{k}$ ($Z_{k}$ for Ising model) and unbroken translation symmetry. Especially, we do not need $\mathcal{O}_{k}$ to be Hermitian or to commute with Hamiltonian.
    \end{remark}
    
It is generally impossible to define coherent superpositions of states on $\A^{ql}$. However, it is indeed possible to do so when two pure states fall into the same superselection sector. Let $\psi_{1}$ and $\psi_{2}$ be two pure states such that $\psi_{1}\simeq\psi_{2}$. By Proposition \ref{prop:equivalence_of_states}, there exists a quasi-local unitary $U$ such that $\psi_{2}=\psi_{1}\circ\Ad_{U}$. Therefore, the GNS vector state $|\psi_{2}\ra$ can be naturally taken as
\beq\label{eq:equivalent_representative}
|\psi_{2}\ra=\pi_{1}(U)|\psi_{1}\ra\in\cH_{1}
\eeq
and the GNS representation of $\psi_2$ can be taken as $\pi_2(a)=\pi_1(U)\pi_1(a)\pi_1(U)^{-1},\,\forall a\in\A^{ql}$, where $(\pi_{1},\cH_{1},|\psi_{1}\ra)$ is the GNS triple of $\psi_{1}$. The relative phase is determined by requiring that $\la\psi_{1}|\psi_{2}\ra\geqslant0$. Although the choice of $U$ may not be unique for a given $\psi_{2}$, one can check that the state Eq.~\eqref{eq:equivalent_representative} does not really depend on the choice of $U$. Thus, the coherent superposition of $\psi_{1}$ and $\psi_{2}$ can be formulated as follows.

\begin{definition}\label{def:coherent_superposition}
    Let $\psi_{1}$ and $\psi_{2}$ be two equivalent pure states, then their coherent superposition is defined to be the vector state $c_{1}|\psi_{1}\ra+c_{2}|\psi_{2}\ra$, where $|\psi_{1}\ra,|\psi_{2}\ra$ are defined above and $|c_{1}|^{2}+|c_{2}|^{2}=1$.
\end{definition}
 
Another lesson from the above example is that states in different superselection sectors cannot be superposed coherently, \ie the sum of two such states is a classical mixture of them. To see this, consider the Greenberger–Horne–Zeilinger (GHZ-)state defined on a finite system with size $L$,
\beq\label{eq:finite_GHZ}
|GHZ\ra:=\frac{1}{\sqrt{2}}(|\uparrow\ra_{L}+|\downarrow\ra_{L})
\eeq
where subscript $L$ means each state is defined on a chain of size $L$. Its density matrix is
\beq
\rho_{GHZ}=\frac{1}{2}(|\uparrow\ra\la\uparrow|_{L}+|\downarrow\ra\la\downarrow|_{L})+\frac{1}{2}(|\uparrow\ra\la\downarrow|_{L}+|\downarrow\ra\la\uparrow|_{L})
\eeq
Note that the cross terms (in the second parenthesis) evaluate to 0 on any local operators in the thermodynamic limit. Thus, these two terms vanish as a state of $\A^{ql}$. As a result, 
\beq\label{eq:mixed_GHZ}
\rho_{GHZ}\stackrel{L\to \infty}{\Longrightarrow}\frac{1}{2}(|\uparrow\ra\la\uparrow|+|\downarrow\ra\la\downarrow|)
\eeq
Namely, this GHZ-state becomes a mixed state in this limit. See Appendix~\ref{sec: thermodynamic limit} for a more precise meaning of thermodynamic limit here.

In Proposition \ref{prop:equivalence_of_states}, we have described the equivalence relation between pure states. As a lesson from the example of the GHZ states (Eq.~\eqref{eq:mixed_GHZ}), mixed states can appear naturally as the thermodynamic limit of some pure states. Thus it will be useful to generalize the equivalence relation in Proposition \ref{prop:equivalence_of_states} to any pair of states that can be either pure or mixed. Below we make this generalization.

To generalize Proposition \ref{prop:equivalence_of_states} to mixed states, we need to introduce the notion of quasi-equivalence of states. Let $\psi$ be a (possibly mixed) state of $\A^{ql}$. We denote its GNS triple by $(\pi_{\psi},\cH_{\psi},|\psi\ra)$ (see Definition \ref{def:GNS_triple}), where $\cH_{\psi}$ is the GNS Hilbert space and $\pi_{\psi}:\A^{ql}\to\B(\cH_{\psi})$ is the $*$-homomorphism. If $\pi_{\psi}$ is reducible with an invariant subspace $\cH_{0}$ (\ie $\pi_{\psi}$ is block-diagonal with respect to $\cH_{\psi}=\cH_{0}\oplus\cH_{0}^{\perp}$), then we say $\pi_{\psi}|_{\cH_{0}}$ (\ie the block on $\cH_{0}$) is a \textit{subrepresentation} of $\pi_{\psi}$.

\begin{definition}[Definition 8.18 of Ref.~\cite{Landsman:2017hpa}]\label{def:quasi_disjoint}
    Let $\pi_{1},\pi_{2}$ be any two representations of $\A^{ql}$, then
    \begin{enumerate}
        \item $\pi_{1}$ and $\pi_{2}$ are said to be quasi-equivalent to each other if any subrepresentation of $\pi_{1}$ has a subrepresentation which is (unitarily) equivalent to a subrepresentation of $\pi_{2}$ and vice versa.
        \item $\pi_{1}$ is said to be disjoint from $\pi_{2}$ if they do not have any equivalent subrepresentations.
    \end{enumerate}
    For two states, we say that they are quasi-equivalent (resp. disjoint) if their GNS representations are quasi-equivalent (resp. disjoint). We write $\psi\sim\phi$ if the states $\psi$ and $\phi$ are quasi-equivalent.
\end{definition}
We remark that a sufficient and necessary condition for two states to be quasi-equivalent is based on von-Neumann algebras, which we will explain later (see Proposition \ref{prop:quasi-eq}). This condition is practically easier to use than the definition above.

We note that two representations can in principle be neither quasi-equivalent nor disjoint. In general, there are four levels of (in-)equivalence between two states.

\begin{itemize}

    \item Two states can be disjoint. For example, $\psi_+$ and $\psi_-$ in example \ref{example:Ising} are disjoint.

    \item Two states can be not disjoint but also not quasi-equivalent. For example, the GHZ state defined in Eq.~\eqref{eq:mixed_GHZ} $\psi_{GHZ}=\frac{1}{2}(\psi_{\uparrow}+\psi_{\downarrow})$ is neither disjoint nor quasi-equivalent to $\psi_{\uparrow}$.

    \item Two states can be quasi-equivalent but not equivalent (the definition of equivalence is in definition \ref{def:inequivalence_GNS}). For example, consider two pure states $\psi_{1}\simeq\psi_{2}$ while $\psi_{1}\not=\psi_{2}$, then the mixed state $\frac{1}{2}(\psi_{1}+\psi_{2})$ is quasi-equivalent to $\psi_{1}$ (and $\psi_{2}$) but they are inequivalent because a mixed state can never be equivalent to a pure state.

    \item Two states can be equivalent, as defined in definition \ref{def:inequivalence_GNS}.
    
\end{itemize}

However, for two pure states, being disjoint is the same as being inequivalent, and being quasi-equivalent is the same as being equivalent.

According to Proposition \ref{prop:equivalence_of_states}, if two pure states are quasi-equivalent (resp. disjoint), it means they are in the same (resp. different) superselection sector(s). There is mixed-state generalization of Proposition \ref{prop:equivalence_of_states}, which would be helpful to verify the quasi-equivalence of states.

\begin{proposition}[Corollary 2.6.11 of Ref.~\cite{bratteli2013operator1}]\label{prop:mixed_states_equivalence}
    Let $\psi_{1}$ and $\psi_{2}$ be states of $\A^{ql}$, then the followings are equivalent.
    \begin{enumerate}
        \item $\psi_{1}$ and $\psi_{2}$ are quasi-equivalent, \ie $\psi_{1}\sim\psi_{2}$.
        \item For any $\epsilon>0$, there exists a finite region $\Gamma_{\epsilon}\subset\Lambda$ such that
        \beq
        |\psi_{1}(A)-\psi_{2}(A)|<\epsilon ||A||,\quad\forall\,A\in\A^{ql}_{\Gamma_{\epsilon}^{c}}.
        \eeq
    \end{enumerate}
\end{proposition}

\subsection{von Neumann algebras}\label{sec:vN}

The notions of disjoint states and quasi-equivalent states play an important role in the study of von Neumann algebras (see Proposition \ref{prop:quasi-eq}), which will be defined below. We will see that von Neumann algebras are very useful for proving some of our results later.

To define a von Neumann algebra, we begin with the definition of commutants.
\begin{definition}
    Let $\cH$ be a (separable) Hilbert space and $\A\subseteq\B(\cH)$ be a $*$-subalgebra. Then its commutant (a.k.a centralizer) is defined to be
    \beq
    \A':=\{x\in\B(\cH)|[x,y]=0,\,\forall\,y\in\A\}
    \eeq
\end{definition}
It is easy to check that $\A'$ is another $*$-subalgebra of $\B(\cH)$, therefore one can form the double commutant,
\beq
\A'':=(\A')'
\eeq
It is easy to see that $\A\subseteq \A''$ from the definition. For example, when $\psi$ is a pure state and $\pi_{\psi}\subseteq\B(\cH_{\psi})$, by Schur's lemma \ref{lemma:Schur}, we have $\pi_{\psi}(\A^{ql})'=\bbC\cdot\id_{\cH_{\psi}}$ and therefore $\pi_{\psi}(\A^{ql})''=\B(\cH_{\psi})$, which is consistent with the relation $\A\subseteq \A''$. On the other hand, from the definition it is also straightforward to check that $(\A')''=\A'$ for any choice of $\A$.
\begin{definition}
    A $*$-subalgebra $\A\subseteq\B(\cH)$ is called a von Neumann algebra if
    \beq
    \A''=\A
    \eeq
    Given a state $\psi$ of $\A^{ql}$, we call
    \beq
    \cM_{\psi}:=\pi_{\psi}(\A^{ql})''\subseteq\B(\cH_{\psi})
    \eeq
     the von Neumann algebra associated with $\psi$. Note that $\cM_{\psi}''=\cM_\psi$ since $(\A')''=\A'$ for any $\A\subseteq\B(\cH)$.
\end{definition}
It is also shown in Ref.~\cite{Landsman:2017hpa} (see Corollary C.129 therein) that every von Neumann algebra is also a $C^*$-algebra (\ie every Cauchy sequence with respect to operator norm must converge).
The first reason why von Neumann algebras are interesting is due to the so-called von Neumann's double commutant theorem. To state it, we need some more general notions of convergence in $\B(\cH)$.
\begin{definition}[Strong and weak operator convergence]
    Let $\cH$ be a Hilbert space and a sequence of operators $a_{n}\in \B(\cH)$. We say that $a_{n}$ converges to $a\in\B(\cH)$ in the strong operator topology if for each $|v\ra\in\cH$ we have $||(a_{n}-a)|v\ra||\stackrel{n\to\infty}{\longrightarrow}0$. In this case, we call $a$ the strong operator limit of $a_{n}$.
    
    Similarly, we say $a_{n}$ converges to $a'\in\B(\cH)$ in the weak operator topology if for any $|u\ra,|v\ra\in \cH$ we have $\la u|a_{n}|v\ra\stackrel{n\to\infty}{\longrightarrow}\la u|a'|v\ra$, \ie each “matrix element” converges. We call $a'$ the weak operator limit of $a_{n}$.
\end{definition}
\begin{remark} \label{remark: operator limits}
    If the Hilbert space $\cH$ is infinite dimensional, then the strong operator convergence is generally stronger than the weak operator convergence. That is, if a sequence of operators has a strong operator limit, then this limit must also be a weak operator limit. On the other hand, in an infinite dimensional $\cH$, both the strong and weak operator convergences are generally weaker than the ``norm convergence", \ie a sequence $a_n$ converges to $a''\in\B(\cH)$ in the sense that $||a_n-a''||\stackrel{n\to\infty}{\longrightarrow}0$. If the Hilbert space $\cH$ is finite dimensional, then all these three versions of operator convergence are equivalent.
\end{remark}
\begin{remark}\label{remark:BBA}
    One reason to consider weak operator topology is because the unit norm ball in $B(\cH)$ is compact in the weak operator topology. This means, for any sequence $x_{n}\in\B(\cH)$ with $||x_{n}||=1$, it is always possible to find a weak operator convergent subsequence. The proof of this statement relies on the Banach-Alaoglu theorem (Theorem \ref{thm:BAB}) and viewing $\B(\cH)$ as the linear dual of trace-class operators.
\end{remark}

Now we are ready to state von Neumann's double commutant theorem.

\begin{theorem}[von Neumann's double commutant theorem, Theorem C.127 of Ref. \cite{Landsman:2017hpa}]\label{thm:vN_double_commutant}
    Let $\A\subseteq\B(\cH)$ be a *-subalgebra and $a\in\B(\cH)$. Then $a\in\A''$ iff $a$ is a weak operator limit of some sequence in $\A$. This statement is still true if the ``weak operator limit" is replaced by ``strong operator limit".
\end{theorem}
In practice, a more useful version of Theorem \ref{thm:vN_double_commutant} is the following Kaplansky density theorem.
\begin{theorem}[Theorem C.131 of Ref.~\cite{Landsman:2017hpa}]\label{thm:Kaplansky_density}
    Let $\A\subseteq\B(\cH)$ and $a\in\B(\cH)$, then $a\in\A''$ with $||a||=1$ iff it is the weak operator limit of a sequence $a_{n}\in\A$ with $||a_{n}||=1$ for each $n$. This statement is still true if the ``weak operator limit" is replaced by ``strong operator limit".
\end{theorem}
As an application of Theorem \ref{thm:Kaplansky_density}, we show that disjoint states (\ie states in different superselection sectors) are orthogonal in the sense of Definition \ref{def:orthogonal}. This proposition will later be used for proving the superselection rule in Theorem \ref{thm:ssr}. Moreover, this theorem is also useful for proving Theorem \ref{thm:clustering}, which is a very profound and powerful result that will be used in proving corollary \ref{coro:SSR} later.
\begin{proposition}[Proposition 10.3.6 of Ref.~\cite{kadison1997fundamentals2}]\label{prop:orthogonal_sectors}
    Let $\psi$ and $\phi$ be disjoint states of $\A^{ql}$, then
    \beq
    ||\psi-\phi||=2
    \eeq
\end{proposition}
\begin{proof}
    Consider the representation $\pi:=\pi_{\psi}\oplus\pi_{\phi}$ on $\cH_{\psi}\oplus\cH_{\phi}$. In this block-diagonal form and use $\pi_{\psi}\not\simeq\pi_{\phi}$, we find
    \beq
    \pi(\A^{ql})'':=\begin{pmatrix}
        \cM_{\psi}&0\\0&\cM_{\phi}
    \end{pmatrix}
    \eeq
    The off-diagonal elements vanish since $\pi_{\psi}\not\simeq\pi_{\phi}$ by assumption.
    In particular, $\mathrm{diag}(1,-1)\in\pi(\A^{ql})''$. Using Theorem \ref{thm:Kaplansky_density}, there exists a sequence $a_{n}\in\A^{ql}$ with $||a_{n}||=1$ such that $\pi(a_{n})\to \mathrm{diag}(1,-1)$ in weak operator topology. In other words,
    \beq
    \psi(a_{n})-\phi(a_{n})=\la\psi|\pi(a_{n})|\psi\ra-\la\phi|\pi(a_{n})|\phi\ra\stackrel{n\to\infty}{\longrightarrow}2
    \eeq
\end{proof}
\begin{remark}
    This proof does not work if $\psi$ and $\phi$ are in the same superselection sector. In that case $\pi_{\psi}(A)=\pi_{\phi}(A)$ after a proper choice of basis, so we have
    \beq
    \pi(\A^{ql})'=M_{2}\otimes\pi_{\psi}(\A^{ql})'
    \eeq
    where $M_{2}$ is the algebra of $2\times 2$ matrices.
    And the double commutant 
    \beq
    \pi(\A^{ql})''=\id_{2\times 2}\otimes \cM_{\psi}
    \eeq
    So the above proof does not apply.
\end{remark}

Based on von Neumann algebras, the following proposition provides the intuition and a simpler way to characterize quasi-equivalent states. This proposition is used to prove Lemma \ref{lemma:typeI_KS} later.

\begin{proposition}[Def.~10.3.1 and corollary 10.3.4 of Ref. \cite{kadison1997fundamentals2}]\label{prop:quasi-eq}
    Two states $\psi$ and $\phi$ of $\A^{ql}$ are quasi-equivalent iff there exists a $*$-isomorphism $f:\cM_{\psi}\to\cM_{\phi}$ such that
    \beq\label{eq:intertwiner}
    f(\pi_{\psi}(a))=\pi_{\phi}(a),\quad\forall\,a\in\A^{ql}
    \eeq
\end{proposition}
\begin{remark}
    One should notice the importance of the map $f$ above. Let $\psi$ and $\phi$ be any two pure states of $\A^{ql}$, then one can show that $\cM_{\psi}\simeq\cM_{\phi}$ as abstract von Neumann algebras (since both of them are isomorphic to $\B(\cH)$ for some separable Hilbert space $\cH$!). Surely this does not mean that $\psi\simeq\phi$ due to the lack of $f$ satisfying Eq.~\eqref{eq:intertwiner}.
\end{remark}

There is also a proposition regarding disjoint states, which will be used to show that the GHZ state Eq.~\eqref{eq:mixed_GHZ} is not a factor state (see Definition~\ref{def:factor}) in Appendix~\ref{sec:clustering}.
\begin{proposition}[Proposition 8.19 of Ref.~\cite{Landsman:2017hpa}]\label{prop:disjoint_decomp}
    Let $\psi$ be a state and if $\psi=t\psi_{1}+(1-t)\psi_{0}$ for some $0<t<1$, then $\psi_{1}$ and $\psi_{0}$ are disjoint iff there is a projector $p\in \pi_{\psi}(\A^{ql})'\cap \cM_{\psi}$ such that
    \beq
    \begin{split}
        \pi_{\psi_{0}}&=\pi_{\psi}|_{p\cH_{\psi}},\\
        \pi_{\psi_{1}}&=\pi_{\psi}|_{(1-p)\cH_{\psi}}.
    \end{split}
    \eeq
    This means, a convex decomposition into disjoint states is always an independent decomposition in the sense of Theorem \ref{thm:independent_decomposition}.
\end{proposition}

\begin{remark}
    This proposition can be generalized to \textbf{countably many} convex combination of disjoint states. But it can fail for uncountable combinations (\eg the maximally mixed state $\la\cdot\ra_{\infty}$ defined in Eq.~\eqref{eq:tracial_state}).
\end{remark}

The following corollary justifies the intuition that disjoint states are in different superselection sectors because their GNS reference vector states cannot be related by local operators.
\begin{corollary}[Hepp's lemma, corollary 8.22 of \cite{Landsman:2017hpa}]\label{coro:Hepp's lemma}
    Let $\pi:\A^{ql}\to\B(\cH)$ be some representation of $\A^{ql}$. For two vector states $|\psi_1\ra$ and $|\psi_2\ra$, define the following two functional states,
    \beq
    \psi_{i}(A):=\la\psi_{i}|\pi(A)|\psi_{i}\ra, \quad|\psi_{i}\ra\in\cH,\,i=1,2
    \eeq
    Then functionals $\psi_1$ and $\psi_2$ are disjoint as states in $\A^{ql}$ iff
    \beq
    \la\psi_{1}|\pi(A)|\psi_{2}\ra=0,\quad\forall\,A\in\A^{ql}.
    \eeq
\end{corollary}
\begin{proof}
    Take $\psi=\frac{1}{2}(\psi_{1}+\psi_{2})$ and use Proposition \ref{prop:disjoint_decomp}.
\end{proof}

As a mathematical application of Theorem \ref{thm:vN_double_commutant} and for future reference in Lemma \ref{lemma:automoprhism_typeI}, we show that finite rank operators on a separable Hilbert space $\cH$ are dense in both the weak and strong operator topologies. Recall that an operator $a\in\B(\cH)$ is called finite rank if $\dim(\im(a))<\infty$.
\begin{lemma}\label{lemma:finite_rank}
    Finite rank operators are dense in both the weak and strong operator topologies, \ie any $x\in\B(\cH)$ is the weak (resp. strong) operator limit of a sequence $a_{n}$ with each $a_{n}$ having a finite rank.
\end{lemma}
The following proof is due to Ref. \cite{Argerami2023finiterank}.
\begin{proof}
    We denote $F\subseteq\B(\cH)$ the algebra of finite rank operators. By Theorem \ref{thm:vN_double_commutant}, it suffices to show $F'=\bbC$. Let $x\in F'$, then $x$ commutes with all rank-1 operators. Fix an orthonormal basis $\{|n\ra\}_{n=1,2,\dots}$ of $\cH$, then
    \beq\label{eq:commuting_finite_rank}
    x|n\ra\la n'|=|n\ra\la n'|x \Rightarrow \la n|x|n'\ra=\delta_{n,n'}\la n|x|n\ra
    \eeq
    which means that $x$ is diagonalized. Let us write $x|n\ra=\lambda_{n}|n\ra$, from Eq.~\eqref{eq:commuting_finite_rank},
    \beq
    \lambda_{n}=\lambda_{n'},\quad\forall n,n'\Rightarrow x=\lambda 1_{\cH},
    \eeq
    which means $F'=\bbC$.
\end{proof}

{\subsection{Clustering and factors}\label{sec:clustering}

To get prepared to prove our main theorems, it is useful to discuss some correlation and entanglement properties of states, which are the subject of this and the next two subsections.

In Appendix \ref{sec:vN}, we have introduced von Neuamnn algebra as a technical tool to discuss superselection sectors. In this section, we will explore the physical meaning of von Neumann algebras themselves. In particular, we will see that the clustering property, \ie quantum correlations decay with distance, can be derived as a natural consequence of the properties of von Neumann algebras. We will also discuss more details of the entanglement properties.

Before diving into those physical applications, let us motivate the question from a pure mathematical point of view: It is natural to ask if a given von Neumann algebra can be decomposed into smaller pieces, just like a representation of a finite group can be decomposed into irreducible representations. An “irreducible” von Neumann algebra is called a factor.
\begin{definition}\label{def:factor}
    A von Neumann algebra $\cM\subseteq\B(\cH)$ is called a factor if
    \beq
    \cM\cap\cM'\simeq\bbC\cdot\id_{\cH}
    \eeq
    Let $\psi$ be a state of $\A^{ql}$, we call $\psi$ a factor state (or primary state) if $\cM_{\psi}$ is a factor.
\end{definition}

Obviously, pure states are factor states because $\cM'\simeq\bbC\cdot\id_{\cH}\subseteq\cM$. But as we will see later, there are also mixed factor states. 

Certainly, the physical meaning of $\cM_{\psi}\cap\cM_{\psi}'$ and the importance of factor states need an explanation. It turns out that they are related to the so-called clustering property.

\begin{definition}[Clustering] \label{def: clustering}
    A state $\psi$ of $\A^{ql}$ is said to be clustering, if for any $a\in\A^{\ell}$ and $\epsilon>0$, there exists a finite $\Gamma_{\epsilon,a}\subseteq\Lambda$, such that
    \beq \label{eq: defining clustering}
    |\psi(ab)-\psi(a)\psi(b)|\leqslant \epsilon ||a||\cdot||b||,\quad\forall\,b\in \A^{ql}_{\Gamma^{c}_{\epsilon,a}}.
    \eeq
\end{definition}
The clustering property is widely used and serves as a locality condition in many areas of physics. However, we note that not all states are clustering.
\begin{example}
    The infinite-size GHZ state $\psi_{GHZ}$ defined in Eq.~\eqref{eq:mixed_GHZ} is \textbf{not} clustering, since
    \beq
    \begin{split}
        \psi_{GHZ}(Z_{k})&=0,\quad\forall k\in\Lambda\\
        \psi_{GHZ}(Z_{i}Z_{j})&=1,\quad\forall i,j\in\Lambda
    \end{split}
    \eeq
    On the contrary, it is shown in Ref.~\cite{Nachtergaele_2006} that locally unique gapped ground states of Hamiltonians (with suitable locality) satisfy the clustering property.
\end{example}
Actually, the necessary and sufficient condition for a state to be clustering is known.
\begin{theorem}[Theorem 2.6.10 of \cite{bratteli2013operator1}]\label{thm:clustering}
    A state $\psi$ on $\A^{ql}$ is clustering iff it is a factor state.
\end{theorem}
We refer the readers to Ref.~\cite{bratteli2013operator1} for the proof of this profound result. For example, $\psi_{GHZ}$ is not a factor state due to Proposition \ref{prop:disjoint_decomp}, so the projector $p\in\cM_{\psi}\cap\cM_{\psi}'$ is nontrivial.

\begin{remark}
    To understand the definition of von Neumann factors heuristically, we remark that $\cM_{\psi}\cap\cM_{\psi}'$ has the meaning of ``algebra at infinity" \cite{bratteli2013operator1,Landsman:2017hpa}. To see this, consider a local operator $a$ and we define $a_{n}:=\pi_{\psi}(\tau^{n}(a))$, where $\tau$ is the lattice translation (by one unit cell). This limit does not exist in the norm topology, but it does exist in the weak operator topology due to Remark \ref{remark:BBA}, thus we denote this “operator at infinity” by $a_{\infty}$. By Theorem \ref{thm:vN_double_commutant}, $a_{\infty}\in\cM_{\psi}$. On the other hand, $a_{\infty}$ commutes with every local operator by construction because it is at the infinity, hence it lies in $\cM_{\psi}'$. Therefore, $\cM_{\psi}\cap\cM_{\psi}'$ is the algebra generated by “operators at infinity”. From this perspective, an intuitive explanation of Theorem \ref{thm:clustering} follows. Since clustering requires the connected two-point correlation function of two local operators far away to be vanishingly small, and we can take one of the two operators to be in the ``algebra at infinity". The assumption $\cM\cap\cM'\simeq \bbC\cdot\id_{\cH}$ means that these operators at inifinity are just $c$-numbers, so the corresponding connected two-point correlation function indeed vanishes and the clustering property holds.
\end{remark}

In some sense, classifying factor states is equivalent to classifying their associated von Neumann algebras. We begin with the relation between factor states.
\begin{lemma}[Corollary 8.22 of Ref. \cite{Landsman:2017hpa}]\label{lemma:factor_states}
    \begin{enumerate}
        \item Two factor states are either disjoint or quasi-equivalent.
        \item A state $\psi$ is factor state iff it has no (countable) convex decomposition into disjoint states.
    \end{enumerate}
\end{lemma}

As a direct consequence of this lemma and Theorem \ref{thm:clustering}, a countable convex combination of disjoint states will never be clustering. For example, the GHZ state is a convex combination of two disjoint states, and it is indeed not clustering. It is important to note that this does not forbid an uncountable combination (\ie an integral) of disjoint states to be a factor state (hence clustering). This can be seen from the following example.

\begin{example}[Tracial state]\label{ex:tracial_state}
    On $\A^{ql}$, there is a unique tracial state which we denoted by $\la\cdot\ra_{\infty}$, which is a normalized (usual) trace. It is uniquely characterized by 
    \beq\label{eq:tracial_state}
    \begin{split}
        \la ab \ra_{\infty}&=\la ba\ra_{\infty},\quad\forall \,a,b\in\A^{ql}\\
        \la \id\ra_{\infty}&=1
    \end{split}
    \eeq
    In more physical terms, the tracial state is the maximally mixed state on $\A^{ql}$ which can be obtained by taking tensor product of maximally mixed states of each site, and it is also the infinite-temperature state. Clearly, 
    \beq
    \la ab\ra_{\infty}=\la a\ra_{\infty}\la b\ra_{\infty}
    \eeq
    when $a, b\in\A^{ql}$ and $a$ and $b$ are supported in non-overlapping regions.Therefore $\la\cdot\ra_{\infty}$ is also a factor state.

    Intuitively, $\la\cdot\ra_{\infty}$ is maximally mixed so it mixes all pure states from all different superselection sectors. Naively, from lemma \ref{lemma:factor_states} one should conclude $\la\cdot\ra_{\infty}$ can never be a factor state since it involves disjoint states. 
    However, $\la\cdot\ra_{\infty}$ is highly mixed such that it does not admit any countable convex decomposition into pure states. So in fact there is no such contradiction.
\end{example}
    As a lesson from the above examples, we see usually a factor state is either not quite mixed (\eg it is a countable convex combination of equivalent pure states) or highly mixed (\eg tracial state).

    The classification of factor states and von Neumann factors are intensively studied during the last century (see appendix C of Ref.~\cite{Landsman:2017hpa} for an exposition). Although a complete classification is still not available, the situation becomes much better understood for the so-called approximately finite dimensional (AFD) factors, which is exactly the case we need.
    \begin{definition}
        A von Neumann algebra $\cM\subseteq\B(\cH)$ is called approximately finite dimensional (AFD), injective or amenable, if there is a sequence of finite dimensional $*$-subalgebras $\{\cM_{n}\}_{n=1,2,\dots}$
        \beq
        \cM_{1}\subseteq\cM_{2}\subseteq\dots\subseteq\cM_{n}\subseteq\dots \subseteq\cM
        \eeq
        such that $\cM=(\bigcup_{n=1}^{\infty}\cM_{n})''$.
    \end{definition}
    The von Neumann algebras we will work with are $\cM_{\psi}$ for some state $\psi$ on $\A^{ql}$. Since $\A^{ql}$ admits a sequence of finite dimensional subalgebra whose union is dense in $\A^{ql}$, the von Neumann algebra $\cM_{\psi}$ is automatically AFD.

    To discuss the classification of von Neumann factors, let us define the notion of trace on von Neumann algebras.
    \begin{definition}
        Let 
        \beq
        \cM_{+}:=\{a\in\cM|\exists b\in\cM,\,a=bb^{*}\}
        \eeq
        be the set of positive operators in $\cM$. Then a trace on $\cM$ is a map
        \beq
        \tr:\cM_{+}\to[0,\infty]
        \eeq
        (its value on $\cM$ can be obtained by linear extension) satisfying:
        \begin{enumerate}
            \item $\tr(\lambda a+b)=\lambda\tr(a)+\tr(b)$ for any $a,b\in\cM_{+}$ and $\lambda\geqslant 0$.
            \item $\tr(u au^{*})=\tr(a),\quad\forall a\in\cM$ and any unitary $u\in\cM$.
        \end{enumerate}
        We say that $\tr$ is finite if $\tr(a)<\infty$ for all $a\in\cM_{+}
        $, or semifinite if it is not finite and for any $b\in\cM_{+}$ there exists $a\in\cM_{+}
        $ such that $b-a\in\cM_{+}$ and $\tr(a)<\infty$. A trace is called infinite if it is not finite or semifinite.
    \end{definition}
    We now introduce the type classification due to von Neuamnn and Murray \cite{Landsman:2017hpa}. 
    
    \begin{definition}
        An AFD factor $\cM$ is said to be of 
        \begin{enumerate}
            \item type-I$_{n}$ if $\cM\simeq M_{n}$ (\ie the algebra of $n\times n$ matrices) of some $n\in\z_{>0}$.
            \item type-I$_{\infty}$ if $\cM\simeq\B(\cH)$ for some (separable) infinite-dimensional Hilbert space $\cH$.
            \item type-II$_{1}$ if $\cM$ is not type-I and admits a finite trace.
            \item type-II$_{\infty}$ if $\cM$ is not type-I and admits a semifinite trace.
            \item type-III if all nonzero traces are infinite.
        \end{enumerate}
        The type of a factor state $\psi$ on $\A^{ql}$ is defined to be the type of $\cM_{\psi}$.

        Type-I$_n$ and type-I$_\infty$ factors together are called type-$I$ factors, while type-II$_1$ and type-II$_\infty$ factors together are called type-II factors.
    \end{definition}
    \begin{example}\label{ex:factors}
        Obviously, any finite-dimensional quantum system is described by type-I$_{n}$ factors for some $n$ and hence they are not interesting to us here. Let $\psi$ be a pure state of $\A^{ql}$, we know that $\cM_{\psi}=\B(\cH_{\psi})$ and hence it is of type-I$_{\infty}$ since $\cH_{\psi}$ is infinite dimensional.

        The tracial state is of type-II$_{1}$ since its GNS Hilbert space is manifestly infinite dimensional and it admits a finite trace (given by tracial state itself). We denote its associated type-II$_{1}$ von Neumann factor by $\cM_{\infty}$.
        A type-II$_{\infty}$ von Neumann factor $\cM$ can be obtained from $\cM_{\infty}\otimes \cM_{\psi}$, where $\psi$ is a pure state.

        Type-III factors are more mysterious and many physically interesting mixed states are of this type. For example, a finite temperature state (a.k.a Kubo-Martin-Schwinger state in $C^*$-algebra context) is typically of type-III (see Sec.~5.4 of Ref.~\cite{bratteli2013operator2}).
        
        Another example of type-III factor is given by the so-called Powers factor (see Sec. III.3.1.7 of Ref.~\cite{Blackadar:2006OA}). Consider an infinite qubit system. For qubit at site $i\in\Lambda$, we define the following mixed state
        \beq
        \rho_{i,\lambda}(A)=\frac{1}{1+\lambda}A_{11}+\frac{\lambda}{1+\lambda}A_{22},\quad\,0<\lambda<1,A\in \A_{i}
        \eeq
        The Powers factor state $\rho_{\lambda}$ is the product state $\bigotimes_{i\in\Lambda}\rho_{i,\lambda}$. When $\lambda=1$ this state is maximally mixed and hence reduces to the tracial state. However, for $0<\lambda<1$, the state is of type-III.
    \end{example}
    
    The above examples indicates that type-II and type-III factors are far more complicated and more entangled than type-I factors. Actually the following result completely determines the structure of type-I factor states.
    \begin{lemma}\label{lemma:typeI}
        Let $\psi$ be a type-I factor state of $\A^{ql}$, then $\psi$ is a countable convex sum,
        \beq
        \psi=\sum_{i=1}^{\infty}\lambda_{i}\psi_{i},\quad 0\leqslant\lambda_{i}\leqslant 1,\,\sum_{i=1}^{\infty}\lambda_{i}=1,
        \eeq
        where each $\psi_{i}$ is a pure state such that $\psi_{i}\simeq\psi_{j}$ but $\psi_{i}\not=\psi_{j}$.
    \end{lemma}
    \begin{proof}
        By the definition of type-I factors, we know $\pi_{\psi}(\A^{ql})\subseteq\cM_{\psi}\simeq\B(\cH)$ for some separable Hilbert space $\cH$. On the other hand, $\B(\cH)\simeq\cM_{\psi}\subseteq\B(\cH_{\psi})$, which means there is a tensor factorization $\cH_{\psi}\simeq\cH\otimes\cH'$ with $\cH'$ some Hilbert space. Note $\pi_{\psi}(\A^{ql})\subseteq\B(\cH)$, so $\A^{ql}$ only acts on $\cH$ but not on $\cH'$.
        We use the Schmidt decomposition for the GNS vector $|\psi\ra$ (see Theorem 4.4.5 of Ref.~\cite{hsing2015theoretical}), then
        \beq
        |\psi\ra=\sum_{n=1}^{\infty}\sqrt{\lambda_{i}}|\psi_{i}\ra\otimes|\xi_{i}\ra,\quad \sum_{i=1}^{\infty}\lambda_{i}=1&\lambda_{i}\geqslant0
        \eeq
        where $|\psi_{i}\ra$ and $|\xi_{i}\ra$ are orthonormal basis of $\cH$ and $\cH'$ respectively. Note that $\pi_{\psi}(\A^{ql})$ acts irreducibly on $\cH$, because $\pi_\psi(\A^{ql})'=\cM_\psi'=(\B(\cH))'=\bbC$.\footnote{Here the commutant is taken in $\B(\cH)$ instead of $\B(\cH_{\psi})$.} Therefore, each $\psi_i(A):=\la\psi_i|\pi_\psi(A)|\psi_{i}\ra$ is a pure state of $\A^{ql}$, because $(\cH, \pi_\psi, |\psi_i\ra)$ can be viewed as a GNS triple associated with $\psi_i$. Thus, we obtain
        \beq
        \begin{split}
            \psi&=\sum_{i=1}^{\infty}\lambda_{i}\psi_{i}.
        \end{split}
        \eeq
        Besides, it is easy to see for each $i$, $\psi_{i}$ must be equivalent from Lemma \ref{lemma:factor_states}.
    \end{proof}
    Before ending this section, let us mention some important properties of factors (without proof).
    \begin{lemma}[Theorem 9.13 of \cite{kadison1997fundamentals2}]
        \begin{enumerate}
        The following results are true.
            \item Let $\psi$ and $\phi$ be two quasi-equivalent factor states, then they are of the same type.
            \item Let $\cM\subseteq\B(\cH)$ be a von Neumann factor, then $\cM'$ is again a factor with the same type.
        \end{enumerate}
    \end{lemma}
    
    In general, the structure of $*$-automorphisms of a $C^{*}$-algebra can be highly complicated. However, the situation significantly simplifies when it is a type-I factor. The following lemma demonstrates this, and it will be useful in proving Lemma \ref{lemma:typeI_KS}.
    
    \begin{lemma}[example II.5.5.14 of Ref.~\cite{Blackadar:2006OA}]\label{lemma:automoprhism_typeI}
        Let $\cM$ be a factor of type-I, then all $*$-automorphisms of $\cM$ are inner, \ie they are all given by $\Ad_{U}$ for some unitary operator $U\in\cM$.
    \end{lemma}
    The following elementary proof is due to Ref. \cite{Argerami2017automorphism}.
    \begin{proof}
        We assume $\cM\simeq\B(\cH)$ and $\dim\cH=\infty$ below. The proof for the type-I$_{n}$ case is the same.
        Let $\theta:\B(\cH)\to\B(\cH)$ be a $*$-automorphism and pick up an orthonormal basis $\{|n\ra\}_{n=1,2,\dots}$ of $\cH$. Consider the projector $P_{n}:=|n\ra\la n|$. It is clear that $\theta(P_{1})$ also has rank 1. Now consider a normalized state $|\eta_{1}\ra\in\im(\theta(P_{1}))$, we then define an operator $U$ by 
        \beq
        U|n\ra=\theta(|n\ra\la 1|)|\eta_{1}\ra
        \eeq
        Below we check $U$ is indeed unitary and $\theta=\Ad_{U}$.

        Notice that
        \beq
        \la n'|U^{\dagger}U|n\ra=\la\eta_{1}|\theta(\la n'|n\ra P_{1})|\eta_{1}\ra=\delta_{n,n'}
        \eeq
        This means $U$ is a unitary operator. 
        
        By Lemma \ref{lemma:finite_rank}, finite rank operators are dense in $\B(\cH)$ with the weak operator topology. Below we show $\theta=\Ad_{U}$ for finite-rank operators first and then we extend this equality to $\B(\cH)$ by continuity{\footnote{$\theta$ is a $*$-automorphism and it is hence continuous in norm topology of $\B(\cH)$. Since weak operator topology is weaker than norm topology, we conclude $\theta$ is also continuous in weak operator topology.}}. To see $\theta=\Ad_{U}$ on finite-rank operators, it suffices to show $\theta=\Ad_{U}$ on rank 1 operators, \ie on $|n\ra\la n'|$,
        \beq
        \begin{split}
            U|n\ra\la n'|U^{\dagger}&=\theta(|n\ra\la 1|)|\eta_{1}\ra\la\eta_{1}|\theta(|1\ra\la n|)\\
            &=\theta(|n\ra\la 1|) \theta(|1\ra\la 1|)\theta(|1\ra\la n'|)\\
            &=\theta(|n\ra\la n'|)
        \end{split}
        \eeq
    In summary, $\theta=\Ad_{U}$ on $\B(\cH)$.
    \end{proof}
    
    This lemma will be used to show Lemma \ref{lemma:typeI_KS} later.

\subsection{Split property and area law}\label{sec:split}

In the previous subsection, we have discussed the correlation properties of states. In this and the next subsections, we will discuss entanglement properties.

Given a subset $\Gamma$ (finite or infinite) of our infinite lattice $\Lambda$, the algebra of quasi-local operators always decomposes as
\beq
\A^{ql}\simeq \A_{\Gamma}\otimes\A_{\Gamma^{c}}
\eeq
Recall that $\A_{\Gamma}$ is the quasi-local operator algebra supported on $\Gamma$ and $\Gamma^{c}:=\Lambda\setminus\Gamma$ is complement of $\Gamma$. Since $\Gamma$ is not necessarily finite, operators in $\A_{\Gamma}$ are not necessarily local. Nevertheless, if $\Gamma$ is finite, $\A_{\Gamma}$ is the local operator algebra supported on $\Gamma$ and hence it is finite dimensional.

Before we proceed, let us define the tensor product of states.
\begin{definition}
    For any two $C^*$-algebra $\cC_{1},\cC_{2}$ and two states $\psi_{i}:\cC_{i}\to\bbC,i=1,2$, the tensor product of states $\psi_{1}\otimes\psi_{2}$ is a state of $\cC_{1}\otimes\cC_{2}$\footnote{We suppress the subtleties in the tensor product of $C^*$-algebras (see Sec 3.2.2 of Ref. \cite{Naaijkens_2017} for details).}, defined as
    \beq
    (\psi_{1}\otimes\psi_{2})(A_{1}\otimes A_{2}):=\psi_{1}(A_{1})\psi_{2}(A_{2})
    \eeq
    for any $A_{i}\in\cC_{i}$.
\end{definition}

Given a pure state $\psi$ of $\A^{ql}$ and the above decomposition $\A^{ql}\simeq \A_{\Gamma}\otimes\A_{\Gamma^{c}}$, we may not be able to factorize $\psi$ into the form $\psi_{\Gamma}\otimes\psi_{\Gamma^{c}}$ with some pure states $\psi_{\Gamma}:\A^{ql}_{\Gamma}\to \bbC$ and $\psi_{\Gamma^{c}}:\A^{ql}_{\Gamma^{c}}\to \bbC$. In the case of 1D system and $\Gamma=(-\infty,0)$, we have the notion of the split property\footnote{The split property in two dimensions is also proposed in Ref. \cite{Naaijkens2022split}, but we do not need it in the present paper.}.
\begin{definition}\label{definition:split_property}
    In a 1D system with the following decomposition
    \beq
    \A^{ql}\simeq \A_{<0}\otimes\A_{\geqslant 0}
    \eeq
    where $\A_{<0}$ is a subalgebra of quasi-local operators which are supported on $\Gamma=(-\infty,0)$ only and similarly for $\A_{\geqslant 0}$.
    A pure state $\psi$ is said to split at the origin if
    \beq\label{eq:splits_0}
    \psi\simeq\psi_{<0}\otimes\psi_{\geqslant 0}
    \eeq
    where $\psi_{<0}$ (resp. $\psi_{\geqslant0}$) is a pure state of $\A_{<0}$ (resp. $\A_{\geqslant 0}$) and the equivalence relation ``$\simeq$" is in the sense of definition \ref{def:inequivalence_GNS}. More generally, a mixed state $\psi$ splits if there exist states $\psi_{\geqslant0}:\A_{\geqslant 0}\to\bbC$ and $\psi_{<0}:\A_{<0}\to\bbC$, such that
    \beq
    \psi\sim\psi_{<0}\otimes\psi_{\geqslant0}
    \eeq
    where the quasi-equivalence ``$\sim$" is in the sense of 
    definition \ref{def:quasi_disjoint}.
\end{definition}

Note that we only define the split property for {\it infinite} chains. Also, the split property above is defined with respect to a particular site on the chain, such as the origin. One may wonder if a state splits at a site, whether it necessarily splits at other sites. In Corollary \ref{coro:split_everywhere}, we show that this is indeed the case, so we do not have to refer to the site with respect to which a state splits.

In quantum mechanics, one can tensor different states to get a state of the composite system. Conversely, one can do the partial trace operation to reduce the degrees of freedom. What does partial trace correspond to in the operator algebra formalism? Suppose $\Gamma\subset\Lambda$ is a subset (it can be finite or infinite) and given a state $\psi:\A^{ql}\to\bbC$, note that $\A_{\Gamma}$ can be viewed as a subalgebra of $\A^{ql}$, hence one can get a state $\psi|_{\Gamma}$ of $\A_{\Gamma}$ by the restriction
\beq
\psi|_{\Gamma}(A)=\psi(A),\forall\,A\in\A_{\Gamma}
\eeq
Here $\psi|_\Gamma$ can be viewed as the state obtained from $\psi$ by partially tracing the degrees of freedom in $\Gamma^c$. Please do not confuse $\psi|_{\Gamma}$ with $\psi_{\Gamma}$ in the definition of the split property, \ie even if the state $\psi$ splits between $\Gamma$ and $\Gamma^c$, $\psi_{\Gamma}$ may {\it not} be obtained by restricting $\psi$ to $\Gamma$ except for some exceptionally special examples. In particular, states obtained by restriction (partial trace), \eg $\psi|_{\Gamma}$, are not necessarily pure even if $\psi$ is pure, but $\psi_\Gamma$ is pure by definition for a pure $\psi$.

To give an example of states with the split property, let us define \textit{product states}.
\begin{definition}\label{def:product_state}
    For any lattice system (in any spatial dimension), a state $\omega$ is said to be factorized or a product state if 
    \beq
    \omega(AB)=\omega(A)\omega(B)
    \eeq
    whenever $A$ and $B$ have disjoint supports. Besides, in the present paper, product states are always assumed to be pure unless otherwise specified.
\end{definition}

As we will see below, a product state is arguably the simplest state in quantum many-body physics, since it means that spins at different site are not entangled with each other at all. A state that is not a product state is referred to as an entangled state.

\begin{lemma}\label{lemma:product_split}
    Let $\Gamma\subset\Lambda$ (finite or infinite and $\Lambda$ is any lattice) and $\omega$ be a (pure) product state, we always have
    \beq\label{eq:product_split}
    \omega=\omega|_{\Gamma}\otimes\omega|_{\Gamma^{c}}
    \eeq
    Moreover, $\omega|_{\Gamma}$ and $\omega|_{\Gamma^{c}}$ are again pure product states.
\end{lemma}

Specializing to 1D, the above lemma shows that all product states split.
\begin{proof}[Proof of Lemma \ref{lemma:product_split}]
    For any $A\in\A_{\Gamma},B\in\A_{\Gamma^{c}}$, we have
    \beq
    \omega(AB)=\omega(A)\omega(B)=\omega|_{\Gamma}(A)\omega|_{\Gamma^{c}}(B)
    \eeq
    Then Eq.~\eqref{eq:product_split} follows. To show that $\omega|_{\Gamma}$ is pure, let $\rho:\A_{\Gamma}\to\bbC$ be a positive linear functional majorized by $\omega|_{\Gamma}$, \ie $\rho\leqslant\omega|_{\Gamma}$ (see definition~\ref{def:alternative_pure}). We then have
    \beq
    \rho\otimes\omega|_{\Gamma^{c}}\leqslant\omega|_{\Gamma}\otimes\omega|_{\Gamma^{c}}=\omega
    \eeq
    By assumption, $\omega$ is a pure state, hence $\rho\otimes \omega|_{\Gamma^{c}}=\lambda\omega=\lambda\omega|_{\Gamma}\otimes\omega|_{\Gamma^{c}}$ for some $\lambda\in [0,1]$. Thus for any $A\in\A_{\Gamma}$, we have
    \beq
    \rho(A)=(\rho\otimes\omega|_{\Gamma^{c}})(A\otimes I)=\lambda\omega(A\otimes I)=\lambda\omega|_{\Gamma}(A)
    \eeq
    Hence $\omega|_{\Gamma}$ is pure by definition \ref{def:alternative_pure}.

    To see that $\omega|_{\Gamma}$ is a product state, let $A,B\in\A_{\Gamma}$ be quasi-local operators with disjoint support. Then
    \beq
    \omega|_{\Gamma}(AB)=\omega(AB)=\omega(A)\omega(B)=\omega|_{\Gamma}(A)\omega|_{\Gamma}(B)
    \eeq
    Thus $\omega|_{\Gamma}$ is indeed a pure product state. So is $\omega|_{\Gamma^{c}}$ for similar reasons.
\end{proof}

\begin{remark}
     However, in general, there are non-product states which satisfy the split property. For example, as shown in Ref.~\cite{Kapustin2020invertible} (see Lemma 4.2 therein), short-range entangled states split, but they are generically not product states.
\end{remark}

In the special case where $\Gamma$ is finite, $\psi|_{\Gamma}$ is a state of $\A_{\Gamma}\simeq \otimes_{k\in\Gamma}\B(\cH_{k})$, \ie the quasi-local algebra is finite dimensional. Therefore, one can represent $\psi|_{\Gamma}$ by a density matrix $\rho_{\psi}$, as we have seen in example \ref{example:denstiy_matrix}. 
\begin{definition}\label{def:area_law}
    Let $\Gamma$ be a finite region as above and a density matrix $\rho_{\psi}$ represent a state $\psi$ on $\Gamma$ (\ie $\psi|_{\Gamma}$), then one defines the entanglement entropy of $\psi|_\Gamma$ as
    \beq
    S(\rho_{\psi})=-\tr(\rho_{\psi}\log\rho_{\psi})
    \eeq
    A state $\psi$ is said to satisfy the area law if
    \beq
    S=O(|\partial\Gamma|)
    \eeq
    for any finite region $\Gamma$. Here $\partial\Gamma$ means the boundary of $\Gamma$.
\end{definition}

\begin{remark}
    A more sophisticated way to define entropy without restricting the quantum states to a finite region is possible. They require more tools from the theory of von Neumann algebras and we will discuss these definitions in Appendix~\ref{sec:entropyI}.
\end{remark}

Especially, in 1D, the area law means that the entanglement entropy of a connected region is upper bounded by a constant that is independent of the size of this region. The following lemma is of fundamental importance.

\begin{lemma}[Theorem 1.5 of Ref. \cite{matsui2011boundedness}]\label{thm:area_law_split}
    In 1D spin chains, a factor state $\psi$ splits if it satisfies the area law.
\end{lemma}
Readers are referred to Sec. 2 of Ref. \cite{matsui2011boundedness} for a proof of Lemma \ref{thm:area_law_split}.

The lemma below characterizes the “uniqueness" of the factorization of states.
\begin{lemma}\label{lemma:separate_equivalence}
    Given a decomposition $\A^{ql}\simeq \A_{\Gamma}\otimes\A_{\Gamma^{c}}$, where $\Gamma$ is a (finite or infinite) subregion of the lattice, and given pure states $\psi_{\Gamma},\psi_{\Gamma}'$ (reps. $\psi_{\Gamma^{c}},\psi_{\Gamma^{c}}'$) of $\A_{\Gamma}$ (resp. $\A_{\Gamma^{c}}$), if 
    \beq
    \psi_{\Gamma}\otimes\psi_{\Gamma^{c}}\simeq \psi_{\Gamma}'\otimes\psi_{\Gamma^{c}}'
    \eeq
    Then $\psi_{\Gamma}\simeq \psi_{\Gamma}'$ and $\psi_{\Gamma^{c}}\simeq\psi_{\Gamma^{c}}'$. Similar results hold if these states are not necessarily pure and one replaces unitary equivalence $\simeq$ by quasi-equivalence $\sim$.
\end{lemma}

\begin{proof}
    We use Proposition \ref{prop:equivalence_of_states} or Proposition \ref{prop:mixed_states_equivalence} for general states. By assumption, for any $\epsilon>0$, there exists a finite region $P_{\epsilon}$ such that for any operator $S\in \A_{P^{c}_{\epsilon}}$, we have
    \beq
    |(\psi_{\Gamma}\otimes\psi_{\Gamma^{c}})(S)-(\psi'_{\Gamma}\otimes\psi'_{\Gamma^{c}})(S)|<\epsilon||S||
    \eeq
    Now we choose $S=A\otimes B$ for any $A\in \A_{\Gamma\setminus P_{\epsilon}}$ and $B\in \A_{\Gamma^{c}\setminus P_{\epsilon}}$, we have
    \beq
    |\psi_{\Gamma}(A)\psi_{\Gamma^{c}}(B)-\psi_{\Gamma}'(A)\psi_{\Gamma^{c}}'(B)|<\epsilon||A||\cdot||B||
    \eeq
    Now let $B=I$, hence $\psi_{\Gamma^{c}}(B)=\psi'_{\Gamma^{c}}(B)=1$,
    \beq
    |\psi_{\Gamma}(A)-\psi'_{\Gamma}(A)|<\epsilon||A||
    \eeq
    So the lemma is proved by using Proposition \ref{prop:equivalence_of_states} or Proposition \ref{prop:mixed_states_equivalence}.
\end{proof}

It is worth pointing out that the split property is only nontrivial if the lattice is divided into two infinite subregions. In fact, we have
\begin{proposition}\label{prop:factorize}
    Let $\Gamma$ be a \textbf{finite} subset of the lattice $\Lambda$ and $\psi$ be a pure state of $\A^{ql}$, then there exists pure states $\phi_{\Gamma}:\A_{\Gamma}\to\bbC$ and $\phi_{\Gamma^{c}}:\A_{\Gamma^{c}}\to\bbC$, such that
    \beq
    \psi\simeq \phi_{\Gamma}\otimes\phi_{\Gamma^{c}}.
    \eeq
    If $\psi$ is mixed, then there exist (possibly mixed) states $\phi_\Gamma: \A_\Gamma\to\bbC$ and $\phi_{\Gamma^c}: \A_{\Gamma^c}\to\bbC$, such that
    \beq
        \psi\sim\phi_\Gamma\otimes\phi_{\Gamma^c}.
    \eeq
    
\end{proposition}
Namely, a pure state can always be ``approximately" factorized with respect to the partition $\Gamma\bigsqcup\Gamma^{c}$ as long as $|\Gamma|<\infty$.

To prove the above proposition, we need the following lemma.
\begin{lemma}\label{lemma:factorization_of_GNS}
    Suppose $\psi:\A^{ql}\to\bbC$ is a pure state and $\Gamma$ is a finite subset of $\Lambda$. Let $\cH_{\Gamma}:=\bigotimes_{k\in\Lambda}\cH_{k}$ be the Hilbert space of $\Gamma$ obtained by tensoring the local Hilbert spaces, then there exists another Hilbert space $\cH_{\Gamma^{c}}$ such that
    \beq
    \cH_{\psi}\simeq \cH_{\Gamma}\otimes \cH_{\Gamma^{c}}
    \eeq
    where $\cH_{\psi}$ is the GNS Hilbert space associated to $\psi$. Moreover, $\cH_{\Gamma^{c}}$ carries an irreducible representation of $\A_{\Gamma^{c}}$.
\end{lemma}
\begin{proof}[Proof of Lemma \ref{lemma:factorization_of_GNS}]
    Let $\pi_{\psi}:\A^{ql}\to\B(\cH_{\psi})$ be the GNS representation of $\psi$. Note that $\cH_{\Gamma}$ is finite-dimensional, and we denote $\dim(\cH_{\Gamma})=d$.

    Consider the restriction $\pi_{\psi}|_{\A_{\Gamma}}$, which gives a representation of $\A_{\Gamma}$. Now this representation decomposes into a direct sum of irreducible representations of $\A_{\Gamma}$, so we have
    \beq
    \cH_\psi\simeq \bigoplus_{i}V_{i}
    \eeq
    where each $V_{i}$ carries an irreducible representation of $\A_{\Gamma}$. By Proposition \ref{prop:fd_C*_rep}, we have $V_{i}\simeq V_{j}\simeq\cH_{\Gamma}$ for all $i,j$. So $\cH_\psi\simeq\cH_{\Gamma}\oplus\cH_{\Gamma}\oplus\cdots$. Hence there exists another Hilbert space $\cH_{\Gamma^{c}}$, such that
    \beq
    \cH_\psi\simeq \cH_{\Gamma}\otimes\cH_{\Gamma^{c}}
    \eeq
    More explicitly, $\cH_{\Gamma^{c}}$ can be obtained as follows. Fixing any rank-1 orthogonal projector $P\in\B(\cH_{\Gamma})$, we can choose $\cH_{\Gamma^{c}}=P(\cH_\psi)$.

    It remains to show that $\cH_{\Gamma^{c}}$ is irreducible as a representation of $\A_{\Gamma^{c}}$. If it is reducible, $\cH_{\psi}$ would be reducible as a representation of $\A^{ql}=\A_{\Gamma}\otimes\A_{\Gamma^{c}}$. This is impossible because $\psi$ is pure.
\end{proof}

We remark that $\cH_{\Gamma^{c}}$ is not obtained by tensoring all local Hilbert spaces in $\Gamma^{c}$, since this infinite tensor product of Hilbert spaces does not make sense. Besides we should keep in mind that $\cH_{\Gamma^{c}}$ depends on the state $\psi$. We also remark if $\psi$ is mixed, the factorized structure $\cH_{\psi}\simeq \cH_{\Gamma}\otimes \cH_{\Gamma^{c}}$ still holds, but $\cH_{\Gamma^c}$ may carry a reducible representation of $\A_{\Gamma^c}$.

Now we turn to the proof of Proposition \ref{prop:factorize}.

\begin{proof}[Proof of Proposition \ref{prop:factorize}]
    We first assume that $\psi$ is pure. Consider the GNS Hilbert space $\cH_{\psi}$ of $\psi$ and the reference state $|\psi\ra$. By Lemma \ref{lemma:factorization_of_GNS} we always have
    \beq \label{eq: tensor product Hilbert space}
    \cH_{\psi}\simeq \cH_{\Gamma}\otimes\cH_{\Gamma^{c}}
    \eeq
    Then we choose any vector states $|\phi_{\Gamma}\ra\in\cH_{\Gamma}$ and $|\phi_{\Gamma^{c}}\ra\in\cH_{\Gamma^{c}}$. Note that $|\phi_{\Gamma}\ra$ gives rise to a linear functional $\phi_{\Gamma}:\A_{\Gamma}\to\bbC$. And similarly we get $\phi_{\Gamma^{c}}:\A_{\Gamma^{c}}\to\bbC$. Since $\cH_{\Gamma}$ (resp. $\cH_{\Gamma^{c}}$) is an irreducible representation of $\A_{\Gamma}$ (resp. $\A_{\Gamma^{c}}$), $\phi_{\Gamma}$ (resp. $\phi_{\Gamma^{c}}$) is a pure state. Since $|\phi_{\Gamma}\ra\otimes|\phi_{\Gamma^{c}}\ra\in\cH_{\psi}$ by construction, $|\phi_{\Gamma}\ra\otimes|\phi_{\Gamma^{c}}\ra$ can be obtained from $|\psi\ra$ by applying to the latter finitely many quasi-local operators. Therefore, the third condition in Proposition \ref{prop:equivalence_of_states} holds and $\psi\simeq\psi_\Gamma\otimes\psi_{\Gamma_c}$.

    If $\psi$ is mixed, the above proof still applies to show the second part of Proposition \ref{prop:factorize}, except that now we cannot conclude that $\phi_\Gamma$ and $\phi_{\Gamma^c}$ are pure.
\end{proof}
\begin{remark}
    It is straightforward to generalize both Lemma \ref{lemma:factorization_of_GNS} and Proposition \ref{prop:factorize} to higher dimensional lattice systems.
\end{remark}

\begin{corollary}\label{coro:split_everywhere}
    Consider a 1D quantum spin chain and let $\psi$ be a state of $\A^{ql}$ which splits at origin. Then $\psi$ splits at any site $n$.
\end{corollary}
\begin{proof}
    Let us start with the case where $\psi$ is pure. By definition $\psi\simeq\psi_{<0}\otimes\psi_{\geqslant0}$ for some pure states $\psi_{<0}$ and $\psi_{\geqslant0}$. We further decompose the right half chain into $[0,n)\sqcup[n,\infty)$. By the same argument as in Proposition \ref{prop:factorize}, there exist pure states $\psi_{[0,n)}$ and $\psi_{\geqslant n}$ such that
    \beq
    \psi_{\geqslant0}\simeq \psi_{[0,n)}\otimes\psi_{[n,\infty)}
    \eeq
    Hence, we see that $\psi$ splits at the site $n$ since $\psi\simeq\psi_{<0}\otimes\psi_{\geqslant n}$, with $\psi_{<n}:=\psi_{<0}\otimes \psi_{[0,n)}$ and $\psi_{\geqslant n}$ constructed above.

    If $\psi$ is mixed, a similar proof still applies.
\end{proof}
Due to this corollary, we do not need to specify at which site the state splits and we will simply say that a state splits if Eq.~\eqref{eq:splits_0} holds.

A direct consequence of Corollary \ref{coro:split_everywhere} is that the converse of Lemma \ref{thm:area_law_split} does not hold in general. To construct a counterexample, let $\psi_{<0}$ and $\psi_{\geqslant 0}$ be pure states of $\A_{<0}$ and $\A_{\geqslant0}$ respectively, and we assume that they violate the area law of entanglement entropy. Then we know that $\psi=\psi_{<0}\otimes\psi_{\geqslant0}$ splits at 0, so by Corollary \ref{coro:split_everywhere} it splits at all sites. However, of course $\psi$ does not satisfy area law by construction.

\subsection{Modular theory and entropy}\label{sec:entropyI}

This section is highly technical and can be safely skipped for the first reading. The take-home message of this section is Proposition \ref{prop:factor_entropy}.

Earlier, given a state $\psi$, we define its entanglement entropy by restricting it on a finite region (see Definition \ref{def:area_law}). This definition is intuitive and useful, but it is not conceptually satisfactory. It is natural to ask: Can we define the entropy on the whole infinite spin chain? This is indeed possible. In this section, we present one possible definition introduced by Ohya and Petz in Ref.~\cite{ohya2004quantum}. For the Ohya-Petz entropy, we show that
\begin{proposition}\label{prop:factor_entropy}
    Let $\psi$ be a factor state of a 1d quantum spin chain on lattice $\Lambda$, if $\psi$ satisfies the entanglement area law, \ie
    \beq
    S(\psi|_{\Gamma})<\const
    \eeq
    for any connected finite subset $\Gamma\subseteq\Lambda$, where $\const$ is a constant that does not depend on the choice of $\Gamma$, then $\psi$ is of type-I.
\end{proposition}

The definition of Ohya-Petz entropy relies on the relative entropy on von Neumann algebras and $C^{*}$-algebras, introduced by Araki \cite{Araki1976entropyI,Araki1977entropyII}. An accessible introduction to Araki's construction can be found in Ref.~\cite{Witten2018entanglement}. We only sketch some main ideas here. Below we begin with some definitions.
\begin{definition}
    Let $\pi:\A^{ql}\to\B(\cH)$ be a representation of $\A^{ql}$, whose associated von Neumann algebra is $\cM:=\pi(\A^{ql})''$. Then,
    \begin{enumerate}
        \item An abstract state $\psi$ is called separating or faithful if $\psi(a^*a)=0$ implies $a=0$ for $a\in\A^{ql}$.
        \item A vector state $|\phi\ra$ is called cyclic with respect to $\A^{ql}$ (resp. $\cM$) if $\pi_{\psi}(\A^{ql})|\psi\ra$ (resp. $\cM|\psi\ra$) is dense in $\cH$ under its norm topology.
    \end{enumerate}
\end{definition}
Recall that in the GNS representation, the GNS vector $|\psi\ra$ is cyclic with respect to $\A^{ql}$ or $\cM_{\psi}$. Besides, pure states are hardly separating since they typically have nonzero GNS ideals.

To define Araki's relative entropy, we need the so-called modular theory of von Neumann algebras due to Tomita and Takesaki, which plays an important role in the classification of type-III von Neumann algebras. This theory begins with defining the Tomita operator\footnote{Note that in our exposition, we use $\cS$ for modular operators and $S$ for entropies, while in Ref.~\cite{Witten2018entanglement} these meanings are reversed.} $\cS_{\psi}$. Let $\psi$ be a state of $\A^{ql}$, and we take its GNS representation with GNS triple $(\cH_\psi, \pi_\psi, |\psi\ra)$. The $*$-operation on $\A^{ql}$ should induce the following operation on $\cH_{\psi}$,
\beq
\cS_{\psi}(\pi_{\psi}(a)|\psi\ra)=\pi_{\psi}(a^{*})|\psi\ra,\quad\forall a\in\A^{ql}
\eeq
This determines the Tomita operator $\cS_{\psi}$ as an anti-linear operator on a dense subspace $\pi_{\psi}(\A^{ql})|\psi\ra$ of $\cH_{\psi}$. However, we note that $\cS_{\psi}$ is well-defined iff $\psi$ is separating for $\A^{ql}$. Otherwise, it can happen that $\pi_{\psi}(a)|\psi\ra=0$ while $\pi_{\psi}(a^*)|\psi\ra\not=0$ and $\cS_{\psi}$ does not exist. Below we assume $\psi$ is indeed separating\footnote{As is mentioned, pure states are \textbf{not} separating in general. But pure states do not interest us here.}.

Any anti-linear operator admits a unique polar decomposition,
\beq
\cS_{\psi}=J_{\psi}\Delta_{\psi}^{1/2}
\eeq
where $J_{\psi}$ is an anti-unitary operator called \textit{modular conjugation} while $\Delta_{\psi}=\cS_{\psi}^{\dagger}\cS_{\psi}$ is a self-adjoint positive operator called \textit{modular operator}\footnote{Recall that the adjoint of an anti-linear operator $\cS$ is defined as $\la\psi'|\cS\psi\ra=\overline{\la\cS^{\dagger}\psi'|\psi\ra}$ for any two states in $\cH_{\psi}$.}. Clearly $\cS_{\psi}^{2}=1$ and it is thus invertible. This implies $\Delta_{\psi}$ is also invertible. It can be shown that $\Delta_{\psi}$ must be trivial if $\psi$ is of type-I or type-II. For type-III states, certain type of spectrum derived from $\Delta_{\psi}$ gives a more refined classification to type-III factors, see \eg Sec.~C.23 of Ref.~\cite{Landsman:2017hpa} for an exposition. We will not present details of this classification.

The definition of relative entropy demands a relative version of modular operators, which will be defined now. Let $\phi$ be another separating state of $\A^{ql}$, the relative Tomita operator $\cS_{\psi|\phi}$ is defined by 
\beq
\cS_{\psi|\phi}(\pi_{\psi}(a)|\psi\ra)=\pi_{\phi}(a^*)|\phi\ra,\quad\forall\,a\in\A^{ql}
\eeq
Conversely, one can define $\cS_{\phi|\psi}$ such that
\beq
\cS_{\phi|\psi}(\pi_{\phi}(a)|\phi\ra)=\pi_{\psi}(a^*)|\psi\ra,\quad\forall\,a\in\A^{ql}
\eeq
And one can easily show that $\cS_{\psi|\phi}\cS_{\phi|\psi}=1=\cS_{\phi|\psi}\cS_{\psi|\phi}$ and thus they are invertible. We can take the polar decomposition,
\beq
\cS_{\psi|\phi}=J_{\psi|\phi}\Delta_{\psi|\phi}^{1/2}
\eeq
where $J_{\psi|\phi}$ is anti-unitary and called the relative modular conjugation while $\Delta_{\psi|\phi}:=\cS_{\psi|\phi}^{\dagger}\cS_{\psi|\phi}$ is self-adjoint, positive and called the relative modular operator.

Now we are ready to define the relative entropy.
\begin{definition}\label{def:relative_entropy}
    Let $\psi$ and $\phi$ be two states on $\A^{ql}$, then their Araki relative entropy is defined to be
    \beq\label{eq:relative_entropy}
    S(\psi,\phi):=-\la\psi|\log(\Delta_{\psi|\phi})|\psi\ra
    \eeq
    More generally, the relative entropy can be defined for any pair of positive linear functionals (\ie unnormalized states).
\end{definition}
Especially, this definition does \textbf{not} require $\psi$ and $\phi$ fall into the same superselection sector.
In general, $S(\psi,\phi)$ takes value in $\R_{\geqslant0}$ and $\infty$.
It is also nontrivial to see why this definition agrees with the relative entropy in classical information theory or quantum mechanics. The relation between Definition \ref{def:relative_entropy} and the traditional relative entropy will be illustrated in the following example.
\begin{example}[Sec.~4 of Ref.~\cite{Witten2018entanglement}]\label{ex:relative_entropy}
    In this example, we study the relative entropy Eq.~\eqref{eq:relative_entropy} in a bi-partite system $\cH_{1}\otimes \cH_{2}$, where we assume $\dim\cH_{1}=\dim\cH_{2}=d<\infty$. We will show that the Araki relative entropy between two potentially mixed states $\psi,\phi\in\cH_1$ is identical to the usually defined relative entropy between these two states. The space $\cH_2$ here is introduced to purify $\psi$ and $\phi$, which simplifies the analysis.
    
    Let $\{|\eta_{k}\ra\}_{k=1,2,\dots,d}$ and $\{|\xi_{k}\ra\}_{k=1,2,\dots,d}$ be the orthonormal basis of $\cH_{1}$ and $\cH_{2}$ respectively, then any $|\psi\ra\in\cH_{1}\otimes \cH_{2}$ can be expanded as
    \beq
    |\psi\ra=\sum_{k=1}^{d}c_{k}|\eta_{k}\ra\otimes |\xi_{k}\ra.
    \eeq
    Let us write $\A_{k}:=\B(\cH_{k}),\,k=1,2$ and $|i,j\ra:=|\eta_{i}\ra\otimes|\xi_{j}\ra$. Then for any $a\in\A_{1}$ we have
    \beq
    (a\otimes 1)|\psi\ra=\sum_{k=1}^{d}c_{k}(a|\eta_{k}\ra)\otimes|\xi_{k}\ra
    \eeq
    If $c_{k}$'s are all nonzero, then $(a\otimes 1)|\psi\ra=0$ iff $a|\eta_{k}\ra=0$ for all $k$ and thus $a=0$. Therefore $|\psi\ra$ is separating for $\A_{1}$ iff all $c_{k}$'s are nonzero. In this case, it is easy to see that $|\psi\ra$ is also separating for $\A_{2}$. Below we assume this condition.
    
    We also need to show that $|\psi\ra$ is cyclic in this case. To this end, view $\A_{1}|\psi\ra$ and $\cH_{1}\otimes\cH_{2}$ as $\A_{1}$-representations and obviously $\A_{1}|\psi\ra\subseteq\cH_{1}\otimes\cH_{2}$. Therefore, there exists an orthogonal projection $P\in\A_{1}'=\A_{2}$ such that $P:\cH_{1}\otimes\cH_{2}\to\A_{1}|\psi\ra$. Therefore $P|\psi\ra=|\psi\ra\Rightarrow (1-P)|\psi\ra=0$. We note that $|\psi\ra$ is separating for $\A_{2}$ and $1-P\in\A_{2}$, thus we conclude $1-P=0$, which means $\A_{1}|\psi\ra=\cH_{1}\otimes \cH_{2}$, \ie $|\psi\ra$ is cyclic for $\A_{1}$. Similarly, it is also cyclic for $\A_{2}$.

    Let us write down the modular operator explicitly by setting $a\in\A_{1}$ to be an elementary matrix,
    \beq
    a|i\ra=|j\ra,\quad a|k\ra=0\,\text{if $k\not=i$}
    \eeq
    Its adjoint is given by
    \beq
    a^{\dagger}|j\ra=|i\ra,\quad a^{\dagger}|k\ra=0\,\text{if $k\not=j$}
    \eeq
    Therefore,
    \beq
    \begin{split}
        (a\otimes1)|\psi\ra&=c_{i}|j,i\ra\\
        (a^{\dagger}\otimes 1)|\psi\ra&=c_{j}|i,j\ra
    \end{split}
    \eeq
    The Tomita operator $\cS_{\psi}$ is given by 
    \beq
    \cS_{\psi}(|j,i\ra)=\frac{c_{j}}{\bar{c}_{i}}|i,j\ra
    \eeq
    Its adjoint is
    \beq
    \cS_{\psi}^{\dagger}(|i,j\ra)=\frac{c_{j}}{\bar{c}_{i}}|j,i\ra
    \eeq
    Finally, the modular operator can be read out
    \beq
    \Delta_{\psi}|i,j\ra=\cS_{\psi}^{\dagger}\cS_{\psi}|i,j\ra=\frac{|c_{j}|^{2}}{|c_{i}|^{2}}|i,j\ra
    \eeq
    In terms of density matrices, let $\rho_{1,2}:=\tr_{\cH_{2,1}}(|\psi\ra\la\psi|)$, then 
    \beq
    \Delta_{\psi}=\rho_{1}^{-1}\otimes \rho_{2}
    \eeq

    Similarly, let $|\phi\ra=\sum_{\alpha=1}^{d}d_{\alpha}|\alpha,\alpha\ra\in\cH_{1}\otimes \cH_{2}$ be another separating and cyclic states (\ie all $d_{\alpha}$'s are nonzero). The relative Tomita operator can be deduced from its definition,
    \beq
    \cS_{\psi|\phi}((a\otimes 1)|\psi\ra)=(a^{\dagger}\otimes 1)|\phi\ra
    \eeq
    As before, for any given $i,\alpha\in\{1,2,\dots,d\}$, we set $a$ to be 
    \beq
    a|i\ra=|\alpha\ra,\quad a|k\ra=0,\,\text{if $k\not=i$}
    \eeq
    It follows that
    \beq
    a^{\dagger}|\alpha\ra=|i\ra,\quad a^{\dagger}|\beta\ra=0,\,\text{if $\beta\not=\alpha$}
    \eeq
    In this case we have
    \beq
    \cS_{\psi|\phi}(|\alpha,i\ra)=\frac{d_{\alpha}}{\bar{c}_{i}}|i,\alpha\ra
    \eeq
    Its adjoint is given by 
    \beq
    \cS_{\psi|\phi}^{\dagger}(|i,\alpha\ra)=\frac{d_{\alpha}}{\bar{c}_{i}}|\alpha,i\ra
    \eeq
    Thus, the relative modular operator $\Delta_{\psi|\phi}$ is given by 
    \beq
    \Delta_{\psi|\phi}|\alpha,i\ra=\frac{|d_{\alpha}|^{2}}{|c_{i}|^{2}}|\alpha,i\ra
    \eeq
    Again, making use of density matrices, $\sigma_{1,2}:=\tr_{\cH_{2,1}}(|\phi\ra\la\phi|)$, then 
    \beq
    \Delta_{\psi|\phi}=\sigma_{1}\otimes \rho_{2}^{-1}
    \eeq

    To complete the computation of relative entropy, we use
    \beq
    \log(\Delta_{\psi|\phi})=\log(\sigma_{1})\otimes 1-1\otimes \log(\rho_{2})
    \eeq
    Therefore, it is easy to see
    \beq\label{eq:qm_relative_entropy}
    \begin{split}
        S(\psi,\phi)&=-\la\psi|\log(\Delta_{\psi|\phi})|\psi\ra\\
        &=-\tr_{12}(\rho_{12}(\log(\sigma_{1})\otimes 1-1\otimes \log(\rho_{2})))\\
        &=-\tr_{1}(\rho_{1}\log(\sigma_{1}))+\tr_{2}(\rho_{2}\log(\rho_{2}))\\
        &=\tr_{1}(\rho_{1}(\log\rho_{1}-\log\sigma_{1}))
    \end{split}
    \eeq
    where we have write $\tr_{12}:=\tr_{\cH_{1}\otimes\cH_{2}}$, $\tr_{i}:=\tr_{\cH_{i}},i=1,2$ and $\rho_{12}:=|\psi\ra\la\psi|$. To derive the last line, we have used the fact that $|\psi\ra$ is pure state so that $\tr_{2}(\rho_{2}\log\rho_{2})=\tr_{1}(\rho_{1}\log\rho_{1})$.
\end{example}

From Eq.~\eqref{eq:qm_relative_entropy}, it is clear that Definition \ref{def:relative_entropy} is a natural generalization of the relative entropy in finite dimensional quantum mechanics. Araki's relaitve entropy enjoys a lot of properties shared by its finite-dimensional counterpart, such as $S(\psi,\phi)\geqslant0$ and the equality holds iff $\psi=\phi$. In fact, we have
\begin{proposition}[Proposition 5.23 of Ref.~\cite{ohya2004quantum}]\label{prop:relative_entropy}
    Let $\psi,\phi$ be any two states on $\A^{ql}$, then
    \begin{enumerate}
        \item $S(\psi,\phi)$ is jointly convex, \ie
        \beq
        S(\lambda\psi_{1}+(1-\lambda)\psi_{2},\lambda\phi_{1}+(1-\lambda)\phi_{2})\leqslant \lambda  S(\psi_{1},\phi_{1})+(1-\lambda)S(\psi_{2},\phi_{2})
        \eeq
        where $0\leqslant\lambda\leqslant 1$.
        
        \item Quantum Pinsker's inequality: $||\psi-\phi||^{2}\leqslant 2 S(\psi,\phi)$ if $\psi(1)=\phi(1)=1$ (\ie they are normalized states).
        \item Coarse-graining reduces differences: Let $\A$ be another $C^*$-algebra and $\alpha:\A^{ql}\to\A$ be a unital Schwarz map\footnote{A linear map $\alpha:\A_{1}\to\A_{2}$ between $C^*$-algebras is called a Schwarz map if $\alpha(a^*a)-\alpha(a^*)\alpha(a)\geqslant0$ and it is unital if $\alpha(1)=1$. In practice, a Schwarz map is often given by taking conditional expectation value, which is a kind of coarse-graining.}, then $S(\psi\circ\alpha,\phi\circ\alpha)\leqslant S(\psi,\phi)$. In particular, this means that the relative entropy is invariant under $*$-automorphisms.
        \item Martingale property: Let $\A_{1}\subseteq\dots\subseteq\A_{n}\subseteq_{\dots}$ be an increasing sequence of $*$-subalgebra such that $\A^{ql}$ can be obtained by completing $\cup_{n=1}^{\infty}\A_{n}$ with respect to operator norm, then
        \beq
        S(\psi,\phi)=\lim_{n\to\infty}S(\psi|_{\A_{n}},\phi|_{\A_{n}})
        \eeq
        
    \end{enumerate}
\end{proposition}

Having defined the relative entropy between 2 states, it is natural to ask if one can define the entropy of one state. One possible definition was proposed by Ohya and Petz in Ref.~\cite{ohya2004quantum}, which makes use of Araki's relative entropy Eq.~\eqref{eq:relative_entropy}. We call it the Ohya-Petz entropy below.
\begin{definition}[Ohya-Petz entropy, Eq.~(6.9) of Ref.~\cite{ohya2004quantum}] \label{def: OP entropy}
    Let $\psi$ be a state of $\A^{ql}$. Its Ohya-Petz entropy is given by 
    \beq
    S(\psi)=\sup\{\sum_{i=1}^{\infty}\lambda_{i}S(\psi_{i},\psi)|\psi=\sum_{i=1}^{\infty}\lambda_{i}\psi_{i}\}
    \eeq
    where the supreme is taken for all countable convex decomposition of $\psi$ and $S(\psi_{i},\psi)$ is the relative entropy between $\psi_{i}$ and $\psi$.
\end{definition}

Clearly, from the non-negativity of relative entropy, $S(\psi)=0$ iff $\psi$ is pure. However, it is not obvious why this definition is a natural generalization of the traditional von Neumann entropy. We illustrate this by an example below (see also Theorem 6.10 of Ref.~\cite{ohya2004quantum}).
\begin{example}\label{ex:OP_entropy}
    Let us consider a finite-dimensional Hilbert space $\cH$ ($\dim\cH=d<\infty$) and a mixed state $\psi$. Suppose $\{|i\ra\}_{i=1,2,\dots,d}$ is an orthonormal basis of $\cH$. Let us write the density matrix of $\psi$ as $\rho_{\psi}$. The spectral decomposition of $\rho_{\psi}$ is given by
    \beq
    \rho_{\psi}=\sum_{i=1}^{d}\lambda_{i}|i\ra\la i|=\sum_{i=1}^{d}\lambda_{i}\psi_{i}
    \eeq
    In this decomposition, by Eq.~\eqref{eq:qm_relative_entropy},
    \beq
    S(\psi_{i},\psi)=-\la i|\log(\rho_{\psi})|i\ra=-\log\lambda_{i}
    \eeq
     As a result,
     \beq
     S(\psi)\geqslant-\sum_{i}\lambda_{i}\log\lambda_{i}=-\tr(\rho_{\psi}\log\rho_{\psi})
     \eeq
    On the other hand, for a general convex decomposition $\psi=\sum_{j}\mu_{j}\phi_{j}$, we have
    \beq
        S(\phi_{j},\psi)=\tr(\rho_{\phi_{j}}(\log\rho_{\phi_{j}}-\log\rho_{\psi}))
    \eeq
    Thus,
    \beq
        \begin{split}
    \sum_{j}\mu_{j}S(\phi_{j},\psi)&=\sum_{j}\mu_{j}\tr(\rho_{\phi_{j}}(\log\rho_{\phi_{j}}-\log\rho_{\psi}))\\
        &=-\tr(\rho_{\psi}\log\rho_{\psi})+\sum_{j}\mu_{j}\tr(\rho_{\phi_{j}}\log\rho_{\phi_{j}})\\
        &\leqslant-\tr(\rho_{\psi}\log\rho_{\psi})
    \end{split}
    \eeq
    To derive the second line we have used $\sum_{j}\mu_{j}\rho_{\phi_{j}}=\rho_{\psi}$.
    By taking supreme, we have $S(\psi)\leqslant-\tr(\rho_{\psi}\log\rho_{\psi})$ and
    we then conclude $S(\psi)=-\tr(\rho_{\psi}\log\rho_{\psi})$.
\end{example}
By generalizing this computation, one can show that Ohya-Petz entropy reduces to von-Neumann entropy for any countable convex sum of type-I factor states (see Theorem 6.10 of Ref.~\cite{ohya2004quantum}).
Our following lemma concerns factor states with other types, which is a $C^*$-algebraic version of Lemma 6.9 in Ref.~\cite{ohya2004quantum}.
\begin{lemma}\label{lemma:divergent_OPentropy}
    Let $\psi$ be a factor state of $\A^{ql}$ of type-II or III, then $S(\psi)=\infty$.
\end{lemma}
It is this feature that makes Ohya-Petz entropy awkward to use in quantum field theories. However it will be useful for proving Proposition \ref{prop:factor_entropy}.

Another (rather unpleasant) feature of Ohya-Petz entropy is that it does not satisfy martingale property in Proposition \ref{prop:relative_entropy}. Instead, we have
\begin{lemma}\label{lemma:sub_martingale}
    Let $\{\A_{n}\}_{n=1,\dots,}$ be an increasing sequence of $*$-subalgebra of $\A^{ql}$, such that their norm completion is $\A^{ql}$. Consider a state $\psi$ on $\A^{ql}$, then
    \beq
    S(\psi)\leqslant \overline{\lim}_{n\to\infty}S(\psi|_{\A_{n}}):=\lim_{n\to\infty}\sup_{k\geqslant n}S(\psi|_{\A_{k}})
    \eeq
\end{lemma}
\begin{proof}
    Note that for any countable convex decomposition $\psi=\sum_{i}\lambda_{i}\psi_{i}$, we have
    \beq
    \sum_{i}\lambda_{i}S(\psi_{i}|_{\A_{n}},\psi|_{\A_{n}})\leqslant S(\psi|_{\A_{n}})\leqslant \sup_{k\geqslant n}S(\psi|_{\A_{k}}),\quad\forall\,n\in\z_{>0}
    \eeq
    By letting $n\to\infty$ and the martingale property of relative entropy (see Proposition~\ref{prop:relative_entropy}), we have
    \beq
    \sum_{i}\lambda_{i}S(\psi_{i},\psi)\leqslant\overline{\lim}_{n\to\infty}S(\psi|_{\A_{n}})
    \eeq
    Taking supreme, we end up with
    \beq
    S(\psi)\leqslant\overline{\lim}_{n\to\infty}S(\psi|_{\A_{n}})
    \eeq
    \begin{remark}
        The equality cannot be saturated in general. To see this, taking $\psi$ to be a pure state on $\A^{ql}$ so we have $S(\psi)=0$. On the other hand, $\psi|_{\A_{n}}$ is a mixed state for arbitrarily large $n$ and we do not expect $S(\psi|_{\A_{n}})$ goes to $0$ in general.
    \end{remark}
\end{proof}
Proposition \ref{prop:factor_entropy} can thus be obtained by combining Lemma \ref{lemma:sub_martingale} and Lemma \ref{lemma:divergent_OPentropy}. In more details,
\begin{proof}[Proof of Proposition \ref{prop:factor_entropy}]
    By assumption, 
    \beq
    S(\psi|_{\Gamma_{n}})<S_{0}=\const,\quad\forall\,n\in\z_{>0}
    \eeq
    Using Lemma \ref{lemma:sub_martingale},
    \beq
    S(\psi)\leqslant\overline{\lim}_{n\to\infty}S(\psi|_{\Gamma_{n}})\leqslant S_{0}<\infty
    \eeq
    From Lemma \ref{lemma:divergent_OPentropy}, this is impossible when $\psi$ is of type-II or type-III.
\end{proof}

\subsection{Hamiltonians and ground states}\label{Sec:hamiltonian_gs}

In the previous few subsections, we focus on the correlation and entanglement properties of states. Starting from this subsection, we will discuss dynamical properties of quantum systems with an infinite size, such as Hamiltonians, ground states, energy gaps and time evolution. In particular, in this subsection, we define Hamiltonians, ground states, and energy gaps in the formalism of operator algebra. Especially, we will explain how these definitions reduce to our more familiar notions in systems with finitely many degrees of freedom. 

Naively, one can define local Hamiltonians as
\beq\label{eq:naive_Hamiltonian}
H=\sum_{j\in\Lambda}h_{j}
\eeq
where $h_{j}\in \A_{B(j,R)}$ is a local term, and $B(j,R):=\{p\in\Lambda|d(p,j)\leqslant R\}$ for some fixed $R>0$. However, this does not work in general since it contains an infinite sum and one has to be careful about its convergence. Actually Eq. \eqref{eq:naive_Hamiltonian} does not converge in general, hence one needs a more careful definition. Nevertheless, recall that in the operator algebra formalism, we work in the Heisenberg picture. In the Heisenberg picture, the Hamiltonian $H$ governs the time evolution of a local operator $A$ as 
\beq
\frac{\dd A}{\dd t}=i[H,A]=i\sum_{j\in\Lambda}[h_{j},A]
\eeq
If $H$ is a local Hamiltonian, then the commutator $[H, A]$ contains at most finitely many non-zero terms by locality. Thus one can define $H$ by this commutator. Moreover, even if $H$ is a long-range interacting Hamiltonian and the commutator $[H, A]$ is an infinite sum, if the interactions decay fast enough as their ranges increase, this commutator may still converge. This motivates the following definition. 

\begin{definition}\label{definition:Hamiltonian}
    A Hamiltonian is a derivation $\delta_{H}$ of the form
    \beq\label{eq:derivation_Hamiltonian}
    \delta_{H}(A)=\sum_{Z}[h_{Z},A],\forall\,A\in \A^{\ell},
    \eeq
    where $h_Z$ is an operator supported on $Z$.
\end{definition}
Here a (symmetric) derivation (which characterizes a Hamiltonian) $\delta$ on $\A^{\ell}$ means the following.

\begin{definition}\label{def:local_derivation}
    A symmetric derivation $\delta$ is a densely defined\footnote{This means the domain $D(\delta)$ is a dense subset in $\A^{ql}$, \eg $D(\delta)=\A^{\ell}$. We always require $\A^{\ell}\subseteq D(\delta)$ below.} map on $D(\delta)\subseteq\A^{ql}$ with the following properties.
    \begin{enumerate}
    \item $\delta:D(\delta)\to\A^{ql}$ is $\bbC$-linear.
    \item $\delta(A^{\dagger})=-\delta(A)^{\dagger}$
    \item $\delta(AB)=\delta(A)B+A\delta(B)$
\end{enumerate}
\end{definition}

Clearly, this definition of Hamiltonian reduces to the usual one in finite quantum systems.

\begin{remark}\label{remark:domain_long_range}
    With long-range interactions, one has to make sure that the right hand side of Eq. \eqref{eq:derivation_Hamiltonian} is convergent. Also, in many cases, $\delta_{H}$ is only defined on local operators $\A^{l}$ or one says that $\delta_{H}$ is only \textbf{densely} defined on $\A^{ql}$.
\end{remark}

Given the notion of Hamiltonians, now we define ground states in the operator algebra formalism.
\begin{definition} \label{def: ground state}
    A state $\psi$ is said to be the ground state of the Hamiltonian $\delta_{H}$ if 
    \beq\label{eq:ground_states}
    \psi(A^{\dagger}\delta_{H}(A))\geqslant 0,\forall\,A\in\A^{l}
    \eeq
\end{definition}
\begin{remark}
    Let us explain why this definition agrees with the usual definition of ground states in finite-dimensional quantum mechanics. Let us assume $\psi$ is a state vector in some Hilbert space $\cH$ denoted by $|\psi\ra$. We also assume there is a Hamiltonian operator $H$ on $\cH$, which has $|\psi\ra$ as a ground state of energy $E$. Then Eq. \eqref{eq:ground_states} is equivalent to
    \beq\label{eq:gs_in_Hilbert_space}
    \la \psi|A^{\dagger}HA|\psi\ra\geqslant E \la\psi|A^{\dagger}A|\psi\ra
    \eeq
    In a finite system, any state in $\cH$ can be prepared by applying some operator to $|\psi\ra$, and hence Eq. \eqref{eq:gs_in_Hilbert_space} indeed means that $|\psi\ra$ has the lowest energy $E$ in Hilbert space $\cH$ and is a ground state.
\end{remark}
\begin{remark}
    As another remark, a drawback of this definition is that it is not obvious if a ground state is pure or not. For example, given a finite size system with the classical Ising Hamiltonian, although its ground state is pure, some of its ground states (\ie the GHZ states) become mixed after taking the thermodynamic limit.
    So given a Hamiltonian $\delta_{H}$, we have to check if its ground state is pure or not. This is drastically different from the finite dimensional cases. We will show that gapped ground states of Hamiltonians with sufficiently short-range interactions are pure indeed (see Eq. \eqref{eq:admissible_H} and Theorem \ref{Thm:pure_gs}).
\end{remark}

Now we are ready to talk about the notion of locally unique gapped ground states and the energy gaps.
\begin{definition}[Locally unique gapped ground state]\label{definition:locally-unique_gs}
    A ground state $\psi$ of Hamiltonian $\delta_{H}$ is a locally unique gapped ground state if there is a $\gamma>0$  such that
    \beq\label{eq:locally_unique_gapped_gs}
    \psi(A^{\dagger}\delta_{H}A)\geqslant\gamma\psi(A^{\dagger}A)
    \eeq
    for any $A\in \A^{l}$ with $\psi(A)=0$. The energy gap $\Delta$ is the largest possible $\gamma$ satisfying the above inequality. A locally-unique gapped ground state is unique if it is the only locally-unique gapped  ground state.
\end{definition}
\begin{remark}
    Again, let us check that this definition reduces to our familiar notions in finite systems. First, the condition $\psi(A)=0$ is to exclude the case where $A$ is a constant multiple of the identity operator, and it can always be achieved by redefining $A\to A-\psi(A)I$. Thus Eq. \eqref{eq:locally_unique_gapped_gs} means that for any
    \beq
    |A\ra:=A|\psi\ra,\,A\in\A^{ql}
    \eeq
    which is orthogonal to $|\psi\ra$, we must have
    \beq
    \frac{\la A|H|A\ra}{\la A|A\ra}\geqslant \Delta+E>E,
    \eeq
    where $E$ is the energy of $\psi$. Recall that any state in $\cH$ can be obtained by applying local operators to $|\psi\ra$. The above inequality means that $|\psi\ra$ is the \textbf{only} state that has an energy smaller than $\Delta$ in the Hilbert space $\cH$, agreeing with our usual definition of a unique gapped ground state. In the context of infinite systems, different ground states may fall into different superselection sectors, and a locally unique ground state is the ground state in a given superselection sector.
\end{remark}
\begin{remark}
    As far as we know, the terminology ``locally unique gapped ground state" is from Ref. \cite{Tasaki2022topological}. The gapped ground state defined in Ref. \cite{kapustin2024anomalous} is actually our locally unique gapped ground state.
\end{remark}
\begin{example}
    We give another example to show that not all ground states in the usual quantum mechanics are locally unique. The model is the (classical) Ising model with the following Hamiltonian
    \beq
    \delta_{H}(A)=-\sum_{j\in\z}[Z_{j}Z_{j+1},A]
    \eeq
    where $Z_{j}$ is usual Pauli operator. Using definition \ref{def: ground state}, one can choose the ground state of this model to be the infinite-system version of the GHZ state:
    \beq
    \psi_{GHZ}=\frac{1}{2}(\psi_{\uparrow}+\psi_{\downarrow})
    \eeq
    However, this state is \textbf{not} a locally-unique gapped ground state of the Ising model. To see it, let $A=Z_{k}$ supported at site $k$.
    Note that
    \beq
    \psi(A^{\dagger}\delta_{H}(A))=0
    \eeq
    since $A$ commutes with Pauli-$Z$ operators. On the other hand,
    \beq
    \begin{split}
      \psi(A)&=0\\
      \psi(A^{\dagger}A)&=1
    \end{split}
    \eeq
    hence $\psi$ is not a locally unique gapped ground state because $0=\psi(A\delta_{H}(A))<\psi(A^{\dagger}A)=1$ and Eq. \eqref{eq:locally_unique_gapped_gs} is violated.

    In fact, we will show in Theorem \ref{Thm:pure_gs} that all locally-unique gapped ground states of any local Hamiltonian are pure, but $\psi_{GHZ}$ is a mixed state.
\end{example}

\subsection{Local Hamiltonians and admissible Hamiltonians} \label{subsec: local and admissible Hamiltonians}

There are two classes of Hamiltonians that are of particular interest. The first class is local Hamiltonians, which are Hamiltonians where in Eq. \eqref{eq:derivation_Hamiltonian} each $h_{Z}$ supports on $Z$ and $h_{Z}=0$ if $\diam(Z)>R$ for some constant $R$. One of the main topics in this paper is long-range interacting systems, where the Hamiltonian is not local. We will work with a particular type of long-range interacting Hamiltonians called {\it admissible Hamiltonians}. They are defined as follows \cite{Kuwahara2019}.
\begin{definition}
    In a 1D lattice system $\Lambda=\z$ where sites are labeled by $i$, the Hamiltonian $\delta_{H}$ is admissible if $\delta_{H}=\sum_{Z:|Z|\leqslant k}[h_{Z},\bullet]$ with $h_{Z}$'s satisfying
    \beq\label{eq:admissible_H}
    \begin{split}
        \sup_{i\in\z}\sum_{\substack{Z:|Z|\leqslant k,Z\owns i\\\diam(Z)=r}}||h_{Z}||&<\frac{J}{r^{\mathfrak{a}}},\,{\rm with\ }\mathfrak{a}>2,\\
        \sup_{i\in\z}||h_{i}||&<B
    \end{split}
    \eeq
    where $\diam(Z):=\max_{x,y\in Z}d(x,y)$, k is an integer, $J$ and $B$ are positive constants and $h_{i}$ is a one-body potential at site $i$.
\end{definition}

We have not shown that the above definition is well-defined. In particular, we have not shown that the $\delta_H$ is indeed a derivation from $\A^l$ to $\A^{ql}$ (see definition \ref{def:local_derivation} and remark \ref{remark:domain_long_range}). Below we will see that this is indeed the case. More precisely, if $A\in \A^{l}_\Gamma$ is a local operator, one can define a sequence
    \beq\label{eq:local_Cauchy_sequence}
     H_{n}(A):=\sum_{\substack{Z:Z\cap\Gamma\not=\emptyset\\ \diam(Z)\leqslant n}}[h_{Z},A]
    \eeq
    then its limit as $n\rightarrow\infty$ exists,
    \beq
    \delta_{H}(A):=\lim_{n\to\infty}H_{n}(A)\in\A^{ql}
    \eeq
    However, as remarked before (see remark.~\ref{remark:domain_long_range}), $\delta_{H}$ cannot be defined on the whole $\A^{ql}$.

\begin{lemma}\label{lemma:existence_Hamiltonian}
    Admissible Hamiltonians are well-defined, that is, if $\delta_H$ is admissible and $A\in\A^{l}$, then $\delta_{H}(A)\in\A^{ql}$, \ie the limit in the above definition exists in $\A^{ql}$.
\end{lemma}
\begin{proof}
    We first show that $H_{n}(A)$ defined in Eq. \eqref{eq:local_Cauchy_sequence} is a local operator. Note for each $n>0$, there are only finitely many subsets $Z$ such that $Z\cap \Gamma\not=\emptyset$, hence $H_{n}(A)\in \A_{B(\Gamma,n)}$.
    
    Below we show $\{H_{n}(A)\}_{n=1,2,...}$ is a Cauchy sequence hence it is convergent in $\A^{ql}$. To this end, consider $m\geqslant n$,
    \beq
    ||H_{m}(A)-H_{n}(A)||=||\sum_{r=n}^{m}\sum_{\substack{Z:Z\cap\Gamma\not=\emptyset\\ \diam(Z)=r}}[h_{Z},A]||\leqslant 2|\Gamma|\cdot||A||\sum_{r=n}^{m}\frac{J}{r^{\mathfrak{a}}}
    \eeq
    where we have used $||\sum_{i}B_{i}||\leqslant \sum_{i}||B_{i}||$ and $||[B,C]||\leqslant 2||B||\cdot ||C||$ for any operators $B_{i},B$ and $C$ in $\A^{ql}$.

    Note for $m\geqslant n\geqslant n_{0}$,
    \beq
    \sum_{r=n}^{m}\frac{1}{r^{\mathfrak{a}}}<\int_{n_{0}}^{\infty}r^{-\mathfrak{a}}dr=\frac{n_{0}^{1-\mathfrak{a}}}{\mathfrak{a}-1}
    \eeq
    which goes to 0 as $n_{0}\to \infty$ since $\mathfrak{a}>2$. Therefore, $\{H_n(A)\}$ is indeed a Cauchy sequence and so it converges to some element in $\A^{ql}$.
\end{proof}

\subsection{Time evolution}\label{subsec:LPA&LR}

In the above, we have introduced the concepts of operators, states and Hamiltonians in the operator algebra formalism. In this subsection, we introduce the notions of time evolution. Because local Hamiltonians can be viewed as special examples of admissible Hamiltonians, below we will focus on admissible Hamiltonians.

Defining the time evolution generated by an admissible Hamiltonian is tricky. Naively, one can use the following exponential 
\beq
\alpha_{t}(A)\stackrel{?}{:=}\exp(it\delta_{H})(A)=\sum_{k=0}^{\infty}\frac{(it)^{k}}{k!}\delta_{H}^{k}(A)
\eeq
However, this definition does not work in general since $\delta_{H}$ is only defined on $D(\delta_{H})$ as defined in Def.~\ref{def:local_derivation}, a dense subset of $A^{ql}$, and typically $\delta_H(A)\in\A^{ql}$ even if $A\in\A^l$! Thus, $\delta_{H}^{2}(A)=\delta_{H}(\delta_{H}(A))$ is ill-defined.

Note that the proof of Lemma \ref{lemma:existence_Hamiltonian} hints us that we can define the $n$-truncated Hamiltonian
\beq
H_{n}:=\sum_{Z:Z\subset [-n,n]}h_{Z}
\eeq
This is a local operator supported on $[-n,n]$ hence it is finite dimensional. One can exponentiate it to define
\beq \label{eq: sequence of evolution}
\alpha_{n}^{t}(A):=e^{iH_{n}t}Ae^{-iH_{n}t}
\eeq
where $A$ is a local operator and $\alpha_{n}^{t}$ is an automorphism of local operators (actually of $\A^{ql}$). Our goal is to show that the limit 
\beq\label{eq:limit_dynamics}
\alpha^{t}(A):=\lim_{n\to\infty}\alpha^{t}_{n}(A)
\eeq
exists and it defines the time evolution of $\delta_{H}$.

\begin{proposition}
    [Existence of dynamics for admissible Hamiltonians]\label{theorem:existence_dynamics}
    The limit defined by Eq. \eqref{eq:limit_dynamics} exists and it defines the time evolution generated by the admissible Hamiltonian $\delta_{H}$, which is a strongly continuous one-parameter subgroup of automorphisms on $\A^{ql}$.\footnote{By strongly continuous, we mean that $\lim_{t\to 0}\alpha_{t}(A)=A$ for any $A\in \A^{ql}$.}
\end{proposition}

This proposition can be easily proved with the help of Theorem 2.2 of Ref.~\cite{Nachtergaele_2006}, which we review below. To this end, we introduce the notion of reproducing functions. Consider a lattice $\Lambda$ with a metric $d$ and let $F:\R^{\geqslant0}\to\R^{\geqslant 0}$ be a decreasing function with $\lim_{r\to \infty}F(r)=0$. This function $F$ is called reproducing if
\beq\label{eq:reproducing}
\begin{split}
    \sup_{y\in\Lambda}\sum_{x\in\Lambda}F(d(x,y))&<\infty\\
    \sum_{l\in\Lambda}F(d(n,l))F(d(l,m))&<C F(d(n,m)),\,\forall\,m,n\in\Lambda
\end{split}
\eeq
for some $0<C<\infty$. Especially, if $\Lambda=\z$, then it can be checked that $F(r)=(1+r)^{-1-\epsilon}$ is reproducing for any $\epsilon>0$. Now we come to the most important lemma in this section, which, in particular, can be used to prove Proposition \ref{theorem:existence_dynamics}.

\begin{lemma}[Theorem 2.2 of Ref. \cite{Nachtergaele_2006}]\label{lemma:existence_dynamics}
    If the Hamiltonian $H=\sum_{Z}h_{Z}$ satisfies
    \beq\label{eq:F_norm}
    ||H||_{F}:=\sup_{m,n\in\Lambda}\frac{1}{F(d(m,n))}\sum_{Z:m,n\in Z}||h_{Z}||<\infty
    \eeq
    for some reproducing function $F(r)$ defined above, then the limit in Eq.~\eqref{eq:limit_dynamics} exists and defines a strongly continuous one-parameter subgroup of $\Aut(\A^{ql})$. This convergence is uniform for $t$ in some compact sets. It also does not depend on the choice of $n$-truncation.
\end{lemma}
\begin{proof}[Proof of Proposition \ref{theorem:existence_dynamics}]
    We only have to check that our admissible Hamiltonian satisfies Eq.~\eqref{eq:F_norm} with some proper choice of $F$. We choose $F(r)=(1+r)^{-1-\epsilon}$ with $0<\epsilon<\mathfrak{a}-2$, then $||H||_{F}$ defined in Eq.~\eqref{eq:F_norm} becomes
\beq
||H||_{F}=\sup_{m,n\in\z}(1+d)^{1+\epsilon}\sum_{Z:m,n\in Z}||h_{Z}||=\sup_{m,n\in\z}(1+d)^{1+\epsilon}\sum_{r=d}^{\infty}\sum_{\substack{Z:m,n\in Z\\ \diam(Z)=r}}||h_{Z}||
\eeq
where $d:=d(m,n)$ and note $\diam(Z):=\sup_{x,y\in Z}d(x,y)\geqslant d $. According to Eq.~\eqref{eq:admissible_H}, we have
\beq
\sum_{\substack{Z:m,n\in Z\\ \diam(Z)=r}}||h_{Z}||\leqslant \sup_{n\in Z}\sum_{\substack{Z:n\in Z\\ \diam(Z)=r}}||h_{Z}||\leqslant \frac{J}{r^{\mathfrak{a}}}
\eeq
Therefore, it only remains to note that
\beq
\sup_{d\in\z^{\geqslant0}}(1+d)^{1+\epsilon}\sum_{r=d}^{\infty}\frac{J}{r^{\mathfrak{a}}}<\sup_{d\in\z^{\geqslant 0}}(1+d)^{1+\epsilon}\int_{d}^{\infty}\frac{J}{r^{\mathfrak{a}}}\dd r=\frac{J}{\mathfrak{a}-1}\sup_{d\in\z^{\geqslant0}}(1+d)^{1+\epsilon}d^{1-\mathfrak{a}}
\eeq
For large $d$, the right hand side can be estimated as $d^{2+\epsilon-\mathfrak{a}}\to 0$ if $d\to\infty$. Hence, $||H||_{F}$ defined in Eq.~\eqref{eq:F_norm} must be bounded. By Lemma \ref{lemma:existence_dynamics}, the limit defined in Eq.~\eqref{eq:limit_dynamics} exists and defines a strongly continuous one-parameter subgroup of $\Aut(\A^{ql})$.
\end{proof}

Given the operator algebra $\A^{ql}$ with a well-defined time evolution $\alpha^{t}$, one says that they together define a $C^*$-dynamical system.

\begin{remark}
    One may wonder whether we can define the Hamiltonian on whole $\A^{ql}$ by 
    \beq
    \delta_{H}(A)\stackrel{?}{:=}\lim_{t\to0}\frac{\alpha^{t}(A)-A}{t},\forall\,A\in\A^{ql}
    \eeq
    However, despite that $\alpha_{t}(A)$ is continuous in $t$, it is in general not differentiable in $t$ if $A\not\in\A^{l}$, so the derivative above does not exist in general. Nevertheless, from the definition in Eq. \eqref{eq: sequence of evolution}, one sees that the derivative exists if $A\in\A^{l}$ and it coincides with our earlier definition of $\delta_{H}(A)$. The quickest way to see this is to write $\alpha^{t}(A)=\lim_{n\to\infty}\alpha^{t}_{n}(A)$ and change the order of limits $t\to 0$ and $n\to\infty$ (this is valid by the uniform convergence in Lemma \ref{lemma:existence_dynamics}). Note for finite $n$, $\frac{\dd}{\dd t}\alpha_{n}^{t}(A)=H_{n}(A)$ where $H_{n}(A)$ is defined in Eq.~\eqref{eq:local_Cauchy_sequence}. It is shown in Lemma \ref{lemma:existence_Hamiltonian} that the limit $\lim_{n\to\infty}H_{n}(A)$ exists if $A$ is a local operator.
\end{remark}

In fact, the above Proposition \ref{theorem:existence_dynamics} can be easily generalized to higher dimensions using exactly the same argument, but with a rather different choice of reproducing function $F$.

\begin{proposition}
[Higher dimensional version of Theorem \ref{theorem:existence_dynamics}]
    For $D$-dimensional lattice $\Lambda\simeq \z^{D}$, if the Hamiltonian satisfies the admissible condition Eq. \eqref{eq:admissible_H} with $\mathfrak{a}>2D$, then the dynamics exists.
\end{proposition}
The proof is also to apply Lemma \ref{lemma:existence_dynamics} but with $F(r)=(1+r)^{-D-\epsilon}$.

Below we present a simple application of the concept of time evolution, which is expected from the usual quantum mechanics.

\begin{corollary}\label{corollary:inv_gs}
    Given a Hamiltonian $\delta_{H}$ which generates a time-evolution $\alpha^{t}$ and $\psi$ the ground state of this Hamiltonian (see Eq. \eqref{eq:ground_states}), then $\psi$ is invariant under $\alpha^{t}$, \ie $\psi\circ \alpha^{t}=\psi$.
\end{corollary}

\begin{proof}[Proof of corollary \ref{corollary:inv_gs}]
    To show that $\psi\circ\alpha^{t}=\psi$, we first show that
    \beq
    \psi(\alpha^{t}(A))=\psi(A),\forall\,A\in\A^{\ell}
    \eeq
    After this step, we extend this result to all $A\in\A^{ql}$.
    
    Since $\alpha^{t}$ is differentiable with respect to $t$ on local operators, this amounts to showing that
    \beq
    \psi(\delta_{H}(A))=0,\,\forall\,A\in\A^{\ell}
    \eeq
    To this end, note that
    \beq
    A=\frac{A+A^{\dagger}}{2}+i\frac{A-A^{\dagger}}{2i}
    \eeq
    So without any loss of generality, we can assume $A$ to be Hermitian due to the linearity of $\psi$ and $\delta_H$. Besides, we can always make a shift $A\to A+\lambda I$ where $\lambda\in\R$, so we can further assume that $A$ has only positive eigenvalues.
    Given the assumption that $A$ is positive (\ie $A$ is self-adjoint and all eigenvalues are positive), there is a positive root of square of $A$, denoted as $\sqrt{A}$, which is again local. Thus, by definition \ref{def:local_derivation},
    \beq
    \psi(\delta_{H}(A))=\psi(\delta_{H}(\sqrt{A}\sqrt{A}))=\psi(\delta_{H}(\sqrt{A})\sqrt{A})+\psi(\sqrt{A}\delta_{H}(\sqrt{A}))
    \eeq
    Note that 
    \beq
    \psi(\sqrt{A}\delta_{H}(\sqrt{A}))^{*}=\psi((\sqrt{A}\delta_{H}(\sqrt{A}))^{\dagger})=-\psi(\delta_{H}(\sqrt{A})\sqrt{A})
    \eeq
    where we have used that $(\delta_{H}(B))^{\dagger}=\sum_{Z}[h_{Z},B]^{\dagger}=-\delta_{H}(B^{\dagger})$ since $h_{Z}^{\dagger}=h_{Z}$.
    On the other hand, since $\psi$ is a ground state of $\delta_{H}$, we have $\psi(\sqrt{A}\delta_{H}(\sqrt{A}))\geqslant 0$, in particular, $\psi(\sqrt{A}\delta_{H}(\sqrt{A}))\in\R$. Therefore, $\psi(\sqrt{A}\delta_{H}(\sqrt{A}))=-\psi(\delta_{H}(\sqrt{A})\sqrt{A})$ and 
    \beq
    \psi(\delta_{H}(A))=0
    \eeq
    So we have shown that 
    \beq
    \psi(\alpha^{t}(A))=\psi(A),\forall\,A\in\A^{\ell}
    \eeq
    For general $A\in\A^{ql}$, we use the standard trick to approximate $A$ by local operators, and the result follows from the fact $\alpha^{t}(\lim_{j\to\infty}A_{j})=\lim_{j\to\infty}\alpha^{t}(A_{j})$ and $\psi$ is continuous (see Proposition 2.3.11 of Ref. \cite{bratteli2013operator1}).
\end{proof}

Notice the corollary above holds as long as the time evolution generated by the Hamiltonian $\delta_H$ exists, regardless of whether this Hamiltonian is admissible or not.

\subsection{The GNS Hamiltonian}

It is natural to ask whether there is a way to realize the abstract Hamiltonian $\delta_{H}$ as an operator on some Hilbert space (especially, the GNS Hilbert space $\cH_{\psi}$ for some state $\psi$, see Sec. \ref{subsec: GNS construction} for the relevant construction). Now we address this question in this subsection. Also, using the concepts and techniques developed so far, in this subsection we will prove Theorem \ref{Thm:pure_gs}, which plays a vital role in one of our proofs of Theorem \ref{thm:main} in the main text (see Appendix \ref{sec: theorem 1}).

We start with the following proposition.

\begin{proposition}[Corollary 2.5.8 of Ref. \cite{Naaijkens_2017}]\label{prop:uniqueness_GNS}
    Suppose $\alpha$ is an automorphism of $\A^{ql}$ and $\psi$ is a state which is invariant under $\alpha$. Then on the GNS Hilbert space $\cH_{\psi}$, there exists an operator $U_{\alpha}\in\B(\cH_{\psi})$ such that
    \beq
    \pi_{\psi}(\alpha(A))=U_{\alpha}\pi_\psi(A)U_{\alpha}^{-1}
    \eeq
\end{proposition}
\begin{proof}
    Suppose $(\pi_\psi, \cH_\psi, |\psi\ra)$ is a GNS triple. To show the above proposition, according to Corollary \ref{corollary:unique:GNS}, it suffices to show that $(\pi_{\psi}\circ\alpha,\cH_{\psi},|\psi\ra)$ is also a GNS triple. To this end, first notice that $\pi_\psi\circ\alpha$ is a homormorphism from $\A^{ql}$ to $\B(\cH_\psi)$. Next, notice that for any $A\in\A^{ql}$, $\la\psi|\pi_\psi\circ\alpha(A)|\psi\ra=\la\psi|\pi_\psi(\alpha(A))|\psi\ra=\psi(\alpha(A))=\psi\circ\alpha(A)=\psi(A)$. Combining these two observations, we conclude that $(\pi_{\psi}\circ\alpha,\cH_{\psi},|\psi\ra)$ is indeed also a GNS triple. Then the proposition follows due to Corollary \ref{corollary:unique:GNS}.
    \end{proof}

As a corollary, we have
\begin{corollary}
    Let $\psi$ be a ground state of a Hamiltonian $\delta_{H}$, which generates a time evolution (\ie a strongly continuous one-parameter subgroup of automorphisms on $\A^{ql}$) on the GNS Hilbert space $\cH_{\psi}$, then there exists a so-called GNS Hamiltonian $H_{\GNS}$ such that
    \beq \label{eq: temp eq}
    \pi_{\psi}(\delta_{H}(A))=[H_{\GNS},\pi_{\psi}(A)],\forall\,A\in\A^{\ell}
    \eeq
    Typically, the GNS Hamiltonian is not in $\pi_{\psi}(A^{ql})$ but in its double commutant $\pi_{\psi}(\A^{ql})''$.
\end{corollary}
\begin{proof}
    Corollary \ref{corollary:inv_gs} implies that
    \beq
    \psi\circ\alpha^{t}=\psi
    \eeq
    Thus by Proposition \ref{prop:uniqueness_GNS}, there exists a unitary operator $U_{t}\in\B(\cH_{\psi})$, such that
    \beq
    \pi_{\psi}(\alpha^{t}(A))=U_{t}\pi_{\psi}(A)U_{t}^{\dagger},\ \forall \,A\in\A^{ql}
    \eeq
    When $A\in\A^{\ell}$, taking derivatives on both sides and writing\footnote{This derivative exists due to Theorem 10.15 of Ref.~\cite{hall2013quantum}.} $H_{\GNS}=-i\frac{\dd}{\dd t}U_{t}$, we get
    \beq
    \pi_{\psi}(\delta_{H}(A))=[H_{\GNS},\pi_{\psi}(A)],\ \forall\,A\in D(\delta_{H})
    \eeq
    which proves Eq. \eqref{eq: temp eq}. To see that the GNS Hamiltonian is in the double commutants $\pi_\psi(\A^{ql})''$, we refer to Corollary 3.2.48 of Ref. \cite{bratteli2013operator1}.
\end{proof}

Now we are ready to establish the following important theorem.

\begin{theorem}\label{Thm:pure_gs}
    Let $\psi$ be a locally-unique gapped ground state of an admissible Hamiltonian, then
    \begin{enumerate}
        \item It is gapped and non-degenerate (\ie unique) in the GNS Hilbert space.
        \item This state $\psi$ is pure.
    \end{enumerate}
\end{theorem}

The proof of this theorem relies on the following few lemmas.

\begin{lemma}[Riesz representation theorem]
Let $\cH$ be a Hilbert space (finite or infinite-dimensional) and $F$ be a bounded linear functional on $\cH$, then there exists a unique state $|F\ra$ in $\cH$, such that
\beq
\la F|\psi\ra=F(|\psi\ra)
\eeq
for any $|\psi\ra\in \cH$. This justifies the usual notion that bras are dual to kets.
\end{lemma}

For a proof, we refer to standard textbooks on functional analysis \eg theorem 2.E in Ref.~\cite{Zeidler1995}. As a corollary, we have

\begin{corollary}\label{corollary:Riesz}
    Consider any bounded sesquilinear form $(\bullet,\bullet):\cH\times\cH\to\bbC$, \ie $(\bullet,\bullet)$ is linear in the second argument and anti-linear in the first argument, and furthermore, \footnote{We are not using bra-ket notations in this example because it makes everything messy.}
    \beq
    \frac{|(u,v)|}{||u||\cdot||v||}<\infty,\,\forall\,u,v\in\cH.
    \eeq
    There exists a bounded linear operator $T$ on $\cH$ such that
    \beq
    (u,v)=\la u,Tv\ra,\forall\,u,v\in \cH
    \eeq
    where $\la\bullet,\bullet\ra$ is the inner product on $\cH$.
\end{corollary}

To understand this corollary, one defines a bounded linear functional by $(u,\bullet):\cH\to\bbC$. By the Riesz representation theorem, there exists another vector $T^{\dagger}(u)\in\cH$ such that
\beq
(u,v)=\la T^{\dagger}u,v\ra
\eeq
It is easy to see that $T^{\dagger}$ defined above is a linear operator. Hence its adjoint exists,
\beq
(u,v)=\la T^{\dagger}u,v\ra=\la u,Tv\ra
\eeq

\begin{lemma}[Theorem 9.17 of Ref.~\cite{Landsman:2017hpa}]
    The GNS time evolution operator $e^{itH_{\GNS}}\in\pi_{\psi}(\A^{ql})''$, where $\psi$ is the ground state of $\delta_{H}$ and $\pi_{\psi}(\A^{ql})''$ means the double commutant of $\pi(\A^{ql})$ in $\B(\cH_{\psi})$.
\end{lemma}
Now we present the proof to Theorem \ref{Thm:pure_gs}.
\begin{proof}[Proof of Theorem \ref{Thm:pure_gs}]
    Our proof here is rather heuristic and readers are refered to theorem A.3 of Ref. \cite{Tasaki2022topological} for a more rigorous one. The technique is exactly the same.
    
    First of all, by the definition of locally-unique ground state, other possible ground states fall into different superselection sectors. Hence $H_{\GNS}$ has a non-degnerate gapped ground state $|\psi\ra$.

    Next, we show that $\psi$ is pure. To this end, we invoke definition \ref{def:alternative_pure}. Let $\rho:\A^{ql}\to\bbC$ be another positive linear functional, such that $\psi-\rho$ is again positive. We have to show that $\rho=\lambda\psi$ for some $\lambda\in [0,1]$.

    Firstly, we define a sesquilinear form on the GNS Hilbert space $\cH_{\psi}$,
    \beq
    (|B\ra,|A\ra):=\rho(B^{\dagger}A)
    \eeq
    On the other hand, by Cauchy-Schwarz inequality and the fact that $\rho\leqslant\psi$,
    \beq
    \begin{split}
        |\rho(B^{\dagger}A)|^{2}&\leqslant \rho(B^{\dagger}B)\rho(A^{\dagger}A)\\
        &\leqslant \psi(B^{\dagger}B)\psi(A^{\dagger}A)\\
        &=||\,|B\ra||^{2}\cdot||\,|A\ra||^{2}
    \end{split}
    \eeq
    Therefore, $(|B\ra,|A\ra):=\rho(B^{\dagger}A)$ is a bounded sesquilinear form on $\cH_{\psi}$.
    By corollary \ref{corollary:Riesz}, there exists a linear operator $T_\rho$ on $\cH_{\psi}$, such that
    \beq
    (|B\ra,|A\ra)=\la B|T_\rho|A\ra=\la\psi|\pi_{\psi}(B)^{\dagger}T_\rho\pi_{\psi}(A)|\psi\ra
    \eeq
    Now note for any $A,B,C\in\A^{ql}$, we have
    \beq
    \begin{split}
    \la\psi|\pi_{\psi}(B)^{\dagger}T_\rho\pi_{\psi}(C)\pi_{\psi}(A)|\psi\ra=\rho(B^{\dagger}CA)=\rho((C^{\dagger}B)^{\dagger}A)&=\la\left(\psi|\pi_\psi(C^\dag B)\right)^\dag T_\rho\pi_\psi(A)|\psi\ra\\
    &=\la\psi|\pi_{\psi}(B)^{\dagger}\pi_{\psi}(C)T_\rho\pi_{\psi}(A)|\psi\ra
    \end{split}
    \eeq
    Hence $\pi_{\psi}(C)T_\rho=T_\rho\pi_{\psi}(C),\forall\,C\in \A^{ql}$. Equivalently, we have $T_\rho\in\pi_{\psi}(\A^{ql})'$. To show that $\psi$ is pure, it suffices to show $T_\rho=\lambda I$ for some $0\leqslant\lambda\leqslant 1$. To this end, note that $H_{\GNS}\in \pi_{\psi}(\A^{ql})''$, especially $[H_{\GNS},T_\rho]=0$, so
\beq
H_{\GNS}T_\rho|\psi\ra=T_\rho H_{\GNS}|\psi\ra=0
\eeq
Therefore, $T_\rho|\psi\ra$ is another ground state of $H_{\GNS}$, which must be proportional to $|\psi\ra$ since $H_{\GNS}$ is non-degenerate on $\cH_{\psi}$. So $T_\rho|\psi\ra=\lambda|\psi\ra$ and hence $T_\rho=\lambda I$ on $\cH_{\psi}$. So $\rho=\lambda\psi$. By assumption, $\rho$ is positive and $\rho\leqslant\psi$, thus we conclude $0\leqslant\lambda\leqslant 1$, \ie $\psi$ is pure.

\end{proof}

In the above proof, the condition that the Hamiltonian is admissible is used to show that the GNS Hamiltonian or its corresponding dynamics exists. In Sec. \ref{subsec: thermodynamic limit of SRE}, we will define a set of Hamiltonians known as nearly local Hamiltonians. These Hamiltonians are not admissible, but their GNS Hamiltonian and dynamics still exist. Hence there one can show that the locally-unique gapped ground states of nearly local Hamiltonians are also pure.

\section{Area law for infinite systems}\label{sec:area_law}

\begin{figure}
    \centering
    \begin{tikzpicture}[node distance=2cm]
\node (start) [process, align=center] {$H_{\phi}$: orgininal (GNS) Hamiltonian\\
with a unique gapped ground $|0\ra$};

\node (trunc) [process, below of=start, yshift=-0.5cm] {Truncated Hamiltonian $H_{t}$ Eq.\eqref{eq:truncated_Hamiltonian}};

\node (folding) [process, below of=trunc, yshift=-0.5cm] {Folded Hamiltonian $H_{t}'$ Eq.~\eqref{eq:folded_truncated_Hamiltonian}};

\node (effective) [process, below of=folding, yshift=-0.5cm] {Effective Hamiltonian $\tH$ Eq.~\eqref{eq:effective_H}};

\node (AGSP) [process, below of=effective, yshift=-0.5cm,align=center] {Approximate ground state projection\\ 
(AGSP) Eq.~\eqref{eq:AGSP}};

\node (low_entangled_state) [process, below of=AGSP, yshift=-1.5cm,align=center] {Prop.~\ref{prop:low_entangled_state}:\\
Existence of a state $|\varphi\ra$ with low Schmidt rank};

\node (decay) [process, right of=start, xshift=7cm, align=center] {Locally unique gapped ground state of\\
an admissible Hamiltonian obeying Eq.~\eqref{eq:admissible}
};

\node (truncatedgap) [process, right of=trunc, xshift=7cm,align=center] {$H_{t}$ again has a unique gapped \\ ground state $|0_t\ra$ with $||\,|0\ra-|0_{t}\ra||\lesssim\frac{ql^{-\bar{\fa}}}{\Delta}$ \\ and energy gap $\Delta_{t}\gtrsim\Delta-ql^{-\bar{\fa}}$};

\node (cutoff) [process, right of=effective, xshift=7cm,align=center] {For $\tau\gg\log q$, $\tH$ has a unique gapped\\ ground state $|\tilde 0_t\ra$ with
$||\,|0_{t}\ra-|\tilde{0}_{t}\ra||\lesssim \sqrt{q}e^{-O(\tau)}$\\
and energy gap $\tilde{\Delta}_{t}\geqslant\frac{1}{2}\Delta$};

\node (AGSP_tool) [process, right of=AGSP, xshift=7cm,align=center] {From AGSP to area law: Prop.~\ref{prop:overlap} and Prop.~\ref{prop:eebound}.\\ 
Schmidt rank of AGSP: Lemma \ref{lemma:schmidt_bound} and Prop.~\ref{prop:bound_SR}.\\
Erorr of AGSP: Eqs.~\eqref{eq:AGSP_delta} and \eqref{eq:AGSP_epsilon}.};

\node (entropy_bound) [process, right of=low_entangled_state, xshift=7cm,yshift=-0.5cm,align=center] {Prop.~\ref{prop:bounding_entropy}:\\
For any $|\psi\ra$ with $||\,|0\ra-|\psi\ra||<\frac{1}{2}$,\\
$S(|0\ra)\leqslant S(|\psi\ra)+S_{0}$\\
where $S_{0}$ depends on $\{k,\bar{\fa},g_{0},d,\Delta\}$.};

\node (area_law) [process, below of=low_entangled_state, yshift=-0.5cm] {Theorem \ref{thm:area_law}: Area law};

\draw [arrow] (start) -- node[anchor=east,align=center] {Keep only\\ nearest-block interactions}(trunc);

\draw [arrow] (decay) -- node[anchor=south,align=center]{$(2q+3)$ blocks\\ (see Fig.~\ref{fig:blocking})} (start);

\draw [arrow] (trunc) -- node[anchor=east,align=center]{Folding\\ (see Fig.~\ref{fig:folding})} (folding);

\draw [arrow] (truncatedgap) -- node[anchor=south,align=center]{Lemma \ref{lemma:spec_stability}} (trunc);

\draw [arrow] (cutoff) -- node[anchor=south,align=center]{Theorem \ref{thm:cut_off}} (effective);

\draw [arrow] (folding) -- node[anchor=east,align=center]{Eerngy cutoff with energy\\ scale $\tau$ below Eq.~\eqref{eq:energy_cutoff}} (effective);

\draw [arrow] (effective) -- node[anchor=east,align=center]{Chebyshev Polynomial} (AGSP);

\draw [arrow] (AGSP_tool) -- node[anchor=east,align=center]{}(AGSP);

\draw [arrow] (AGSP) -- node[anchor=east,align=center]{Prop.~\ref{prop:overlap} and Prop.~\ref{prop:bound_SR}\\ $q=O(l^{\bar{\fa}})$} (low_entangled_state);

\draw [arrow] (AGSP) -- node[anchor=south,align=center,xshift=0.2cm]{Prop.~\ref{prop:eebound}\\ Lemma \ref{lemma:schmidt_bound}\\
$q=2$}(entropy_bound);

\draw [arrow] (low_entangled_state) -- node[anchor=south,align=center]{}(area_law);
\draw [arrow] (entropy_bound) -- node[anchor=south,align=center]{}(area_law);

\end{tikzpicture}
    \caption{The flow chart for proving the entanglement area law. Some of the results are stated informally above, and the corresponding precise statements can be found in the paper.}
    \label{fig:flowchart}
\end{figure}

As explained in the main text, the area law of gapped ground states plays an important role in quantum many-body physics. The area law was first proved in 1D finite-size systems by Hastings, using his factorization lemma \cite{Hastings2007}. Later, this lemma was extended to infinite-size systems \cite{matsui2011boundedness} and higher dimensions \cite{Ukai2024}. In particular, Ref. \cite{Ukai2024} established the area law for gapped ground states of Hamiltonians satisfying Eq.~\eqref{eq:reproducing}. However, in our proof, an alternative method called approximate ground state projection (AGSP) is employed, which can lead to a better estimation and efficient algorithms for the ground states \cite{Arad2013, Landau2013AGSP}. Previous applications of the ASGP method are for finite-size systems, and part of the novelty of our proof is to extend it to infinite-size systems.

Using the so-called approximate ground state projection (AGSP) technique, it is proved in Ref.~\cite{Kuwahara2019} that, for a finite-size system with a gapped admissible Hamiltonian that has a unique ground state, its ground state satisfies the area law of entanglement entropy. In this appendix, we generalize the AGSP techniques to infinite-size systems, which can be applied to the proof of the area law for locally unique gapped ground states of an admissible Hamiltonian (see Fig. \ref{fig:flowchart} for the strategy of the proof). This generalization further allows us to prove the area law for a sequence of finite systems with increasing sizes, which can have multiple gapped ground states (see Appendix \ref{subapp: SSB area law}}).

Before diving into the details, we explain two main difficulties in this generalization. First, in the proof of Ref.~\cite{Kuwahara2019}, one needs to regroup the lattice and truncate the Hamiltonian for better locality. It is a priori not clear if this construction is well-defined in infinite-size systems, and we have to show that the truncated Hamiltonian is indeed an operator on the GNS Hilbert space. This is resolved in Appendix~\ref{sec:truncation}.
Then the effective Hamiltonian obtained by cutting off energy in Ref.~\cite{Kuwahara2019} can be defined by functional calculus (see e.g. Chapter 10 of Ref.~\cite{Hall2013}).

The second difficulty is that in the case of infinite-size systems, although one can discuss the Ohya-Petz entropy of a half infinite chain (see definition \ref{def: OP entropy}), we would like to consider the von Neumann entropies of a sequence of intervals with increasing sizes, because the latter provides a more refined characterization of the entanglement scaling behavior. To fully use the results established in Ref.~\cite{Kuwahara2019}, which divide a finite-size chain into two halves, we employ a folding trick to reduce our setup to the case of Ref.~\cite{Kuwahara2019} (see Appendix~\ref{sec:folding} and Fig.~\ref{fig:folding}).

After resolving these two difficulties, the construction of AGSP in Ref.~\cite{Kuwahara2019} largely follows, with the modifications highlighted in Appendix \ref{subapp: comparison}. Hence the area law of entanglement entropy can be proved for locally unique gapped ground states of admissible Hamiltonians in quantum spin chains with an infinite size.

\subsection{Factorization of GNS Hilbert space and Schmidt decomposition}\label{sec:factorization&Schmidt}

Let $H=\sum_{Z}h_{Z}$ be an admissible Hamiltonian with a locally unique gapped ground state $\phi$. To prove the area law, the idea is to consider the GNS Hilbert space $\cH_{\phi}$ of $\phi$. According to Theorem \ref{Thm:pure_gs}, $H$ can be represented as an unbounded operator $H_{\phi}$ on $\cH_{\phi}$, which has $|0\ra$ as the \textit{unique gapped ground state} with an energy gap $\Delta$. Because the entanglement entropy in an interval is completely determined by the reduced density matrix of this interval, which is in turn determined by the expectation values of operators supposed on this interval (see example \ref{example:denstiy_matrix}), the entanglement entropy of an interval in $\phi$ is identical to the entanglement entropy of the same interval in $|0\ra$. Therefore, in the following we will focus on the vector state $|0\ra$ in $\cH_\phi$ and show that it satisfies the entanglement area law.
We also assume $H_{\phi}|0\ra=0$. In this setup, we will see that many techniques in Ref.~\cite{Kuwahara2019} can be applied to $H_{\phi}$, since $H_\phi$ has a unique gapped ground state $|0\ra$. But we do note there are some differences in our treatment compared with the analysis in Ref.~\cite{Kuwahara2019}, since we are studying an infinite system. As we will exclusively work on the GNS Hilbert space $\cH_{\phi}$, we will often identify a quasi-local operator with its image under $\pi_{\phi}$.

Note that $\cH_{\phi}$ is factorized in the sense of Proposition \ref{prop:factorize}: Let $\Lambda=\z$ be the lattice and $\Gamma\subset \Lambda$ be a finite subset, then we have
\beq\label{eq:GNS_factorize}
\cH_{\phi}:=\cH_{\Gamma}\otimes \cH_{\Gamma^{c}}
\eeq
 where $\cH_{\Gamma}:=\bigotimes_{k\in\Gamma}\cH_{k}$ is a finite-dimensional Hilbert space. In the contrary $\cH_{\Gamma^{c}}$ is defined by Eq.~\eqref{eq:GNS_factorize} (hence it depends on the choice of $\phi$) and \textbf{not} by tensoring local Hilbert spaces.

The Schmidt decomposition works as follows \cite{matsui2011boundedness}. Given a unit vector $|\psi\ra\in\cH_{\phi}$, one has
\beq
|\psi\ra=\sum_{j=0}^{l}\sqrt{\lambda_{j}}|\xi_{j}\ra\otimes|\eta_{j}\ra
\eeq
where $0<\lambda_{l}<\lambda_{l-1}<...<\lambda_{1}$, $\sum_{l}\lambda_{l}=1$ and $|\xi_{j}\ra\in \cH_{\Gamma}$ (respectively $|\eta_{j}\ra\in\cH_{\Gamma^{c}}$) with following orthogonality condition:
\beq
\begin{split}
    \la\xi_{i}|\xi_{j}\ra&=\delta_{ij}\\
    \la\eta_{i}|\eta_{j}\ra&=\delta_{ij}
\end{split}
\eeq
The minimal choice of $l$ is called the Schmidt rank of $|\psi\ra$ (with respect to $\Gamma$), denoted by $\SR(|\psi\ra,\Gamma)$.

The entanglement entropy can be computed from the Schmidt decomposition via
\beq
S(|\psi\ra,\Gamma)=-\sum_{j=1}^{\SR(|\psi\ra,\Gamma)
}\lambda_{j}^{2}\log(\lambda_{j}^{2})
\eeq

To prove the area law, it is useful to introduce the concept of the Schmidt rank of an operator. An operator $O$ on $\cH_{\phi}$ can be written as
\beq
O=\sum_{j=1}^{m}O_{\Gamma,j}\otimes O_{\Gamma^{c},j}
\eeq
where $O_{\Gamma,j}\in\A_{\Gamma}$ and $O_{\Gamma^{c},j}$ is an operator on $\cH_{\Gamma^{c}}$. The minimal choice of $m$ is also defined as the Schmidt rank of $O$, denoted by $\SR(O,\Gamma)$.

Here we state some basic properties of the Schmidt ranks, which can be proved straightforwardly.
\begin{lemma}\label{lemma:SRproperty}
The followings properties are true for Schmidt ranks.
  \begin{enumerate}
    \item $\SR(A|\psi\ra,\Gamma)\leqslant \SR(A,\Gamma)\cdot\SR(|\psi\ra,\Gamma)$
    \item $\SR(AB,\Gamma)\leqslant \SR(A,\Gamma)\cdot\SR(B,\Gamma)$
    \item $\SR(A+B,\Gamma)\leqslant \SR(A,\Gamma)+\SR(B,\Gamma)$
    \item $\SR(O_{Y},\Gamma)\leqslant d^{|Y|}$, where $O_{Y}$ is an operator supported on $Y\subseteq\z$ and $d$ is the dimension of local Hilbert space.
    \item For any disjoint union $\Gamma=\Gamma_{1}\sqcup\Gamma_{2}$, we have
    \beq
    \SR(A,\Gamma_{1})\leqslant d^{2|\Gamma_{2}|} \SR(A,\Gamma)
    \eeq
    \item If $A$ supports only on $\Gamma$ or $\Gamma^{c}$, the we have $\SR(A,\Gamma)=1$.
\end{enumerate}
\end{lemma}

In the following, we will focus on the case $\Gamma=[0,n]$. To simplify the notations, we write $\SR(O,[0,n])=\SR(O)$.

\subsection{Approximate ground state projection} \label{subapp: defining AGSP}

A powerful tool that can be used to prove the area law is the approximate ground state projection (AGSP), which will be defined now.

For any state $|\Omega\ra\in\cH_{\phi}$, one can define the AGSP as follows.

\begin{definition}\label{def:AGSP}
    Let $K$ be a bounded self-adjoint operator on $\cH_{\phi}$ such that
    \begin{enumerate}
        \item There exists a state $|\Omega_{K}\ra$ such that $K|\Omega_{K}\ra=|\Omega_{K}\ra$.
        \item There exists $\delta_{K}\in[0,1)$ such that
        \beq
        |\,|\Omega\ra-|\Omega_{K}\ra|<\delta_{K}
        \eeq
        \item There exists $\epsilon_{K}\in [0,1)$ such that
        \beq
        ||K(1-|\Omega_{K}\ra\la\Omega_{K}|)||<\epsilon_{K}
        \eeq
        \item There exists $D_{K}>0$ such that $\SR(K)<D_{K}$.
    \end{enumerate}
    Then we say $K$ is a $(\delta_{K},\epsilon_{K},D_{K})$-approximate ground state projection (AGSP) of $|\Omega\ra$, or simply AGSP.
\end{definition}
Let us emphasize that AGSP $K$ is hardly a projection (\ie $K^{2}\not=K$ in general). In practice, the construction of AGSP often requires more input than the state $|\Omega\ra$, such as the parent Hamiltonian and the spectral gap. Of course, it is usually a nontrivial task to construct suitable AGSP.

We are not concerned with the construction of AGSP for now, and it will be constructed in Appendix \ref{subapp: AGSP}. In this and the following two sections, we answer the question: What can AGSP do for bounding the entanglement entropy? The following two propositions explain why AGSP is useful to show the entanglement area law, and their proofs are given in Appendices \ref{subapp: prop overlap} and \ref{subapp: prove eebound}, respectively.

\begin{proposition}\label{prop:overlap}
    Let $K$ be an AGSP with $\epsilon_{K}^{2}D_{K}<\frac{1}{2}$, then there exists a state $|\psi\ra$ with $\SR(|\psi\ra)<D_{K}$ and
    \beq\label{eq:low_rank_app}
    |\,|\psi\ra-|\Omega\ra|<\epsilon_{K}\sqrt{2D_{K}}+\delta_{K}
    \eeq
\end{proposition}
\begin{proposition}\label{prop:eebound}
    Let $|\psi_{D}\ra\in\cH_{\phi}$ be a state with
    \beq
        \begin{split}
        \SR(|\psi_{D}\ra)&<D,\\
        |\,|\psi_{D}\ra-|\Omega\ra|&<\nu_{0}
    \end{split}
    \eeq
    And let $\{K_{p}\}_{p=0}^{\infty}$ (with $K_{0}:=1$) be a sequence of AGSP with parameters $(\delta_{p},\epsilon_{p},D_{p})$ such that $\epsilon_{p},\delta_{p}\searrow0$ as $p\to\infty$. Then for each $K_{p}$, there exists $\theta_{p}\in\R$ such that
    \beq\label{eq:bounded_overlap}
    ||\frac{e^{-i\theta_{p}}K_{p}|\psi_{D}\ra}{||K_{p}|\psi_{D}\ra||}-|\Omega\ra||<\gamma_{p}
    \eeq
    where 
    \beq
    \gamma_{p}:=\frac{\epsilon_{p}}{1-\delta_{p}-\nu_{0}}+\delta_{p}
    \eeq
    In addition, if $\gamma_{p}\leqslant 1$ for all $p$, then the entanglement entropy of $|\Omega\ra$ is bounded by
    \beq
    S\leqslant \log(D)-\sum_{p=0}^{\infty}\gamma_{p}^{2}\log(\frac{\gamma_{p}^{2}}{3D_{p+1}})
    \eeq
    where we have defined $\gamma_{0}=1$.
\end{proposition}

Note that the convergence of the second term depends on the detailed construction of AGSP operator, and for suitable constructions it is indeed convergent (see Eq.~\eqref{eq:series_area_law})

Combining these two propositions, the desired area law can be proved as long as one can construct “good” enough AGSPs where that $D_{K} \epsilon_{K}^{2}<\frac{1}{2}$ and $\gamma_{p}<1$, as required in these propositions.

\subsection{Proof of Proposition \ref{prop:overlap}} \label{subapp: prop overlap}
To prove this proposition, we need two lemmas. The first one is the famous Eckart-Young theorem.

\begin{lemma}[Eckart-Young theorem]\label{lemma:EY}
    Let $\cH_{1}$ and $\cH_{2}$ be (possibly countably infinite-dimensional) Hilbert spaces and $|\varphi\ra\in\cH_{1}\otimes\cH_{2}$ with following Schmidt decomposition
    \beq
    |\varphi\ra=\sum_{i=1}\lambda_{i}|\eta_{i}\ra\otimes |\theta_{i}\ra
    \eeq
    where $\lambda_{1}\geqslant\lambda_{2}\geqslant...>0$, and $\{|\eta_{i}\ra\},\{|\theta_{i}\ra\}$ orthonormal basis of $\cH_{1}$ and $\cH_{2}$ respectively. Then
    \beq
    \min_{\substack{\SR(|\psi\ra)\leqslant r\\ |\,|\psi\ra|=1}}|\,|\psi\ra-|\varphi\ra|^{2}=\sum_{i>r}\lambda_{i}^{2}
    \eeq
    Or equivalently
    \beq
    \max_{\substack{\SR(|\psi\ra)\leqslant r\\ ||\,|\psi\ra||=1}}|\la\psi|\varphi\ra|^{2}=\sum_{i=1}^{r}\lambda_{i}^{2}
    \eeq
\end{lemma}
See, \eg Ref.~\cite{hsing2015theoretical} for a proof.
\begin{lemma}\label{lemma:schmidt_bound}
    Let $K$ be an AGSP with $\epsilon_K^{2}D_{K}\leqslant\frac{1}{2}$, then the Schmidt decomposition of $|\Omega_{K}\ra$
    \beq
    |\Omega_{K}\ra=\sum_{i=1}\lambda_{i}|P_{i}\ra
    \eeq
    satisfies
    \beq
    \lambda_{1}>\frac{1}{\sqrt{2D_{K}}}
    \eeq
\end{lemma}
\begin{proof}
    Let us first write
    \beq
    |P_{1}\ra=\lambda_{1}|\Omega_{K}\ra+\sqrt{1-\lambda_{1}^{2}}|\Omega_{K}^{\perp}\ra
    \eeq
    We have
    \beq
    K|P_{1}\ra=\lambda_{1}|\Omega_{K}\ra+b |\Omega_{K}^{\perp'}\ra
    \eeq
    where $|b|<\epsilon_{K}$. Consider the state
    \beq
    |\psi\ra:=\frac{K|P_{1}\ra}{|K|P_{1}\ra|}
    \eeq
    We immediately have
    \begin{enumerate}
        \item $\SR(|\psi\ra)<D_{K}$
        \item $\la\Omega_{K}|\psi\ra<\frac{\lambda_{1}}{\sqrt{\lambda_{1}^{2}+\epsilon_{K}^{2}}}$
    \end{enumerate}
    Consider the Schmidt decomposition of $|\psi\ra$,
    \beq
    |\psi\ra=\sum_{i=1}^{D_{K}}\mu_{i}|\psi_{i}\ra
    \eeq
    Then
    \beq
    \begin{split}
            \frac{\lambda_{1}}{\sqrt{\lambda_{1}^{2}+\epsilon_{K}^{2}}}&<|\la\Omega_{K}|\psi\ra|\\
            &=|\sum_{i=1}^{D_{K}}\mu_{i} \la\Omega_{K}|\psi_{i}\ra|\\
            &\leqslant\sqrt{\sum_{i=1}^{D_{K}}\mu_{i}^{2}} \sqrt{\sum_{i=1}^{D_{K}}|\la\Omega_{K}|\psi_{i}\ra|^{2}}\\
            &\leqslant \sqrt{D_{K}\lambda_{1}^{2}}
    \end{split}
    \eeq
    In the third inequality, we have used the Cauchy-Schwarz inequality. In the last inequality, we have used that fact that $\sum_{i}\mu_{i}^{2}=1$ and the Eckart-Young theorem (Lemma \ref{lemma:EY}):
    \beq
    |\la\Omega_{K}|\psi_{i}\ra|\leqslant \lambda_{1},\quad i=1,2,...,D_{K}
    \eeq
    By the assumption $\epsilon_{K}^{2}D_{K}<\frac{1}{2}$, we now deduce that
    \beq
    \lambda_{1}>\frac{1}{\sqrt{2D_{K}}}
    \eeq
\end{proof}

Now we come back to the proof of Proposition \ref{prop:overlap}.
\begin{proof}[Proof of Proposition \ref{prop:overlap}]
    Consider the Schmidt decomposition of $|\Omega_{K}\ra$,
    \beq
    |\Omega_{K}\ra=\sum_{i=1}\lambda_{i}|P_{i}\ra
    \eeq
    Let us choose 
    \beq
    |\psi\ra:=\frac{K|P_{1}\ra}{|K|P_{1}\ra|}
    \eeq
    Obviously $\SR(|\psi\ra)<D_{K}$, but we need to bound
    \beq
    |\,|\psi\ra-|\Omega\ra|\leqslant |\,|\Omega_{K}\ra-|\Omega\ra|+|\,|\Omega_{K}\ra-|\psi\ra|<\delta_{K}+|\,|\Omega_{K}\ra-|\psi\ra|
    \eeq
    We write 
    \beq
    |\psi\ra=\mu|\Omega_{K}\ra+\sqrt{1-\mu^{2}}|\Omega_{K}^{\perp}\ra
    \eeq
    In the proof of Lemma \ref{lemma:schmidt_bound} we have noticed that $\mu=\la\Omega_{K}|v\ra\geqslant\frac{\lambda_{1}}{\sqrt{\lambda_{1}^{2}+\epsilon_{K}^{2}}}$, so
    \beq
    |\,|\Omega_{K}\ra-|\psi\ra|=\sqrt{2-2\mu}\leqslant\sqrt{2-\frac{2}{\sqrt{1+(\epsilon_{K}/\lambda_{1}})^{2}}}
    \eeq
    Let us note $\sqrt{2-2(1+x^{2})^{-1/2}}<x$ for $x>0$, which can be proved by AM-GM inequality, then
    \beq
    |\,|\Omega_{K}\ra-|\psi\ra|<\frac{\epsilon_{K}}{\lambda_{1}}\leqslant \epsilon_{K}\sqrt{2D_{K}}
    \eeq
    In the second inequality, we have used Lemma \ref{lemma:schmidt_bound}.
    This completes the proof. 
\end{proof}

\subsection{Proof of Proposition \ref{prop:eebound}} \label{subapp: prove eebound}

The first part of the proof follows the same idea as in the proof of Proposition \ref{prop:overlap}. First, let $\{K_{p}\}_{p=1}^{\infty}$ be the sequence of AGSPs described in the proposition. We will not be concerned with the construction of these AGSPs for now. Let us denote $|\Omega_{p}\ra:=|\Omega_{K_{p}}\ra$ for simplicity. The phase $e^{i\theta_{p}}$ is chosen such that
\beq
\la\Omega_{p}|e^{-i\theta_{p}}|\psi_{D}\ra=|\la\Omega_{p}|\psi_{D}\ra|
\eeq
We also denote $\nu_{p}:=|\la\Omega_{p}|\psi_{D}\ra|$.
Now let us write
\beq
\Gamma_{p}:=| \frac{e^{-i\theta_{p}}K_{p}|\psi_{D}\ra}{||K_{p}|\psi_{D}\ra||}-|\Omega\ra|
\eeq
Our first goal is to show 
\beq\label{eq:gammap}
\Gamma_{p}\leqslant\gamma_{p}:=\frac{\epsilon_{p}}{1-\delta_{p}-\nu_{0}}+\delta_{p}
\eeq
To this end, note
\beq
\Gamma_{p}\leqslant | \frac{e^{-i\theta_{p}}K_{p}|\psi_{D}\ra}{|K_{p}|\psi_{D}\ra|}-|\Omega_{p}\ra|+|\,|\Omega_{p}\ra-|\Omega\ra|<| \frac{e^{-i\theta_{p}}K_{p}|\psi_{D}\ra}{|K_{p}|\psi_{D}\ra|}-|\Omega_{p}\ra|+\delta_{p}
\eeq
So it remains to show that
\beq
| \frac{e^{-i\theta_{p}}K_{p}|\psi_{D}\ra}{|K_{p}|\psi_{D}\ra|}-|\Omega_{p}\ra|<\frac{\epsilon_{p}}{1-\nu_{0}-\delta_{p}}
\eeq
We decompose $e^{-i\theta_{p}}|\psi_{D}\ra$ by following orthogonal projection,
\beq
e^{-i\theta_{p}}|\psi_{D}\ra=\nu_{p}|\Omega_{p}\ra+\sqrt{1-\nu_{p}^{2}}|\Omega_{p}^{\perp}\ra
\eeq
Then we have 
\beq
e^{-i\theta_{p}}K_{p}|\psi_{D}\ra=\nu_{p}|\Omega_{p}\ra+b_{p}|\Omega_{p}^{\perp'}\ra,\quad |b_{p}|<\sqrt{1-\nu_{p}^{2}}\epsilon_{p}\leqslant \epsilon_{p}
\eeq
Therefore,
\beq
| \frac{e^{-i\theta_{p}}K_{p}|\psi_{D}\ra}{|K_{p}|\psi_{D}\ra|}-|\Omega_{p}\ra|=\sqrt{2-\frac{2\nu_{p}}{\sqrt{\nu_{p}^{2}+|b_{p}|^{2}}}}<\frac{|b_{p}|}{\nu_{p}}<\frac{\epsilon_{p}}{\nu_{p}}
\eeq
In the first inequality, we have used $\sqrt{2-2(1+x^{2})^{-\frac{1}{2}}}<x$ for $x>0$ (this follows from AM-GM inequality). So it remains to show $\nu_{p}\geqslant 1-\delta_{p}-\nu_{0}$. This can be seen from
\beq
\begin{split}
    \nu_{p}&:=|\la\Omega_{p}|\psi_{D}\ra|\\&\geqslant 1- |\,|\psi_{D}\ra-|\Omega_{p}\ra|\\&\geqslant1-|\,|\psi_{D}\ra-|\Omega\ra|-|\,|\Omega\ra-|\Omega_{p}\ra|\\&>1-\delta_{p}-\nu_{0}
\end{split}
\eeq
The first inequality is due to the fact $2\mathrm{Re}(z)-1\leq|\mathrm{Re}(z)|\leq|z|$ for any $z\in\bbC$ with $|z|\leqslant1$ and take $z=\la\Omega_{p}|\psi_{D}\ra$. This completes the proof of Eq.~\eqref{eq:bounded_overlap}.

Our next task is to bound the entanglement entropy. Let us write the Schmidt decomposition of $|\Omega\ra$ as
\beq
|\Omega\ra:=\sum_{i=1}\mu_{i}|P_{i}\ra
\eeq
where $\mu_{1}\geqslant\mu_{2}\geqslant...>0$ and $\{|P_{i}\ra\}_{i=1,2,...}$ is an unentangled orthonormal basis. The entanglement entropy of $|\Omega\ra$ is given by
\beq
S(|\Omega\ra)=-\sum_{i=1}\mu_{i}^{2}\log(\mu_{i}^{2})
\eeq
By Eckart-Young theorem (lemma \ref{lemma:EY}), we have
\beq
\Gamma_{p}^{2}=| \frac{e^{-i\theta_{p}}K_{p}|\psi_{D}\ra}{|K_{p}|\psi_{D}\ra|}-|\Omega\ra|^{2}\geqslant \sum_{i=1}^{D_{p}'}\mu_{i}^{2}
\eeq
where $D_{p}':=\SR(K_{p}|\psi_{D}\ra)< D D_{p}$. We also define
\beq
\Gamma_{p,p+1}^{2}:=\sum_{D_{p}'<i\leqslant D_{p+1}'}\mu_{i}^{2}
\eeq
Then 
\beq
\Gamma_{p,p+1}^{2}\leqslant\Gamma_{p}^{2}\leqslant\gamma_{p}^{2}\leqslant 1
\eeq
For $D_{p}'< i\leqslant D_{p+1}'$, the contribution to the entropy can be bounded as follows. First, we define two probability distributions, for $i=D_{p}'+1,...,D_{p+1}'$
\beq
\begin{split}
    q_{i}&:=\frac{\mu_{i}^{2}}{\sum_{D_{p}'<i\leqslant D_{p+1}'}\mu_{i}^{2}}\\
    Q_{i}&:=\frac{1}{D_{p+1}'-D_{p}'}
\end{split}
\eeq
Convexity gives the following bound on the contribution from Schmidt coefficients indexed by $D_{p}'<i\leqslant D_{p+1}'$.
\beq
\begin{split}
    S_{p}&:=-\sum_{D_{p}'< i\leqslant D_{p+1}'}\mu_{i}^{2}\log{\mu_{i}^{2}}\\
    &\leqslant-\sum_{D_{p}'<i\leqslant D_{p+1}'}\frac{\Gamma_{p,p+1}^{2}}{D_{p+1}'-D_{p}'}\log\frac{\Gamma_{p,p+1}^{2}}{D_{p+1}'-D_{p}'}\\
    &=-\Gamma_{p,p+1}^{2}\log\frac{\Gamma_{p,p+1}^{2}}{D'_{p+1}-D'_{p}}\\
    &<-\Gamma_{p,p+1}^{2}\log\frac{\Gamma_{p,p+1}^{2}}{D_{p+1}}+\Gamma_{p,p+1}^{2}\log(D)
\end{split}
\eeq
To derive the last line, we have used $\log(D_{p+1}'-D_{p}')<\log(D_{p+1}')\leqslant\log(DD_{p+1})$. Now, the entropy of $|\Omega\ra$ can be bounded by
\beq
\begin{split}
    S(|\Omega\ra)&=\sum_{p=0}^{\infty}S_{p}\\&\leqslant\sum_{p=0}^{\infty}(\Gamma_{p,p+1}^{2}\log(D)-\Gamma_{p,p+1}^{2}\log\frac{\Gamma_{p,p+1}^{2}}{D_{p}})\\
    &=\log(D)-\sum_{p=0}^{\infty}\Gamma_{p,p+1}^{2}\log\frac{\Gamma_{p,p+1}^{2}}{D_{p+1}}\\
    &\leqslant\log(D)-\sum_{p=0}^{\infty}\gamma_{p}^{2}\log\frac{\gamma_{p}^{2}}{3D_{p+1}}
\end{split}
\eeq
where we have used the fact that $-x\log(x)<-x\log(x/3)<-y\log(y/3)$ for $0<x\leqslant y\leqslant 1$. This completes the proof.

\subsection{Outline of the construction of AGSP}
In this section, we briefly outline the main idea of constructing good AGSPs that are applicable for proving entanglement area laws, as is indicated in proposition \ref{prop:overlap} and \ref{prop:eebound}.

Recall that, in our case the state $|\Omega\ra$ is chosen as the unique gapped ground state $|0\ra$ of the GNS Hamiltonian $H_{\phi}$. So naively, the AGSP could be taken as $K\propto e^{-\beta H_{\phi}}$. As $\beta\to\infty$, this operator projects out all excited states while leaves the ground state $|0\ra$ invariant. However, this operator has a large Schmidt rank, which violates the last requirement in definition \ref{def:AGSP}.
In the case of local Hamiltonians \cite{Arad2013arealaw}, the Hamiltonian $H_{\phi}$ itself is of controllable Schmidt rank. So instead, one uses a polynomial function of the Hamiltonian to construct the AGSP.

However, note that in our definition, the AGSP must be a bounded self-adjoint operator. Surely any polynomial of Hamiltonian is unbounded so it cannot serve as AGSP. To resolve this problem, we need to cut-off the Hamiltonian and replace it by a bounded operator.

\subsection{Interaction Truncation of Hamiltonian}\label{sec:truncation}

We note the following convenient lemma.
\begin{lemma}\label{lemma:quasi_local_perturb}
    Let $H=\sum_{Z}h_{Z}$ be an admissible Hamiltonian and $X,Y\subseteq\z$  be any two concatenated subsets with $\dist(X,Y)=L>0$, then
    \beq
    V_{X,Y}(\Gamma):=\sum_{\substack{Z\subset\Gamma \\Z\cap X\not=\emptyset\\Z\cap Y\not=\emptyset}}h_{Z}
    \eeq
    is a well-defined element in $\A^{ql}$ for any choice of $\Gamma\subset\z$ (note $X,Y,\Gamma$ can be infinite). Besides, there exists a constant $g_{0}\geqslant 1$ such that
    \beq
    ||V_{X,Y}(\Gamma)||<g_{0}L^{-\bar{\fa}}
    \eeq
    where 
    \beq\label{eq:g0}
    g_{0}=\max\{\frac{\fa}{\fa-2}J,1\},\quad \bar{\fa}=\fa-2>0
    \eeq
\end{lemma}
\begin{proof}
    Let us define
    \beq
    V^{(n)}_{X,Y}(\Gamma):=\sum_{\substack{Z\subset \Gamma\\Z\cap X\not=\emptyset\\Z\cap Y\not=\emptyset\\\diam(Z)\leqslant n}}h_{Z}
    \eeq
    It is easy to see that for a finite $n$, each $V^{(n)}_{X,Y}(\Gamma)$ contains only finite local terms. Next, we show it is a Cauchy sequence. 

    For $m\geqslant n$, consider
    \beq
    ||V^{(m)}_{X,Y}(\Gamma)-V^{(n)}_{X,Y}(\Gamma)||\leqslant\sum_{r=n}^{m}\sum_{\substack{i\in X\\ \dist(i,Y)\leqslant r}}\sum_{\substack{Z:Z\ni i\\ \diam(Z)=r}}||h_{Z}||
    \eeq
    The RHS can be estimated as follows
    \beq
    \sum_{r=n}^{m}\sum_{\substack{i\in X\\ \dist(i,Y)\leqslant r}}\sum_{\substack{Z:Z\ni i\\ \diam(Z)=r}}||h_{Z}||<\sum_{r=n}^{m}\frac{J}{r^{\mathfrak{a}}}(r-L)<J(\frac{-L}{(\fa-1)(n-1)^{\fa-1}}+\frac{1}{(\fa-2)(n-1)^{\fa-2}})\stackrel{n\to\infty}{\to}0
    \eeq
    In the second step we note that there are at most $(r-L)$ choices of $i\in X$ such that $\dist(i,Y)\leqslant r$ since both $X,Y$ are concatenated subsets.
    The last step is obtained by replacing summation by integrals. So it is clear that $\{V^{(n)}_{X,Y}(\Gamma)\}_{n=1,2...}$ is a Cauchy sequence in $\A^{\ell}$. Hence it converges to an element in $\A^{ql}$. 

    The estimation of its norm follows from Lemma 1 of Ref.~\cite{Kuwahara2019}.
\end{proof}

Furthermore, we immediately obtain
\beq\label{eq:g_extensive}
\begin{split}
    \sup_{i\in\z}\sum_{Z:Z\ni i}||h_{Z}||&< B+\sup_{i\in\z}\sum_{r=1}^{\infty}\sum_{\substack{Z:Z\ni i\\\diam(Z)=r}}||h_{Z}||\\
    &<B+\sum_{r=1}^{\infty}\frac{J}{r^{\fa}}\\
    &\leqslant B+\frac{\fa J}{\fa-1}=:g
\end{split}
\eeq
In Ref.~\cite{Kuwahara2016exponential}, Eq.~\eqref{eq:g_extensive} is called $g$-extensive condition. We adopt this term from now on.

Now, we regroup the lattice into $(2q+3)$ blocks $\{B_{s}\}_{s=0,1...,2q+2}$, where $q$ is an even integer. We require that $|B_{s}|=l$ for $s=1,...,q,q+2,...,2q+2$ and $|B_{q+1}|=n-ql$, where $n\gg ql$ is another integer (see Fig.~\ref{fig:blocking}). The right boundary of $B_{\frac{q}{2}}$ (resp. $B_{\frac{3q}{2}+1}$) is exactly the site $0$ (resp. site $n$).
For now, $q$ and $l$ are free parameters, and their values will be specified later. We always assume $n\gg ql$.
\begin{figure}[h!]
    \centering
    \includegraphics[width=0.8\linewidth]{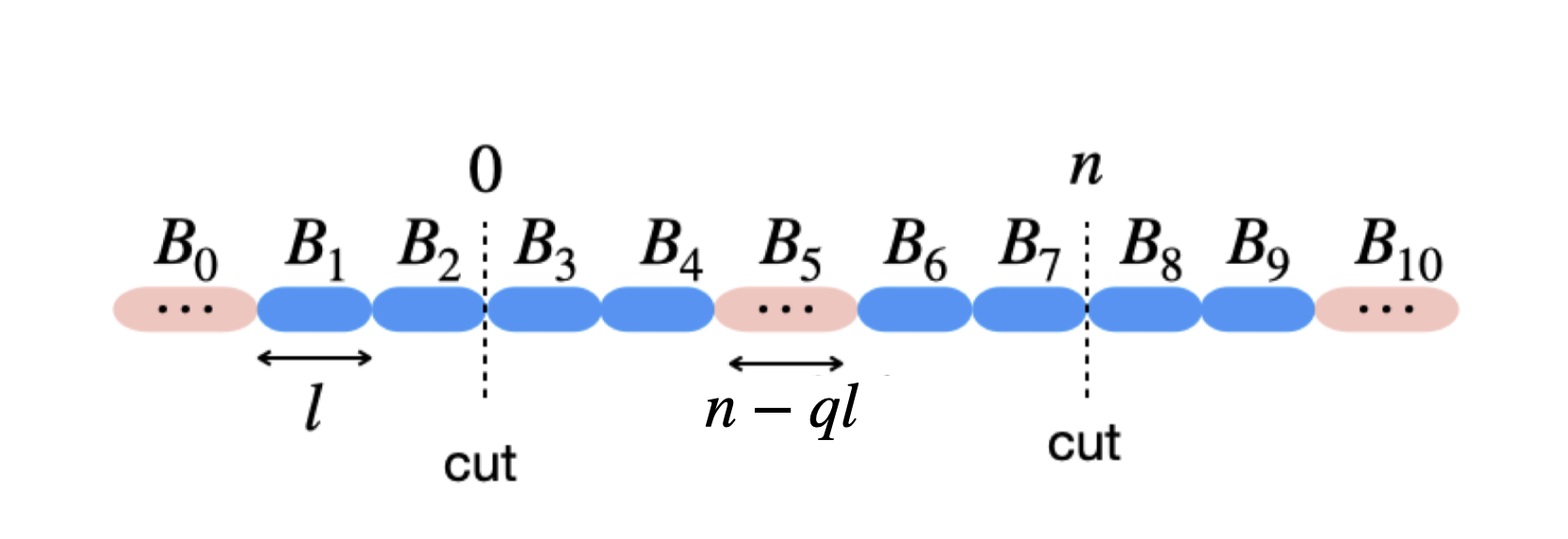}
    \caption{The original lattice $\z$ is regrouped into $(2q+3)$ blocks. The case $q=4$ is shown in the figure. The length of blue blocks is $l$ while the length of pink blocks is infinite or depends on $n$.}
    \label{fig:blocking}
\end{figure}

The interaction truncation of $H_{\phi}$ is obtained by discarding all interactions between non-adjacent blocks, \ie
\beq\label{eq:truncated_Hamiltonian}
H_{t}:=H_{\phi}-\delta H=H_{\phi}-\sum_{s=0}^{2q+2} V_{B_{s},X_{s}}(\Lambda_{s})
\eeq
where we have defined
\beq
\begin{split}
    X_{s}&:=\bigcup_{j\geqslant s+2}B_{j}\\
    \Lambda_{s}&:=\bigcup_{j\geqslant s}B_{j}\\
    \delta H&=\sum_{s=0}^{2q+2} V_{B_{s},X_{s}}(\Lambda_{s})
\end{split}
\eeq
Note that $H_{t}$ is a well defined (unbounded) self-adjoint operator on $\cH_{\phi}$ because it is obtained by subtracting finitely many quasi-local terms (by Lemma \ref{lemma:quasi_local_perturb}). It is also easy to see that
\beq\label{eq:norm_perturbation}
||\delta H||\leqslant \sum_{s=0}^{2q+3}||V_{B_{s},X_{s}}(\Lambda_{s})||\leqslant2(q+2)g_{0}l^{-\bar{\fa}}
\eeq
we have used $n\gg ql$ and lemma \ref{lemma:quasi_local_perturb}. Hence it is a bounded operator on $\cH_{\phi}$. 

Since $\cH_{\phi}$ is infinite dimensional and $H_{\phi}$ is unbounded, it is only defined\footnote{In general, $D(H_\phi)$ is smaller than $\cH_\phi$. For example consider a quantum mechanical model, say 1d harmonic oscillator. We denote its eigenstate as $|n\ra$ and $H|n\ra=n |n\ra$. Consider the state $|\psi\ra:=\sum_{n=0}^{\infty}\frac{6}{(n+1)\pi^{2}}|n\ra$, which is square-integrable. However, $H|\psi\ra=\sum_{n=0}^{\infty}\frac{6n}{(n+1)\pi^{2}}|n\ra$ does not converge, thus $|\psi\ra$ is not in the domain of $H$, although it is square-integrable hence is a valid state.} on a dense subset $D(H_{\phi})\subseteq\cH_{\phi}$ (called the domain of $H_{\phi}$) \cite{hall2013quantum}. In particular, the self-adjointness $H_{\phi}=H_{\phi}^{\dagger}$ requires that $H_{\phi}$ and $H_{\phi}^{\dagger}$ have the same domain, which can fail in general. For example, momentum operators in an infinite potential well is \textbf{not} self-adjoint (proposition 9.27 of Ref. \cite{hall2013quantum}). In our situation, the self-adjointness of $H_{t}$ follows from the following proposition,
\begin{proposition}[Proposition 9.13 of Ref.~\cite{hall2013quantum}]\label{prop:domain}
    Let $A$ be an unbounded self-adjoint operator on a separable Hilbert space $\cH$ and $B$ be a bounded self-adjoint operator, then $A+B$ is again an unbounded self-adjoint operator with the same domain as $A$.
\end{proposition}
See Ref.~\cite{hall2013quantum} for a proof. In our situation, $A=H_{\phi}$ and $B=\delta H$. Below we will show that under some mild conditions (see Lemma \ref{lemma:spec_stability}), $H_{t}$ is again uniquely gapped with the ground state $|0_{t}\ra$. Proposition \ref{prop:domain} ensures that $H_{\phi}|0_{t}\ra$ and $H_{t}|0\ra$ are well defined, which will be vital for the proof of lemma \ref{lemma:spec_stability}, especially for Eq.~\eqref{eq:truncated_gs}. 

Another complication in the infinite-size setup is that although the spectral decomposition still makes sense, the notion of spectrum and eigenvalues must be distinguished.
\begin{definition}
    Let $\cH$ be a (possibly infinite-dimensional) Hilbert space and $A$ be an operator on $\cH$, then its spectrum is defined to be the set
    \beq
    \sigma(A):=\{z\in\bbC|(z-A)\text{ is not invertible}\}
    \eeq
    On the other hand, $\lambda\in\bbC$ is called an eigenvalue of $A$ if there exists a vector $|\psi\ra\in \cH$ (called the eigenvector or eigenstate) such that
    \beq
    A|\psi\ra=\lambda |\psi\ra
    \eeq
\end{definition}
In finite-dimensional spaces, $z\in\sigma(A)$ is equivalent to saying $z$ is an eigenvalue of $A$. However, they are not equivalent in infinite-dimensional spaces. As an example, consider $\cH=L^{2}(\R)$ and $A=-i\frac{\dd}{\dd x}$. It is easy to verify that $\sigma(A)=\R$ but it has no eigenvalues since plane waves are not normalizable and do not belong to $L^{2}(\R)$. Nevertheless, for any $\epsilon>0$, one can define the $\epsilon$-almost eigenstate.
\begin{definition}[$\epsilon$-almost eigenstate]\label{def:ep_eigen}
    Let $\cH$ be a Hilbert space and $A$ be an operator on $\cH$. For $\epsilon>0$, a (normalized) state $|\psi\ra\in\cH$ is said to be an $\epsilon$-almost eigenstate with eigenvalue $\lambda$ if
    \beq
    |A|\psi\ra-\lambda|\psi\ra|<\epsilon
    \eeq
\end{definition}
For future reference, we state the spectral theorem for unbounded operators and its important corollaries.
\begin{theorem}[Theorem 10.4 of Ref.~\cite{hall2013quantum}] \label{thm: spectral decomposition of generic operator}
    Let $A$ be a (possibly unbounded) self-adjoint operator on a Hilbert space $\cH$, then there exists a unique projection-valued measure $\mu^{A}$ on the spectrum $\sigma(A)$ of $A$, such that
    \beq
    \int_{\sigma(A)}\lambda\, \dd\mu^{A}(\lambda)=A
    \eeq
\end{theorem}
In a more familiar (and less rigorous) notation, $\int_{\sigma(A)}\dd\mu^{A}(\lambda)\sim\sum_{\lambda\in\sigma(A)}|\lambda\ra\la\lambda|$. The reason to avoid objects like $|\lambda\ra$ is because $A$ may contain continuous spectrum so it can happen that $A$ does not possess a (normalizable) $\lambda$-eigenstate for $\lambda\in\sigma(A)$.

Using this spectral decomposition, we can define the so-called functional calculus below.
\begin{definition}[Functional calculus, Def.~10.5 of Ref. \cite{hall2013quantum}]
    For any measurable function $f$ on $\sigma(A)$, we can define a (possibly unbounded) operator $f(A)$,
    \beq
    f(A):=\int_{\sigma(A)}f(\lambda)\dd\mu^{A}(\lambda)
    \eeq
\end{definition}

It is useful to define the spectral projection,
\beq
\Pi_{I}^{A}:=\int_{\lambda\in I}\dd\mu^{A}(\lambda)
\eeq
where $I\subset \R$ is a Borel subset\footnote{Recall that a Borel subset of $\R$ is a subset which can be obtained from open intervals by taking countable unions, countable intersections and complements.}. It is well known that if $\lambda\in\sigma(A)$, then $\Pi_{I}^{A}$ has a nonzero image for any neighborhood $I$ of $\lambda$ (see Theorem 10.20 of Ref.~\cite{hall2013quantum} for a proof). For example, an $\epsilon$-almost eigenstate $|\psi\ra$ of $A$ for $\lambda\in\sigma(A)$ can be constructed as follows. By assumption, $\lambda\in \sigma(A)$, so $\Pi^{A}_{(\lambda-\epsilon,\lambda+\epsilon)}$ has a nonzero image, and we choose $|\psi\ra\in \im(\Pi^{A}_{(\lambda-\epsilon,\lambda+\epsilon)})$. One can check that this $|\psi\ra$ is indeed an $\epsilon$-almost eigenstate. Using a slightly different form of spectral theorem, one can also define spectral multiplicity (\eg ground state degeneracy) of an unbounded self-adjoint operator $A$, see Sec.~7.3 of Ref.~\cite{hall2013quantum}. For our purpose, we only need spectral multiplicity for point spectrum (\ie $\lambda\in\sigma(A)$ and it has a neighborhood $I$ such that $I\cap \sigma(A)=\lambda$).
\begin{definition}
    Let $\lambda\in\sigma(A)$ be a point spectrum of self-adjoint operator $A$, we define its spectral multiplicity to be 
    \beq
    M_{\lambda}:=\lim_{\epsilon\to 0^{+}}\rk(\Pi_{(\lambda-\epsilon,\lambda+\epsilon)})
    \eeq
\end{definition}
\begin{remark}
    Since $\lambda$ is a point spectrum, it follows that $\rk(\Pi_{(\lambda-\epsilon,\lambda+\epsilon)})$ becomes a constant for sufficiently small $\epsilon$. This ensures it is well-defined. Following the same argument, one can show that every point spectrum is actually an eigenvalue.
\end{remark}

Now we show that $H_{t}$ also has a unique gapped ground state for sufficiently small perturbations $||\delta H||$. The following is the infinite-
dimensional generalization of Lemma 3 and Lemma 4 in Ref.~\cite{Kuwahara2019}.
\begin{lemma}\label{lemma:spec_stability}
     If $||\delta H||<\Delta/4$, then $H_{t}$ has a unique ground state $|0_{t}\ra$ with energy gap $\Delta_{t}>\Delta-2||\delta H||$. Besides,
    \beq\label{eq:truncated_gs}
    | |0\ra-|0_{t}\ra|<\frac{||\delta H||}{\Delta-4||\delta H||},
    \eeq
    if the global phase factors of $|0\ra$ and $|0_t\ra$ are chosen such that $\la 0|0_t\ra\geqslant 0$.
\end{lemma}

The proof of this lemma in the finite dimensional case relies on the so-called Weyl's inequality, which is not directly applicable to the infinite dimensional case. Here we present a general proof that is applicable to both finite and infinite dimensional cases based on resolvents, which will be briefly reviewed below (see Sec.17.7 of Ref.~\cite{hassani2013mathematical} for more details).

Let $A$ be any self-adjoint operator (bounded or not), its resolvent is defined as (by functional calculus)
\beq\label{eq:resolvent}
    R(z,A):=\frac{1}{z-A}, \quad z\not\in \sigma(A)
\eeq
We note $R(z,A)$ is nothing but the Green's function.
Note by definition, $R(z,A)$ exists iff $z\not\in\sigma(A)$. One important property of $R(z,A)$ is that its operator norm
\beq
||R(z,A)||\leqslant\frac{1}{d(z,\sigma(A))}
\eeq
where $d(z,\sigma(A))$ is the usual distance in $\bbC$.
Besides, let $\Gamma$ be a contour which encircles exactly one point $\lambda\in\sigma(A)$, then
\beq
P_{\lambda}:=\oint_{\Gamma}\frac{\dd z}{2\pi i}R(z,A)
\eeq
is the projection to $\lambda$-eigensubspace. In general (\ie $\lambda$ may not be an eigenvalue), $P_{\lambda}$ defined via the above expression still makes sense and it still serves as a spectral projection.

\begin{proof}[Proof of Lemma \ref{lemma:spec_stability}]
    If $||\delta H||<\Delta/2$, then for any $||\delta H||<z<\Delta-||\delta H||$, 
    we have
    \beq
    ||R(z,H_{\phi})\delta H||\leqslant \frac{||\delta H||}{d(z,\sigma(H_{\phi}))}<1
    \eeq
    This means for such a choice of $z$, $(1+R(z,H_{\phi})\delta H)$ is invertible. This implies that
    \beq
    R(z,H_{t})=R(z,H_{\phi})(1+R(z,H_{\phi})\delta H)^{-1}
    \eeq
    is well-defined and hence $z\not\in\sigma(H_{t})$. The same argument is true if $z<-||\delta H||$. This means the ground state energy $E_{0}$ of $H_{t}$ falls into either $[-||\delta H||,||\delta H||]$ or $[\Delta-||\delta H||,\infty)$. 
    
    Below we show that the ground state $|0_{t}\ra$ is non-degenerate if $||\delta H||<\Delta/4$ and has energy in $(-||\delta H||,||\delta H||)$.
    
    We choose $\Gamma$ to be the circle centered at origin with radius $r$ satisfying $2||\delta H||<r\leqslant\Delta/2$. Consider
    \beq
    \begin{split}
        P_{\phi}&=\oint_{\Gamma}\frac{\dd z}{2\pi i}R(z,H_{\phi})\\
        P_{t}&=\oint_{\Gamma}\frac{\dd z}{2\pi i}R(z,H_{t})
    \end{split}
    \eeq
    The goal is to show $||P_{\phi}-P_{t}||<1$, which implies $\rk(P_{t})=\rk(P_{\phi})=1$. To see this, suppose $\rk(P_{t})>\rk(P_{\phi})$, let us denote $V:=\mathrm{im}(P_{t})$ and $W:=\mathrm{im}(P_{\phi})$. By assumption, $\dim V>\dim W=1$. The restriction of $P_{\phi}|_{V}:V\to W$ cannot be injective for dimensional reasons. Thus there exists $v\in V$ such that $P_{t}v=v$ but $P_{\phi}(v)=0$, so we have $||P_{t}-P_{\phi}||\geqslant 1$. The same argument also shows if $\rk(P_{t})<\rk(P_{\phi})$ then $||P_\phi-P_t||\geqslant 1$.

    Now we note
    \beq
    P_{\phi}-P_{t}=-\oint_{\Gamma}\frac{\dd z}{2\pi i}(z-H_{\phi})^{-1}\delta H (z-H_{t})^{-1}
    \eeq
    The norm of RHS can be estimated as follows. Note that $d(r,\sigma(H_{t}))> r-||\delta H|| $. So
    \beq\label{eq:rank_bound}
    ||\mathrm{RHS}||\leqslant \frac{r ||\delta H||}{d(r,\sigma(H_{\phi})) d(r,\sigma(H_{t}))}<\frac{||\delta H||}{r-||\delta H||}<1
    \eeq
    for any $r$ satisfying $2||\delta H||<r\leqslant\frac{\Delta}{2}$. Therefore, as long as $||\delta H||<\Delta/4$, we have $||P_{\phi}-P_{t}||<1$, which means $\rk(P_{t})=\rk(P_{\phi})=1$ and $H_{t}$ has a unique ground state $|0_{t}\ra$ with energy $E_{0}\in(-r,r)$. The interval $(-r, r)$ does not intersect with $[\Delta-||\delta H||, \infty)$, so $E_0\in[-||\delta H||, ||\delta H||]$. Since there is no other energy eigenstates with energy $E\in (||\delta H||,\Delta -||\delta H||)$, the energy gap can be estimated as $\Delta_{t}>(\Delta-||\delta H||)-||\delta H||=\Delta-2||\delta H||$.

    In the rest of Sec.~\ref{sec:area_law} and Sec.~\ref{sec:proof_cut_off}, we shift the energy by a constant such that $H_{t}|0_{t}\ra=0$. In this convention, $H_{\phi}|0\ra\not=0$ in general.

    To prove Eq.~\eqref{eq:truncated_gs}, we choose the phase of $|0_{t}\ra$ such that
    \beq
    \la0|0_{t}\ra:=\mu_{1}\geqslant0
    \eeq
    Then one can expand $|0\ra$ as
    \beq
    |0\ra=\mu_{1}|0_{t}\ra+\mu_{2}|\psi_{t,\perp}\ra
    \eeq
    where $\la0_{t}|\psi_{t,\perp}\ra=0$ and $\mu_{1}^{2}+|\mu_{2}|^{2}=1$.The vector $(\mu_{1},\mu_{2})^T$ should be an eigenvector of the following $2\times2$ matrix,
    \beq\label{eq:eigenproblem}
    \begin{pmatrix}
        \la0_{t}|H_{\phi}|0_{t}\ra\ra& \la0_{t}|H_{\phi}|\psi_{t,\perp}\ra\\
        \la\psi_{t,\perp}|H_{\phi}|0_{t}\ra&\la\psi_{t,\perp}|H_{\phi}|\psi_{t,\perp}\ra
    \end{pmatrix}=:\begin{pmatrix}
        f_{0}&f\\
        f^{*}&f_{\perp}
    \end{pmatrix}
    \eeq
    Note that clearly $|0\ra$ and $|0_{t}\ra$ are in the domain of $H_{\phi}$ (the latter is due to proposition \ref{prop:domain}), so $|\psi_{t,\perp}\ra$ is also in this domain. Therefore, the matrix in Eq.~\eqref{eq:eigenproblem} is well defined. Using the fact that $|0\ra$ is the lowest-energy eigenstate of $H$, by direct calculations we find that (we will further show that $f_\perp>f_0$ later)
    \beq
    \begin{split}
        \left|\frac{\mu_2}{\mu_1}\right|=\left(\frac{\sqrt{(\frac{f_0-f_\perp}{2})^2+|f|^2}+\frac{f_0-f_\perp}{2}}{\sqrt{(\frac{f_0-f_\perp}{2})^2+|f|^2}-\frac{f_0-f_\perp}{2}}\right)^{1/2}\leqslant\frac{|f|}{f_\perp-f_0}
    \end{split}
    \eeq

    Also
    \beq
    \mu_{1}=\frac{1}{\sqrt{1+|\mu_{2}/\mu_{1}|^{2}}}\geqslant 1-\frac{1}{2}|\frac{\mu_{2}}{\mu_{1}}|^{2}
    \eeq
    Therefore,
    \beq
    ||\,|0\ra-|0_{t}\ra||^{2}=2-2\mu_{1}\leqslant|\frac{\mu_{2}}{\mu_{1}}|^{2}<\frac{|f|^{2}}{(f_{\perp}-f_{0})^{2}}
    \eeq
    which is equivalent to
    \beq
    ||\,|0\ra-|0_{t}\ra||\leqslant\frac{|f|}{f_{\perp}-f_{0}}
    \eeq
    Now let us estimate the matrix elements in Eq.~\eqref{eq:eigenproblem}. First of all, 
    \beq
    |f|=|\la0_{t}|H_{\phi}|\psi_{t,\perp}\ra|=|\la0_{t}|\delta H|\psi_{t,\perp}\ra|\leqslant||\delta H|0_{t}\ra||\leqslant ||\delta H||
    \eeq
    To estimate $f_{\perp}$, we have
    \beq
    f_{\perp}=\la\psi_{t,\perp}|H_{\phi}|\psi_{t,\perp}\ra\geqslant\Delta_{t}+\la\psi_{t,\perp}|\delta H|\psi_{t,\perp}\ra> \Delta-3||\delta H||
    \eeq
    where we have used $\la\psi_{t,\perp}|H_{t}|\psi_{t,\perp}\ra\geqslant\Delta_{t}>\Delta-2||\delta H||$.

    At last, we have
    \beq
    f_{0}=\la0_{t}|H_{\phi}|0_{t}\ra=\la0_{t}|\delta H|0_{t}\ra\leqslant ||\delta H||
    \eeq
    So $f_{\perp}-f_{0}>\Delta-4||\delta H||>0$ by assumption and we conclude
    \beq
    ||\,|0\ra-|0_{t}\ra||\leqslant\frac{|f|}{f_{\perp}-f_{0}}<\frac{||\delta H||}{\Delta-4||\delta H||}.
    \eeq
\end{proof}

From this point on, we will always assume $||\delta H||<\frac{\Delta}{4}$ and keep the convention that $H_{t}|0_{t}\ra=0$ rather than $H_{\phi}|0\ra=0$. For future convenience, we also provide an upper bound for $\Delta$ and $\Delta_{t}$.
\begin{lemma}\label{lemma:upbound_origgap}
    Let $H$ be a gapped admissible Hamiltonian with a unique ground state $|0\ra$ and energy gap $\Delta$, then $\Delta$ is upper bounded by $2g$, where
    \beq
    g:=B+\frac{\fa J}{\fa-1}
    \eeq
\end{lemma}
\begin{proof}
    Note that
    \beq
    \sup_{i\in\z}\sum_{Z:Z\ni i}||h_{Z}||\leqslant B+\sup_{i\in\z}\sum_{r=1}^{\infty}\sum_{\substack{Z:Z\ni i\\\diam(Z)=r}}||h_{Z}||\leqslant B+\frac{\fa J}{\fa-1}=g
    \eeq
    Then, let us write $H=V_{i}+H_{\Lambda_{i}}$, where
    \beq
    V_{i}:=\sum_{Z:Z\ni i}h_{Z}, \quad i\in\z
    \eeq
     and $\Lambda_{i}:=\Lambda\setminus\{i\}$ and $H_{\Lambda_{i}}$ only supports on $\Lambda_{i}$. We have $||V_{i}||\leqslant g$. So a similar analysis used in Lemma \ref{lemma:spec_stability} shows that $H_{\Lambda_{i}}$ is again a self-adjoint operator with lower-bounded spectrum. Now we set $\min (\sigma(H_{\Lambda_{i}}))=E_{i}$ and for any $\epsilon>0$, the $\epsilon$-almost eigenstate $|0_{\Lambda_{i}}^{\epsilon}\ra$ (see definition~\ref{def:ep_eigen}) satisfies
     \beq
    E_{i}\leqslant\la0_{\Lambda_{i}}^{\epsilon}|H_{\Lambda_{i}}|0_{\Lambda_{i}}^{\epsilon}\ra<\epsilon+E_{i}
     \eeq
     We consider
    \beq
    |\psi^{\epsilon}\ra=|\psi_{i}\ra\otimes |0_{\Lambda_{i}}^{\epsilon}\ra
    \eeq
    where $|\psi_{i}\ra$ is a state on site $i$ such that $\la0|\psi^{\epsilon}\ra=0$. Then we have
    \beq
    \la\psi^{\epsilon}|H|\psi^{\epsilon}\ra-\la0|H|0\ra\geqslant \Delta
    \eeq
    On the other hand, we have
    \beq
    \la\psi^{\epsilon}|H|\psi^{\epsilon}\ra\leqslant \la 0^{\epsilon}_{\Lambda_{i}}|H_{\Lambda_{i}}|0^{\epsilon}_{\Lambda_{i}}\ra+||V_{i}||< E_{i}+\epsilon+g
    \eeq
    and
    \beq
    \la 0|H|0\ra=\la 0|H_{\Lambda_{i}}|0\ra+\la 0|V_{i}|0\ra\geqslant E_{i}-||V_{i}||\geqslant E_{i}-g
    \eeq
    Therefore,
    \beq
    \Delta<\epsilon+2g
    \eeq
    Note that $\epsilon$ is arbitrarily small, we conclude
    \beq
    \Delta\leqslant 2g
    \eeq
    \end{proof}
    The upper bound for $\Delta_{t}$ is obtained by applying this lemma to $H_{t}$.
    In the following, we rescale the Hamiltonian to ensure that $g=1$. Note that we always have $g_{0}\geqslant g=1$ by choosing a larger $g_{0}$ in Lemma \ref{lemma:quasi_local_perturb}.

\subsection{Folding trick and Schmidt rank}\label{sec:folding}

After the truncation introduced in last section, only interactions within each block and between two adjacent blocks remain. This allows us to use a folding trick (see Fig.~\ref{fig:folding}), which reduces our setup to the case of Ref.~\cite{Kuwahara2019}. More precisely, we fold the whole chain with respect to the midpoint of the interval $[0,n]$. Then there are $q+2$ blocks in total, denoted by $B'_s$ with $s=0, 1, \cdots, q+1$, and the dimension of each local Hilbert space is now $d^{2}$, where $d$ is the dimensional of local Hilbert space of original chain. By abusing notations, we again denote the folded lattice by $\Lambda$.

\begin{figure}[h!]
    \centering
    \includegraphics[width=0.8\linewidth]{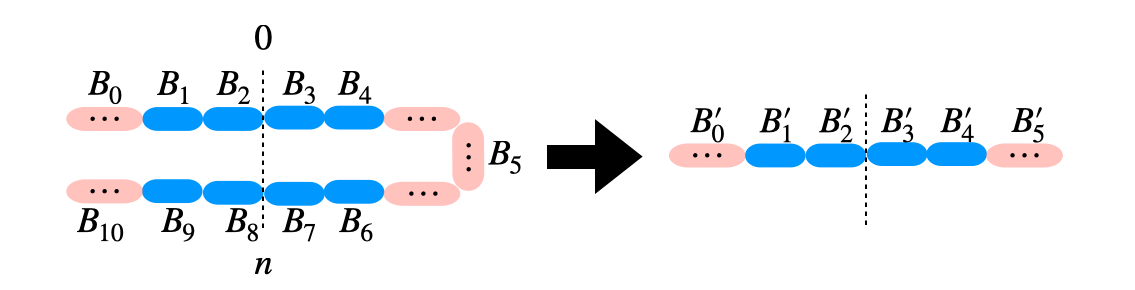}
    \caption{The whole chain is folded with respect to the midpoint of the interval $[0,n]$.}
    \label{fig:folding}
\end{figure}

Since the truncated Hamiltonian $H_{t}$ only contains interactions within each block and nearest-neighbor blocks, the folded Hamiltonian $H_{t}'$ again satisfies this property. For convenience, let us write the Hamiltonian $H_{t}'$ as 
\beq\label{eq:folded_truncated_Hamiltonian}
H_{t}':=\sum_{s=0}^{q+1}h_{s}+\sum_{s=0}^{q}h_{s,s+1}
\eeq
where $h_{s}$ is only supported on\footnote{To be more precise, for $s>0$, each $B_{s}'$ is a compact subset therefore $h_{s}:=\sum_{Z\subset B_{s}'}h_{Z}$ is a finite sum and hence well-defined. We then define $h_{0}:=H_{t}'-\sum_{s=1}^{q+1}h_{s}-\sum_{s=0}^{q}h_{s,s+1}$. By Prop.~\ref{prop:domain} $h_{0}$ has the same domain as $H_{t}'$.} $B_{s}'$ and $h_{s,s+1}$ is the interaction between $B_{s}'$ and $B_{s+1}'$. More explicitly, we have
\beq
h_{s,s+1}= V_{B_{s},B_{s+1}}(B_{s}\sqcup B_{s+1})+V_{B_{2q+2-s},B_{2q+1-s}}(B_{2q+2-s}\sqcup B_{2q+1-s})
\eeq
Using Lemma \ref{lemma:quasi_local_perturb} with $L=1$, we find
\beq\label{eq:bounding_interaction}
||h_{s,s+1}||= ||V_{B_{s},B_{s+1}}(B_{s}\sqcup B_{s+1})+V_{B_{2q+2-s},B_{2q+1-s}}(B_{2q+2-s}\sqcup B_{2q+1-s})||\leqslant 2g_{0}
\eeq

Next, we estimate the Schmidt rank of power of $H_{t}'$, which will be used for the construction of AGSP later.
\begin{lemma}\label{lemma:SR_bound_weak}
    \beq
    \SR((H_{t}')^{m})\leqslant(2+(2d^{2}l)^{k})^{m}
    \eeq
\end{lemma}
\begin{proof}
    Note that $H'_{t}$ can be decomposed into
    \beq
    H'_{t}=H'_{t,L}+h_{\frac{q}{2},\frac{q}{2}+1}+H'_{t,R}
    \eeq
    where $H'_{t,L}$ is supported only on $\bigcup_{s=0}^{\frac{q}{2}}B'_{s}$, while $H'_{t,R}$ is supported only on $\bigcup_{s=\frac{q}{2}+1}^{q+1}B'_{s}$. Therefore, $\SR(H'_{t,L})=\SR(H'_{t,R})=1$ and only $h_{\frac{q}{2},\frac{q}{2}+1}$ contributes to Schmidt rank in a nontrivial way. We note that
    \beq\label{eq:middle_block_inter}
    h_{\frac{q}{2},\frac{q}{2}+1}=\sum_{\substack{Z\cap B'_{\frac{q}{2}}\cap B'_{\frac{q}{2}+1}\not=\emptyset,\\|Z|\leqslant k}}h_{Z}
    \eeq
    For each $Z$, we have $\SR(h_{Z})\leqslant d^{2k}$ since it is supported on at most $k$-sites and each site has now $d^{2}$-dimensional local Hilbert space. If $k\geqslant 3$, the number of possible $Z$ appearing Eq.~\eqref{eq:middle_block_inter} can be estimated as\footnote{Below we have used $\sum_{j=0}^{k}\begin{pmatrix}
        2l\\j
    \end{pmatrix}<(\frac{2el}{k})^{k}$, which essentially follows from Stirling's approximation (see \cite{Jizhe2018binomial} for a detailed proof).}
    \beq
    \sum_{j=1}^{k}(\begin{pmatrix}
        2l\\j
    \end{pmatrix}-2
    \begin{pmatrix}
        l\\j
    \end{pmatrix})< (\frac{2el}{k})^{k}<(2l)^{k}
    \eeq
    On the other hand, if $k=2$, it can be checked that above summation is still bounded by $(2l)^{k}$. We then conclude that
    \beq
    \SR(H_{t}')\leqslant 2+(2d^{2}l)^{k}
    \eeq
    and the desired result follows from $\SR(AB)\leqslant\SR(A)\SR(B)$ for any operators $A$ and $B$.
\end{proof}

For large $q$, a much better bound is possible.
\begin{proposition}\label{prop:bound_SR}
The Schmidt rank of $(H_{t}')^{m}$ is bounded by
\beq
\SR((H_{t}')^{m})\leqslant d^{2ql}(q+m+1)^{q+1}[e(q+1)^{2}(2d^{2}l)^{k}]^{\frac{m}{q+1}}
\eeq
\end{proposition}
The proof is completely identical to Proposition 4 of Ref.~\cite{Kuwahara2019}, but with $d$ replaced by $d^{2}$ because of the folding trick. Later, we will fix $l$ and $m$ as power functions of $q$ (up to some logarithmic factors). In that case, we have $(q+m+1)^{q+1}\leqslant d^{2ql}$ for large $q$. By Proposition~\ref{prop:bound_SR},
\beq
\SR((H_t')^{m})\leqslant d^{4ql}[e(q+1)^{2}(2d^{2}l)^{k}]^{\frac{m}{q+1}}
\eeq
\begin{remark}
    Later we will cut-off $H_{t}'$ and obtain an effective Hamiltonian $\tH$ with a bounded norm. The Schmidt rank estimation in Proposition~\ref{prop:bound_SR} is also valid for $\tH$ since they have the same block-wise locality property. See Proposition 4 and its proof in Ref.~\cite{Kuwahara2019}.
\end{remark}

\subsection{Energy cut-off}

Recall that our AGSP must be a bounded operator in order to be defined on the whole Hilbert space $\cH_{\phi}$. It is hard to construct such an operator out of polynomials of $H_{t}$ because $H_{t}$ itself is unbounded. In this section, we briefly sketch the solution to this problem.

First, let us recall the spectral decomposition of unbounded operators, as stated in Theorem \ref{thm: spectral decomposition of generic operator}, which is based on functional calculus.
Previously, we have encountered a special example of functional calculus, \ie the resolvents defined by Eq.~\eqref{eq:resolvent}, which is used to derive the spectral gap of the truncated Hamiltonian. Similarly, if one chooses $f(x)=e^{-ix}$, then one automatically gets the exponential $e^{-iA}$ of $A$.

To cut off the spectrum of $H_{t}$, let us choose the following
\beq
f_{\lambda}(x)=\begin{cases}
    x,x<\lambda\\
    \lambda,x\geqslant\lambda
\end{cases}
\eeq
Now for each $s=0,1,...,q+1$, we denote $E_{s,0}:=\inf(\sigma(h_{s}))$. In the case where $h_{s}$ has discrete spectrum, $E_{s,0}$ reduces to the lowest eigenvalue of $h_{s}$. We then define
\beq\label{eq:energy_cutoff}
\tilde{h}_{s}:=f_{\tau_{s}}(h_{s})
\eeq
where $\tau_{s}:=\tau+E_{s,0}$ and $\tau$ is a free parameter which will be specified later. Less formally, one could think of
\beq
\tilde{h}_{s}\sim\sum_{E_{s,j}<\tau_{s}}E_{s,j}|E_{s,j}\ra\la E_{s,j}|+\sum_{E_{s,j}\geqslant\tau_{s}}\tau_{s}|E_{s,j}\ra\la E_{s,j}|
\eeq
although the RHS makes sense only for discrete spectra. After the cut-off, the new Hamiltonian
\beq\label{eq:effective_H}
\tilde{H}_{t}':=\sum_{s=0}^{q+1}\tilde{h}_{s}+\sum_{s=0}^{q}h_{s,s+1}
\eeq
The norm of this effective Hamiltonian is bounded by
\beq
||\tilde{H}_{t}'||\leqslant\sum_{s=0}^{q+1}||\tilde{h}_{s}||+\sum_{s=0}^{q}||h_{s,s+1}||\leqslant\sum_{s=0}^{q+1}(\tau+|E_{s,0}|)+2(q+1)g_{0}
\eeq
Now we utilize the following lemma.
\begin{lemma}\label{lemma:energy_shift}
    There exists a set of constants $\{\mathcal{E}_{s}\}_{s=0}^{q+1}$ such that
    \beq
    \begin{split}
        \sum_{s=0}^{q+1}\mathcal{E}_{s}&=0\\
        |E_{s,0}+\mathcal{E}_{s}|&<2g_{0}
    \end{split}
    \eeq
    Thus we can shift $\tilde{h}_{s}\to\tilde{h}_{s}+\mathcal{E}_{s}$. After this shift, $\tilde{h}_{s}$ has ground state energy $E_{s,0}+\mathcal{E}_{s}$ and this operation leaves $\tH$ invariant.
\end{lemma}
\begin{proof}[Proof of Lemma \ref{lemma:energy_shift}]
    For $\epsilon>0$, let us define $|E_{0,s}^{\epsilon}\ra\in\im(\Pi^{(s)}_{[E_{s,0},E_{s,0}+\epsilon)})$, where $E_{s,0}:=\inf(\sigma(h_{s}))$. Consider the state
    \beq
    |\psi\ra:=\bigotimes_{s=0}^{q+1}|E_{0,s}^{\epsilon}\ra
    \eeq
    Since we have set the ground state energy $E_{t,0}$ of $H_{t}'$ to be zero,
    \beq
    \la\psi|H_{t}'|\psi\ra\geqslant E_{t,0}=0
    \eeq
    On the other hand,
    \beq
    \begin{split}
        0\leqslant\la\psi|(\sum_{s=0}^{q+1}h_{s}+\sum_{s=0}^{q}h_{s,s+1})|\psi\ra&\leqslant\sum_{s=0}^{q+1}E_{s,0}+2(q+1)g_{0}
    \end{split}
    \eeq
    where we have used the fact that $\epsilon$ can be taken as any positive number and that $||h_{s,s+1}||\leqslant 2g_{0}$ (see Eq.~\eqref{eq:bounding_interaction}). Moreover,
    \beq
    0=E_{t,0}=\sum_{s=0}^{q+1}\la0_{t}|h_{s}|0_{t}\ra+\sum_{s=0}^{q}\la0_{t}|h_{s,s+1}|0_{t}\ra\geqslant \sum_{s=0}^{q+1}E_{s,0}-2(q+1)g_{0}
    \eeq
    Therefore, we conclude,
    \beq\label{eq:cut_off_gs_energy}
    -2(q+1)g_{0}\leqslant \sum_{s=0}^{q+1}E_{s,0}\leqslant 2(q+1)g_{0}
    \eeq
    We now shift $h_{s}\to h_{s}+\mathcal{E}_{s}$ such that $\sum_{s=0}^{q+1}\mathcal{E}_{s}=0$ and $E_{s,0}+\mathcal{E}_{s}=E_{s',0}+\mathcal{E}_{s'},\,\forall \,s,s'$, which implies
    \beq
    (q+2)(E_{s,0}+\mathcal{E}_{s})=\sum_{s=0}^{q+1}(E_{s,0}+\mathcal{E}_{s})=\sum_{s=0}^{q+1}E_{s,0}
    \eeq
    Equivalently,
    \beq
    \mathcal{E}_{s}=-E_{s,0}+\frac{1}{q+2}\sum_{s=0}^{q+1}E_{s,0}
    \eeq
    Consequently,
    \beq
    -\frac{q+1}{q+2}2g_{0}\leqslant E_{s,0}+\mathcal{E}_{s}\leqslant\frac{q+1}{q+2}2g_{0},\quad\forall\,s
    \eeq
\end{proof}
By Lemma \ref{lemma:energy_shift}, we have
\beq\label{eq:norm_bound}
||\tilde{H}_{t}'||\leqslant(q+2)\tau+2g_{0}(2q+3)\leqslant(q+2)(\tau+4g_{0})
\eeq
For large $\tau$, we should expect $\tilde{H}_{t}'$ should have similar spectral properties as $H_{t}$. Especially, we expect that for large $\tau$, $\tilde{H}_{t}'$ has a unique gapped ground state and its ground state should be close to $|0_{t}\ra$. This can be made precise by the following theorem.
\begin{theorem}\label{thm:cut_off}
    Suppose $\tau$ satisfies
    \beq\label{eq:critical_cut_off_main}
    \tau\geqslant \max\{16g_{0}+\frac{1}{\lambda'}\log(\frac{176g_{0}(q+1)(q+2)}{\Delta_{t}}),8g_{0}+\frac{1}{\lambda}\log(\frac{6480(q+2)g_{0}}{\lambda\Delta_{t}^{2}})\}
    \eeq
    where
    \beq
    \begin{split}
        \lambda&:=\frac{1}{8gk+16g_{0}}\\
        \lambda'&:=\min\{\frac{1}{224g_{0}},\frac{1}{8gk}\}
    \end{split}
    \eeq
    Recall that $k$ is defined by the condition that our original Hamiltonian $H$ contains at most $k$-body interactions. Then $\tilde{H}_{t}'$ is uniquely gapped whose spectral gap $\tilde{\Delta}_{t}$ satisfies
    \beq
    \tilde{\Delta}_{t}\geqslant\frac{1}{2}\Delta_{t}
    \eeq
    Moreover, the norm distance between the ground state $|\tilde{0}_{t}\ra$ of $\tilde{H}_{t}$ and $|0_{t}\ra$ is exponentially small in $\tau$,
    \beq
    |\,|0_{t}\ra-|\tilde{0}_{t}\ra|\leqslant114\sqrt{\frac{(q+2)g_{0}}{\lambda\Delta_{t}^{2}}}e^{-\frac{\lambda}{2}(\tau-8g_{0})},
    \eeq
    if the relative phase of $|0_t\ra$ and $|\tilde 0_t\ra$ are chosen such that $\la 0_t|\tilde 0_t\ra\geqslant 0$.
\end{theorem}
The proof of Theorem \ref{thm:cut_off} is postponed to Appendix~\ref{sec:proof_cut_off}.

\subsection{Construction of AGSP} \label{subapp: AGSP}

In this subsection, we aim to construct a (family of) AGSP with respect to $|0_{t}\ra$ (\ie with $|\Omega\ra=|0_{t}\ra$ in Definition~\ref{def:AGSP}).

Let us assume $\tau$ satisfies Eq.~\eqref{eq:critical_cut_off_main}, therefore by Theorem \ref{thm:cut_off}, the effective Hamiltonian Eq.~\eqref{eq:effective_H} has a unique gapped ground state $|\tilde{0}_{t}\ra$ with energy gap $\tilde{\Delta}_{t}\geqslant\frac{1}{2}\Delta_{t}$. The AGSP is constructed as a polynomial of $\tH$ of degree $m$, denoted by $K(m,\tH)$. The construction for $K(m,\tH)$ is based on the so-called Chebyshev polynomials,
\beq
T_{m}(x):=\frac{(x+\sqrt{x^{2}-1})^{m}+(x-\sqrt{x^{2}-1})^{m}}{2}
\eeq
It is shown in Ref.~\cite{Kuwahara2015local} that
\beq\label{eq:Chebeyshev_bound}
\begin{split}
    |T_{m}(x)|&\leqslant 1,\quad |x|\leqslant1\\
    \frac{1}{2}\exp(2m\sqrt{\frac{|x|-1}{|x|+1}})&\leqslant |T_{m}(x)|\leqslant\frac{1}{2}(2x)^{m},\quad |x|\geqslant 1
\end{split}
\eeq
Using functional calculus, the AGSP is constructed as \cite{Kuwahara2019},
\beq\label{eq:AGSP}
K(m,\tH):=\frac{T_{m}(\frac{2\tH-(||\tH||+\tilde{E}_{t,1})}{||\tH||-\tilde{E}_{t,1}})}{T_{m}(-\frac{||\tH||+\tilde{\Delta}_{t}-\tilde{E}_{t,0}}{||\tH||-\tilde{E}_{t,1}})}
\eeq 
where $\tilde{E}_{t,1}:=\tilde{E}_{t,0}+\tilde{\Delta}_{t}$ and the denominator is chosen to ensure $K(m,\tH)|\tilde{0}_{t}\ra=|\tilde{0}_{t}\ra$. In the rest of this subsection, we will review which parameters determine the AGSP $K(m, \tilde H_t')$ and estimate the parameters $\delta_{K(m, \tilde H_t')}$, $\epsilon_{K(m, \tilde H_t')}$ and $D_{K(m, \tilde H_t')}$ associated with $K(m, \tilde H_t')$ (see Appendix \ref{subapp: defining AGSP} for the meanings of these parameters).

First, the AGSP $K(m, \tilde H_t')$ clearly depends on $m$. It also depends on the Hamiltonian $H_\phi$ we start with, which determines the parameters $g_0$ (see Lemma \ref{lemma:quasi_local_perturb}) and $g$ (see Eq. \eqref{eq:g_extensive}). For a given $H_\phi$, $\tilde H_t'$ depends on $q$ (related to the number of blocks), $l$ (the length of the finite blocks), and $\tau$ (the energy cutoff). The AGSP $K(m, \tilde H_t')$ is ultimately determined by all these parameters.

Using the above parameters, we can bound the parameters $\delta_{K(m, \tilde H_t')}$, $\epsilon_{K(m, \tilde H_t')}$ and $D_{K(m, \tilde H_t')}$ associated with the AGSP $K(m, \tilde H_t')$.

By Theorem \ref{thm:cut_off}, we know that the parameter $\delta$ for $K(m,\tH)$ is given by
\beq\label{eq:AGSP_delta}
|\,|\tilde{0}_{t}\ra-|0_{t}\ra|\leqslant\delta_{K(m,\tH)}:=114\sqrt{\frac{(q+2)g_{0}}{\lambda\Delta_{t}^{2}}}e^{-\frac{\lambda}{2}(\tau-8g_{0})}
\eeq
and by Eq.~\eqref{eq:Chebeyshev_bound}, the parameter $\epsilon_{m}$ of this AGSP is bounded as
\beq\label{eq:AGSP_epsilon}
\begin{split}
    \epsilon_{m}&=\sup_{\tilde E_{t, 0}+\tilde{\Delta}_{t}\leqslant x\leqslant ||\tH||}|K(m,x)|\leqslant \frac{1}{|T_{m}(-\frac{||\tH||+\tilde{\Delta}_{t}-\tilde{E}_{t,0}}{||\tH||-\tilde{\Delta}_{t}-\tilde{E}_{t,0}})|}\\&\leqslant 2\exp(-2m\sqrt{\frac{\tilde{\Delta}_{t}}{||\tH||-\tilde{E}_{t,0}}})\\
    &\leqslant2\exp(-2m\sqrt{\frac{\tilde{\Delta}_{t}}{2||\tH||}})\\
    &\leqslant2\exp(-m\sqrt{\frac{\Delta_{t}}{(q+2)(\tau+4g_{0})}})
\end{split}
\eeq
To obtain third line, we have used $\tilde{E}_{t,0}\geqslant-||\tH||$. In the last line, we have used $\tilde{\Delta}_{t}\geqslant\frac
{1}{2}\Delta_{t}$ (see Theorem \ref{thm:cut_off}) and $||\tH||\leqslant (q+2)(\tau+4g_{0})$ (see Eq. \eqref{eq:norm_bound}).
This also shows that $K(m,\tH)$ is a bounded operator with $||K(m,\tH)||\leqslant 2$.

We note that the Schmidt rank for $K(m,\tH)$ is given by Proposition~\ref{prop:bound_SR}, \ie
\beq\label{eq:AGSP_SR}
\SR(K(m,\tH))\leqslant md^{4ql}[e(q+1)^{2}(2d^{2}l)^{k}]^{\frac{m}{q+1}}
\eeq

\subsection{Proof of the area law} \label{subapp: proving area law}

With all the previous results in hand, the goal of this section is to prove the following entanglement area law.

\begin{theorem}[area law for locally unique gapped ground states]\label{thm:area_law}
    Let $\phi$ be a locally unique gapped ground state of an admissible Hamiltonian $H$ (see Eq.~\eqref{eq:admissible}), then the entanglement entropy 
\beq
S(\phi|_{[0,n]})< S_{1}
\eeq
where $S_{1}$ is a constant depending only on $\{k,J,B,\fa,\Delta,d\}$, but not on $n$.
\end{theorem}

The proof is divided into two propositions. Recall that $|0\ra$ is the GNS vector state representing $\phi$ in $\cH_{\phi}$.
\begin{proposition}\label{prop:low_entangled_state}
    There exists a quantum state $|\varphi\ra$ such that
    \beq
    ||\,|0\ra-|\varphi\ra||\leqslant\frac{1}{2}
    \eeq
    with 
    \beq
    \log(\SR(|\varphi\ra))\leqslant D_{0}
    \eeq
    where $D_{0}$ is a constant that depends only on $\{k,g_{0},\bar{\fa},d,\Delta\}$, but not on $n$.
\end{proposition}

\begin{proposition}\label{prop:bounding_entropy}
    Let $|\varphi\ra$ be any quantum state such that
    \beq
    ||\,|0\ra-|\varphi\ra||\leqslant\frac{1}{2}
    \eeq
    and let $D_{\varphi}:=\SR(|\varphi\ra)$, then the entropy of $\phi$ is bounded as
    \beq
    S(\phi|_{[0,n]})\leqslant \log(D_{\varphi})+S_{0}:=S_{1}
    \eeq
    where $S_{0}$ is a constant depending only on $\{k,g_{0},\bar{\fa},\Delta,d\}$, but not on $n$.
\end{proposition}

Theorem \ref{thm:area_law} can thus be proved by combining these two propositions and noting $\bar{\fa}=\fa-2,g_{0}=\max\{B+\frac{\fa}{\fa-1}J,\frac{\fa}{\fa-2}J\}$. So in the rest of this subsection, we will prove these two propositions. We note that $q$ and $l$ in Fig. \ref{fig:blocking} can be chosen arbitrarily in these proofs.

\begin{proof}[Proof of Proposition~\ref{prop:low_entangled_state}]
    In the following, we choose $q$ and $l$ such that
    \beq
    ||\delta H||\leqslant \frac{\Delta}{8},
    \eeq
    which leads to $\Delta_{t}\geqslant \frac{3}{4}\Delta$ by Lemma \ref{lemma:spec_stability}. From Eq.~\eqref{eq:norm_perturbation}, this can be achieved by 
    \beq
    2g_{0}(q+2)l^{-\bar{\fa}}\leqslant \frac{\Delta}{8}
    \eeq
    So we can take $l=\lceil(\frac{16g_{0}(q+2)}{\Delta})^{\frac{1}{\bar{\fa}}}\rceil=O((\frac{q}{\Delta})^{\frac{1}{\bar{\fa}}})$, where $\lceil x\rceil$ with $x$ a positive number denotes the smallest integer that is not smaller than $x$. Again, by Lemma \ref{lemma:spec_stability}, this choice of $l$ ensures that
    \beq
    |\,|0\ra-|0_{t}\ra|< \frac{1}{4}.
    \eeq
    By Proposition \ref{prop:overlap}, if there exists an AGSP $K$ for $|0_{t}\ra$ with
    \beq\label{eq:AGSP_property1}
    \epsilon_{K}^{2}D_{K}< \frac{1}{2},
    \eeq
    then we obtain a quantum state $|\varphi\ra$ such that
    \beq
    \begin{split}
        |\,|0_{t}\ra-|\varphi\ra|&<\epsilon_{K}\sqrt{2D_{K}}+\delta_{K},\\
        \log(\SR(|\varphi\ra)&< \log(D_{K}).
    \end{split}
    \eeq
    This also implies
    \beq
    |\,|0\ra-|\varphi\ra|< \frac{1}{4}+\epsilon_{K}\sqrt{2D_{K}}+\delta_{K}
    \eeq
    Thus, we can finish the proof if we can find an AGSP $K$ with
    \beq \label{eq: to be satisfied}
    \begin{split}
        \epsilon_{K}\sqrt{2D_{K}}+\delta_{K}&\leqslant \frac{1}{4},\\
        \log(D_{K})&\leqslant D_{0}.
    \end{split}
    \eeq
    where $D_{0}$ is a constant depending on $\{k,g_{0},\bar{\fa},\Delta,d\}$ only, but not on $n$.
    
    We will take the AGSP to be $K=K(m, \tilde H_t')$ discussed in Appendix \ref{subapp: AGSP}, which depends on the original GNS Hamiltonian $H_\phi$ and parameters $\tau$, $m$ and $q$ (it also depends on $l$, but $l$ is fixed to be $\lceil(\frac{16g_{0}(q+2)}{\Delta})^{\frac{1}{\bar{\fa}}}\rceil$, as discussed above). Our goal below is to show that the following choices of the parameters $\tau$, $m$, and $q$ will make Eqs. \eqref{eq: to be satisfied} and \eqref{eq:critical_cut_off_main} hold:
    \beq
    \begin{split}
        \tau&=c_\tau\log(q),\\
       q&=q_{0}:=\lceil\max\{(\frac{480c_{2}}{b_{1}})^{4},(c_{3}^{-1}(3+\bar{\fa}^{-1}))^{\bar{\fa}},\frac{4c_{1}}{b_{1}},(\frac{7\log 2}{6c_{1}})^{\frac{\bar{\fa}}{\bar{\fa}+1}}\}\rceil\\
        m&=\lceil\frac{4c_{1}}{b_{1}}q_{0}^{\frac{3}{2}+\frac{1}{\bar{\fa}}}\log^{1/2}(q_{0})\rceil
    \end{split}
    \eeq
    where the parameters $c_\tau$, $c_1$, $c_2,c_{3},b_{1}$ depend on $\{\bar\fa, g_0, k, \Delta, d\}$ (parameters in the Hamiltonian and the dimension of local Hilbert spaces), but not on $n$ (the size of the subregion under consideration).

    Below we justify this statement.
    
    Recall that the parameter $\delta_{K}$ is bounded by Eq.~\eqref{eq:AGSP_delta},
    \beq
    \delta_{K}\leqslant\frac{114}{\Delta_{t}}\sqrt{\frac{(q+2)g_{0}}{\lambda}}e^{-\frac{\lambda}{2}(\tau-8g_{0})}\leqslant\frac{152}{\Delta}\sqrt{\frac{(q+2)g_{0}}{\lambda}}e^{-\frac{\lambda}{2}(\tau-8g_{0})},
    \eeq
    where we have used $\Delta_{t}\geqslant\frac{3}{4}\Delta$. From Eq.~\eqref{eq:AGSP_epsilon}, we have
    \beq\label{eq:AGSP_epsilon2}
    \epsilon_{K}\leqslant2\exp(-m\sqrt{\frac{\Delta_{t}}{(q+2)(\tau+4g_{0})}})\leqslant 2\exp(-\frac{m}{2}\sqrt{\frac{3\Delta}{(q+2)(\tau+4g_{0})}}),
    \eeq
    Finally, from Eq.~\eqref{eq:AGSP_SR} (with the assumption that $(q+m+q)^{q+1}\leqslant d^{2ql}$),
    \beq\label{eq:AGSP_SR2}
    D_{K}\leqslant md^{4ql}(e(q+1)^{2}(2d^{2}l)^{k})^{\frac{m}{q+1}}.
    \eeq
    
    We would like to choose $\tau$ such that $\delta_{K}\leqslant \frac{1}{8}$, which can be achieved if
    \beq
    \tau\geqslant 8g_{0}+\frac{2}{\lambda}\log(\frac{1216}{\Delta}\sqrt{\frac{(q+2)g_{0}}{\lambda}})
    \eeq
    In turn, the above inequality can be solved by $\tau=c_{\tau}\log(q)$, where
    \beq
    c_{\tau}:=\max\{2(8g_{0}+\frac{1}{\lambda}+\frac{2}{\lambda}\log(\frac{1216}{\Delta}\sqrt{\frac{g_{0}}{\lambda}})),4(16g_{0}+\frac{1}{\lambda'}+\frac{1}{\lambda'}\log(\frac{704g_{0}}{3\Delta}))\}
    \eeq
    One can verify that this choice of $\tau$ satisfies
    Eq.~\eqref{eq:critical_cut_off_main}.

    Next, we will choose $m$ and $q$ such that 
    \beq\label{eq:AGSP_property_2}
    \epsilon_{K}\sqrt{2D_{K}}\leqslant \frac{1}{8}.
    \eeq
    From Eq.~\eqref{eq:AGSP_epsilon2} and Eq.~\eqref{eq:AGSP_SR2},
    \beq\label{eq:AGSP_property3}
    \begin{split}
        \epsilon_{K}&\leqslant \exp(-b_{1}m\frac{1}{\sqrt{q\log(q)}})\\
        D_{K}&\leqslant\exp(c_{1}q^{1+\frac{1}{\bar{\fa}}}+c_{2}\frac{m\log(q)}{q})
    \end{split}
    \eeq
    where the constants $b_{1},c_{1},c_{2}$ only depend on $\{g_{0},k,\Delta,d\}$.
    Note we have used $l=O((\frac{q}{\Delta})^{\frac{1}{\bar{\fa}}})$ and $\log(m)\leqslant \log(q+m+1)\leqslant ql\log d$. Notice that if $\epsilon_{K}D_{K}\leqslant \epsilon_{K}^{\frac{1}{2}}$ (by choosing $m,q$ properly), then Eq.~\eqref{eq:AGSP_property1} and Eq.~\eqref{eq:AGSP_property_2} can be satisfied if $\epsilon_K\leqslant 2^{-\frac{14}{3}}$. The condition $\epsilon_{K}D_{K}\leqslant\sqrt{\epsilon_{K}}$ is satisfied if
    \beq\label{eq:epsilon&D}
    \begin{split}
        c_{1}q^{1+\frac{1}{\bar{\fa}}}&\leqslant \frac{b_{1}}{4}\frac{m}{\sqrt{q\log q}}\\
        c_{2}m \frac{\log q}{q}
        &\leqslant \frac{b_{1}}{4}\frac{m}{\sqrt{q\log q}}
    \end{split}
    \eeq
    From the second inequality, we obtain
    \beq
    \frac{q^{1/2}}{\log^{3/2}(q)}\geqslant \frac{4c_{2}}{b_{1}}
    \eeq
    This can be solved by noticing $\frac{q^{1/2}}{\log^{3/2}(q)}\geqslant \frac{1}{120}q^{1/4}$ for any $q>1$, so we require
    \beq
    q\geqslant(\frac{480c_{2}}{b_{1}})^{4}
    \eeq
    The first inequality in Eq.~\eqref{eq:epsilon&D} shows that
    \beq
    m\geqslant\frac{4c_{1}}{b_{1}}q^{\frac{3}{2}+\frac{1}{\bar{\fa}}}\log^{1/2}(q)
    \eeq
    However, we need to check the consistency that $\log(m)\leqslant ql \log(d)=c_{3}q^{1+\frac{1}{\bar{\fa}}}$, where $c_{3}$ only depends $\{\Delta,d\}$. This can be ensured by choosing
    \beq
    q\geqslant \max\{(c_{3}^{-1}(3+\frac{1}{\bar{\fa}}))^{\bar{\fa}},\frac{4c_{1}}{b_{1}}\}
    \eeq
    Finally for $\epsilon_{K}\leqslant 2^{-\frac{14}{3}}$, we need $q\geqslant(\frac{7\log2}{6c_{1}})^{\frac{\bar{\fa}}{\bar{\fa}+1}}$.

    To summarize, all requirements are satisfied by setting
    \beq
    \begin{split}
        \tau&=c_\tau \log q\\
        l&=\lceil(\frac{16g_{0}(q+2)}{\Delta})^{\frac{1}{\bar{\fa}}}\rceil\\
         q&=q_{0}:=\lceil\max\{(\frac{480c_{2}}{b_{1}})^{4},(c_{3}^{-1}(3+\bar{\fa}^{-1}))^{\bar{\fa}},\frac{4c_{1}}{b_{1}},(\frac{7\log 2}{6c_{1}})^{\frac{\bar{\fa}}{\bar{\fa}+1}}\}\rceil\\
        m&=\lceil\frac{4c_{1}}{b_{1}}q_{0}^{\frac{3}{2}+\frac{1}{\bar{\fa}}}\log^{1/2}(q_{0})\rceil
    \end{split}
    \eeq
    where $c_\tau, c_{1},c_{2},c_{3},b_{1}$ are constants only depending on $\{g_{0},k,d,\Delta\}$, but not on $n$.
    
    We then conclude from Eq.~\eqref{eq:AGSP_property3},
    \beq
    \log(D_{K})\leqslant D_{0}
    \eeq
    where $D_{0}$ is a constant depending only on $\{\bar{\fa},k,g_{0},\Delta,d\}$, but not on $n$ (the size of the interval).
\end{proof}
\begin{proof}[Proof of Proposition~\ref{prop:bounding_entropy}]
    In this proof, we set $q=2$ (recall that parameters like $q$ can be chosen arbitrarily, and we only need to prove this proposition for one choice of this set of parameters). From Eq. \ref{eq:norm_perturbation} and Lemma \ref{lemma:spec_stability}, we have
    \beq
    \begin{split}
        ||\delta H||&\leqslant 8g_{0}l^{-\bar{\fa}},\\
        \Delta_{t}&\geqslant \Delta-2||\delta H||.
    \end{split}
    \eeq
    We require $\Delta_{t}\geqslant\frac{3}{4}\Delta$. This can be achieved by $l\geqslant(\frac{64g_{0}}{\Delta})^{\frac{1}{\bar{\fa}}}$. 
    
    With $q=2$, we choose
    \beq
    \tau\geqslant\max\{16g_{0}+\frac{1}{\lambda'}\log(\frac{2112g_{0}}{\Delta_{t}}),8g_{0}+\frac{1}{\lambda}\log(\frac{25920g_{0}}{\lambda \Delta_{t}})\}
    \eeq
    so $\tau$ satisfies Eq.~\eqref{eq:critical_cut_off_main}. By Theorem \ref{thm:cut_off}, we have
    \beq
    \begin{split}
        |\,|\tilde{0}_{t}\ra-|0_{t}\ra|\leqslant\frac{228}{\Delta_{t}}\sqrt{\frac{g_{0}}{\lambda}}e^{-\frac{\lambda}{2}(\tau-8g_{0})}&\leqslant\frac{304}{\Delta}\sqrt{\frac{g_{0}}{\lambda}}e^{-\frac{\lambda}{2}(\tau-8g_{0})},\\
        \tilde{\Delta}_{t}\geqslant\frac{1}{2}\Delta_{t}&\geqslant\frac
        {3}{8}\Delta.
    \end{split}
    \eeq
    
    Below, we construct AGSP $K_{m,l,\tau}$ for $|0\ra$ (given by Eq.~\eqref{eq:AGSP}), which are parameterized by $\{m,l,\tau\}$. The corresponding $(\delta,\epsilon,D)$ of $K_{m,l,\tau}$ is denoted by $(\delta_{m,l,\tau},\epsilon_{m,l,\tau},D_{m,l,\tau})$.

    First, we note from Lemma \ref{lemma:spec_stability} and Theorem \ref{thm:cut_off}, $\delta_{m,l,\tau}$ is bounded as
    \beq\label{eq:AGSP_delta3}
    \delta_{m,l,\tau}=|\,|0\ra-|\tilde{0}_{t}\ra|\leqslant\frac{304}{\Delta}\sqrt{\frac{g_{0}}{\lambda}}e^{-\frac{\lambda}{2}(\tau-8g_{0})}+\frac{16g_{0}l^{-\bar{\fa}}}{\Delta}.
    \eeq

    Similarly, from Eq.~\eqref{eq:AGSP_epsilon},
    \beq\label{eq:AGSP_epsilon3}
    \epsilon_{m,l,\tau}\leqslant 2\exp(-\frac{m}{4}\sqrt{\frac{3\Delta}{\tau+4g_{0}}}).
    \eeq
    Also, from Lemma \ref{lemma:SR_bound_weak},
    \beq\label{eq:AGSP_SR3}
    D_{m,l,\tau}\leqslant (2+(2d^{2}l)^{k})^{m}=e^{mO(\log(dl))}
    \eeq
    In the following, consider a sequence of AGSP $\{K_{p}\}_{p=1,2,\dots}$ with parameters,
    \beq
    \begin{split}
        \delta_{p}&:=\delta_{m_{p},l_{p},\tau_{p}}\\
        \epsilon_{p}&:=\epsilon_{m_{p},l_{p},\tau_{p}}\\
        D_{p}&:=D_{m_{p},l_{p},\tau_{p}}
    \end{split}
    \eeq
    By Proposition~\ref{prop:eebound}, the entanglement entropy is bounded as
    \beq
    S(\phi,[0,n])\leqslant \log(3D_{\varphi}D_{1})-\sum_{p=1}^{\infty}\gamma_{p}^{2}\log\frac{\gamma_{p}^{2}}{3D_{p+1}},\quad\gamma_{p}:=\frac{\epsilon_{p}}{1-\nu_{0}-\delta_{p}}+\delta_{p}
    \eeq
    where $\nu_{0}:=|\,|0\ra-|\varphi\ra|\leqslant \frac{1}{2}$ and $D_{\phi}:=\SR(|\varphi\ra)$. Recall that $\gamma_{p}$ measures the norm distance between $|0\ra$ and $K_{p}|\varphi\ra$ (see Eq.~\eqref{eq:gammap}). 
    
    Below we choose $\{m_{p},l_{p},\tau_{p}\}$ such that
    \beq
    \gamma_{p}\leqslant\frac{1}{p}
    \eeq
    which can be satisfied if
    \beq\label{eq:AGSP_error}
    \begin{split}
        \delta_{p}&\leqslant \frac{1}{3p}\\
        \frac{\epsilon_{p}}{1-\nu_{0}-\delta_{p}}&\leqslant \frac{2}{3p}
    \end{split}
    \eeq
    Note the second inequality can be achieved if $\epsilon_{p}\leqslant \frac{1}{9p}$. To obtain AGSP with Eq.~\eqref{eq:AGSP_error}, we require,
    \beq
    \begin{split}
        \frac{304}{\Delta}\sqrt{\frac{g_0}{\lambda}}e^{-\frac{\lambda}{2}(\tau_{p}-8g_{0})}&\leqslant\frac{1}{6p}\\
        \frac{16g_{0}l_{p}^{-\bar{\fa}}}{\Delta}&\leqslant\frac{1}{6p}
    \end{split}
    \eeq
    This can be solved by
    \beq\label{eq:AGSP_tau_l}
    \begin{split}
        \tau_{p}&\geqslant8g_{0}+\frac{2}{\lambda}\log(\frac{1824p\sqrt{g_{0}}}{\Delta\sqrt{\lambda}})\\
        l_{p}&\geqslant(\frac{96g_{0}p}{\Delta})^{\frac{1}{\bar{\fa}}}
    \end{split}
    \eeq
    For $\epsilon_{p}\leqslant\frac{1}{9p}$, we have
    \beq
    2\exp(-\frac{m}{4}\sqrt{\frac{3\Delta}{\tau+4g_{0}}})\leqslant\frac{1}{9p}
    \eeq
    Combining with Eq.~\eqref{eq:AGSP_tau_l},
    \beq
    m_{p}\geqslant\frac{4\log(18p)}{\sqrt{3\Delta}}\sqrt{12g_{0}+\frac{2}{\lambda}\log(\frac{1824p\sqrt{g_{0}}}{\Delta\sqrt{\lambda}})}
    \eeq
    Under this choice, the Schmidt-rank is bounded by
    \beq
    \log(3D_{p})\leqslant b_{2}\log^{\frac{5}{2}}(p)
    \eeq
    where $b_{2}$ is a constant depending only on $\{k,g_{0},\bar{\fa},\Delta,d\}$. Thus, by Proposition~\ref{prop:eebound}, we conclude
    \beq\label{eq:series_area_law}
    S(\phi|_{[0,n]})\leqslant \log(3D_{\varphi}D_{1})+\sum_{p=1}^{\infty}\frac{1}{p^{2}}(\log(p^{2})+b_{2}\log^{\frac{5}{2}}(p))\leqslant \log(D_{\varphi})+S_{0}
    \eeq
    where $S_{0}$ is a constant depending only on $\{k,g_{0},\Delta,d,\bar{\fa}\}$ and we have noticed that both $\sum_{p=1}^{\infty}p^{-2}\log(p)$ and $\sum_{p=1}^{\infty}p^{-2}\log^{\frac{5}{2}}(p)$ converge.
\end{proof}

\subsection{Comparison with previous works} \label{subapp: comparison}

As mentioned before, our proof is largely inspired by Ref.~\cite{Kuwahara2019}. However, there are important technical differences between their proof and ours, which will be summarized below.

The most important difference is that our proof works for locally unique gapped ground states (\eg systems with spontaneously broken discrete symmetries). So we do not need to assume the uniqueness of gapped ground states as in Ref.~\cite{Kuwahara2019}.

The second difference is that some tools employed in Ref.~\cite{Kuwahara2019} are no longer available in infinite-size systems. For example, when establishing their Lemma 3 and Lemma 4 (\ie the finite-size counterpart of our Lemma \ref{lemma:spec_stability}), the Weyl inequality on operator spectrum is used. This inequality is not available for general self-adjoint operators on infinite-dimensional Hilbert spaces (\eg when this operator has a continuous spectrum). In this case, an alternative tool, called resolvent (\ie Green's function) becomes necessary for spectral analysis. Indeed, resolvent analysis plays an important role through out our proof.

Furthermore, in the proof, many auxiliary Hamiltonians are constructed (\eg the truncated Hamiltonian $H_{t}$ and effective Hamiltonian $\tH$). We note that in Ref.~\cite{Kuwahara2019} the uniqueness of ground states of these auxiliary Hamiltonians is often implicitly assumed in their proof. We carefully prove the uniqueness of ground states in, \eg Lemma \ref{lemma:spec_stability} and Proposition~\ref{prop:unigapped_gs}.

In addition, a key ingredient in the proof of the area law in Ref. \cite{Kuwahara2019} is Lemma 14 therein. In the proof of Lemma 14, it is assumed that the state defined in Eq. (S.242) there always takes the form of Eqs. (S.244) and (S.241) there. However, we find this is not always the case, so the statement and proof of Lemma 14 should be modified. Our modified version is our Lemma \ref{lemma:gap_exp_value}.

On the other hand, our approach is also different from the one in Ref. \cite{Ukai2024}, which also proves Theorem \ref{thm:area_law}. The approach in Ref. \cite{Ukai2024} is based on a generalization of Hasting's factorization lemma \cite{Hastings2007}, while our approach is based on AGSP, which can potentially lead to a better estimation and efficient algorithms for the ground states \cite{Arad2013, Landau2013AGSP}.

\section{Proof of Theorem \ref{thm:cut_off}}\label{sec:proof_cut_off}

In Appendix \ref{sec:area_law}, we have proved the entanglement area law for locally unique gapped ground states of an adimissible Hamiltonian in an infinite quantum spin chain, assuming that Theorem \ref{thm:cut_off} is correct. In this appendix, we prove this theorem.

Before diving into the details, let us introduce several new notations. We first denote the spectral projections of $h_{s}$ by
\beq
\Pi_{I}^{(s)}=\int_{\lambda\in I}\dd\mu^{h_{s}}(\lambda)
\eeq
where $I\subseteq\R$ is a Borel subset\footnote{Recall that a Borel subset of $\R$ is a subset which can be obtained from open intervals by taking countable unions, countable intersections and compliments.}.
In particular, we write $\Pi_{<x}^{(s)}$ (resp. $\Pi^{(s)}_{\leqslant x}$) for $\Pi_{(-\infty,x)}^{(s)}$ (resp. $\Pi_{(-\infty,x]}^{(s)}$). We also define $\Pi_{>x}^{(s)}$ and $\Pi_{\geqslant x}^{(s)}$ in a similar way.

For the truncated Hamiltonian $H_{t}'$ and the effective (\ie cutoff) Hamiltonian $\tilde{H}_t'$, we define $\Pi_{I}$ and $\tilde{\Pi}_{I}$ as follows:
\beq
\begin{split}
    \Pi_{I}&=\int_{\lambda\in I}\dd\mu^{H_{t}'}(\lambda)\\
    \tilde{\Pi}_{I}&=\int_{\lambda\in I}\dd\mu^{\tilde{H}_{t}'}(\lambda)
\end{split}
\eeq
where $I\subseteq \R$ is again a Borel subset.

Intuitively, if $\tau$ is large enough, then we should expect that $\tilde{H}_{t}$ behaves like $H_{t}$ and their ground states must have large overlap. The following two propositions are fundamental to formalize this intuition by spectral projections.
\begin{proposition}\label{prop:norm_spec_proj}
    The overlap between projection $\Pi_{\geqslant E'}^{(s)}$ and $\Pi_{\leqslant E}$ is bounded from above by
    \beq\label{eq:norm_proj}
    ||\Pi_{>E'}^{(s)}\Pi_{\leqslant E}||\leqslant\frac{4e^{3/2}}{e-1}e^{-\lambda(\delta E_{s}'-\delta E-8g_{0})},
    \eeq
    with
    \beq
    \lambda:=\frac{1}{8g k+16g_{0}},
    \eeq
    where $\delta E'_{s}:=E'-E_{s,0}$, $\delta E:=E-E_{t,0}$ with $E_{s,0}$ and $E_{t,0}$ the ground state energy of $h_{s}$ and $H_{t}'$, respectively. And $\tau_{s}:=E_{s,0}+\tau$. A similar result for $\tilde{\Pi}_{<E}$ is given by
    \beq\label{eq:tilde_norm_proj}
    ||\Pi_{>E'}^{(s)}\tilde{\Pi}_{<E}||\leqslant\frac{4e^{3/2}}{e-1}e^{-\lambda'(\min\{E',\tau_{s}\}-E_{s,0}-\delta \tilde{E}-8g_{0})}
    \eeq
    where $\delta\tilde{E}:=E-\tilde{E}_{t,0}$ with $\tilde{E}_{t,0}$ being the ground state energy of $\tH$ and
    \beq
    \lambda':=\min\{\frac{1}{224g_{0}},\frac{1}{8gk}\}
    \eeq
\end{proposition}
\begin{proposition}\label{prop:Hamiltonian_norm_difference}
    For $\tau>0$,
    \beq
    ||(H_{t}'-\tilde{H}_{t}')\Pi_{<E}||\leqslant\frac{27(q+2)}{\lambda}e^{-\lambda(\tau-\delta E-8g_{0})}
    \eeq
    where $\delta E:=E-E_{t,0}$.
\end{proposition}
Proposition \ref{prop:norm_spec_proj} and  \ref{prop:Hamiltonian_norm_difference} are the infinite-dimensional generalizations of proposition 8 and 9 of Ref.~\cite{Kuwahara2019}.
The proof of Proposition \ref{prop:norm_spec_proj} is deferred to Sec.~\ref{sec:proof_norm_spec_proj_1} and \ref{sec:proof_norm_spec_proj_2}. Below we prove Prop.~\ref{prop:Hamiltonian_norm_difference} assuming Prop.~\ref{prop:norm_spec_proj}.
\begin{proof}[Proof of Proposition~\ref{prop:Hamiltonian_norm_difference}]
    Note that
    \beq
        (H_{t}'-\tilde{H}_{t}')\Pi_{<E}= \sum_{s=0}^{q+1}(h_{s}-\tau_{s})\Pi_{\geqslant \tau_{s}}^{(s)}\Pi_{<E}
    \eeq
    For any $y>0$, we can write
    \beq
    \sum_{s=0}^{q+1}(h_{s}-\tau_{s})\Pi_{\geqslant \tau_{s}}^{(s)}\Pi_{<E}=\sum_{s=0}^{q+1}\sum_{j=0}^{\infty}(h_{s}-\tau_{s})\Pi_{[\tau_{s}+jy,\tau_{s}+(j+1)y)}^{(s)}\Pi_{<E}
    \eeq
    According to Prop.~\ref{prop:norm_spec_proj}, 
    \beq
    ||\Pi_{[\tau_{s}+jy,\tau_{s}+(j+1)y)}^{(s)}\Pi_{<E}||\leqslant\frac{4e^{3/2}}{e-1}e^{-\lambda (jy+\tau-\delta E-8g_{0})}
    \eeq
    As a consequence, we have
    \beq
    \begin{split}
        ||(H_{t}'-\tilde{H}_{t}')\Pi_{<E}||&\leqslant \sum_{s=0}^{q+1}\sum_{j=0}^{\infty}||(h_{s}-\tau_{s})\Pi_{[\tau_{s}+jy,\tau_{s}+(j+1)y)}^{(s)}\Pi_{<E}||\\
        &\leqslant\frac{4e^{3/2}}{e-1}\sum_{s=0}^{q+1}\sum_{j=0}^{\infty}(j+1)ye^{-\lambda (\tau+jy-\delta E-8g_{0})}\\
        &=(q+2)\frac{4e^{3/2}}{e-1}e^{-\lambda(\tau-\delta E-8g_{0})}\frac{ye^{2\lambda y}}{(e^{\lambda y}-1)^{2}}
    \end{split}
    \eeq
    Setting $y=\lambda^{-1}$, we obtain
    \beq
     ||(H_{t}'-\tilde{H}_{t}')\Pi_{<E}||\leqslant(q+2)\frac{4e^{3/2}}{e-1}\frac{e^{2}}{\lambda(e-1)^{2}}e^{-\lambda(\tau-\delta E-8g_{0})}\leqslant\frac{27(q+2)}{\lambda}e^{-\lambda(\tau-\delta E-8g_{0})}
    \eeq
\end{proof}

Next, we show $\tH$ is uniquely gapped for sufficiently large $\tau$. The key is the following lemma, which is a generalization of lemma 15 of Ref.~\cite{Kuwahara2019}.
\begin{lemma}\label{lemma:lower_bound_energy}
For any $E\geqslant \tilde{E}_{t,0}$ and a state $|\psi\ra$ satisfying $\tilde{\Pi}_{\leqslant E}|\psi\ra=|\psi\ra$ and $\la 0_{t}|\psi\ra=0$, we have
\beq
\la \psi|\tH|\psi\ra\geqslant E_{\perp}
\eeq
where
\beq
\begin{split}
    E_{\perp}&:=\Delta_{t}(1-\kappa(E))^{2}-4g_{0}\kappa(E)(1+\kappa(E))(q+1)\\
    \kappa(E)&:=\sum_{s=0}^{q+1}||\Pi^{(s)}_{>\tau_{s}}\tilde{\Pi}_{\leqslant E}||
\end{split}
\eeq
\end{lemma}
\begin{proof}
    Let us define
    \beq
    \begin{split}
        P_{s}&:=\Pi^{(s)}_{\leqslant \tau_{s}}\\
        Q_{s}&:=1-P_{s}
    \end{split}
    \eeq
    with $s=0,1,\dots,q+1$. We also define $P^{(-1)}:=1$ and
    \beq
    \begin{split}
        P^{(m)}&:=P_{0}P_{1}\dots P_{m}\\
        Q^{(m)}&:=1-P^{(m)}
    \end{split}
    \eeq
    By construction, we have
    \beq\label{eq:restricted_equality}
    \tH P^{(q+1)}=H_{t}'P^{(q+1)}
    \eeq
    Note that 
    \beq
    P^{(q+1)}-1=\sum_{s=0}^{q+1}(P^{(s)}-P^{(s-1)})=-\sum_{s=0}^{q+1}P^{(s-1)}Q^{(s)}
    \eeq
    Consequently, for any $|\psi\ra$ with $\tilde{\Pi}_{\leqslant E}|\psi\ra=|\psi\ra$,
    \beq\label{eq:difference_proj}
    \begin{split}
        ||P^{(q+1)}|\psi\ra-|\psi\ra||&\leqslant\sum_{s=0}^{q+1}||P^{(s)}Q^{(s)}|\psi\ra||\\
        &\leqslant \sum_{s=0}^{q+1}||Q^{(s)}|\psi\ra||\\
        &=\sum_{s=0}^{q+1}||Q^{(s)}\tilde{\Pi}_{\leqslant E}|\psi\ra||\\
        &\leqslant \sum_{s=0}^{q+1}||\Pi_{>\tau_{s}}^{(s)}\tilde{\Pi}_{\leqslant E}||=:\kappa(E)
    \end{split}
    \eeq
    As a result,
    \beq
    \begin{split}
            \la\psi|\tH|\psi\ra\geqslant \la\psi|P^{(q+1)}\tH P^{(q+1)}|\psi\ra-|\la\psi|P^{(q+1)}\tH Q^{(q+1)}|\psi\ra|\\-|\la\psi|Q^{(q+1)}\tH P^{(q+1)} |\psi\ra|+\la\psi|Q^{(q+1)}\tH Q^{(q+1)}|\psi\ra
    \end{split}
    \eeq
    To proceed, we need to lower bound each term on the right-hand side.

    For the first term,
    \beq
    \begin{split}
        \la\psi|P^{(q+1)}\tH P^{(q+1)}|\psi\ra&=\la\psi|P^{(q+1)}H_{t}'P^{(q+1)}|\psi\ra\\
        &=\la\psi|P^{(q+1)}\Pi_{\geqslant\Delta_{t}}H_{t}'\Pi_{\geqslant\Delta_{t}}P^{(q+1)}|\psi\ra\\
        &\geqslant \Delta_{t} ||\Pi_{\geqslant\Delta_{t}}P^{(q+1)}|\psi\ra||^{2}
    \end{split}
    \eeq
    where we have used Eq.~\eqref{eq:restricted_equality} in the first line and the fact that $H_{t}|0_{t}\ra=0$ in the second line. From $\la 0_{t}|\psi\ra=0$ we have $\Pi_{\geqslant\Delta_{t}}|\psi\ra=|\psi\ra$ since $H_{t}'$ is non-degenerate, hence
    \beq
    \begin{split}
            ||\Pi_{\geqslant\Delta_{t}}P^{(q+1)}|\psi\ra||&\geqslant ||\Pi_{\geqslant\Delta_{t}}|\psi\ra||-||\Pi_{\geqslant\Delta_{t}}(P^{(q+1)}-1)|\psi\ra||\\
            &\geqslant 1-||(P^{(q+1)}-1)|\psi\ra||\\
            &\geqslant 1-\kappa
    \end{split}
    \eeq
    In summary,
    \beq
     \la\psi|P^{(q+1)}\tH P^{(q+1)}|\psi\ra\geqslant \Delta_{t}(1-\kappa)^{2}
    \eeq

    For the second term, we note that $[\tilde{h}_{s},P^{(q+1)}]=[\tilde{h}_{s},Q^{(q+1)}]=0$ while $P^{(q+1)}Q^{(q+1)}=0$, we thus conclude that
    \beq
    \begin{split}
        |\la\psi|P^{(q+1)}\tH Q^{(q+1)}|\psi\ra|&\geqslant -\sum_{s=0}^{q}|\la\psi|P^{(q+1)}h_{s,s+1}Q^{(q+1)}|\psi\ra|\\
        &\geqslant -\sum_{s=0}^{q}||h_{s,s+1}P^{(q+1)}|\psi\ra||\cdot ||Q^{(q+1)}|\psi\ra||\\
        &\geqslant -\sum_{s=0}^{q}||h_{s,s+1}||\cdot ||(P^{(q+1)}-1)|\psi\ra||\\
        &\geqslant -2g_{0}(q+1)\kappa
    \end{split}
    \eeq
    In the second line, we have used Cauchy-Schwartz inequality. We have also used $||h_{s,s+1}||\leqslant 2g_{0}$ and Eq.~\eqref{eq:difference_proj} to derive the last line. The same bound holds for the third term,
    \beq
     |\la\psi|Q^{(q+1)}\tH P^{(q+1)}|\psi\ra|\geqslant -2g_{0}(q+1)\kappa
    \eeq
    For the last term, we have
    \beq
    \la\psi|Q^{(q+1)}\tH Q^{(q+1)}|\psi\ra\geqslant \tilde{E}_{t,0}||Q^{(q+1)}|\psi\ra||^{2}\geqslant -4g_{0}(q+1)\kappa^{2}
    \eeq
    To see this, we apply Eq.~\eqref{eq:cut_off_gs_energy} to
    \beq
    \tilde{E}_{t,0}\geqslant\sum_{s=0}^{q+1}E_{s,0}-\sum_{s=0}^{q}||h_{s,s+1}||=-4(q+1)g_{0}
    \eeq

    We complete the proof by combining these 4 pieces.
\end{proof}

Now we can complete the proof for ground-state non-degeneracy for $\tH$ with sufficiently large $\tau$.
 \begin{proposition}\label{prop:unigapped_gs}
     The effective Hamiltonian $\tH$ is uniquely gapped if
     \beq\label{eq:min_cut_off}
     \tau\geqslant \max\{16g_{0}+\frac{1}{\lambda'}\log(\frac{176g_{0}(q+1)(q+2)}{\Delta_{t}}),8g_{0}+\frac{1}{\lambda}\log(\frac{432(q+2)}{\lambda\Delta_{t}})\}
     \eeq
     where 
     \beq
     \begin{split}
         \lambda&:=\frac{1}{8gk+16g_{0}}\\
         \lambda'&:=\min\{\frac{1}{224g_{0}},\frac{1}{8gk}\}
     \end{split}
     \eeq
 \end{proposition}
 \begin{proof}
     Assuming the contrary. If $\tilde{E}_{t,0}$ is contained in a continuous spectrum, there exists an $\epsilon_{0}>0$ such that for all $0<\epsilon<\epsilon_{0}$, both of $\tilde{\Pi}_{[\tilde{E}_{t,0}+\epsilon,\tilde{E}_{t,0}+2\epsilon]}$ and $\tilde{\Pi}_{[\tilde{E}_{t,0}+3\epsilon,\tilde{E}_{t,0}+4\epsilon]}$ have nonzero images. Choose $|\psi_{1}\ra\in\im(\tilde{\Pi}_{[\tilde{E}_{t,0}+\epsilon,\tilde{E}_{t,0}+2\epsilon]})$ and $|\psi_{2}\ra\in\im(\tilde{\Pi}_{[\tilde{E}_{t,0}+3\epsilon,\tilde{E}_{t,0}+4\epsilon]})$, then it is easy to see $\la\psi_{1}|\psi_{2}\ra=0$. Let
     \beq
     |\psi\ra:=a_{1}|\psi_{1}\ra+a_{2}|\psi_{2}\ra
     \eeq
     The coefficients $a_{1},a_{2}$ are chosen such that $\la 0_{t}|\psi\ra=0$ and $\la\psi|\psi\ra=1$. Meanwhile, $\tilde{\Pi}_{\leqslant E}|\psi\ra=|\psi\ra$ for any $E>\tilde{E}_{t,0}+4\epsilon$ by construction.

     A similar construction for $|\psi\ra$ holds if $\tH$ is gapped with multiple ground states, where we replace $|\psi_{1}\ra,\,|\psi_{2}\ra$ with orthogonal ground states and set $\epsilon=0$ in this case.
     In either cases (with $0\leqslant \epsilon<\epsilon_{0}$), we have
     \beq\label{eq:upper_bound_energy1}
     \la\psi|\tH|\psi\ra\leqslant \tilde{E}_{t,0}+4\epsilon
     \eeq

     Choosing $\epsilon<\min\{\epsilon_{0},\frac{1}{8}\Delta_{t}\}$ and noting that $\Delta_{t}\leqslant2g\leqslant 2g_{0}$ (by Lemma \ref{lemma:spec_stability}, Lemma \ref{lemma:upbound_origgap} and the convention of $g$ and $g_0$ such that $g\leqslant g_0$), we get $\tilde{E}_{t,0}+8g_{0}>\tilde{E}_{t,0}+4\epsilon$. According to lemma \ref{lemma:lower_bound_energy}, under the condition $\tilde{E}_{t,0}+8g_{0}>\tilde{E}_{t,0}+4\epsilon$, we can conclude,
     \beq\label{eq:lower_bound_energy1}
     \la\psi|\tH|\psi\ra>E_{\perp}
     \eeq
     where
    \beq
   \begin{split}
    E_{\perp}&:=\Delta_{t}(1-\kappa)^{2}-4g_{0}\kappa(1+\kappa)(q+1)\\
    \kappa(\tilde{E}_{t,0}+8g_{0})&:=\sum_{s=0}^{q+1}||\Pi^{(s)}_{>\tau_{s}}\tilde{\Pi}_{\leqslant \tilde{E}_{t,0}+8g_{0}}||
    \end{split}
    \eeq
    By Prop.~\ref{prop:norm_spec_proj} and Eq.~\eqref{eq:min_cut_off}
    \beq
    ||\Pi^{(s)}_{>\tau_{s}}\tilde{\Pi}_{\leqslant \tilde{E}_{t,0}+8g_{0}}||\leqslant\frac{4e^{3/2}}{e-1}e^{-\lambda'(\tau-16g_{0})}\leqslant\frac{4e^{3/2}}{e-1} \frac{\Delta_{t}}{176g_{0}(q+1)(q+2)}<\frac{\Delta_{t}}{16(q+1)(q+2)g_{0}}
    \eeq
    Consequently,
    \beq
    \kappa(\tilde{E}_{t,0}+8g_{0})<\frac{\Delta_{t}}{16(q+1)g_{0}}\leqslant\frac{1}{24}
    \eeq
    where we have used $\Delta_{t}\leqslant2g\leqslant 2g_{0}$. In this case
    \beq
    E_{\perp}=\Delta_{t}(1-\kappa)^{2}-4g_{0}\kappa(1+\kappa)(q+1)>\frac{379}{576}\Delta_{t}>\frac{1}{2}\Delta_{t}>0,
    \eeq
    where we have used that $q\geqslant 2$. 
    
    On the other hand, since $H_{t}'-\tH\geqslant0$, we have $\tilde{E}_{t,0}\leqslant E_{t,0}=0$. Therefore, Eq.~\eqref{eq:lower_bound_energy1} contradicts Eq.~\eqref{eq:upper_bound_energy1}. 
 \end{proof}

    Below we always assume $\tau$ satisfies Eq.~\eqref{eq:min_cut_off} and $q$ is large enough such that $\tau\geqslant 8g_{0}$. We then denote the uniquely gapped ground state of $\tH$ by $|\tilde{0}_{t}\ra$.
    \begin{remark}
        In Ref.~\cite{Kuwahara2019}, it is implicitly assumed that $\tH$ has a unique ground state without proving it. Here we have filled this gap.
    \end{remark}
    
    Our next task is to bound its energy gap $\tilde{\Delta}_{t}$ from below.
    For later convenience, we first estimate the upper bound of the spectral gap $\tilde{\Delta}_{t}$ for the effective Hamiltonian $\tH$. We start from the following lemmas.
    \begin{lemma}\label{lemma:upper_bound_eff_gap}
        The spectral gap of $\tilde{H}_{t}'$ is bounded from above by 
        \beq
        \tilde{\Delta}_{t}\leqslant 8g_{0}
        \eeq
        where $g_{0}$ is defined by Eq.~\eqref{eq:g0}.
        \end{lemma}
        \begin{proof}
            First we focus on the block $B_{0}'$, and we then define $H_{\Lambda_{0}}:=\tH-\tilde{h}_{0}-h_{0,1}$, which is supported only on $\Lambda_{0}:=\Lambda\setminus B_{0}'$. We denote $E_{\Lambda_{0}}:=\inf(\sigma(H_{\Lambda_{0}}))$ and by the spectral theorem of self-adjoint operators, for any $\epsilon'>0$, there exists a state $|E_{\Lambda_{0}}^{\epsilon'}\ra$ such that
            \beq
            E_{\Lambda_{0}}\leqslant \la E_{\Lambda_{0}}^{\epsilon'}|H_{\Lambda_{0}}|E_{\Lambda_{0}}^{\epsilon'}\ra<E_{\Lambda_{0}}+\epsilon'
            \eeq

            Let $E_{0,0}:=\inf(\sigma(h_{0}))$, we claim that there exists $0<\epsilon<4g$ such that there are two (nonzero) states $|\psi_{0}\ra$ and $|\psi_{1}\ra$, which satisfy
            \beq
            \begin{split}
                |\psi_{0}\ra&\in\im(\Pi^{(0)}_{[E_{0,0},E_{0,0}+\epsilon)})\\
                |\psi_{1}\ra&\in\im(\Pi^{(0)}_{[E_{0,0}+\epsilon,E_{0,0}+4g]})\\
                \la\psi_{0}|\psi_{1}\ra&=\la\psi_{0}|h_{0}|\psi_{1}\ra=0
            \end{split}
            \eeq
            where $\Pi^{(0)}$ is the spectral projector of $h_{0}$.
            Assuming the contrary, if there are no such two states, then $h_{0}$ is uniquely gapped with energy gap larger than $4g$, contradicting lemma \ref{lemma:upbound_origgap} since $h_{0}$ is an admissible Hamiltonian (with an extra $2$ due to folding trick).
            Now we consider the quantum state
            \beq
            |\psi\ra=(a_{0}|\psi_{0}\ra+a_{1}|\psi_{1}\ra)\otimes |E_{\Lambda_{0}}^{\epsilon'}\ra
            \eeq
            The coefficients $a_{0},a_{1}$ are chosen such that $\la\psi|\tilde{0}_{t}\ra=0$.

            By Eq.~\eqref{eq:min_cut_off} and the assumption that $q$ is large, we know $\tau\geqslant8g_{0}>4g$. Therefore, we have $\tilde{h}_{0}|\psi_{i}\ra=\tilde{h}_{0}\Pi^{(0)}_{<E_{0,0}+\tau}|\psi_{i}\ra=h_{0}\Pi^{(0)}_{<E_{0,0}+\tau}|\psi_{i}\ra=h_{0}|\psi_{i}\ra$ for $i=0,1$. Consequently
            \beq
            \begin{split}
                E_{0,0}\leqslant\la\psi_{0}|\tilde{h}_{0}|\psi_{0}\ra&< E_{0,0}+\epsilon\\
                E_{0,0}+\epsilon\leqslant\psi_{1}|\tilde{h}_{0}|\psi_{1}\ra&\leqslant E_{0,0}+4g\\
            \end{split}
            \eeq
            As a result,
            \beq
            \begin{split}
                \la\psi|\tilde{H}_{t}'|\psi\ra&\leqslant(|a_{0}|^{2}(E_{0,0}+\epsilon)+|a_{1}|^{2}(E_{0,0}+4g))+\epsilon'+E_{\Lambda_{0}}+\la\psi|h_{0,1}|\psi\ra\\&\leqslant E_{0,0}+4g+E_{\Lambda_{0}}+||h_{0,1}||+\epsilon'\\
                &\leqslant E_{0,0}+E_{\Lambda_{0}}+4g+2g_{0}+\epsilon'
            \end{split}
            \eeq
            where we have used $||h_{0,1}||\leqslant 2g_{0}$ (see Eq. \eqref{eq:bounding_interaction}). On the other hand,
            \beq
            \la\tilde{0}_{t}|\tilde{H}_{t}|\tilde{0}_{t}\ra\geqslant E_{0,0}+E_{\Lambda_{0}}-\epsilon'-||h_{0,1}||\geqslant E_{0,0}+E_{\Lambda_{0}}-2g_{0}-\epsilon'
            \eeq
            Therefore,
            \beq
            \tilde{\Delta}_{t}\leqslant\la\psi|\tilde{H}_{t}'|\psi\ra-\la\tilde{0}_{t}|\tilde{H}_{t}|\tilde{0}_{t}\ra\leqslant 4g+4g_{0}+2\epsilon'\leqslant8g_{0}+2\epsilon'
            \eeq
            since $\epsilon'$ can be arbitrarily small, we obtain
            \beq
            \tilde{\Delta}_{t}\leqslant8g_{0}
            \eeq
        \end{proof}

         Now we bound $\tilde{\Delta}_{t}$ from below. By definition and lemma \ref{lemma:upper_bound_eff_gap}, $\forall\,\epsilon>0$, $\tilde{\Pi}_{[\tilde{E}_{t,0}+\tilde{\Delta}_{t},\tilde{E}_{t,0}+\tilde{\Delta}_{t}+\epsilon)}$ has a nonzero image, where $\tilde{E}_{t,0}:=\inf(\sigma(\tH))$. Let us choose $\epsilon<\frac{1}{2}\tilde{\Delta}_{t}$ for later convenience, and\footnote{Note that this $|\psi_{1}\ra$ and $|\psi\ra$ defined in Eq.~\eqref{eq:low_lying_excitation2} have nothing to do with $|\psi_{1}\ra$ and $|\psi\ra$ in previous proof.} $|\psi_{1}\ra\in\im(\tilde{\Pi}_{[\tilde{E}_{t,0}+\tilde{\Delta}_{t},\tilde{E}_{t,0}+\tilde{\Delta}_{t}+\epsilon)})$. This ensures $\la\tilde{0}_{t}|\psi_{1}\ra=0$. Then there is an orthogonal decomposition,
         \beq
         |0_{t}\ra=a_{0}|\tilde{0}_{t}\ra+a_{1}|\psi_{1}\ra+\sqrt{1-|a_{0}|^{2}-|a_{1}|^{2}}|\psi_{2}\ra
         \eeq
         where $a_{0},a_{1}\in\bbC$ and $|\psi_{2}\ra$ is orthogonal to $|\tilde{0}_{t}\ra$ and $|\psi_{1}\ra$.
         Consider
         \beq\label{eq:low_lying_excitation2}
         |\psi\ra:=\frac{1}{\sqrt{|a_{0}|^{2}+|a_{1}|^{2}}}(\bar{a}_{1}|\tilde{0}_{t}\ra-\bar{a}_{0}|\psi_{1}\ra)
         \eeq
         where $\bar{z}$ is the complex conjugation for any $z\in\bbC$. This state $|\psi\ra$ satisfies $\la0_{t}|\psi\ra=0$. Thus by repeating the calculation in the proof of Proposition \ref{prop:unigapped_gs}, we find, for $\epsilon\leqslant\frac{\log(2)}{\lambda'}$,
        \beq
        \kappa(\tilde{E}_{t,0}+8g_{0}+\epsilon)\leqslant\frac
        {\Delta_{t}}{16g_{0}(q+1)}e^{\lambda'\epsilon}\leqslant\frac
        {\Delta_{t}}{8g_{0}(q+1)}\leqslant \frac{1}{12},
        \eeq
        where Lemmas \ref{lemma:spec_stability} and \ref{lemma:upbound_origgap} have been used in the last inequality.
        
        The above inequalities lead to that
        \beq\label{eq:lower_bound_expectation}
        \la\psi|\tH|\psi\ra\geqslant\frac{41}{72}\Delta_{t}
        \eeq
        
        On the other hand, we apply Prop.~\ref{prop:Hamiltonian_norm_difference} with $\tau$ satisfying Eq.~\eqref{eq:min_cut_off} and $\delta E\to 0^{+}$, we obtain
        \beq\label{eq:norm_distance_1}
        ||(H_{t}'-\tH)|0_t\ra||=||\tH|0_t\ra||\leqslant
        \frac{27(q+2)}{\lambda}e^{-\lambda(\tau-8g_{0})}\leqslant\frac{\Delta_{t}}{16}
        \eeq
        
        Therefore, the lower bound for $\tilde{\Delta}_{t}$ can be obtained from the following lemma,
        \begin{lemma}\label{lemma:gap_exp_value}
            The spectral gap $\tilde{\Delta}_{t}$ is lower bounded by
            \beq\label{eq: bound on tilde Delta}
            \tilde{\Delta}_{t}\geqslant \la\psi|\tH|\psi\ra-\epsilon
            \eeq
            where $0<\epsilon<\min\{\frac{1}{2}\tilde{\Delta}_{t}
            ,\frac{\log(2)}{\lambda'}\}$ is used to define $|\psi\ra$ and $|\psi_{1}\ra$ above. Therefore,
            \beq
            \tilde{\Delta}_{t}>\frac{41}{72}\Delta_{t}>\frac{1}{2}\Delta_{t}
            \eeq
        \end{lemma}
        \begin{proof}
            Notice that $\la\tilde{0}_{t}|\tH|\psi_{1}\ra=0$ since
            \beq
            \begin{split}
                \la\tilde{0}_{t}|\tH|\psi_{1}\ra&=\la\tilde{0}_{t}|\tH\tilde{\Pi}_{[\tilde{E}_{t,0}+\tilde{\Delta}_{t}-\epsilon,\tilde{E}_{t,0}+\tilde{\Delta}_{t}+\epsilon]}|\psi_{1}\ra\\
                &=(\la\tilde{0}_{t}|\tilde{\Pi}_{[\tilde{E}_{t,0}+\tilde{\Delta}_{t}-\epsilon,\tilde{E}_{t,0}+\tilde{\Delta}_{t}+\epsilon]})\tH|\psi_{1}\ra\\
                &=0
            \end{split}
            \eeq
            where we have used $\epsilon<\frac{1}{2}\tilde{\Delta}_{t}$. We then have
            \beq
                \tilde{E}_{t,0}\leqslant\la\psi|\tH|\psi\ra&\leqslant \frac{|a_{0}|^{2}(\tilde{E}_{t,0}+\tilde{\Delta}_{t}+\epsilon)+|a_{1}|^{2}\tilde{E}_{t,0}}{|a_{0}|^{2}+|a_{1}|^{2}}\leqslant(\tilde{E}_{t,0}+\tilde{\Delta}_{t}+\epsilon),
            \eeq
            which proves Eq. \eqref{eq: bound on tilde Delta}.
            
            On the other hand, $\la\psi|\tH|\psi\ra\geqslant\frac{41}{72}\Delta_{t}$, it follows that
            \beq\label{eq:lower_bound_effective_gap1}
            \tilde\Delta_t\geqslant\la\psi|\tilde H_t'|\psi\ra-\epsilon>\frac{41}{72}\Delta_t-\epsilon>\frac{\Delta_t}{2}-\epsilon
            \eeq
            As a consequence,
            \beq
            \tilde{\Delta}_{t}>\frac{1}{2}\Delta_{t},
            \eeq
            since $\epsilon$ can be arbitrarily small.
        \end{proof}
        We also obtain
        \begin{corollary}
            The first excited energy of $\tH$ is lower bounded as
            \beq
            \tilde{E}_{t,0}+\tilde{\Delta}_{t}\geqslant\frac{41}{72}\Delta_{t}
            \eeq
        \end{corollary}
        
        Finally, we estimate the norm distance of the ground state $|0_{t}\ra$ and the effective one $|\tilde{0}_{t}\ra$.
        \begin{lemma}\label{lemma:norm_distance_effective}
            For any $\tau$ satisfies
            \beq\label{eq:critical_cut_off}
            \tau\geqslant \max\{16g_{0}+\frac{1}{\lambda'}\log(\frac{176g_{0}(q+1)(q+2)}{\Delta_{t}}),8g_{0}+\frac{1}{\lambda}\log(\frac{6480(q+2)g_{0}}{\lambda\Delta_{t}^{2}})\}
            \eeq
            we have
            \beq
            ||\,|0_{t}\ra-|\tilde{0}_{t}\ra||\leqslant114\sqrt{\frac{(q+2)g_{0}}{\lambda\Delta_{t}^{2}}}e^{-\frac{\lambda}{2}(\tau-8g_{0})}
            \eeq
        \end{lemma}
        \begin{proof}
           Note that $\Delta_{t}\leqslant2g\leqslant 2g_{0}$ (due  to Lemmas \ref{lemma:spec_stability} and \ref{lemma:upbound_origgap}), so
           \beq
            8g_{0}+\frac{1}{\lambda}\log(\frac{432(q+2)}{\lambda\Delta_{t}})<8g_{0}+\frac{1}{\lambda}\log(\frac{6480(q+2)g_{0}}{\lambda\Delta_{t}^{2}})
           \eeq
           Therefore, any $\tau$ satisfing Eq.~\eqref{eq:critical_cut_off} automatically satisfies Eq.~\eqref{eq:min_cut_off}.
           
           Let us write
           \beq
           \begin{split}
               P_{0}&:=|0_{t}\ra\la0_{t}|\\
               \tilde{P}_{0}&:=|\tilde{0}_{t}\ra\la\tilde{0}_{t}|
           \end{split}
           \eeq
           We rewrite these projectors in terms of resolvents,
           \beq
           \begin{split}
               P_{0}&=\oint_{C}\frac{\dd z}{2\pi i}\frac{1}{z-H_{t}'}\\
               \tilde{P}_{0}&=\oint_{C}\frac{\dd z}{2\pi i}\frac{1}{z-\tH}
           \end{split}
           \eeq
           where the contour is $C:=\{z\in\bbC:|z-\tilde{E}_{t,0}|=r\}$ and $r:=\min\{\tilde{\Delta}_{t}-\frac{1}{3}\Delta_{t},\frac{1}{3}\Delta_{t}-\tilde{E}_{t,0}\}$. From lemma \ref{lemma:upper_bound_eff_gap}, $\tilde\Delta_{t}>\frac{41}{72}\Delta_{t}$ (see Lemma \ref{lemma:gap_exp_value}) and $\tilde{E}_{t,0}\leqslant 0$, we conclude that $\frac{17}{72}\Delta_{t}<r<8g_{0}$.
           
           We aim to bound $|\la 0_t|\tilde 0_t\ra|$ by bounding $||\tilde{P}_{0}P_{0}-P_{0}^{2}||$. To this end,
           \beq
           \tilde{P}_{0}P_{0}-P_{0}^{2}=\oint_{C}\frac{\dd z}{2\pi i}\frac{1}{z-\tH}(\tH-H_{t}')P_{0}\frac{1}{z-H_{t}'}
           \eeq
           For any $z\in C$, we note
           \beq
           \begin{split}
               ||(z-\tH)^{-1}||&\leqslant\max\{r^{-1},(\tilde{\Delta}_{t}-r)^{-1}\}\leqslant\frac{72}{17\Delta_{t}}<\frac{5}{\Delta_{t}}\\
               ||(z-H_{t}')^{-1}||&\leqslant\max\{(r+\tilde{E}_{t,0})^{-1},(\Delta_{t}-r-\tilde{E}_{t,0})^{-1}\}\leqslant\frac{6}{\Delta_{t}}\\
               ||(\tH-H_{t}')P_{0}||&\leqslant\frac{27(q+2)}{\lambda}e^{-\lambda(\tau-8g_{0})}
           \end{split}
           \eeq
           As a consequence, for $\tau$ satisfying Eq.~\eqref{eq:critical_cut_off},
           \beq
           ||\tilde{P}_{0}P_{0}-P_{0}^{2}||\leqslant \frac{810(q+2)}{\lambda\Delta_{t}^{2}}e^{-\lambda(\tau-8g_{0})}r<\frac{6480g_{0}(q+2)}{\lambda\Delta_{t}^{2}}e^{-\lambda(\tau-8g_{0})}<1
           \eeq
           where we have used $r<\tilde{\Delta}_{t}\leqslant8g_{0}$. Note that $||\tilde{P}_{0}P||=|\la0_{t}|\tilde{0}_{t}\ra|$, we have
           \beq
            |\la0_{t}|\tilde{0}_{t}\ra|\geqslant 1-\frac{6480g_{0}(q+2)}{\lambda\Delta_{t}^{2}}e^{-\lambda(\tau-8g_{0})}
           \eeq
           Let us choose the relative phase so that $\la0_{t}|\tilde{0}_{t}\ra\geqslant 0$, then we conclude
           \beq
           | \,|0_{t}\ra-|\tilde{0}_{t}\ra|\leqslant\sqrt{\frac{12960(q+2)g_{0}}{\lambda\Delta_{t}^{2}}}e^{-\frac{\lambda}{2}(\tau-8g_{0})}\leqslant114\sqrt{\frac{(q+2)g_{0}}{\lambda\Delta_{t}^{2}}}e^{-\frac{\lambda}{2}(\tau-8g_{0})}
           \eeq
        \end{proof}
        This completes the proof of Theorem \ref{thm:cut_off}.

        Below we prove Proposition \ref{prop:norm_spec_proj}. We divide the proof into two parts, based on the two statements in this proposition. In fact, the proof of the first part can be viewed as a simpler warm-up of the second part.

\subsection{Proof of Proposition \ref{prop:norm_spec_proj}: First part}\label{sec:proof_norm_spec_proj_1}

\begin{lemma}\label{lemma:norm_proj}
    For any normalized state $|\psi\ra$, we define
    \beq
    |\phi\ra:=\Pi_{>E'}^{(s)}\Pi_{\leqslant E}|\psi\ra
    \eeq
    Then,
    \beq
    || \,|\phi\ra||<\frac{4e^{3/2}}{e-1}e^{-\lambda(\la H_{t}'\ra_{\phi}-E)}
    \eeq
    where
    \beq
    \la H_{t}'\ra_{\phi}:=\begin{cases}
        0,\quad\text{if $|\phi\ra=0$}\\
        \frac{\la\phi| H_{t}'|\phi\ra}{||\,|\phi\ra||^{2}},\quad\text{\rm otherwise}
    \end{cases}
    \eeq
\end{lemma}
The proof of this lemma will be postponed to Sec.~\ref{sec:proof_norm_proj}. We also assume $|\phi\ra\not=0$ for at least one $|\psi\ra$, otherwise the lemma is trivially true.
Below we prove the first part of proposition \ref{prop:norm_spec_proj} assuming lemma \ref{lemma:norm_proj}.

First, we rewrite $H_{t}'$ as 
\beq
H_{t}'=h_{s}+(h_{s,s+1}+h_{s-1,s})+\delta H_{s},
\eeq
where we define $h_{-1,0}:=0$ and $h_{q+1,q+2}:=0$ and 
\beq
\delta H_{s}:=H_{t}'-(h_{s,s+1}+h_{s-1,s})-h_{s}.
\eeq
Next, we note
\beq
\begin{split}
    \la h_{s-1,s}+h_{s,s+1}\ra_{\phi}&\geqslant-||h_{s-1,s}||-||h_{s,s+1}||\geqslant -4g_{0},\\
    \la h_{s}\ra_{\phi}&:=\frac{\la\phi|h_s|\phi\ra}{|||\phi\ra||}=\frac{1}{||\,|\phi\ra||^{2}}\la\psi|\Pi_{\leqslant E}\Pi_{>E'}^{(s)}h_{s}\Pi_{>E'}^{(s)}\Pi_{\leqslant E}|\psi\ra> E'.
\end{split}
\eeq
Finally, we denote $E_{\Lambda_{s},0}:=\inf(\sigma(\delta H_{s}))$ and $E_{s,0}:=\inf(\sigma(h_{s}))$. For any $\epsilon>0$, we define normalized states $|E_{s}^{\epsilon}\ra$ and $|E_{\Lambda_{s}}^{\epsilon}\ra$ such that
\beq
\begin{split}
    ||h_{s}|E_{s,0}^{\epsilon}\ra-E_{s,0}|E_{s,0}^{\epsilon}\ra||&<\epsilon\\
    ||\delta H_{s}|E_{\Lambda_{s},0}^{\epsilon}\ra-E_{\Lambda_{s},0}|E_{\Lambda_{s},0}^{\epsilon}\ra||&<\epsilon
\end{split}
\eeq
Now consider $|\varphi^{\epsilon}\ra:=|E_{s,0}^{\epsilon}\ra\otimes |E_{\Lambda_{s},0}^{\epsilon}\ra$, then
\beq
\begin{split}
    0&=E_{t,0}\leqslant \la\varphi^{\epsilon}|H_{t}'|\varphi^{\epsilon}\ra\\
    &<2\epsilon+E_{s,0}+E_{\Lambda_{s},0}+||h_{s,s+1}||+||h_{s-1,s}||\\
    &\leqslant2\epsilon+E_{s,0}+E_{\Lambda_{s},0}+4g_{0}
\end{split}
\eeq
Since $\epsilon$ is arbitrary, we conclude that
\beq
E_{\Lambda_{s},0}\geqslant E_{t,0}-E_{s,0}-4g_{0}
\eeq
So we have 
\beq
\la \delta H_{s}\ra_{\phi}\geqslant E_{\Lambda_{s},0}\geqslant E_{t,0}-E_{s,0}-4g_{0}
\eeq
Collecting all above pieces we obtain that $\la H'_t\ra_\phi\geqslant E'+E_{t, 0}-E_{s, 0}-8g_0$. Combining this result with Lemma \ref{lemma:norm_proj}, we get
\beq
||\,|\phi\ra||\leqslant\frac{4e^{3/2}}{e-1}e^{-\lambda(\delta E_{s}'-\delta E-8g_{0})},
\eeq
where $\delta E'_s=E'-E_{s, 0}$ and $\delta E=E-E_{t, 0}$.

Therefore,
\beq
||\Pi_{<E'}^{(s)}\Pi_{\geqslant E}||=\sup_{|\psi\ra}||\,|\phi\ra||\leqslant\frac{4e^{3/2}}{e-1}e^{-\lambda(\delta E'_{s}-\delta E-8g_{0})}
\eeq
So the first part of Proposition \ref{prop:norm_spec_proj} has been proved.

\subsubsection{Proof of Lemma \ref{lemma:norm_proj}}\label{sec:proof_norm_proj}

The proof is based on the following lemma, which is proved in Appendix \ref{sec:proof_projected_local}.

\begin{lemma}\label{lemma:projected_local}
    Let $O_{s}$ be an operator supported on $B_{s}'$ such that $[O_{s},h_{s}]=0$, then
    \beq
    ||\Pi_{\geqslant E'}O_{s}\Pi_{\leqslant E}||<4||O_{s}||e^{-\lambda(E'-E)}
    \eeq
    where $\Pi_{\geqslant E'}$ and $\Pi_{\leqslant E}$ are spectral projections associated to $H_{t}'$, and 
    \beq
    \lambda:=\frac{1}{8gk+16g_{0}}
    \eeq
\end{lemma}
\begin{proof}[Proof of Lemma \ref{lemma:norm_proj}]
For any $x,y$, we note that
\beq
\begin{split}
    \la\phi|H_{t}'|\phi\ra&=\la\phi|\Pi_{\leqslant x}H_{t}'\Pi_{\leqslant x}|\phi\ra+\sum_{j=1}^{\infty}\la\phi|\Pi_{[x+(j-1)y,x+jy)}H_{t}'\Pi_{[x+(j-1)y,x+jy)}|\phi\ra\\
    &\leqslant x ||\Pi_{\leqslant x}|\phi\ra||^{2}+\sum_{j=1}^{\infty}(x+jy)||\Pi_{[x+(j-1)y,x+jy)}|\phi\ra||^{2}\\
    &=x||\,|\phi\ra||^{2}+y\sum_{j=1}^{\infty}j ||\Pi_{[x+(j-1)y,x+jy)}\Pi_{>E'}^{(s)}\Pi_{\leqslant E}|\psi\ra||^{2}\\
    &=x||\,|\phi\ra||^{2}+y\sum_{j=1}^{\infty}j ||\Pi_{[x+(j-1)y,x+jy)}\Pi_{>E'}^{(s)}\Pi_{\leqslant E}||^{2}
\end{split}
\eeq
We then make use of the lemma \ref{lemma:projected_local} since $[\Pi^{(s)}_{>E'},h_{s}]=0$, we have
\beq
||\Pi_{[x+(j-1)y,x+jy)}\Pi_{>E'}^{(s)}\Pi_{\leqslant E}||\leqslant ||\Pi_{\geqslant (x+(j-1)y}\Pi_{>E'}^{(s)}\Pi_{\leqslant E}||<4 e^{-\lambda(x+(j-1)y-E)}
\eeq
Choosing $y:=\frac{1}{2\lambda}$ and using that $\sum_{j=1}^{\infty}je^{-j}=\frac{e}{(e-1)^{2}}$, we obtain
\beq 
||\,|\phi\ra||^{2}\la H_{t}'\ra_{\phi}=\la\phi|H_{t}'|\phi\ra< x||\,|\phi\ra||^{2}+\frac{8e^{2}}{\lambda(e-1)^{2}}e^{-2\lambda(x-E)}
\eeq
Choosing $x=\la H_{t}'\ra-\frac{1}{2\lambda}$, we then conclude
\beq
||\,|\phi\ra||< \frac{4e^{3/2}}{e-1}e^{-\lambda(\la H_{t}'\ra_{\phi}-E)}
\eeq
\end{proof}

\subsubsection{Proof of Lemma \ref{lemma:projected_local}}\label{sec:proof_projected_local}

The following proof is based on the following lemma.
\begin{lemma}\label{lemma:p-body}
   Let $\delta_{H'_{t}}$ be the derivation defined by $H_{t}'$, then for any $n\in\mathbb{N}$ we have
   \beq
       ||\ad_{H_{t}'}^{n}(h_{s,s+1})||\leqslant (4gk)^{n}n!(2g_{0}),\quad s=0,1,...,q.
   \eeq
\end{lemma}
The proof of Lemma \ref{lemma:p-body} is postponed to Appendix~\ref{sec:p-body}.
\begin{proof}[Proof of Lemma \ref{lemma:projected_local}]
    For any $\nu\in\R$, we have
    \beq
    ||\Pi_{\geqslant E'}O_{s}\Pi_{\leqslant E}||=||\Pi_{\geqslant E'}e^{-\nu H_{t}'}e^{\nu H_{t}'}O_{s}e^{-\nu H_{t}'}e^{\nu H_{t}'}\Pi_{\leqslant E}||
    \eeq
    where $e^{-\nu H_{t}'}$ is the exponential of $H_{t}'$ defined by functional calculus. Note that
    \beq
    \begin{split}
        ||\Pi_{\geqslant E'}e^{-\nu H_{t}'}||&\leqslant e^{-\nu E'}\\
        ||e^{\nu H_{t}'}\Pi_{\leqslant E'}||&\leqslant e^{\nu E}
    \end{split}
    \eeq
    Therefore,
    \beq\label{eq:Ad_series}
    \begin{split}
        ||\Pi_{\geqslant E'}O_{s}\Pi_{\leqslant E}||&\leqslant e^{\nu(E-E')}||\Ad_{e^{\nu H_{t}'}}(O_{s})||\\
    &=e^{\nu(E-E')}||e^{\nu\ad_{H_{t}'}}(O_{s}||\\
    &\leqslant e^{\nu(E-E')}\sum_{m=0}^{\infty}\frac{\nu^{m}}{m!}||\ad_{H_{t}}^{m}(O_{s})||
    \end{split}
    \eeq
    where we have used $\ad_{X}(Y):=[X,Y]$ and $\Ad_{e^{X}}=e^{\ad_{X}}$ for any operators $X,Y$.

    Below we show by induction that 
    \beq\label{eq:ad_bound}
    ||\ad_{H_{t}'}^{m}(O_{s})||\leqslant (4gk+8g_{0})^{m}m! ||O_{s}||
    \eeq
    For $m=0$, we have $\ad_{H_{t}'}^{0}=\mathrm{id}$, so the above inequality holds trivially. Let us assume that this inequality is true for all $m\leqslant m_{0}$. For $m=m_{0}+1$,
    \beq
    &||\ad_{H_{t}'}^{m_{0}+1}(O_{s})||=||\ad_{H_{t}'}^{m_{0}}([h_{s,s+1}+h_{s-1,s},O_{s}])||\\&\leqslant ||\ad_{H_{t}'}^{m_{0}}((h_{s-1,s}+h_{s,s+1})O_{s})||+||\ad_{H_{t}'}^{m_{0}}(O_{s}(h_{s-1,s}+h_{s,s+1}))||
    \eeq
    where we have used the facts that $O_{s}$ is supported on $B_{s}'$ only and that $[h_{s},O_{s}]=0$. 
    
    Next, we employ the following identity
    \beq\label{eq:binomial_expansion}
    \ad_{A}^{m}(BC)=\sum_{m'=0}^{m}\ad_{A}^{m'}(B)\ad_{A}^{m-m'}(C)\begin{pmatrix}
        m\\
        m'
    \end{pmatrix}
    \eeq
    where 
    \beq
    \begin{pmatrix}
        m\\
        m'
    \end{pmatrix}:=\frac{m!}{(m')!(m-m')!}
    \eeq
    is the binomial coefficient. Eq.~\eqref{eq:binomial_expansion} can be proved by induction and the Lebniz rule: $\ad_{A}(BC)=\ad_{A}(B)C+B\ad_{A}(C)$. We have
    \beq
    \begin{split}
        ||\ad_{H_{t}'}^{m_{0}}((h_{s-1,s}+h_{s,s+1})O_{s})||+||\ad_{H_{t}'}^{m_{0}}(O_{s}(h_{s-1,s}+h_{s,s+1}))||\\ \leqslant2\sum_{m_{1}+m_{2}=m_{0}}||\ad_{H_{t}'}^{m_{1}}(h_{s-1,s}+h_{s,s+1})||\,||\ad_{H_{t}'}^{m_{2}}(O_{s})||\begin{pmatrix}
        m_{0}\\
        m_{1}
    \end{pmatrix}.
    \end{split}
    \eeq
    By the inductive assumption, 
    \beq
    ||\ad_{H_{t}'}^{m_{2}}(O_{s})||\leqslant (4gk+8g_{0})^{m_{2}}(m_{2})!||O_{s}||
    \eeq
    Using Lemma \ref{lemma:p-body}, we have
    \beq
    ||\ad_{H_{t}'}^{m_{1}}(h_{s-1,s}+h_{s,s+1})||\leqslant 4g_{0} (4gk)^{m_{1}}(m_{1})!
    \eeq
    Thus we obtain
    \beq
    \begin{split}
        ||\ad_{H_{t}'}^{m_{0}+1}(O_{s})||&\leqslant 8g_{0}(m_{0})!||O_{s}||\sum_{m_{1}+m_{2}=m_{0}}(4gk)^{m_{1}}(4gk+8g_{0})^{m_{2}}\\
        &<8g_{0}(m_{0})!(4gk+8g_{0})^{m_{0}}||O_{s}||\sum_{m_{2}=0}^{m_{0}}1\\
        &<(m_{0}+1)!(4gk+8g_{0})^{m_{0}+1}||O_{s}||,\\
    \end{split}
    \eeq
    which proves Eq.~\eqref{eq:ad_bound}.

    Combining Eq.~\eqref{eq:Ad_series} and Eq.~\eqref{eq:ad_bound}, we obtain
    \beq
    \begin{split}
        ||\Pi_{\geqslant E'}O_{s}\Pi_{\leqslant E}||&\leqslant 2e^{\nu(E-E')}\sum_{m=0}^{\infty}\frac{\nu^{m}}{m!}(4gk+8g_{0})^{m}m!||O_{s}||\\
        &=\frac{2}{1-\nu(4gk+8g_{0})}e^{-\nu(E'-E)}||O_{s}||
    \end{split}
    \eeq
    We choose $\nu=\lambda:=\frac{1}{8gk+16g_{0}}$ so we have
    \beq
    ||\Pi_{\geqslant E'}O_{s}\Pi_{\leqslant E}||<4e^{-\lambda(E'-E)}||O_{s}||
    \eeq
\end{proof}

\subsubsection{Proof of Lemma \ref{lemma:p-body}}\label{sec:p-body}

We start by proving the following: For each $h_{Z}$ appearing in $H'_{t}$, we have
\beq\label{eq:ad_power}
||\ad_{H'_{t}}^{n}(h_{Z})||\leqslant (4gk)^{n}n! ||h_{Z}||
\eeq
To see this, note that
\beq
\begin{split}
    ||\ad_{H_{t}'}^{n}(h_{Z})||&=||\sum_{Z_{n}}\sum_{Z_{n_{0}}}\dots\sum_{Z_{1}}\ad_{h_{Z_{n}}}...\ad_{h_{Z_{1}}}(h_{Z})||\\
    &\leqslant\sum_{Z_{n}}\dots\sum_{Z_{1}}||\ad_{h_{Z_{n}}}...\ad_{h_{Z_{1}}}(h_{Z})||
\end{split}
\eeq
Below we prove a slightly stronger result than Eq.~\eqref{eq:ad_power} by induction:
\beq\label{eq:adjoint_power}
\sum_{Z_{n}}\dots\sum_{Z_{1}}||\ad_{h_{Z_{n}}}...\ad_{h_{Z_{1}}}(h_{Z})||\leqslant (4gk)^{n}n!||h_{Z}||
\eeq

For $n=0$, the summation becomes trivial and there is no adjoint action, so it holds trivially. Let us assume that it is true for $n\leqslant n_{0}-1$. When $n=n_{0}$, we have
\beq
\begin{split}
    \sum_{Z_{n}}\dots\sum_{Z_{1}}||\ad_{h_{Z_{n}}}...\ad_{h_{Z_{1}}}(h_{Z})||&=\sum_{Z_{n-1}}...\sum_{Z_{1}}\sum_{Z_{n}\cap(Z\cup Z_{1}...\cup Z_{n-1})\not=\emptyset}||\ad_{h_{Z_{n}}}...\ad_{h_{Z_{1}}}(h_{Z})||\\
    &\leqslant \sum_{Z_{n-1}}...\sum_{Z_{1}}\sum_{i\in Z\cup Z_{1}...\cup Z_{n-1}}\sum_{Z_{n}\ni i}||\ad_{h_{Z_{n}}}...\ad_{h_{Z_{1}}}(h_{Z})||\\
    &\leqslant \sum_{Z_{n-1}}...\sum_{Z_{1}}\sum_{i\in Z\cup Z_{1}...\cup Z_{n-1}}\sum_{Z_{n}\ni i}2||h_{Z_{n}}||\,||\ad_{h_{Z_{n-1}}}...\ad_{h_{Z_{1}}}(h_{Z})||\\
    &< \sum_{Z_{n-1}}...\sum_{Z_{1}}nk(4g)||\ad_{h_{Z_{n-1}}}...\ad_{h_{Z_{1}}}(h_{Z})||\\
    &\leqslant4ngk (4gk)^{n-1}(n-1)!||h_{Z}||\\
    &=(4gk)^{n}n! ||h_{Z}||
\end{split}
\eeq
To derive the third last line, we have used the generalization of Eq. \eqref{eq:g_extensive} to the setup after folding, and we have also used the fact that $|Z\cup Z_{1}\cup...\cup Z_{n-1}|\leqslant nk$ since each of them is supported on at most $k$-sites. In the second last line, we have used the inductive assumption.
As a result,
\beq
\begin{split}
    ||\ad_{H_{t}'}^{n}(h_{s,s+1})||&\leqslant \sum_{\substack{|Z|\leqslant k\\Z\cap B_{s}'\not=\emptyset\\Z\cap B_{s+1}'\not=\emptyset}}||\ad_{H'_{t}}^{n}(h_Z)||\\
    &<(4gk)^{n}n! \sum_{\substack{|Z|\leqslant k\\Z\cap B_{s}'\not=\emptyset\\Z\cap B_{s+1}'\not=\emptyset}}||h_{Z}||\\
    &\leqslant 2g_{0}(4gk)^{n}n! 
\end{split}
\eeq
where we have used the same estimation as in lemma \ref{lemma:quasi_local_perturb} in the second last line. This completes the proof.

\subsection{Proof of Proposition \ref{prop:norm_spec_proj}: Second part}\label{sec:proof_norm_spec_proj_2}

In the second part, we aim at bounding $||\Pi_{>E'}^{(s)}\tilde{\Pi}_{\leqslant E}||$. The idea is similar to the first part. More specifically, the proof will be established by proving an $\tilde{H}'_{t}$-version of Lemma \ref{lemma:norm_proj} and Lemma \ref{lemma:projected_local}. However, there is an important distinction. Note $\tilde{H}'_{t}$ may not be $k$-local in general, so there is no analogue of Lemma \ref{lemma:p-body}.

We start with the analogue of lemma \ref{lemma:norm_proj}. Consider any normalized state $|\psi\ra$, and we define
\beq
|\tilde{\phi}\ra:=\Pi^{(s)}_{>E'}\tilde{\Pi}_{\leqslant E}|\psi\ra
\eeq
So we have
\beq
||\Pi^{(s)}_{>E'}\tilde{\Pi}_{\leqslant E}||=\sup_{|\psi\ra}||\,|\tilde{\phi}\ra||
\eeq
\begin{lemma}\label{lemma:tilde_proj_norm}
    The norm of $|\tilde{\phi}\ra$ is bounded as follows
    \beq
    ||\,|\tilde{\phi}\ra||\leqslant \frac{4e^{3/2}}{e-1}e^{-\lambda'(\la \tilde{H}_{t}'\ra_{\tilde{\phi}}-E)}
    \eeq
    where
    \beq
    \la\tilde{H}_{t}'\ra_{\tilde{\phi}}:=\begin{cases}
        \frac{\la\tphi|\tH|\tphi\ra}{||\,|\tphi\ra||^{2}},\quad\text{if}\,|\tphi\ra\not=0\\
        0, \quad\text{\rm otherwise}
    \end{cases}
    \eeq
\end{lemma}
Assuming lemma \ref{lemma:tilde_proj_norm} for now, we rewrite $\la\tH\ra_{\tphi}$ in the following form
\beq
\la\tH\ra_{\tphi}=\la\tilde{h}_{s}\ra_{\tphi}+\la(h_{s,s+1}+h_{s-1,s})\ra_{\tphi}+\la\delta \tilde{H}_{s}\ra
\eeq
where we have defined $\delta\tilde{H}_{s}:=\tH-\tilde{h}_{s}-h_{s,s+1}-h_{s-1,s}$. We denote $\Lambda_{s}:=\Lambda\setminus B_{s}'$. Then for any $\epsilon>0$, consider the $\epsilon$-almost ground state $|\tilde{E}_{\Lambda_{s},0}^{\epsilon}\ra$ where $\tilde{E}_{\Lambda_{s},0}:=\inf(\sigma(\delta\tilde{H}_{s}))$.
We have
\beq
\begin{split}
    \la \tilde{h}_{s}\ra_{\tphi}=\frac{\la\psi|\tilde{\Pi}_{\leqslant E}\Pi^{(s)}_{>E'}\tilde{h}_{s}\Pi^{(s)}_{>E'}\tilde{\Pi}_{\leqslant E}|\psi\ra}{||\,|\tphi\ra||^{2}}&\geqslant \min(E',\tau_{s})\\
    \la(h_{s,s+1}+h_{s-1,s}\ra_{\tphi}\geqslant -||h_{s,s+1}||-||h_{s-1,s}||&\geqslant-4g_{0}\\
    \la\delta \tilde{H}_{s}\ra_{\tphi}\geqslant \tilde{E}_{\Lambda_{s},0}\geqslant \tilde{E}_{t,0}-E_{s,0}-4g_{0}
\end{split}
\eeq
where in the first inequality we have used that $||\tilde{h}_{s}||=\tau_{s}$ and in the last inequality, we have used
\beq
\begin{split}
    \tilde{E}_{t,0}&\leqslant (\la E_{s,0}^{\epsilon}|\otimes \la \tilde{E}_{\Lambda_{s},0}^{\epsilon}|)\tH(|E_{s,0}^{\epsilon}\ra\otimes|\tilde{E}_{\Lambda_{s},0}^{\epsilon}\ra)\\
    &\leqslant E_{s,0}+\epsilon+E_{\Lambda_{s},0}+\epsilon+4g_{0}
\end{split}
\eeq
where $|E_{s,0}^{\epsilon}\ra$ is the $\epsilon$-almost ground state of $h_{s}$.
Since $\epsilon$ is arbitrarily small, we have $\tilde{E}_{\Lambda_{s},0}\geqslant \tilde{E}_{t,0}-E_{s,0}-4g_{0}$. Then Eq.~\eqref{eq:tilde_norm_proj} follows from Lemma \ref{lemma:tilde_proj_norm}.

Now let us turn to the proof of Lemma \ref{lemma:tilde_proj_norm}. The proof is almost identical to the proof of Lemma \ref{lemma:norm_proj} except that lemma \ref{lemma:projected_local} does not hold for $\tH$ since it is not $k$-local, so there will not be an analogue of lemma \ref{lemma:p-body}. Nevertheless, we can establish the following analogue to lemma \ref{lemma:projected_local}, which will be sufficient to prove lemma \ref{lemma:tilde_proj_norm}.
\begin{lemma}\label{lemma:tilde_projected_local}
    Let $O_{s}$ be any operator supported on $B_{s}'$ such that $[O_{s},h_{s}]=0$, then we have
    \beq
    ||\tilde{\Pi}_{\geqslant E'}O_{s}\tilde{\Pi}_{\leqslant E}||\leqslant 4||O_{s}||e^{-\lambda'(E'-E)}
    \eeq
    where $\lambda':=\min\{\frac{1}{8gk},\frac{1}{224g_{0}}\}$.
\end{lemma}
As is explained, the proof of this lemma requires new ideas and the full proof is given in Sec.~\ref{sec:tilde_proj_local}.

\subsection{Proof of Lemma \ref{lemma:tilde_projected_local}}\label{sec:tilde_proj_local}
Following the proof of Lemma \ref{lemma:projected_local}, 
\beq
\begin{split}
    ||\tilde{\Pi}_{\geqslant E'}O_{s}\tilde{\Pi}_{\leqslant E}||&=||\tilde{\Pi}_{\geqslant E'}e^{-\nu \tH}e^{\nu\tH}O_{s}e^{-\nu \tH}e^{\nu \tH}\tilde{\Pi}_{\leqslant E}||\\
    &\leqslant e^{-\nu(E'-E)} ||\Ad_{e^{\nu\tH}}(O_{s})||\\
\end{split}
\eeq
Next, we decompose $\tH$ into the following \textit{free} and \textit{interacting} parts
\beq
\begin{split}
    G&:=\sum_{s=0}^{q+1}\tilde{h}_{s}\\
    F&:=\sum_{s=0}^{q}h_{s,s+1}
\end{split}
\eeq
Consider
\beq
U(x):=e^{x \tH}e^{-xG}
\eeq
It is easy to check $\partial_{x}U=U F(x)$ where $F(x):=e^{xG}Fe^{-xG}$. Solving this differential equation of $U(x)$, one can rewrite $U(x)$ as Dyson series by iterative integration,
\beq\label{eq:Dyson}
    U(x)=\T_{\rightarrow}e^{\int_{0}^{x}F(t)\dd t}
\eeq
Here $\T_{\rightarrow}$ is the ordering operator such that
\beq
\T_{\rightarrow}(F(x_{1})F(x_{2})...F(x_{n}))=F(x_{\sigma(1)})F(x_{\sigma(2)})\dots F(x_{\sigma(n)})
\eeq
where $\sigma\in S_{n}$ such that $x_{\sigma(1)}\leqslant x_{\sigma(2)}\leqslant\dots \leqslant x_{\sigma(n)}$. It is useful to rewrite $U(x)$ as an infinite product,
\begin{lemma}[Theorem 5.5.10 of Ref.~\cite{slavik2007product}]
    The operator $U(x)$ can be written as 
    \beq\label{eq:infintie_product}
    U(x)=\lim_{N\to\infty}\prod_{m=1}^{N}(1+F(x_{m})\Delta x)=\lim_{N\to\infty}\prod_{m=1}^{N}e^{F(x_{m})\Delta x}
    \eeq
    where $\Delta x:=\frac{x}{N}$ and $x_{m}:=\frac{mx}{N}$.
\end{lemma}
It is also easy to check that both Eq.~\eqref{eq:infintie_product} and Eq.~\eqref{eq:Dyson} converge if
\beq
\int_{0}^{\nu}||F(x)||\dd x<\infty
\eeq
for $\nu\in[0, \infty)$. The following lemma tells us that this is always true.
\begin{lemma}[(S.358) of Ref.~\cite{Kuwahara2019}]\label{lemma:interaction_norm}
    For any $0\leqslant x<\frac{1}{8gk}$, we have
    \beq
    ||h_{s,s+1}(x)||< \fg:=56g_{0}
    \eeq
    where $h_{s,s+1}(x):=e^{xG}h_{s,s+1}e^{-xG}$ is still supported on $B_{s}'\cup B_{s+1}'$.
\end{lemma}
Note $||F(x)||\leqslant\sum_{s=0}^{q}||h_{s,s+1}(x)||\leqslant 2(q+1)\fg$, therefore $\int_{0}^{\nu}||F(x)||\dd x<\infty$ for all $\nu\geqslant0$.

By assumption $[O_{s},G]=[O_{s},\tilde{h}_{s}]=0$, we have
\beq
\begin{split}
    e^{\nu \tH}O_{s}e^{-\nu\tH}&=U(\nu)O_{s}U(\nu)^{-1}\\
    &=\Ad_{U(\nu)}(O_{s})\\
    &=\lim_{N\to\infty}\prod_{m=1}^{N}\Ad_{e^{F(x_{m})\Delta x}}(O_{s})\\
    &=\lim_{N\to\infty}\prod_{m=1}^{N}e^{\ad_{F(x_{m})}\Delta x}(O_{s})\\
    &=\T_{\rightarrow}e^{\int_{0}^{\nu}\ad_{F(x)}\dd x}(O_{s})\\
    &=\sum_{m=0}^{\infty}\int_{0\leqslant x_{1}\leqslant\dots\leqslant x_{m}\leqslant \nu}\dd x_{1}\dots\dd x_{m}\ad_{F(x_{1})}\dots \ad_{F(x_{m})}(O_{s})
\end{split}
\eeq
All terms of order $o(\frac{1}{N})$ on the exponent are suppressed.
Note that $\ad_{F(x)}(O_{s})=[h_{s,s+1}(x)+h_{s-1,s}(x),O_{s}]$, by lemma \ref{lemma:interaction_norm} we have
\beq
||[h_{s,s+1}(x),O_{s}]||\leqslant 2||h_{s,s+1}(x)||\cdot ||O_{s}||\leqslant 2\fg ||O_{s}||
\eeq
Consequently, $\ad_{F}(O_{s})$, which is supported on $B_{s-1}'\cup B_{s}'\cup B_{s+1}'$, satisfies
\beq
||\ad_{F}(O_{s})||\leqslant 4\fg ||O_{s}||
\eeq
Repeating the same argument, one can show that
\beq
||\ad_{F(x_{1})}\dots \ad_{F(x_{m})}(O_{s})||\leqslant (m+1)!(2\fg)^{m}||O_{s}||
\eeq
It follows that
\beq
\begin{split}
    ||e^{\nu \tH}O_{s}e^{-\nu\tH}||&\leqslant\sum_{m=0}^{\infty}\int_{0\leqslant x_{1}\leqslant\dots\leqslant x_{m}\leqslant \nu}\dd x_{1}\dots\dd x_{m}(m+1)!(2\fg)^{m}||O_{s}||\\
    &=\sum_{m=0}^{\infty}(m+1)(2\fg \nu)^{m}||O_{s}||\\
    &=\frac{1}{(1-2\fg \nu)^{2}}||O_{s}||
\end{split}
\eeq
where we have used $\sum_{m=0}^{\infty}(m+1)r^{m}=(1-r)^{-2}$. Therefore, choosing $\nu\leqslant\lambda':=\min\{\frac{1}{8gk},\frac{1}{4\fg}\}$, we have
\beq
||e^{\nu \tH}O_{s}e^{-\nu\tH}||\leqslant 4||O_{s}||
\eeq
This completes the proof.

\subsubsection{Proof of Lemma \ref{lemma:interaction_norm}}
This proof is adapted from Ref.~\cite{Kuwahara2019} and we include it here for the convenience of readers. 
Let us begin with the following observation: for any $s=0,1,\dots,q+1$, we have
\beq
e^{x\tilde{h}_{s}}=\Pi_{<\tau_{s}}^{(s)}e^{xh_{s}}+\Pi_{\geqslant\tau_{s}}^{(s)}e^{x\tau_{s}}
\eeq
Therefore, for any operator $W$, we have
\beq
\begin{split}
    e^{x\tilde{h}_{s}}We^{-x\tilde{h}_{s}}=\Pi^{(s)}_{<\tau_{s}}e^{xh_{s}}We^{-xh_{s}}\Pi_{<\tau_{s}}^{(s)}+\Pi_{<\tau_{s}}^{(s)}e^{x(h_{s}-\tau_{s})}W\Pi_{\geqslant\tau_{s}}^{(s)}
    \\+\Pi_{\geqslant\tau_{s}}^{(s)} We^{-x(h_{s}-\tau_{s})}\Pi_{<\tau_{s}}^{(s)}+\Pi_{\geqslant\tau_{s}}^{(s)}W\Pi_{\geqslant\tau_{s}}^{(s)}
\end{split}
\eeq
For each term on the right-hand side, we note for any $x\geqslant 0$,
\beq
\begin{split}
    ||\Pi^{(s)}_{<\tau_{s}}e^{xh_{s}}We^{-xh_{s}}\Pi_{<\tau_{s}}^{(s)}||&\leqslant ||e^{xh_{s}}We^{-xh_{s}}||,\\
    ||\Pi_{<\tau_{s}}^{(s)}e^{x(h_{s}-\tau_{s})}W\Pi_{\geqslant\tau_{s}}^{(s)}||&\leqslant ||\Pi_{<\tau_{s}}^{(s)}e^{x(h_{s}-\tau_{s})}||\cdot ||W\Pi_{\geqslant\tau_{s}}^{(s)}||< ||W||,\\
    ||\Pi_{\geqslant\tau_{s}}^{(s)} We^{-x(h_{s}-\tau_{s})}\Pi_{<\tau_{s}}^{(s)}||&=||\Pi_{\geqslant\tau_{s}}^{(s)}e^{-x(h_{s}-\tau_{s})}e^{xh_{s}} We^{-xh_{s}}\Pi_{<\tau_{s}}^{(s)}||\leqslant ||e^{xh_{s}} We^{-xh_{s}}||,\\
    ||\Pi_{\geqslant\tau_{s}}^{(s)}W\Pi_{\geqslant\tau_{s}}^{(s)}||&\leqslant ||W||.
\end{split}
\eeq
Thus we obtain,
\beq\label{eq:reduced_conjugation}
||e^{x\tilde{h}_{s}}We^{-x\tilde{h}_{s}}||< 2||W||+2||e^{xh_{s}} We^{-xh_{s}}||
\eeq
Recall that our goal is to bound the norm of $h_{s,s+1}(x)=e^{x\tilde{h}_{s}}e^{x\tilde{h}_{s+1}}h_{s,s+1}e^{-x\tilde{h}_{s+1}}e^{-x\tilde{h}_{s}}$. By Eq.~\eqref{eq:reduced_conjugation}
\beq
\begin{split}
    ||e^{x\tilde{h}_{s}}h_{s,s+1}e^{-x\tilde{h}_{s}}||&< 2||h_{s,s+1}||+2||e^{xh_{s}} h_{s,s+1}e^{-xh_{s}}||\\
    ||e^{x\tilde{h}_{s+1}}h_{s,s+1}e^{-x\tilde{h}_{s+1}}||&< 2||h_{s,s+1}||+2||e^{xh_{s+1}} h_{s,s+1}e^{-xh_{s+1}}||
\end{split}
\eeq
Note that
\beq
\begin{split}
    e^{xh_{s}}h_{s,s+1}e^{-xh_{s}}&=\sum_{m=0}^{\infty}\frac{x^{m}}{m!}\ad_{h_{s}}^{m}(h_{s,s+1})\\
\end{split}
\eeq
Note that for each $Z\subset\Lambda$ with $|Z|\leqslant k$,
\beq
\begin{split}
    ||\ad_{h_{s}}^{m}(h_{Z})||&=||\sum_{Z_{1},\dots,Z_{m}\subset B_{s}'}\ad_{h_{Z_{1}}}\dots\ad_{h_{Z_{m}}}(h_{Z})||\\
    &\leqslant\sum_{Z_{1},\dots,Z_{m}\subset B_{s}'}||\ad_{h_{Z_{1}}}\dots\ad_{h_{Z_{m}}}(h_{Z})||\\
    &\leqslant\sum_{Z_{1},\dots,Z_{m}}||\ad_{h_{Z_{1}}}\dots\ad_{h_{Z_{m}}}(h_{Z})||\\
    &\leqslant (4gk)^{m}m!||h_{Z}||
\end{split}
\eeq
where we have used Eq.~\eqref{eq:adjoint_power}. Therefore, when $0\leqslant x<\frac{1}{8gk}$,
\beq
||e^{xh_{s}}h_{Z}e^{-xh_{s}}||\leqslant\frac{1}{1-4gk x}||h_{Z}||<2||h_{Z}||
\eeq

Similarly for $e^{xh_{s+1}}h_{Z}e^{-xh_{s+1}}$, we have
\beq
||e^{xh_{s+1}}h_{Z}e^{-xh_{s+1}}||\leqslant 2||h_{Z}||
\eeq
Now consider
\beq
\begin{split}
    ||e^{xh_{s}}h_{s,s+1}e^{-xh_{s}}||&=||\sum_{\substack{Z\cap B_{s}'\not=\emptyset\\Z\cap B_{s+1}'\not=\emptyset\\|Z|\leqslant k}}e^{xh_{s}}h_{Z}e^{-xh_{s}}||\\
    &\leqslant\sum_{\substack{Z\cap B_{s}'\not=\emptyset\\Z\cap B_{s+1}\not=\emptyset\\|Z|\leqslant k}}||e^{xh_{s}}h_{Z}e^{-xh_{s}}||\\
    &\leqslant 2\sum_{\substack{Z\cap B_{s}'\not=\emptyset\\Z\cap B_{s+1}\not=\emptyset\\|Z|\leqslant k}}||h_{Z}||\\
    &\leqslant 4g_{0}
\end{split}
\eeq
where we have used the definition of $g_{0}$. Similarly we have
\beq
\begin{split}
    ||e^{xh_{s+1}}h_{s,s+1}e^{-xh_{s+1}}||&\leqslant 4g_{0}\\
    ||e^{x(h_{s+1}+h_{s})}h_{s,s+1}e^{-x(h_{s}+h_{s+1})}||&\leqslant 4g_{0}
\end{split}
\eeq
We end up with 
\beq
\begin{split}
    ||e^{x\tilde{h}_{s}}e^{x\tilde{h}_{s+1}}h_{s,s+1}e^{-x\tilde{h}_{s+1}}e^{-x\tilde{h}_{s}}||&\leqslant 2||e^{x\tilde{h}_{s+1}}h_{s,s+1}e^{-x\tilde{h}_{s+1}}||+2||e^{xh_{s}}e^{x\tilde{h}_{s+1}}h_{s,s+1}e^{-x\tilde{h}_{s+1}}e^{-xh_{s}}||\\
    &\leqslant4||h_{s,s+1}||+4||e^{xh_{s+1}}h_{s,s+1}e^{-xh_{s}}||\\&+4||e^{xh_{s}}h_{s,s+1}e^{-xh_{s}}||+4||e^{x(h_{s}+h_{s+1})}h_{s,s+1}e^{-x(h_{s}+h_{s+1})}||\\
    &\leqslant56g_{0}=:\fg
\end{split}
\eeq
This completes the proof.

\section{Group cohomology and differentiable group cohomology}\label{sec:group_cohomology}

In this section, we review the basics of group cohomology and differentiable group cohomology. For group cohomology, there are many materials in the literature \cite{brown2012cohomology,Weibel_1994_group,Chen2010}. For differentiable group cohomology, see appendix A.1 of Ref. \cite{kapustin2024anomalous} and Ref. \cite{brylinski2000differentiable}. We will only cover the motivations and basics here.

\subsection{Projective representations in quantum mechanics}

To motivate group cohomology, we start with projective representations in quantum mechanics. Suppose we have a symmetry group $G$ (assumed to be unitary and discrete for simplicity) acting on a Hilbert space $\cH$. Usually this symmetry action is given by a homomorphism $\rho:G\to U(\cH)$, \ie a unitary representation of $\cH$. More explicitly, for each $g\in G$, we assign a unitary operator $\rho(g)$ such that
\beq
\rho(g)\rho(h)=\rho(gh),\,\forall\,g,h\in G
\eeq
However, in quantum mechanics, states are {\it not} really a vector in $\cH$, but a {\it{ray}}. That means a state $|\psi\ra$ is the same as $e^{i\theta}|\psi\ra$ as a quantum state. Thus, the space of states is not literally $\cH$, but the projective space $P(\cH)$. This for allows more general symmetry actions as
\beq
\rho(g)\rho(h)=\omega(g,h)\rho(gh)
\eeq
where $\omega(g,h)\in \U$.\footnote{In principle, one has to show that the phase $\omega(g,h)$ is the same on each quantum state. This relies the coherence of these states and one can find the proof in Sec. 2.2 of Ref. \cite{weinberg2005quantum}.} This $\rho$ is a representation up to a phase $\omega$ and is called a projective representation.
Moreover, the matrix multiplication is associative, so 
\beq
(\rho(g)\rho(h))\rho(k)=\rho(g)(\rho(h)\rho(k))
\eeq
This imposes the following constraint on $\omega$,
\beq\label{eq:2-cocycle_condition}
\omega(g,h)\omega(gh,k)=\omega(g,hk)\omega(h,k)
\eeq
Any function $G\times G\to \U$ satisfying Eq. \eqref{eq:2-cocycle_condition} is called a 2-cocycle. Furthermore, one can redefine the phase of $\rho(g)\to \tilde{\rho}(g)=\rho(g)\eta(g),\eta(g)\in\U$ (we do not require $\eta:G\to \U$ to be a homomorphism), and the resulting 2-cocycle is
\beq\label{eq:shift_2-coboundary}
\tilde{\omega}(g,h)=\omega(g,h)\eta(g)\eta(h)\eta(gh)^{-1}
\eeq
One can easily check that $\tilde{\omega}$ again satisfies the 2-cocycle condition, Eq. \eqref{eq:2-cocycle_condition}. If there exists $\eta(g)$ such that $\tilde{\omega}(g,h)=1$ for all $g,h\in G$, then we say that $\omega$ is a 2-coboundary or trivial. Any two 2-cocycles $\omega$ and $\tilde{\omega}$ related by Eq. \eqref{eq:shift_2-coboundary} are viewed as equivalent, since they differ only by the artificial choice of phase factors $\eta(g)$ of representation matrix $\rho(g)$. We write $\omega\sim \tilde{\omega}$ if $\omega$ and $\tilde\omega$ are equivalent. The space of 2-cocycles modulo this equivalence $\sim$ is the so-called the degree 2 group cohomology of $G$, denoted by $\rH^{2}(G;\U)$.
\begin{example}
    Let us consider $G=\z_{2}\times\z_{2}$. We write its elements as $(a,b)$ where $a,b=0,1 \mod{2}$. Then we define a projective representation $\rho$ as follows:
    \beq
    \begin{split}
        \rho(0,0)=I,\,\rho(1,0)=\sigma_{y}\\
        \rho(0,1)=\sigma_{x},\,\rho(1,1)=\sigma_{z}
    \end{split}
    \eeq
    Note that $\rho(0,1)\rho(1,0)=i\rho(1,1)$ hence $\omega((0,1),(1,0))=i$. Similarly, $\omega((1,0),(0,1))=-i$. One can show that this 2-cocycle is not a 2-coboundary and hence defines the nontrivial class in $\rH^{2}(\z_{2}\times\z_{2};\U)\simeq \z_{2}$. In the context of symmetry-protected topological phases, this projective representation describes the boundary of the cluster state \cite{Son_2011}.
     
\end{example}
\begin{example}\label{example:0+1d_SO3}
    Consider the case where $G=SO(3)$, the spin rotation symmetry \footnote{Actually, this a subtler case because $SO(3)$ is a Lie group so it requires more careful treatment, which will be left to later sections. We omit this subtlety for now.}. One can show that $\rH^{2}(SO(3);\U)\simeq \rHom(\pi_{1}(SO(3)),\U)\simeq \z_{2}$, and this class is trivial if the (total) spin quantum number $S\in \z$ and it is nontrivial if $S\in\z+\frac{1}{2}$.
\end{example}

A projective representation provides the following constraint on quantum states.
\begin{proposition}
        If $G$ acts on the Hilbert space $\cH$ via a projective representation $\rho$ whose associated 2-cocycle $\omega\not =1\in\rH^{2}(G;\U)$, then there cannot be a nonzero $G$-symmetric state.
\end{proposition}
\begin{proof}
    Suppose $|\psi\ra$ is a $G$-symmetric state, that is 
    \beq
    \rho(g)|\psi\ra=\eta(g)^{-1}|\psi\ra
    \eeq
    where $\eta(g)\in \U$ is any $\U$-valued function on $G$. Then one redefines $\tilde{\rho}(g)=\rho(g)\eta(g)$, this shifts $\omega$ by a 2-coboundary and the resulting $\tilde{\omega}$ (see Eq. \eqref{eq:shift_2-coboundary}) is nontrivial, \ie there exists $g,h\in G$ such that $\tilde{\omega}(g,h)\not=1$. Now
    \beq
    \tilde{\rho}(g)|\psi\ra=|\psi\ra,\,\forall \,g\in G
    \eeq
    One can calculate $\tilde{\rho}(g)\tilde{\rho}(h)|\psi\ra$ in 2 different ways
    \beq
    \begin{split}
        \tilde{\rho}(g)(\tilde{\rho}(h)|\psi\ra)&=\tilde{\rho}(g)|\psi\ra=|\psi\ra\\
        (\tilde{\rho}(g)\tilde{\rho}(h))|\psi\ra&=\tilde{\omega}(g,h)\tilde{\rho}(gh)|\psi\ra=\tilde{\omega}(g,h)|\psi\ra
    \end{split}
    \eeq
    By assumption, $\tilde{\omega}(g,h)\not =1$ for some $g,h\in G$. Hence $|\psi\ra=0$, which shows that there is no nonzero $G$-symmetric state.
\end{proof}

As a corollary, consider a $G$-symmetric Hamiltonian $H$ which has a $G$ symmetry that acts projectively. We have

\begin{corollary}
    If $\rho$ is nontrivial projective representation, then a $G$-symmetric Hamiltonian must have degenerate ground states which break the $G$-symmetry.
\end{corollary}

This can be viewed as $(0+1)d$ version of anomaly constraints.

\begin{example}
    Consider a system made of $N$ qubits (or equivalently, spin $\frac{1}{2}$'s), whose Hamiltonian $H$ has a $G=SO(3)$ symmetry encountered in example \ref{example:0+1d_SO3}. If $N=1\mod{2}$, then this system must be at least 2-fold degenerate. For example, consider $N=1$, for the Hamiltonian $H$ to be $SO(3)$-symmetric, it has to commute with all Pauli operators. It is easy to check that $H$ must be $\lambda I$ for some $\lambda\in \bbC$ and $I$ is the identity operator. Hence the ground states are trivially 2-fold degenerate. However, for $N=2$ where the total spin is an integer, one can take 
    \beq
    H=J\vec{S}_{1}\cdot \vec{S}_{2},J>0
    \eeq
    where the ground state is non-degnerate.
\end{example}

\subsection{Group cohomology}

Now we present the definition of group cohomology in general. Let $G$ be a discrete group, one defines a space $BG$ which is a collection of spaces $\{G^{n}\}_{n=1,2,...}$ equipped with a collection of maps $d_{k}:G^{n}\to G^{n-1},\,k=0,1,...,n$ (called face maps). Explicitly,
\beq\label{eq:face_maps}
d_{k}(g_{1},g_{2},...,g_{n})=\begin{cases}
    (g_2,...,g_n),\,k=0\\
    (g_1,...,g_{k}g_{k+1},...,g_{n}),\,0<k<n\\
    (g_1,...,g_{n-1}),k=n
\end{cases}
\eeq
One can check that if $d=\sum_{k=0}^{n}(-1)^{k}d_{k}$, then $d^{2}=0$.
Let $A$ be an Abelian group (with {\it discrete topology}). For example $A$ can be $\z_{2}$, $\z$, $\R$ or $\U$. We denote all $A$-valued functions on $BG$ as $C^{\bullet}(BG,A)$. For example, one writes $\omega\in C^{2}(BG,A)$ if $\omega:G^{2}\to A$.
Consider an $A$-valued function $\omega$ on $G^{n-1}$. The maps $d_{k}:G^{n}\to G^{n-1}$ induces a pullback of $\omega$, \ie $d_{k}^{*}\omega:=\omega\circ d_{k}$ on $G^{n}$. We denote $\delta=d^{*}$ (it follows that $\delta^{2}=0$), thus $C^{\bullet}(BG,A)$ together with $\delta$ becomes a cochain complex.

\begin{definition}
    A function $\omega:G^{n}\to A$ is said to be an $n$-cocycle if $\delta \omega=0$. We denote the space of all $n$-cocycles by $\mathrm{Z}^{n}(G;A)$. Besides, if an $n$-cocycle $\omega$ satisfies $\omega=\delta \eta$ for some $\eta\in C^{n-1}(G;A)$, it is called an $n$-coboundary. The space of all $n$-coboundary is denoted as $\mathrm{B}^{n}(G;A),n>1$. Besides, $\mathrm{B}^{1}(G;A)$ is defined to be 0.
\end{definition}

\begin{definition}
    The degree $n$ group cohomology of $G$ is defined to be
    \beq\label{eq:group_coho}
    \rH^{n}(G;\U)=\frac{\mathrm{Z}^{n}(G;A)}{\mathrm{B}^{n}(G;A)}
    \eeq
    In more details, $\rH^{n}(G;A)$ are defined to be  equivalence classes of $n$-cocycles under the equivalence relation $\omega\simeq \omega+\delta\eta$ where $\omega\in\mathrm{Z}^{n}(G;A)$ and $\delta\eta\in \mathrm{B}^{n}(G;A)$.
\end{definition}

\begin{example}
    Let us consider a function $\omega:G\to A$ or equivalently $\omega$ here is a 1-cochain. Now we compute $\delta\omega$
    \beq
    \delta\omega(g_{1},g_{2})=(d_{0}^{*}\omega-d_{1}^{*}\omega+d^{*}_{2}\omega)(g_{1},g_{2})=\omega(g_{1})+\omega(g_{2})-\omega(g_{1}g_{2})
    \eeq
    where we have used Eq. \eqref{eq:face_maps}, \eg
    \beq
    d_{1}^{*}\omega(g_{1},g_{2})=\omega(d_{1}(g_{1},g_2))=\omega(g_{1}g_{2})
    \eeq
    Then $\omega$ is a 1-cocycle iff it is a homomorphism, \ie $\omega(g_1 g_2)=\omega(g_{1})+\omega(g_{2})$. We conclude
    \beq
    \rH^{1}(G;A)=\rHom(G,A)
    \eeq
\end{example}

\begin{example}
    Now we consider a 2-cochain, again denoted by $\omega:G^{2}\to \A$. Then one calculates $\delta\omega$ as follows
    \beq
    \delta\omega(g_{1},g_{2},g_{3})=\omega(g_{2},g_{3})-\omega(g_{1}g_{2},g_{3})+\omega(g_{1},g_{2}g_{3})-\omega(g_{1},g_{2})
    \eeq
    If one writes the group action in $A$ as multiplication rather than addition, one immediately recognizes $\delta\omega=0$ is exactly the 2-cocycle condition Eq. \eqref{eq:2-cocycle_condition} in projective representations. One can shift $\omega$ by a 2-coboundary $\delta\eta$. As we computed in the last example, this corresponds to
    \beq
    \omega(g_{1},g_{2})\to \tilde{\omega}(g_{1},g_{2})=\omega(g_{1},g_{2})+\eta(g_{1})+\eta(g_{2})-\eta(g_{1}g_{2})
    \eeq
    In the context of projective representation, this amounts to redefining our representation matrices by a phase Eq. \eqref{eq:shift_2-coboundary}.
\end{example}
Group cohomology of higher degrees are used to classify 't Hooft anomalies in physics. We will explain this in some more details in Sec. \ref{sec:anomaly_index}.

\begin{remark}
    The geometry behind Eqs.
    \eqref{eq:face_maps} and \eqref{eq:group_coho} is that we are doing simplicial cohomology on the space $BG$ (which is known as classifying space in mathematics), see, \eg, Ref. \cite{Weibel_1994_simplicial} for more details.
\end{remark}

\subsection{Differentiable group cohomology}

It is tempting to generalize the above definition of group cohomology of discrete to group cohomology of Lie group. Naively, we should require group cochains to be smooth, \ie $\omega:G^{n}\to A$ should be a smooth function on $G^{n}$ where $A$ is an abelian Lie group. We denote the cohomology of smooth cochains by $\rH^{*}_{s}(G;A)$. However, this definition does not work due to the van Est theorem \cite{Stasheff1978ContinuousCO}, which says that for connected compact Lie group $G$ we have
\beq
\rH^{n}_{s}(G;A)=0,n>0
\eeq
Moreover, if $G$ is compact but not connected, we have
\beq
\rH^{n}_{s}(G;A)\simeq \rH^{n}(\pi_{0}(G);A),n\geqslant 0
\eeq
This means that this cohomology group does not capture the smooth structure of $G$ at all.

Roughly speaking, there are three different ways to define useful cohomology theory for Lie groups \cite{Wagemann2011cohomology}. 

\begin{enumerate}

    \item The first is to use measurable cochains rather than smooth ones. The resulting cohomology is known as the Borel group cohomology in the physics literature\footnote{We remark that this cohomology has nothing to do with the so-called Borel equivariant cohomology, which is often referred to as Borel cohomology in the mathematical literature.}, see Refs.~\cite{Chen2010a,Ogata_2021}. Following the convention in the physics literature, we denote this cohomology by $\rH^{*}_{\text{B}}(G;A)$.
    
    \item The second way is to replace smooth cochains by \textit{locally} smooth cochains. By locally smooth we mean that $\omega:G^{n}\to A$ is smooth in a neighborhood of $(1,1,...,1)$. We denote this cohomology theory by $\rH^{*}_{loc,s}(G;A)$.
    
    \item The last way is to use simplicial method. Intuitively, one fixes a set of charts $\{U_{i}\}_{i\in J}$ on the Lie group $G$, and cochains are defined as smooth functions on each of trivialization chart $U_{i}$. On intersections, one needs \textit{transition functions} to patch them together. The resulting cohomology theory is denoted as $\rH^{*}_{\diff}(G;A)$, which is exactly the same as $\rH^{*}_{\text{simp},s}(G;A)$ in Ref. \cite{Wagemann2011cohomology}.
    
\end{enumerate}

Let us look at an example of differential group cohomology.

\begin{example}
    Suppose $\rho$ is a smooth projective representation of $G$ on a finite dimensional Hilbert space $V$, \ie $\rho:G\to \mathrm{PU}(V):=\mathrm{U}(V)/\U$. On each trivialization chart $U_{i}$ of $G$, one can lift $\rho$ to be $\rho_{i}:U_{i}\to \mathrm{U}(V)$, which may not be a representation of $G$ in general. On each chart, 
    \beq
    \rho_{i}(g)\rho_{i}(h)=\rho_{i}(gh)\omega_{i}(g,h)
    \eeq
    where $\omega_{i}(g,h)\in\U$ and $g,h,gh\in U_{i}$. Of course, $\omega_{i}$ is further constrained by the usual 2-cocycle condition Eq. \eqref{eq:2-cocycle_condition}. On the intersection $U_{i}\cap U_{j}$, one notes that two different liftings at most differ by a phase, \ie
    \beq
    \rho_{i}(g)=\rho_{j}(g)\eta_{ij}(g)
    \eeq
    where $\eta_{ij}(g)\in\U$ and $g\in U_{i}\cap U_{j}$. Thus we have
\beq
\omega_{i}(g,h)=\delta \eta_{ij}(g,h)\omega_{j}(g,h)
\eeq
We say $[\omega,\eta]$ defines a differentiable group cohomology class in $\rH^{2}_{\diff}(G;\U)$. Note that here $\eta_{ij}$ plays the role of transition functions in usual bundle theory. 
\end{example}

A priori, these cohomology groups may not be the same. However, it turns out that they are isomorphic for finite dimensional Lie groups with suitable coefficients.

\begin{theorem}[Corollary IV.9 and remark IV.13 of Ref. \cite{Wagemann2011cohomology}]
    For finite dimensional Lie group $G$ which acts smoothly on $\U$, we have
    \beq
    \rH^{*}_{loc,s}(G;\U)\simeq \rH^{*}_{\mathrm{B}}(G;\U\simeq \rH^{*}_{\diff}(G;\U)
    \eeq
\end{theorem}
Despite of being isomorphic, $\rH^{*}_{\diff}(G;\U)$ is more convenient for constructing the anomaly index (see Ref.~\cite{kapustin2024anomalous}).

There are some useful properties of $\rH_{\diff}$. We list them here and the readers are referred to Refs. \cite{brylinski2000differentiable,Wagemann2011cohomology} for proofs.

\begin{proposition}
    If either $G$ or $A$ is discrete, then
    \beq
    \rH^{*}_{\diff}(G;A)\simeq \rH^{*}(BG;A)
    \eeq
    where the right hand side is the singular cohomology of the classifying space $BG$.
\end{proposition}

\begin{proposition}
    If $G$ is compact
    \beq
    \rH^{n}_{\diff}(G;\R)=0,n>0
    \eeq
\end{proposition}
By Bockstein homomorphism,
\begin{corollary}
    For compact Lie group $G$, we have
    \beq
    \rH^{n}_{\diff}(G;\U)\simeq \rH^{n+1}(BG;\z),\,n\geqslant 1
    \eeq
\end{corollary}
\begin{proposition}[Kunneth formula, Appendix B of Ref.~\cite{Cheng2015}]
    Let $G$ and $H$ be finite-dimensional Lie groups (including discrete groups) then
    \beq
    \rH^{n}_{\diff}(G\times H;\U)\simeq \bigoplus_{p+q=n}\rH^{p}_{\diff}(G;\rH_{\diff}^{q}(H;\U))
    \eeq
\end{proposition}

\section{Construction of anomaly index}\label{sec:anomaly_index}

\subsection{Decomposition}

We now present the construction of the anomaly index in Ref. \cite{kapustin2024anomalous}. The spirit of this construction is similar to Ref. \cite{Else_2014}, \ie we need to cut our chain at the origin and decompose the symmetry action. This cut induces a factorization on $\A^{ql}$, such that
\beq\label{eq:algebra_decomp}
\A^{ql}\simeq \A_{<0}\otimes\A_{\geqslant 0}
\eeq

To study how symmetry action decomposes under this cut, it is useful to define some more subgroups of $\G^{lp}$, the group of all symmetry actions.

\begin{definition}\label{def:subgroups}
    We define the following useful subgroups of $\G^{lp}$,
    \begin{enumerate}
    
    \item $\G^{lp}_{0}$: The subgroup generated by $\Ad_{U}$, where $U$ is a quasi-local unitary operators.
    
    \item $\G^{lp}_{<0}$: The subgroup that acts trivially on $\A_{\geqslant 0}$ and maps $\A_{<0}$ to itself.
    
    \item $\G^{lp}_{\geqslant 0}$: The subgroup acts trivially on $\A_{<0}$ and maps $\A_{\geqslant 0}$ to itself.
    
    \item  $\G^{lp}_{+}$ (resp. $\G^{lp}_{-}$) is the subgroup generated by $\G^{lp}_{\geqslant 0}\G^{lp}_{0}$ (resp. $\G^{lp}_{<0}\G^{lp}_{0}$).
    
\end{enumerate}
\end{definition}

Note that $\G^{lp}_{\geqslant 0}$ and $\G^{lp}_{<0}$ act on $\G^{lp}_{0}$ by conjugation. More explicitly, for $\alpha\in\G_{\geqslant 0}^{lp}$ or $\alpha\in\G_{< 0}^{lp}$ and $U$ a quasi-local unitary,
\beq\label{eq:group_action}
\alpha\triangleright\Ad_{U}:=\alpha\Ad_{U}\alpha^{-1}=\Ad_{\alpha(U)}.
\eeq

\begin{figure}[h!]
    \centering
    \includegraphics[width=0.5\textwidth]{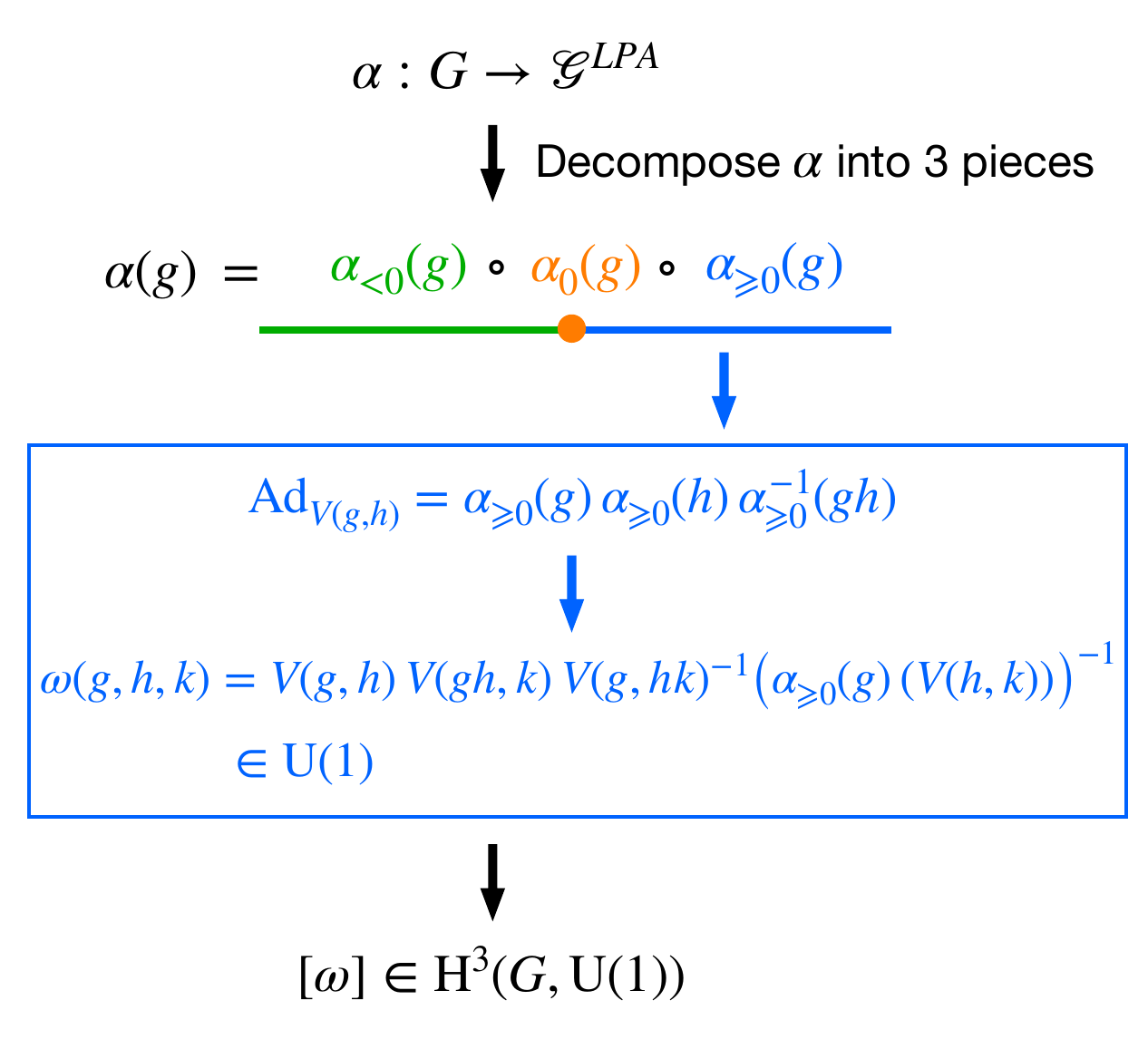}
    \caption{The anomaly index $\omega\in\rH^{3}(G;\U)$ from decomposing the symmetry action.}
    \label{fig:anomaly_index_LPA}
\end{figure}

To accomplish the decomposition in Fig.~\ref{fig:anomaly_index_LPA}, we need a few lemmas in Ref. \cite{kapustin2024anomalous}. Readers are referred to Ref. \cite{kapustin2024anomalous} for rigorous proofs of these lemmas. Here we only state the lemmas and comment on why some of these lemmas are intuitively true.

\begin{lemma}[Lemma 2.1 of Ref. \cite{kapustin2024anomalous}]\label{lemma:0_index_decomp}
    Let $\alpha\in\G^{lp}$ be an LPA, then the following statements are equivalent.
    \begin{enumerate}
        \item The GNVW index (see Appendix \ref{subsec:QCA_LPA} for a brief review) $\ind(\alpha)=0$.
        \item The element $\alpha$ admits following decomposition
        \beq\label{eq:decomp}
        \alpha=\alpha_{<0}\alpha_{0}\alpha_{\geqslant 0}.
        \eeq
    \end{enumerate}
\end{lemma}

Given the decomposition Eq. \eqref{eq:algebra_decomp}, one may want to restrict the symmetry action $\alpha$ to each half chain, which gives the $\alpha_{\geqslant 0}$ and $\alpha_{<0}$ parts. However, generically $\alpha$ can expand the support of an operator. If an operator is supported on one of the two half chains, under $\alpha$ it will generically acquire some support on the other half chain. But $\alpha_{\geqslant 0}$ and $\alpha_{<0}$ cannot achieve this, so the $\alpha_{0}$ part is also expected. 

\begin{lemma}[Lemma 2.2 of Ref. \cite{kapustin2024anomalous}]\label{lemma:non_unique_decomp}
    Suppose $\alpha\in\G^{lp}$ has a vanishing GNVW index and it admits two different decompositions
    \beq
    \alpha=\alpha_{<0}\alpha_{0}\alpha_{\geqslant 0}=\tilde{\alpha}_{<0}\tilde{\alpha}_{0}\tilde{\alpha}_{\geqslant 0}
    \eeq
    then $\alpha_{<0}\tilde{\alpha}_{<0}^{-1}\in\G^{lp}_{0}$ and $\alpha_{\geqslant 0}\tilde{\alpha}_{\geqslant 0}^{-1}\in \G^{lp}_{0}$.
\end{lemma}

The following lemma is also often useful when dealing with $\G_{+}^{lp}$ and $\G_{-}^{lp}$.

\begin{lemma}[Corollary 2.1 of Ref.~\cite{kapustin2024anomalous}]
    The intersection $\G^{lp}_{+}\cap \G^{lp}_{-}=\G^{lp}_{0}$
\end{lemma}

Note that any $\alpha_{+}\in\G^{lp}_{+}$ can be uniquely written as 
\beq
\alpha_{+}=\alpha_{0}\alpha_{\geqslant 0},\ 
\alpha_{0}\in\G^{lp}_{0},\ 
\alpha_{\geqslant 0}\in \G^{lp}_{\geqslant 0}
\eeq
Note $\Ad_{U}\alpha_{\geqslant 0}=\alpha_{\geqslant 0}\Ad_{\alpha_{\geqslant0}^{-1}(U)}$. Similarly, one can always write $\beta_- \in \G^{lp}_{-}$ as $\beta_{0}\beta_{<0},\beta_{0}\in\G^{lp}_{0},\beta_{<0}\in\G^{lp}_{<0}$. So one conclude the third lemma by comparing these two.

\begin{lemma}[Remark 2.3 of  Ref. \cite{kapustin2024anomalous}]\label{lemma:normal_subgroup}
    The subgroups $\G^{lp}_{+}$ and $\G^{lp}_{-}$ are both normal subgroups of $\G^{lp}$.
\end{lemma}

\subsection{Anomaly index from decomposition}

With the above background, now we can perform the decomposition (see Fig.~\ref{fig:anomaly_index_LPA}) and derive the anomaly index.

Let $\alpha:G\to \G^{lp}$ be a symmetry action implemented by LPA{\footnote{When $G$ is a Lie group, it is tempting to demand $\alpha$ to be a smooth map. However, there is no obvious smooth structure on $\G^{lp}$. Instead, it is noted in Ref.~\cite{kapustin2024anomalous} one should work with a smooth subgroup $\G^{al}$ of $\G^{lp}$}}, \ie
\beq
\alpha(g)\alpha(h)=\alpha(gh),\forall\,g,h\in G
\eeq
We first assume that $\ind(\alpha(g))=0,\forall\,g\in G$. Thus, for each $g$, it admits the following decomposition as in Eq. \eqref{eq:decomp}
\beq
\alpha(g)=\alpha(g)_{<0}\alpha(g)_{0}\alpha(g)_{\geqslant 0}
\eeq
where $\alpha(g)_{<0}\in\G^{lp}_{<0},\alpha(g)_{0}\in\G^{lp}_{0}$ and $\alpha(g)_{\geqslant 0}\in\G^{lp}_{\geqslant 0}$. On the other hand, note that $\alpha(g)_{\geqslant 0}$ trivially commutes with $\alpha(h)_{<0}$ because they act on disjoint domains,
\beq
\alpha(g)_{<0}\alpha(g)_{0}\alpha(g)_{\geqslant 0}\alpha(h)_{<0}\alpha(h)_{0}\alpha(h)_{\geqslant 0}=\alpha(g)_{<0}\alpha(h)_{<0}\beta_{0}(g,h)\alpha(g)_{\geqslant 0}\alpha(h)_{\geqslant 0}
\eeq
where $\beta(g,h)_0:=(\alpha(h)_{<0}^{-1}\triangleright \alpha(g)_0)(\alpha(g)_{\geqslant 0}\triangleright\alpha(h)_{0})\in\G^{lp}_{0}$ (see Eq.~\eqref{eq:group_action} for definition of group action $\triangleright$). According to lemma \ref{lemma:non_unique_decomp}, there exists a map $V:G\times G\to\cU^{ql}$ (which may {\it not} be a homomorphism) such that
\beq\label{eq:non-homomorphism}
\alpha(g)_{\geqslant 0}\alpha(h)_{\geqslant 0}=\Ad_{V(g,h)}\alpha(gh)_{\geqslant 0}
\eeq
Furthermore, this $V$ is constrained by the associativity of $\alpha_{\geqslant 0}$ as follows
\beq
\alpha(g)_{\geqslant 0}(\alpha(h)_{\geqslant 0}\alpha(k)_{\geqslant 0})=(\alpha(g)_{\geqslant 0}\alpha(h)_{\geqslant 0})\alpha(k)_{\geqslant 0}
\eeq
The above equation further simplifies to
\beq
\Ad_{\omega(g,h,k)}=1
\eeq
where
\beq\label{eq:anomaly_index}
\omega(g,h,k)=V(g,h)V(gh,k)V(g,hk)^{-1}(\alpha(g)_{\geqslant 0}(V(h,k)))^{-1}
\eeq
Since $\omega(g,h,k)$ commutes with all quasi-local operators, it must be a pure phase, \ie $\omega(g,h,k)\in \U$. It can be checked that $\omega$ satisfies the 3-cocycle condition (a.k.a pentagon identity) and shifting $V(g,h)$ by a phase $\rho(g,h)\in\U$ will change $\omega$ by a 3-coboundary (see appendix B of Ref. \cite{kapustin2024anomalous} for more details). Thus $\omega:G^{3}\to\U$ is a well-defined degree 3 group cohomology class, \ie $[\omega]\in\rH^{3}(G;\U)$ (see Sec. \ref{sec:group_cohomology} for a review of group cohomology).

If GNVW index of $\alpha(g)$ is nontrivial, $\alpha$ induces a map $\tau:G\stackrel{\alpha}{\to} \G^{lp}\to\G^{T}$ where $\G^{T}$ is the group of generalized translations. In this case, we stack our system with another (decoupled) copy, on which the $G$-symmetry acts as $\tau(g)^{-1}$. On the composite system, this $G$-symmetry acts as
\beq
\beta(g):=\alpha(g)\otimes\tau(g)^{-1}
\eeq
This is again a group action thanks to the fact that $\G^{T}$ is abelian. Besides, $\ind(\beta(g))=0$, which allows us to define
\beq
\omega_{\alpha}:=\omega_{\beta}
\eeq
This construction \textit{internalizes} the translation symmetry. 

We call $\omega$ the anomaly index of the symmetry action $\alpha$. If $\omega\neq 1\in\rH^3(G; U(1))$, we say that the symmetry action $\alpha$ is anomalous, otherwise $\alpha$ is said to be anomaly-free or non-anomalous.

Although the cohomology class of $\omega$ does not depend on the choice of $\alpha_{\geqslant 0}$ and $V$ \cite{kapustin2024anomalous}, in the previous work such as Refs. \cite{Else_2014,kapustin2024anomalous}, it is not clear if the cohomology class of $\omega$ depends on where we cut the chain. For example, one can also cut the chain at $1$ rather than $0$, do they yield the same anomaly index? Below we show that the anomaly index is independent of the choice of the cut.

\begin{proposition}\label{prop:site_independence}
        Assume the dimension of the local Hilbert space $\cH_{k}$ is bounded by a constant $K$, \ie $\dim\cH_{k}<K$ for any $k\in\z$. Then the anomaly index $\omega$ constructed above does not depend on the choice of the cut.
\end{proposition}

\begin{proof}[Proof of Proposition \ref{prop:site_independence}]
    If all local Hilbert spaces are of the same dimension, we have an operation which is translation by $+1$, denoted by $\tau$. Otherwise, we add some decoupled degrees of freedom at each site, on which our symmetry acts trivially. This step is to ensure that all local Hilbert spaces have the same dimension. After this prescription, we again have a well defined translation operation.

    Define $\tilde{\alpha}(g):=\tau\alpha(g)\tau^{-1}$, then the decomposition of $\tilde{\alpha}$ at 0 gives a decomposition of $\alpha$ at 1. We fix a decomposition of $\tilde{\alpha}$ at $0$ as
    \beq
    \tilde{\alpha}=\tilde{\alpha}_{<0}\tilde{\alpha}_{0}\tilde{\alpha}_{\geqslant 0}
    \eeq
    Similarly, we fix a decomposition of $\alpha$ at 0 as
    \beq
    \alpha=\alpha_{<0}\alpha_{0}\alpha_{\geqslant 0}
    \eeq
    We only have to show that they give rise to the same anomaly index. To this end, note that by lemma \ref{lemma:normal_subgroup}, there exist $\beta_{0}'\in\G^{lp}_{0}$ and $\beta_{\geqslant 0}\in\G^{lp}_{\geqslant 0}$ such that 
    \beq
    \tau\alpha_{\geqslant 0}\tau^{-1}=\beta'_{0}\beta_{\geqslant 0}
    \eeq
    Similarly,
    \beq
    \begin{split}
        \tau\alpha_{<0}\tau^{-1}&=\beta_{<0}\beta''_{0}\\
        \tau\alpha_{0}\tau^{-1}&=\beta_{0}'''
    \end{split}
    \eeq
    So
    \beq
    \tilde{\alpha}_{<0}\tilde{\alpha}_{0}\tilde{\alpha}_{\geqslant 0}=\tilde{\alpha}=\beta_{<0}\beta_{0}\beta_{\geqslant 0}
    \eeq
    where $\beta_{0}:=\beta_{0}''\beta_{0}'''\beta_{0}'$. Thus, we have two different decompositions for $\tilde{\alpha}$ at 0. Combining this result, the lemma 2.2 of Ref. \cite{kapustin2024anomalous} and proposition 3.1 of Ref. \cite{kapustin2024anomalous}, we conclude that the two decompositions give the same anomaly index.
\end{proof}

\begin{remark}
    Actually one can replace the lattice translation in above proof by any locality-preserving automorphism. Similar argument shows that the resulting anomaly index is the same as the original one.
\end{remark}

\subsection{Consequences of an anomalous symmetry} \label{subapp: consequence of anomaly}

Now we discuss some consequences of an anomalous symmetry. We first prove a generalized version of the main theorem of Ref.~\cite{kapustin2024anomalous} (Remark 4.1 therein).
\begin{lemma}\label{lemma:typeI_KS}
    Let $\psi$ be a type-I factor state with the split property and $\alpha:G\to\G^{lp}$ be a symmetry action. If $\psi\circ\alpha(g)=\psi$ for all $g\in G$, then the anomaly index Eq.~\eqref{eq:anomaly_index} of $\alpha$ must be trivial.
\end{lemma}
\begin{proof}
    Here we will prove this result for the case where $\tilde\alpha(g)$ have a vanishing GNVW index for all $g\in G$. To extend our proof to the case where some $\tilde\alpha(g)$ may have a non-vanishing GNVW index, one just needs to 1) add another copy of the system, 2) demand that the symmetry acts as the translation in the opposite direction in the added copy, 3) let this added copy to be in a translation invariant pure product state, and 4) repeat the proof for the case where all symmetry actions have a vanishing GNVW index.

    Recall that for a general state, we say $\psi$ splits if $\psi\sim\psi_{<0}\otimes\psi_{\geqslant0}$, where $\sim$ means quasi-equivalence. We then have
    \beq
    \psi_{<0}\otimes \psi_{\geqslant0}\sim\psi=\psi\circ\alpha(g)\sim (\psi_{<0}\circ \alpha(g)_{<0})\otimes(\psi_{\geqslant0}\circ\alpha(g)_{\geqslant0})
    \eeq
    One then deduces from Lemma \ref{lemma:separate_equivalence} that $\psi_{\geqslant 0}\sim\psi_{\geqslant 0}\circ\alpha(g)_{\geqslant 0}$. So $\psi\sim\psi_{<0}\otimes\psi_{\geqslant 0}\sim\psi_{<0}\otimes\psi_{\geqslant 0}\circ\alpha(g)_{\geqslant0}\sim\psi\circ\alpha(g)_{\geqslant 0}$. 
    
    Let $(\pi_{\psi},\cH_{\psi},|\psi\ra)$ be the GNS triple of $\psi$. Thus the GNS representation of $\psi\circ\alpha(g)_{\geqslant0}$ is given by $\pi_{\psi}\circ\alpha(g)_{\geqslant 0}$ acting on $\cH_{\psi}$.
    By Proposition \ref{prop:quasi-eq}, there exists a $*$-isomorphism $F_{g}:\cM_{\psi}\to\cM_{\psi\circ\alpha(g)_{\geqslant0}}$ such that
    \beq
    F_{g}(\pi_{\psi}(a))=\pi_{\psi}(\alpha(g)_{\geqslant 0}(a))
    \eeq
    Moreover, $\pi_{\psi}(\A^{ql})'=\pi_{\psi}(\alpha(g)_{\geqslant0}\A^{ql})'$ since $\alpha(g)_{\geqslant0}$ is a $*$-automorphism. This in turn means $\cM_{\psi}=\cM_{\psi\circ\alpha(g)_{\geqslant0}}$ as an abstract von Neumann algebra. Thus, $F_{g}$ can be thought of as a $*$-automorphism on $\cM_{\psi}$. By assumption, $\cM_{\psi}$ is a type-I factor acting on $\cH_{\psi}$, which means $\cM_{\psi}=\B(\cH)\otimes 1_{\cK}$ with $\cK$ another Hilbert space such that $\cH_{\psi}=\cH\otimes \cK$.

    Recall that all $*$-automorphisms of $\B(\cH)$ are inner (Lemma \ref{lemma:automoprhism_typeI}). This ensures that there exists a unitary operator $U_{g}$ on $\cH$ such that
    \beq
    \pi_{\psi}(\alpha(g)_{\geqslant0}(A))=U_{g}\pi_{\psi}(A)U_{g}^{\dagger}
    \eeq
    where $U_{g}$ should be understood as $U_{g}\otimes 1_{\cK}$. We also note that $U_{g}$ can fail to be a homomorphism in general. On the other hand, by Eq.~\eqref{eq:non-homomorphism},
    \beq
    \pi_{\psi}(V(g,h)AV(g,h)^{-1})=\pi_{\psi}(\alpha(g)_{\geqslant0}\alpha(h)_{\geqslant0}\alpha(gh)_{\geqslant0}^{-1}(A))=U_{g}U_{h}U_{gh}^{\dagger}\pi_{\psi}(A)U_{gh}U_{h}^{\dagger}U_{g}^{\dagger}.
    \eeq
    This means $U_{gh}U_{h}^{\dagger}U_{g}^{\dagger}\pi_{\psi}(V(g,h))$ commutes with $\pi_{\psi}(\A^{ql})$ for each $A\in\A^{ql}$, \ie
    \beq
    U_{gh}U_{h}^{\dagger}U_{g}^{\dagger}\pi_{\psi}(V(g,h))\in\pi_{\psi}(\A^{ql})'=\cM_{\psi}'
    \eeq
    On the other hand, $U_{gh}U_{h}^{\dagger}U_{g}^{\dagger}\pi_{\psi}(V(g,h))$ lies in $\cM_{\psi}=\B(\cH)\otimes1_{\cK}$, since $U_{gh}U_h^\dag U_g^\dag\in\cM_\psi$ and $\pi_\psi(V(g, h))\in\pi_\psi(\A^{ql})\subseteq\pi_\psi(\A^{ql})''=\cM_\psi$. Therefore, we have
    \beq
    \eta(g,h):=U_{gh}U_{h}^{\dagger}U_{g}^{\dagger}\pi_{\psi}(V(g,h))\in \U
    \eeq
    where we have used $\cM_{\psi}\cap\cM_{\psi}'=\bbC$. By Eq.~\eqref{eq:anomaly_index}, this implies
    \beq
    \omega(g,h,k)=(\delta\eta)(g,h,k)
    \eeq
    Therefore, $\omega$ is trivial in $\rH^{3}(G;\U)$.
\end{proof}

\begin{remark}
    In Ref. \cite{kapustin2024anomalous}, a similar result is obtained, but with ``type-I factor state" in Lemma \ref{lemma:typeI_KS} replaced by ``pure state" (see Remark 4.1 therein and Lemma \ref{lemma:Kapustin_Sopenko} in the present paper). Because pure states are all type-I factor states, our version is more general than the version in Ref. \cite{kapustin2024anomalous}. Besides being more general, another advantage of our version is that Lemma \ref{lemma:generalized_KS} immediately follows, since a 1D state that is clustering and satisfies the entanglement area law must be a type-I factor state with the split property (due to Proposition \ref{prop:factor_entropy} and Lemma \ref{lemma:split again}). Lemma \ref{lemma:generalized_KS} is easier to use than the original statement in Ref. \cite{kapustin2024anomalous}, since purity of a quantum state is typically harder to check compared with the clustering property and area law.
    The latter two properties are often more accessible from a practical point of view.
\end{remark}

Another consequence of an anomalous symmetry is related to edgeability.

\begin{corollary}[Edgeablity \cite{Senthil2013edge,Han2017edge}]
    If $\omega\not=1\in \rH^{3}(G;\U)$, then this spin chain is not edgeable, \ie one cannot put this chain on a manifold with boundaries while preserving the $G$-symmetry.
\end{corollary}
\begin{proof}
    We imagine the chain with boundary is a half chain obtained by cutting an infinite chain at some point (say, 0). Then the symmetry acts on the half chain by
    \beq
    \alpha(g)_{\geqslant 0}\alpha(h)_{\geqslant 0}=\Ad_{V(g,h)}\alpha(gh)_{\geqslant 0}
    \eeq
    To preserve the $G$ symmetry on this half chain, one needs $\alpha_{\geqslant 0}$ to be a group homomorphism, \ie $\Ad_{V(g,h)}$ is trivial. Hence $V(g,h)$ commutes with all local operators, and it must be a phase. Thus, by definition
    \beq
    \omega(g,h,k)=(\delta V)(g,h,k)
    \eeq
    This means the anomaly index must vanish. In other words, if $\omega\neq 1$, the chain cannot have a symmetric boundary condition.
    
    However, one can still refine $V(g,h)\to V(g,h)\eta(g,h)$ where $\eta\in\rH^{2}(G;\U)$, which leaves the relation $\omega=\delta V$ invariant. That is, the symmetric boundary conditions is a module over $\rH^{2}(G;\U)$. One can understand this shift as stacking a 1d SPT phase to our spin chain.
\end{proof}
\begin{remark}
    It is recently shown in Refs. \cite{Seifnashri2025anomaly,Bols2025sym,Kapustin2025higher} that the converse is also true, \ie if $\omega=1\in\rH^{3}(G;\U)$, then one can chooses $V$ properly to obtain a symmetry action (approximately) localized on right-half chain.
\end{remark}

\section{Proof of Theorem \ref{thm:main}} \label{sec: theorem 1}

In this appendix, we give a proof of our Theorem \ref{thm:main} in the main text. This proof is slightly different from the one in the main text.

\begin{theorem}[Theorem \ref{thm:main} in the main text]
    Let the $G$-symmetry act on a quantum spin chain via a locality-preserving automorphisms (see Eq. \eqref{eq:LPA}). Suppose this symmetry action has a non-vanishing anomaly index and $H$ is an admissible Hamiltonian (see Eq. \eqref{eq:admissible_H}) that commutes with the $G$-action. Then there cannot be a locally-unique gapped $G$-symmetric ground state of $H$.
\end{theorem}

To prove this theorem, we use the following established lemmas.

\begin{lemma}[Theorem 1.5 of Ref. \cite{matsui2011boundedness}]\label{lemma:split again}
If the entanglement entropy of a factor state $\psi$ satisfies the area law, then it splits (see definition \ref{definition:split_property}).
\end{lemma}

From Theorem 2.8 of Ref. \cite{Hastings2005}, we know that a locally unique gapped ground state of an adimissible Hamiltonian is a factor state. Combining this result, Theorem \ref{thm:area_law} and Lemma \ref{lemma:split again}, one deduces
\begin{corollary}\label{corollary:GS_splits}
If $\psi$ is a locally unique gapped ground state of an admissible Hamiltonian $H$, then it splits.
\end{corollary}

The final lemma is

\begin{lemma}[Theorem 2 and remark 4.1 of Ref. \cite{kapustin2024anomalous}]\label{lemma:Kapustin_Sopenko}
Given a symmetry action $\alpha:G\to \G^{lp}$ (or $\G^{al}$ for continuous $G$) on a quantum spin chain, if there exists a pure state $\psi$ which splits and is $G$-symmetric, then the associated anomaly index $\omega=1$.
\end{lemma}

\begin{proof}[Proof of Theorem \ref{thm:main}]
    By Corollary \ref{corollary:GS_splits} and Lemma \ref{lemma:Kapustin_Sopenko}, we only have to show that the locally unique gapped ground state of an admissible Hamiltonian is pure, which is exactly the statement of Theorem \ref{Thm:pure_gs}.
\end{proof}

\section{Connection between infinite systems and sequences of finite systems} \label{sec: thermodynamic limit}

Many of our theorems in the previous appendices are proved for infinite systems, but real physical systems are all of finite sizes. In this appendix, we bridge finite and infinite systems together and use our results to obtain some important implications on finite systems. 

For now we will mostly work with finite-size systems with open boundary conditions. See Appendix~\ref{sec:periodic} for a discussion on periodic boundary conditions.

\subsection{Thermodynamic limits of states, automorphisms and Hamiltonians}

We start with a general discussion of the thermodynamic limits of states, automorphisms and Hamiltonians.

Let $\Lambda\simeq\z$ be the infinite lattice and fix an increasing sequence of finite subsets,
\beq\label{eq:increasing}
\Gamma_{1}\subseteq\Gamma_{2}\dots \subseteq\Lambda,\quad |\Gamma_{L}|<\infty
\eeq
indexed by an integer $L$, which can be thought of as the size of a finite system. We assume each $\Gamma_{L}$ is connected and it eventually exhausts $\Lambda$, \ie
\beq\label{eq:exhausting}
\bigcup_{L=1}^{\infty}\Gamma_{L}=\Lambda
\eeq

Let $\psi_{L}:\A_{\Gamma_{L}}\to\bbC$ be a sequence of quantum states. we formally extend them to be a state of $\A^{ql}$ as follows. Let $\Omega$ be a fixed pure product state of $\A^{ql}$ (which models the environment), as shown in Lemma \ref{lemma:product_split}, $\Omega_{L}:=\Omega|_{\Gamma_{L}^{c}}$ is another pure product state on $\A_{\Gamma_{L}^{c}}$. We define $\tilde{\psi}_{L}:=\psi_{L}\otimes\Omega_{L}$, which is a state on $\A^{ql}$. Formalizing this idea, we have
\begin{definition}\label{eq:sequence_states}
     Fix a pure product state $\Omega:\A^{ql}\to \bbC$, a sequence of finite-size states $\{\tilde{\psi}_{L}\}_{L=1,2,\dots}$ means
    \beq
    \tilde{\psi}_{L}:=\psi_{L}\otimes\Omega|_{\Gamma_{L}^{c}}
    \eeq
    where $\psi_{L}$ is a state of $\A_{\Gamma_{L}}$.
\end{definition}

Having a sequence of states, a natural task is to define the notion of convergence. In the following, we will only use the so-called {\it weak-$*$ convergence} of sequence of states, which physically means that local observables converge in the thermodynamic limit.

\begin{definition}[weak-$*$ convergence of states]\label{def:weak-$*$}
    Let $\tilde\psi_{L}$ be a sequence of states of $\A^{ql}$ defined as above, we say it weak-$*$ converges to a state $\psi$ if for any $\epsilon>0$ and $A\in \A^{l}$, there exists a constant $L_{\epsilon,A}$ such that $\forall\,L>L_{\epsilon,A}$, we have
    \beq
    |\tilde\psi_{L}(A)-\psi(A)|<\epsilon ||A||
    \eeq
    In this case, we say that $\psi$ is the weak-$*$ limit of $\psi_{L}$ as $L\to\infty$ and write $\lim_{L\to\infty}\psi_{L}=\psi$.
\end{definition}
\begin{remark}
    It is easy to see the weak-$*$ limit is unique if it exists.
    An equivalent way to formulate weak-$*$ convergence is: For any fixed finite subset $\Gamma\subseteq\Lambda$, the reduced density matrix of $\{\psi_{L}\}$ on $\Gamma$ converges in the trace norm. For example, let us first take $\Gamma_L=[-L, L]$. Then the weak-$*$ convergence of a sequence of states $\psi_{L}$ on $\A_{\Gamma_L}$ means the following. For any $N\in\z^+$ and $\epsilon>0$, there exists an integer $L>N$, such that for any $L_{1,2}>L$ and $A\in\A_{[-N, N]}$,  we have $|\psi_{L_1}(A)-\psi_{L_2}(A)|<\epsilon||A||$. Note that in this formulation we do not need to introduce the product state of the environment, $\Omega|_{\Gamma^c_L}$.
\end{remark}

By slightly abusing notations, we will simply write $\psi_{L}$ rather than $\tilde{\psi}_{L}$, although one should keep in mind that $\tilde{\psi}_{L}$ depends on the choice of $\Omega$ (\ie the environment).

We emphasize that weak-$*$ limit is strictly weaker than norm convergence, \ie $\lim_{L\to\infty}\psi_{L}=\psi$ does not imply $\lim_{L\to\infty}||\psi_{L}-\psi||=0$, where the norm is taken on the whole $\A^{ql}$, although we do have $\lim_{L\to\infty}||(\psi_{L}-\psi)|_{\Gamma}||=0$ for any $|\Gamma|<\infty$.

Given a sequence of states, it may or may not admit a weak-$*$ limit. However, a deep theorem in functional analysis shows, for any sequence of states, it is always possible to pick up a weak-$*$ convergent subsequence.
\begin{theorem}[Banach-Alaoglu-Bourbaki theorem 3.16 of Ref.~\cite{brezis2010functional}]\label{thm:BAB}
    Any sequence of states $\{\psi_{L}\}_{L=1,2,\dots}$ has a weak-$*$ convergent subsequence.
\end{theorem}
The proof requires Tychonoff theorem in general topology, which is equivalent to the axiom of choice.

\begin{remark}
    Earlier, we have indicated that the GHZ state (Eq.~\eqref{eq:finite_GHZ}) becomes mixed in the thermodynamic limit (see Eq.~\eqref{eq:mixed_GHZ}). In terms of weak-$*$ convergence, we say that the weak-$*$ limit of GHZ state is a mixed state.
\end{remark}

It will also be useful to define the notion of convergence for Hamiltonians and automorphisms. Let $\alpha_{L}$ be a $*$-automorphism of $\A_{\Gamma_{L}}$, and we formally extend it to $\A^{ql}$ by viewing it as $\alpha_{L}\otimes\mathrm{id}_{\Gamma_{L}^{c}}$. Again, by abusing notations, we will write $\alpha_{L}$ rather than $\alpha_{L}\otimes \mathrm{id}_{\Gamma_{L}^{c}}$. Therefore we have a natural notion of finite-size sequence of automorphisms (and similarly, sequence of finite-size Hamiltonians).
\begin{definition}\label{def:sequence_aut}
    A sequence of finite-size $*$-automorphisms on $\A^{ql}$ means $\{\alpha_{L}\otimes\mathrm{id}_{\Gamma_{L}^{c}}\}_{L=1,2,\dots}$, where each $\alpha_{L}$ is a $*$-automorphism of $\A_{\Gamma_{L}}$. Similarly, a sequence of finite-size Hamiltonians $\{H_{L}\}_{L=1,2,\dots}$ means $\{H_{L}\otimes \mathrm{id}_{\Gamma^{c}_{L}}\}_{L=1,2,\dots}$, where each $H_{L}$ is a self-adjoint operator on $\Gamma_{L}$.
\end{definition}
However, merely a sequence does not buy us much. We now define the strong convergence of such finite-size sequences of $*$-automorphisms and finite-size Hamiltonians.
\begin{definition}\label{def:strong_convergence}
    If for any fixed $A\in\A^{l}$ and $\epsilon>0$, there exists $\alpha\in\Aut(\A^{ql})$ and $L_{\epsilon,A}>0$, such that for any $L>L_{\epsilon,A}$,
    \beq
    ||\alpha_{L}(A)-\alpha(A)||<\epsilon ||A||
    \eeq
    Then we say $\{\alpha_{L}\}$ strongly converges to $\alpha$ as $L\to\infty$ and write $\lim_{L\to\infty}\alpha_{L}=\alpha$.

    Similarly, if for any fixed $A\in\A^{\ell}$ and $\epsilon>0$, there exists a densely defined derivation $\delta_{H}$ and $L_{\epsilon,A}>0$, such that for any $L>L_{\epsilon,A}$,
    \beq
    ||\delta_{H}(A)-[H_{L},A]||<\epsilon||A||
    \eeq
    Then we say $\ad_{H_{L}}:=[H_{L},\cdot]$ strongly converge to $\delta_{H}$ and write $\lim_{L\to\infty}\ad_{H_{L}}=\delta_{H}$. We say $H_{L}$ is admissible if $H$ is admissible.
\end{definition}

The following lemma shows these notions of convergences fit nicely.
\begin{lemma}\label{lemma:composition}
    Consider sequences of finite-size states $\{\psi_{L}\}_{L=1,2,\dots}$, finite-size $*$-automorphisms $\{\alpha_{L}\}_{L=1,2,\dots}$ and finite-size Hamiltonians $\{H_{L}\}_{L=1,2,\dots}$. Assuming that 
    \beq
    \begin{split}
        \lim_{L\to\infty}\psi_{L}&=\psi\\
        \lim_{L\to\infty}\alpha_{L}&=\alpha\\
        \lim_{L\to\infty}[H_{L},\cdot]&=\delta_{H}
    \end{split}
    \eeq
    Then
    \begin{enumerate}
        \item $\lim_{L\to\infty}\psi_{L}(A)=\psi(A),\quad\forall\,A\in\A^{ql}$.
        \item $\lim_{L\to\infty}\psi_{L}\circ\alpha_{L}=\psi\circ\alpha$.
        \item $\lim_{L\to\infty}\psi_{L}(B\,\ad_{H_{L}}(A))=\psi(B\,\delta_{H}(A)),\quad\forall\,A\in\A^{\ell},\,B\in\A^{ql}$.
    \end{enumerate}
\end{lemma}
\begin{proof}
    For property 1, we choose a norm-convergent sequence of local operators $a_{n}\in A^{l}$ such that $\lim_{n\to\infty}a_{n}=A$, \ie for any $\epsilon>0$, there exists $n_{\epsilon}$ such that $||a_{n}-A||<\epsilon ||A||$ for all $n>n_{\epsilon}$. For a fixed $n$ and $\epsilon$, there exists  an $L_{\epsilon,n}$ such that $||\psi_{L}(a_{n})-\psi(a_{n})||<\epsilon ||a_{n}||$ for all $L>L_{\epsilon,n}$. Therefore,
    \beq
    \begin{split}
        |\psi_{L}(A)-\psi(A)|&\leqslant |\psi_{L}(A)-\psi_{L}(a_{n})|+|\psi_{L}(a_{n})-\psi(a_{n})|+|\psi(a_{n})-\psi(A)|\\
        &\leqslant ||\psi_{L}||\cdot||A-a_{n}||+|\psi_{L}(a_{n})-\psi(a_{n})|+||\psi||\cdot ||a_{n}-A||\\
        &\leqslant 2\epsilon ||A||+\epsilon ||a_{n}||\\
        &\leqslant2\epsilon||A||+\epsilon(1+\epsilon)||A||\\
        &\leqslant\epsilon(3+\epsilon)||A||,\quad\forall\,L>L_{n,\epsilon}
    \end{split}
    \eeq
    So property 1 is proved.
    
    To show property 2, we note that for any $A\in\A^{\ell}$,
    \beq
    \begin{split}
        |\psi_{L}(\alpha_{L}(A))-\psi(\alpha(A))|&\leqslant|\psi_{L}(\alpha_{L}(A))-\psi_{L}(\alpha(A))|+|\psi_{L}(\alpha(A))-\psi(\alpha(A))|\\
        &\leqslant||\alpha_{L}(A)-\alpha(A)||+|\psi_{L}(\alpha(A))-\psi(\alpha(A))|
    \end{split}
    \eeq
    The first term goes to 0 as $L\rightarrow\infty$, since $\lim_{L\to\infty}\alpha_{L}=\alpha$. And we use the property 1 so the second term also vanishes as $L\to\infty$. This proves property 2.

    To see property 3, note
    \beq
    \begin{split}
        |\psi_{L}(B[H_{L},A])-\psi(B\,\delta_{H}(A))|&\leqslant|\psi_{L}(B[H_{L},A])-\psi_{L}(B\,\delta_{H}(A))|+|\psi_{L}(B\,\delta_{H}(A))-\psi(B \,\delta_{H}(A))|\\
        &\leqslant ||B||\cdot||[H_{L},A]-\delta_{H}(A)||+|\psi_{L}(B\,\delta_{H}(A))-\psi(B \,\delta_{H}(A))|
    \end{split}
    \eeq
    As $L\rightarrow\infty$, the first term goes to 0 since $\lim_{L\to\infty}\ad_{H_{L}}=\delta_{H}$, and the second term also vanishes in that limit because of property 1. This proves property 3.

\end{proof}
\begin{remark}
    We believe these properties are well-known to experts. For instance, the second property of \ref{lemma:composition} can be found in the Lemma 5.6 of Ref.~\cite{Bachmann2012automorphic}. However, we include the complete proof here for the convenience of readers.
\end{remark}

After clarifying the meanings of convergence of states, automorphisms and Hamiltonians, in the rest of this appendix, we will derive the finite-size analogs of our three main results, the entanglement area law, the LSM theorem and symmetry-enforced long-range entanglement.

\subsection{LSM theorem for sequences of finite-size systems}

It turns out that the LSM theorem is the simplest to be extended to sequences of finite-size systems.

Let us consider a sequence of finite-size systems indexed by a positive integer $L$ with finite-size Hamiltonian $H_{L}$, such that $\lim_{L\to\infty}\ad_{H_{L}}=\delta_{H}$ for some admissible Hamiltonian $H$ defined on $\Lambda$, where the limit is in the sense of Definition~\ref{def:strong_convergence}.
If for each $L$, $H_{L}$ has a uniquely gapped ground state $|\psi_{L}\ra$ with a gap $\Delta_{L}>0$, then one can define a state $\psi$ on infinite chain
\beq \label{eq: state in thermodynamic limit}
\psi(A):=\lim_{L\to\infty}\la\psi_{L}|A|\psi_{L}\ra,\quad\forall\,A\in\A^{\ell}
\eeq
Of course, this limit may not exist. But Theorem \ref{thm:BAB} ensures that there always exists a convergent subsequence, \ie a sequence $L_{n}$ which increases to $\infty$ such that $\lim_{L_n\rightarrow\infty}\la\psi_{L_{n}}|A|\psi_{L_{n}}\ra$ exists for all local operators $A$ \cite{Tasaki2022topological}. Below we only focus on a weak-$*$ convergent subsequence.

Our goal here is to show the following (Theorem \ref{theorem: finite and infinite} in the main text).

\begin{theorem}\label{thm:limit_ad_gs}
    The state $\psi$ defined by Eq. \eqref{eq: state in thermodynamic limit} is a locally unique gapped ground state of $\delta_{H}$, if $|\psi_{L}\ra$ is the unique gapped ground state of $H_{L}$ with $\Delta_{L}>\Delta>0$ for all $L$, and $\{H_L\}$ strongly converges to $\delta_H$.
\end{theorem}

\begin{proof}
    Without loss of generality, we assume $A\in\A^{\ell}$ satisfies $\psi(A)=0$ and define $A_{L}:=A-\psi_{L}(A)$. Therefore $A_{L}|\psi_{L}\ra$ is orthogonal to $|\psi_{L}\ra$. Since $H_{L}$ has a unique gapped ground state, we then have
    \beq
    \la\psi_{L}|A_{L}^{\dagger}H_{L}A_{L}|\psi_{L}\ra\geqslant (E_{L,0}+\Delta_{L})\la\psi_{L}|A_{L}^{\dagger}A_{L}|\psi_{L}\ra
    \eeq
    where $E_{L,0}:=\la\psi_{L}|H_{L}|\psi_{L}\ra$ is the ground state energy of $H_{L}$. Therefore,
    \beq
    \psi_{L}(A_{L}^{\dagger}[H_{L},A_{L}])\geqslant \Delta_{L}\psi_{L}(A_{L}^{\dagger}A_{L})
    \eeq
    where we have defined $\psi_{L}(\cdot):=\la\psi_{L}|\cdot|\psi_{L}\ra$.
    By using $A_L=A-\psi_L(A)$, this is equivalent to 
    \beq
    \psi_{L}(A^{\dagger}[H_{L},A])\geqslant\Delta_{L}\psi_{L}(A^{\dagger}A)+\psi_{L}(A)^{*}\psi_{L}(\ad_{H_{L}}(A))-\Delta_{L}|\psi_{L}(A)|^{2}
    \eeq
    Taking the limit $L\rightarrow\infty$ on both sides, by Lemma \ref{lemma:composition},
    \beq
    \psi(A^{\dagger}\delta_{H}(A))\geqslant \Delta\psi(A^{\dagger}A),\quad\,\forall\,A\in\A^{\ell}
    \eeq
    where we have used $\lim_{L\to\infty}\psi_{L}(A)=\psi(A)=0$, $\lim_{L\to\infty}\psi_{L}(\ad_{H_{L}}(A))=\psi(\delta_{H}(A))$ and $\lim_{L\to\infty}\psi_{L}(A^{\dagger}\ad_{H_{L}}(A))=\psi(A^{\dagger}\delta_{H}(A))$.

    Therefore, according to Definition \ref{definition:locally-unique_gs}, $\psi$ is a locally unique gapped ground state of $\delta_H$.
     
\end{proof}
\begin{remark}
    The above result holds for any strongly-convergent sequence of finite-size Hamiltonians.
    The local Hamiltonian version is proved in Ref.~\cite{Tasaki2022topological} (see Theorem 2.6 therein), which is a special case of the above theorem. Also, the same proof shows the unique ground state of $H_{L}$ (without assuming the gap) becomes a ground state of $\delta_{H}$ in the sense of Definition~\ref{def: ground state}.
\end{remark}

Hence by Theorem \ref{thm:area_law}, Lemma \ref{lemma:split again} and Theorem \ref{Thm:pure_gs}, we obtain

\begin{corollary} \label{coroapp: pure and split}
    The state $\psi$ defined by Eq. \eqref{eq: state in thermodynamic limit} is pure and splits, if $\delta_H$ is an admissible Hamiltonian.
\end{corollary}

Let $G$ be a symmetry group and $\alpha_{L}:G\to \Aut(\A_{\Gamma_{L}})$ be a sequence of homomorphisms.
For a finite chain with size $L$, we write $\alpha_{L}(g)(A)=\rho_{L}(g)A\rho_{L}^{-1}(g)$,where $\rho_{L}$ is a (projective) representation of $G$. If $\lim_{L\to\infty}\alpha_{L}(g)=\alpha(g),\forall g\in G$ where $\alpha:G\to\G^{LPA}$ is a homomorphism, then we say $\alpha_{L}$ is anomalous iff $\alpha$ is anomalous.

\begin{theorem}\label{thm:finite_size_LSM}
    Let $\alpha_{L}:G\to\Aut(\A_{\Gamma_{L}})$ be a strongly convergent sequence of symmetry actions and $H_{L}$ be a strongly convergent sequence of Hamiltonians. If 
    \begin{enumerate}
        \item $\alpha:=\lim_{L\to\infty}\alpha_{L}$ is a symmetry action described by LPA.
        \item $\delta_{H}:=\lim_{L\to\infty}\ad_{H_{L}}$ is an admissible Hamiltonian.
        \item $H_{L}$ is invariant under $\alpha_{L}(g)$ for all $L$ and $g\in G$.
        \item $H_{L}$ has a unique gapped ground state with uniformly lower bounded energy gap $\Delta_{L}\geqslant \Delta>0$.
    \end{enumerate}
    Then $\alpha$ has trivial anomaly index.
\end{theorem}
\begin{proof}
    Let $\psi_{L}$ be the unique ground state of $H_{L}$. By Theorem \ref{thm:limit_ad_gs} and Corollary \ref{coroapp: pure and split}, its weak-$*$ accumulation point\footnote{It means a possible weak-$*$ limit of some subsequences of $\psi_{L}$ since $\psi_{L}$ may not weak-$*$ converge itself.} $\psi$ is a locally unique gapped ground state of $\delta_{H}$, hence it is pure and splits. Besides, by Lemma \ref{lemma:composition}, $\psi$ must be symmetric under $\alpha$. By Lemma \ref{lemma:Kapustin_Sopenko}, $\alpha$ is anomaly free.
\end{proof}
The proof of this theorem is simple, but it has the following important corollary, which is Corollary \ref{corollary: LSM finite} in the main text.
\begin{corollary}
    If a sequence of finite-size spin chains have admissible Hamiltonians $H_{L}$ which are symmetric under an anomalous symmetry, then $H_{L}$ cannot have a unique gapped ground state for sufficiently large $L$.
\end{corollary}
Note that the maximal size of such systems that allow a unique gapped ground state depends on the details of $H_{L}$ and $\alpha_{L}$, and is hence non-universal.

\subsection{Entanglment area law in finite-size systems with a spontaneously broken discrete symmetry} \label{subapp: SSB area law}

In Appendix \ref{sec:area_law}, we have proved the entanglement area law for locally unique gapped ground states of an infinite-size system with an admissible Hamiltonian. In this subsection, we will extend the area law to sequences of finite-size systems that may have multiple almost degenerate ground states due to a spontaneously broken discrete symmetry.

We first define the entanglement area law for a sequence of finite-size states below.
\begin{definition}\label{def:finite_area_law}
    Let $\{\psi_{L}\}_{L=1,2,\dots,}$ be a sequence of finite-size states and $\Gamma\subseteq\Lambda$ be a fixed finite, connected subset. We say that $\{\psi_{L}\}$ satisfies the area law of entanglement entropy if
    \beq\label{eq:finite_area_law}
    S_{L}(\Gamma):=-\tr(\rho_{L,\Gamma}\log(\rho_{L,\Gamma}))<\const
    \eeq
    where $\rho_{L,\Gamma}$ is the density matrix of the state $\psi_{L}|_{\Gamma}$ and the $\const$ does not depend on the choice of $\Gamma$ or $L$.
\end{definition}

Definition~\ref{def:finite_area_law} and Definition~\ref{def:area_law} can be related by the following result.

\begin{lemma}\label{lemma:finite_area_law}
    Let $\{\psi_{L}\}_{L=1,2,\dots}$ be a sequence of finite-size states satisfying the area law (\ie Definition~\ref{def:finite_area_law}), if $\psi_{L}$ weak-$*$ converges to $\psi$, then $\psi$ satisfies the entanglement area law (in the sense of Definition~\ref{def:area_law}).
\end{lemma}
\begin{proof}
    Let $\rho$ and $\psi$ be two states on $\B(\cH)$ with $\dim\cH=d_{\cH}$. Then the Fannes-Audenaert inequality (Theorem 5.39 of Ref.~\cite{Watrous_2018}) shows
    \beq
    |S(\rho)-S(\psi)|\leqslant T \log(d_{\cH}-1)-T\log T-(1-T)\log(1-T)
    \eeq
    where $T:=\frac{1}{2}||\psi-\rho||$. We apply this inequality to $\psi_{L}|_{\Gamma},\psi|_{\Gamma}$ and $\cH_{\Gamma}$. By the definition of weak-$*$ convergence, we obtain $||\psi_{L}|_{\Gamma}-\psi|_{\Gamma}||\to 0$ as $L\to\infty$. So we conclude that 
    \beq
    S(\psi,\Gamma)\leqslant \const
    \eeq
    where $\const$ is defined in Eq.~\eqref{eq:finite_area_law}.
\end{proof}

Next, we clarify the meaning of almost degenerate gapped ground states.
\begin{definition} \label{def: almost generate gs}
    Let $H_{L}$ be a self-adjoint operator with $\sigma(H_{L})\subseteq[0,\delta E_{L}]\cup [\delta E_{L}+\Delta_{L},\infty)$ on $\cH_{\Gamma_{L}}:=\bigotimes_{j\in\Gamma_{L}}\cH_{j}$ for some $\delta E_{L}\geqslant0,\Delta_{L}>0$. We say that a sequence of Hamiltonians $\{H_L\}$ has almost degenerate gapped ground states if $\lim_{L\to\infty}\delta E_{L}=0$ and $\lim_{L\to\infty}\Delta_{L}\geqslant \Delta>0$. Let $\Pi_{0,L}$ be the spectral projection of $H_{L}$ with energy $0\leqslant E\leqslant\delta E_{L}$, then the image $\cH_{0,L}:=\im(\Pi_{0,L})$ is called the ground state subspace of $H_{L}$.
\end{definition}

Almost degenerate gapped ground states are ubiquitous in systems with spontaneously broken discrete symmetries \cite{Tasaki2019SSB,Beekman2019SSB}. In this subsection, to simplify notations, we will write $|\psi\ra$ instead of $|\psi_{L}\ra$ and its weak-$*$ limit (if exists) will be denoted by $\psi^{\infty}$.

Another ingredient is the notion of superselection rule in finite-size systems, which often appears in systems with a spontaneously broken discrete symmetry.
\begin{definition} \label{def: superselection rule}
    Let $z\in\mathbb{N}$, and let $\{|\psi_{j}\ra\}$ be $z$ sequences of finite size states with $j=1,2,\dots,z$, with each of $|\psi_{j}\ra$ depending on $L$ (implicitly). We say these sequences satisfy superselection rule if for any $i\not=j$
    \beq\label{eq:sss}
    \lim_{L\to\infty}|\la\psi_{i}|A|\psi_{j}\ra|=0,\quad\forall A\in\A^{\ell}.
    \eeq
    Note that in the above equation $|\psi_{i}\ra$ means the state in the $i$-th sequence and it is defined in a system with $L$ sites. We say that a subspace $\cH_{0,L}$ satisfies the superselection rule if it has a basis satisfying superselection rule.
\end{definition}

The relation between this superselection rule and the superselection sectors studied in Appendix~\ref{subsec:Superselection_sectors} is made clear by Lemma \ref{lemma:weak*_lsc} and Corollary \ref{coro:SSR} below.
\begin{lemma}\label{lemma:weak*_lsc}
    The norm of states are weak-$*$ lower semi-continuous, \ie for any weak-$*$ convergent sequences $\psi\to\psi^{\infty}$ and $\phi\to\phi^{\infty}$, we have
    \beq
    ||\psi^{\infty}-\phi^{\infty}||\leqslant\lowlim_{L\to\infty}||\psi-\phi||:=\lim_{L\to\infty}\inf_{k\geqslant L}||\psi-\phi||
    \eeq
\end{lemma}
Loosely speaking, this lemma means if two finite-size states are close to each other, then so are their weak-$*$ limits\footnote{We emphasize that it can happen that two drastically different states (in finite-size) converge to the same state in the weak-$*$ topology (\eg the product state and the $W$-state).}.
This lemma is a special version of a theorem in convex analysis, which states that an $\R$-valued convex function is weak-$*$ lower semi-continuous iff it is lower semi-continuous in the norm topology. We refer to Ref.~\cite{normlsc} for the proof of this general version.

To see the superselection rules, we have
\begin{corollary}\label{coro:SSR}
    Let $\psi\to\psi^{\infty}$ and $\phi\to\phi^{\infty}$ be two weak-$*$ convergent sequences. If $\psi^{\infty}$ and $\phi^{\infty}$ are inequivalent pure states (\ie they are in different superselection sectors), then $\psi$ and $\phi$ satisfy superselection rule,
    \beq\label{eq:ssr}
    \lim_{L\to\infty}|\la\psi|A|\phi\ra|=0,\quad\forall\,A\in\A^{\ell}
    \eeq
\end{corollary}
\begin{proof}
    We first show Eq.~\eqref{eq:ssr} for the special case $A=1$. We will reduce the general case to $A=1$ later.
    
    By Proposition~\ref{prop:orthogonal_sectors}, we have $||\psi^{\infty}-\phi^{\infty}||=2$. Therefore,
    \beq
    2\geqslant\lowlim_{L\to\infty}||\psi-\phi||\geqslant 2\Rightarrow \lim_{L\to\infty}||\psi-\phi||=2
    \eeq
    And the special case of Eq.~\eqref{eq:ssr} with $A=1$ follows by noting
    \beq
    \lim_{L\to\infty}|\la\psi|\phi\ra|=\lim_{L\to\infty}\sqrt{2-||\psi-\phi||}=0
    \eeq
    
    For the general case with a generic $A\in \A^l$, we note that Eq.~\eqref{eq:ssr} is linear in $A$, so we can expand $A$ into Pauli basis{\footnote{For a $d$-dimensional Hilbert space with orthonormal basis denoted by $\{|j\ra\}$, where $j=1, 2, \cdots d$, the Pauli basis of operators are generated by $\sum_{j=1}^{d-1}|j+1\ra\la j|+|1\ra\la d|$ and $\sum_j^de^{\frac{2\pi i j}{d}}|j\ra\la j|$.}}, where each basis operator is a local unitary operator. Therefore, we can restrict ourselves to unitary $A$. In this case, we note that $\phi\circ\Ad_{A}$ weak-$*$ converges to $\phi^{\infty}\circ\Ad_{A}\simeq\phi^{\infty}$. In particular, $\phi^{\infty}\circ\Ad_{A}$ falls into a different superselection sector than $\psi^{\infty}$.
    Then the same argument as above applied to $\psi$ and $\phi\circ\Ad_{A}$ proves Eq.~\eqref{eq:ssr}.
\end{proof}

As discussed in the main text, locally unique gapped ground states can be viewed as gapped ground states in a superselection sector. It is natural to expect that almost degenerate gapped ground states of finite-size systems with the superselection rule can converge to locally unique gapped ground states, as the size increases. This is made precise below.

\begin{corollary}\label{coro:almost_degeneracy}
    Let $H_{L}$ be a sequence of finite-size Hamiltonians with $\lim_{L\to\infty}H_{L}=H$. If $H_{L}$ has almost degenerate gapped ground states, and its ground state subspace $\cH_{0,L}$ has an $L$-independent dimension $z$ and satisfies the superselection rule, then there exists a subsequence $L_{n}$ such that $\psi_{i}:=\lim_{L_{n}\to\infty}\la\psi_{i}|\cdot|\psi_{i}\ra$ are locally unique gapped ground states of $H$, where $|\psi_i\ra$ is the orthornormal basis of $\cH_{0,L}$ that satisfies the superselection rule.
\end{corollary}
\begin{proof}
    Let $|\psi_{i}\ra$ with $i=1,2,\dots,z$ be a basis of the ground state subspace of $\cH_{0,L}$ that satisfies the superselection rule. We first show there exists a subsequence $L^{(n)}$ such that $\la\psi_{i}|\cdot|\psi_{i}\ra$ are weak-$*$ convergent for all $i$. This proceeds as follows. Focus on $|\psi_{1}\ra$ and by remark \ref{remark:BBA} there exists a subsequence $L^{(1)}$ such that $\la\psi_{1}|\cdot|\psi_{1}\ra$ is weak-$*$ convergent. Next, similar argument shows there exists a subsequence $L^{(2)}$ of $L^{(1)}$ such that $\la\psi_{2}|\cdot|\psi_{2}\ra$ is weak-$*$ convergent. Repeating this construction, one ends up with a subsequence $L^{(n)}$ such that $\psi_{i}^{\infty}:=\lim_{L^{(n)}\to\infty}\la\psi_{i}|\cdot|\psi_{i}\ra$ is weak-$*$ convergent for any $i=1,\dots,z$. Below we only work with this subsequence and rename $L^{(n)}$ as $L$.

    Next, we show that these weak-$*$ limits are actually locally unique gapped ground states of $H$. Namely, according to Definition \ref{definition:locally-unique_gs}, we will show that $\psi_i^\infty(A^\dag\delta_H(A))\geqslant\Delta\psi_i^\infty(A^\dag A)$ for any $A\in\A^l$ such that $\psi^\infty_i(A)=0$, where $\Delta$ is defined in Definition \ref{def: almost generate gs}.
    
    To this end, we will discuss two cases separately. Let us first assume that $A$ satisfies $\lim_{L\rightarrow\infty}|A|\psi_i\ra|=0$. In this case, we have $\psi^\infty_i(A^\dag A)=0$ and $\psi^\infty_i(A^\dag\delta_H(A))=0$. So the condition $\psi^\infty_i(A^\dag\delta_H(A))\geqslant\Delta\psi^\infty_i(A^\dag A)$ is satisfied.

    Next, we consider the case where $|A|\psi_i\ra|$ does not converge to 0. In this case, there is at least one subsequence such that $|A|\psi_i\ra|>v>0$ for large enough sizes. For our purpose, it suffices to focus on this subsequence and consider large enough sizes, and we can expand
    \beq
    \frac{A|\psi_{i}\ra}{|A|\psi_{i}\ra|}=\sum_{j=1}^{z}c_{ij}|\psi_{j}\ra+\sqrt{1-\sum_{j=1}^{z}|c_{ij}|^{2}}|A^{\perp}\ra
    \eeq
    where $c_{ij}:=\la\psi_{j}|A|\psi_{i}\ra/|A|\psi_i\ra|$ and $|A^{\perp}\ra\in\cH_{0,L}^{\perp}$. Therefore,
    \beq \label{eq: average energy}
    \begin{split}
        \frac{\la\psi_{i}|A^{\dagger}H_{L}A|\psi_{i}\ra}{|A|\psi_{i}\ra|^{2}}&\geqslant 0+(1-\sum_{j=1}^{z}|c_{ij}|^{2})\la A^{\perp}|H_{L}|A^{\perp}\ra\\
        &\geqslant (1-\sum_{j=1}^{z}|c_{ij}|^{2})(\delta E_{L}+\Delta_{L})
    \end{split}
    \eeq
    where we have used $\la \phi|H_{L}|\phi\ra\geqslant 0$ for any $|\phi\ra$ and $\la A^{\perp}|H_{L}|A^{\perp}\ra\geqslant \Delta_{L}+\delta E_{L}$.

    According to Eq.~\eqref{eq:sss}, we have $\lim_{L\to\infty}|c_{ij}|=0$ for any $j\neq i$. Moreover, $\lim_{L\rightarrow\infty} |c_{ii}|=0$ because $\psi^\infty_i(A)=0$ and $|A|\psi_\ra|$ has a nonzero lower bound.
    Taking $L\to\infty$ in Eq. \ref{eq: average energy}, we have
    \beq
    \psi^{\infty}_{i}(A^{\dagger}\delta_{H}(A))\geqslant \Delta \psi^{\infty}_{i}(A^{*}A) 
    \eeq

    In summary, $\psi_i^\infty$ is a locally unique gapped ground state of $H$.
\end{proof}

Combining Corollary \ref{coro:almost_degeneracy} and Lemma \ref{lemma:finite_area_law}, we end up with
\begin{proposition}\label{prop:degenerate_area_law}
    Let $H_{L}$ be a sequence of finite-size Hamiltonians such that
    \begin{enumerate}
        \item $\lim_{L\to\infty}H_{L}=H$ is an admissible Hamiltonian.
        \item $H_{L}$ has almost degenerate gapped ground states and its ground state subspace $\cH_{0,L}$ has a constant dimension.
        \item The ground state subspace satisfies the superselection rule.
    \end{enumerate}
    Then any state $|\psi\ra\in\cH_{0,L}$ satisfies area law of entanglement entropy for sufficiently large $L$, in the sense of Definition \ref{def:finite_area_law}.
\end{proposition}
\begin{proof}
    Let $|\psi\ra:=\sum_{i=1}^{z}c_{i}|\psi_{i}\ra\in\cH_{0,L}$ be a normalized state and $\Gamma\in\Lambda$ be a finite connected subset. We write $D_{\Gamma}:=\dim\cH_{\Gamma}$ in this proof, where $\cH_{\Gamma}$ is the local Hilbert space on $\Gamma$. On $\Gamma$ we have
    \beq
    ||(\psi|_{\Gamma}-\sum_{i=1}^{z}|c_{i}|^{2}\psi_{i}|_{\Gamma})||
   =\sup_{||A||=1} |\tr(|\psi\ra\la\psi|A)-\sum_{i=1}^z|c_i|^2\tr(|\psi_i\ra\la\psi_i|A)|
   =\sup_{||A||=1}|\sum_{i\neq j}c_i\bar c_j\tr(|\psi_i\ra\la\psi_j|A)|
    \eeq
    Therefore, by Definition \ref{def: superselection rule}, for any $\epsilon>0$, there exists $L_{\epsilon}$ such that for any $L>L_{\epsilon}$, we have
    \beq
     ||(\psi|_{\Gamma}-\sum_{i=1}^{z}|c_{i}|^{2}\psi_{i}|_{\Gamma})||< \epsilon
    \eeq
    And by Fannes–Audenaert inequality, we have
    \beq
    \begin{split}
        S(\psi|_{\Gamma})&< \sum_{i=1}^{z}|c_{i}|^{2}S(\psi_{i}|_{\Gamma})-\sum_{i=1}|c_{i}|^{2}\log(|c_{i}|^{2})+\frac{\epsilon}{2}\log(D_{\Gamma}-1)\\
        &-\frac{\epsilon}{2}\log(\frac{\epsilon}{2})-(1-\frac{\epsilon}{2})\log(1-\frac{\epsilon}{2})
    \end{split}
    \eeq
    The first summation is bounded since $\psi_{i}$ weak-$*$ converges to locally unique gapped ground state of an admissible Hamiltonian and they satisfy area law by Theorem \ref{thm:area_law}. The remaining terms are $\Gamma$-independent except for $\frac{\epsilon}{2}\log(D_{\Gamma}-1)$. However, we note that this term is arbitrarily small when $\epsilon\to 0$ and $L\to\infty$. Therefore we conclude that
    \beq
    S(\psi|_{\Gamma})<\const
    \eeq
    where $\const$ does not depend on $\Gamma$ for sufficiently large\footnote{However, we must remember that the minimal system size for this area law bound depends on the choice of $\Gamma$. In more mathematical terms, the convergence of entropy is not uniform in the choice of $\Gamma$.} $L$.
\end{proof}

Proposition \ref{prop:degenerate_area_law} gives a criterion for us to check if a set of almost degenerate gapped ground states satisfy the entanglement area law, and the superselection rule is important there. It would be helpful to know when the superselection rule is present. To this end, we introduce the notion of complete distinguishability.

In the following, to simplify the notations, we drop the size-dependence in the Hamiltonians.
Let $H$ be a Hamiltonian of a system of size $L$ with almost degenerate gapped ground states, which span a subspace $\cH_{0}$. The spectrum $\sigma(H)\subseteq[0,\delta E_{L}]\cup[\delta E_{L}+\Delta_{L},\infty)$, where $\lim_{L\to\infty}\delta E_{L}=0$ and $\Delta_{L}\geqslant\Delta>0$ for any $L$.

Denote by the projector $P:=\Pi_{\leqslant\delta E_{L}}$ the spectral projection of $H$. Let $\{O_{\beta}\}_{\beta\in J}$ be a set of local Hermitian operators where $J$ is an indexing set. The notion of complete distinguishability is defined below.

\begin{definition}\label{def:comp_disting}
   If $\{O_{\beta}\}_{\beta\in J}$ satisfies the following
\begin{enumerate}
    \item For any $\beta\in J$, the lowest eigenvalue of $PO_{\beta}P|_{\im P}$ is unique with a uniformly lower-bounded spectral gap $\epsilon_{\beta}$. The unique eigenvector is denoted by $|\psi_{\beta}\ra$.
    \item There exists $I\subseteq J$ such that $\{|\psi_{\alpha}\ra\}_{\alpha\in I}$ span $\cH_{0}$. This implies $|I|\geqslant z$.
    \item Fixing such $I\subseteq J$, for $\alpha,\beta\in I$ with $\alpha\not=\beta$, there exists system-size independent $\eta_{\alpha,\beta}>0$, such that for any finite region $\Gamma$, there exists $\gamma\in J$ satisfying $\supp(O_{\gamma})\cap\Gamma=\emptyset$, such that
    \beq
    |\psi_{\alpha}(O_{\gamma})-\psi_{\beta}(O_{\gamma})|\geqslant \eta_{\alpha,\beta}||O_{\gamma}||
    \eeq
\end{enumerate}
then we say that $\cH_{0}$ is completely distinguishable by local operators $\{O_{\beta}\}_{\beta\in J}$.
\end{definition}

Let us explain the intuition and motivation behind this notion of complete distinguishability. Property 1 essentially means that if the Hermitian operator $O_\beta$ is added to the original Hamiltonian, then a unique $|\psi_\beta\ra$ can be selected from the almost degenerate gapped ground states. Property 2 roughly means that a complete basis in the ground state subspace can be selected out by adding a perturbation of the form $O_\beta$ to the Hamiltonian. Property 3 means that there are many local operators that can be used to distinguish any two states $|\psi_\alpha\ra$ and $|\psi_\beta\ra$, so that their weak-$*$ limits cannot be quasi-equivalent.

In the situation with a spontaneously broken discrete symmetry, the set $\{O_{\beta}\}_{\beta\in J}$ can often be provided by local order parameters (or more precisely, their certain Hermitian combinations).

Below we shift $\{O_{\beta}\}_{\beta\in J}$ by constants so that $\inf(\sigma(PO_{\beta}P|_{\im(P)}))=0$ for all $\beta\in J$. Besides, we rescale them so that $||O_{\beta}||=1,\forall\,\beta\in J$ for later convenience. Therefore, after all these redefinitions, the spectrum of $O_\beta$ is contained in $\{0\}\cup[\epsilon_\beta, 1]$.

The main theorem of this section is the following.
\begin{theorem}\label{thm:ssr}
    Let $H$ be an admissible Hamiltonian with uniformly finite-dimensional almost degenerate gapped ground state subspace $\cH_{0}$, if these ground states are completely distinguishable by local operators $\{O_{\beta}\}_{\beta\in J}$, then there exists a choice of basis $\{|\psi_{\alpha}\ra\}_{\alpha\in I}$ in $\cH_{0}$ that satisfies superselection rule Eq.~\eqref{eq:ssr}.
\end{theorem}
This theorem becomes obvious by combining Corollary \ref{coro:SSR} with the following lemma.

\begin{lemma}\label{lemma:purity}
    Assuming $\{O_{\alpha}\}_{\alpha\in J}$ satisfies property 1 in the definition \ref{def:comp_disting}, then for any $\alpha\in J$, the weak-$*$ limit of $|\psi_{\alpha}\ra$ is pure.
\end{lemma}
\begin{proof}
    The desired result follows if we show that the weak-$*$ limit of $\psi_{\alpha}$ is pure.

    Let us consider a perturbed Hamiltonian $H':=H+\kappa O_{\alpha},\alpha\in I$, where $\kappa>0$ is a coupling constant. Below we show that for a sufficiently small $\kappa$, $H'$ has a unique gapped ground state whose energy gap is proportional to $\kappa$. Furthermore, its unique ground state $|\psi_{\alpha}(\kappa)\ra$ is close to $|\psi_{\alpha}\ra$, and their distance is also controlled by $\kappa$. The argument is similar to the analysis in Lemma \ref{lemma:spec_stability}.

    By the standard resolvent argument, it can be shown that the spectrum of $H'$ is contained in $[0,\delta E_{L}+\kappa]\cup[\Delta_{L}+\delta E_{L},\infty)$. Below we assume $\kappa<\frac{\Delta}{8}$, so these two parts of spectrum remain disjoint.

    We utilize the so-called Feshbach-Schur map \cite{Gustafson2020} to show that $H'$ actually has a unique gapped ground state. For any Hermitian operator $a$ and projector $\Pi$ (whose orthogonal complement is $\Pi^{\perp}:=1-\Pi$), if $a^{\perp}:=\Pi^{\perp}a\Pi^{\perp}|_{\im(\Pi^{\perp})}$ is invertible, we define its Feshbach-Schur transformation by
    \beq
    F_{\Pi}(a):=(\Pi(a-aRa)\Pi)|_{\im \Pi}
    \eeq
    where $R:=\Pi^{\perp}(a^{\perp})^{-1}\Pi^{\perp}$. It is shown in Ref.~\cite{Gustafson2020} that, assuming $a^{\perp}|_{\im \Pi}$ is invertible, then the followings are equivalent,
    \begin{enumerate}
        \item $\lambda\in\sigma(a)$.
        \item $0\in \sigma(F_{\Pi}(a-\lambda))$.
    \end{enumerate}
    
   We now apply this transformation to $H'=H+\kappa O_{\alpha}$ and projector $P$. After some calculation, we find
   \beq
   F_{P}(H'-\lambda)=(PH'P-\kappa^{2}S_{\lambda}-\lambda)|_{\im P}
   \eeq
    where $S_{\lambda}:=PO_{\alpha}R_{\lambda}O_{\alpha}P|_{\im P}$, $R_{\lambda}:=P^{\perp}((P^{\perp}H'P^{\perp}-\lambda)|_{\im P^{\perp}})^{-1}P^{\perp}$. The above form suggests an eigenvalue problem for $PH'P-\kappa^{2}S_{\lambda}$ can tell us the spectrum of $F_P(H'-\lambda)$. However, it is not really the case, since $S_{\lambda}$ depends on $\lambda$ nonlinearly. Since we aim at studying the ground state and ground state energy of $H'$, we assume $0\leqslant\lambda\leqslant\delta E_{L}+\kappa$ below. For this $\lambda\in\sigma(H')$, we have
   \beq
   \lambda\in\sigma(PH'P|_{\im P}-\kappa^{2}S_{\lambda})
   \eeq
    By expanding $PH'P=PHP+\kappa PO_{\alpha}P$ and noting $||PHP||\leqslant\delta E_{L}$, we can treat $PHP$ as a perturbation to $PO_{\alpha}P$. On the other hand, on $\im P^{\perp}$ we have
    \beq
    \begin{split}
        P^{\perp}H'P^{\perp}|_{\im P^{\perp}}-\lambda &=P^{\perp}HP^{\perp}|_{\im P^{\perp}}+\kappa P^{\perp}O_{\alpha}P^{\perp}|_{\im P^{\perp}}-\lambda\\
        &\geqslant (\Delta_{L}+\delta E_{L})-\kappa-\delta E_{L}-\kappa\\
        &\geqslant \Delta-2\kappa>\frac{3}{4}\Delta
    \end{split}
    \eeq
    Thus it is invertible on $\im P^{\perp}$. In addition,
    \beq
    ||S_{\lambda}||\leqslant||(P^{\perp}H'P^{\perp}|_{\im P^{\perp}}-\lambda)^{-1}||<\frac{4}{3\Delta}
    \eeq
    The gap of $\kappa PO_{\alpha}P$ is not closed by $PHP$ and $-\kappa^{2}S_{\lambda}$ if
    \beq\label{eq:choice_kappa}
    \begin{split}
        \delta E_{L}&<\frac{1}{12}\kappa \epsilon_{\alpha}\\
        \frac{4\kappa^{2}}{3\Delta}&<\frac{1}{24}\kappa\epsilon_{\alpha}
    \end{split}
    \eeq
    To summarize, we require,
    \beq
    \frac{12\delta E_{L}}{\epsilon_{\alpha}}<\kappa<\min\{\frac{1}{32}\epsilon_{\alpha}\Delta,\frac{\Delta}{8}\}
    \eeq
    This can always be achieved for large $L$ since $\delta E_{L}\to0$. With this choice, the spectrum of $PH'P|_{\im P}$ is contained in $[0, \delta E_L+\frac{4\kappa^2}{3\Delta})\cup(\kappa\epsilon_\alpha-\frac{4\kappa^2}{3\Delta}, \kappa+\delta E_L+\frac{4\kappa^2}{3\Delta})$.

    Next, we show that the ground state of $H'$ is unique and give more precise estimations of the ground state energy and gap. Fix a parameter $r$ satisfying $\frac{1}{4}\kappa\epsilon_{\alpha}<r<\frac{1}{2}\kappa\epsilon_{\alpha}$.
    Consider the following projections,
    \beq
    \begin{split}
        \Pi&:=\oint_{|t|=r}\frac{1}{t-\kappa  PO_{\alpha}P}\frac{\dd t}{2\pi i}\\
        \Pi'&:=\oint_{|t|=r}\frac{1}{t-\kappa PO_{\alpha}P-PHP+\kappa^{2}S_{\lambda}}\frac{\dd t}{2\pi i}
    \end{split}
    \eeq
    By assumption, $PO_{\alpha}P$ has a unique ground state $|\psi_{\alpha}\ra$, and thus $\Pi=|\psi_{\alpha}\ra\la\psi_{\alpha}|$. We aim to bound $||\Pi-\Pi'||$ from above. To this end, note
    \beq
        \Pi-\Pi'&=\oint_{|t|=r}(t-\kappa PO_{\alpha}P)^{-1}(\kappa^{2}S_{\lambda}-PHP)(t-\kappa PO_{\alpha}P-PHP+\kappa^{2}S_{\lambda})^{-1}\frac{\dd t}{2\pi i}
    \eeq
    Thus we have
    \beq
        ||\Pi-\Pi'||&\leqslant r||(r-\kappa PO_{\alpha}P)^{-1}||\times||(\kappa^{2}S_{\lambda}-PHP)||\times||(r-\kappa PO_{\alpha}P-PHP+\kappa^{2}S_{\lambda})^{-1}||
    \eeq
    With our choice of $\frac{1}{4}\kappa\epsilon_{\alpha}<r<\frac{1}{2}\kappa\epsilon_{\alpha}$,
    \beq\label{eq:resolvent_bound_3}
    \begin{split}
        ||(r-\kappa PO_{\alpha}P)^{-1}||&=\frac{1}{r}\\
        ||(\kappa^{2}S_{\lambda}-PHP)||&\leqslant \frac{4\kappa^{2}}{3\Delta}+\delta E_{L}\\
        ||(r-\kappa PO_{\alpha}P-PHP+\kappa^{2}S_{\lambda})^{-1}||&\leqslant\max\{(r-\delta E_{L}-\frac{4\kappa^{2}}{3\Delta})^{-1},(\kappa\epsilon_{\alpha}-\frac{4\kappa^{2}}{3\Delta}-r)^{-1}\}\leqslant\frac{8}{\kappa\epsilon_{\alpha}}
    \end{split}
    \eeq
    where we have used Eq.~\eqref{eq:choice_kappa} in the last inequality.
    Consequently,
    \beq
    ||\Pi-\Pi'||\leqslant\frac{8}{\kappa\epsilon_{\alpha}}(\frac{4\kappa^{2}}{3\Delta}+\delta E_{L})<1
    \eeq
    This implies $\rk(\Pi')=\rk(\Pi)=1$.
    Therefore, $H'$ has a unique gapped ground state $|\psi_{\alpha}(\kappa)\ra$, and the ground state energy lies in $[0, \delta E_L+\frac{4\kappa^2}{3\Delta})$ with a gap lower bounded by
    \beq
    \Delta'>\kappa\epsilon_\alpha-\delta E_L-\frac{8\kappa^2}{3\Delta}
    \eeq
    
    Moreover, the trace norm distance between $\psi_{\alpha}$ and $\psi_{\alpha}(\kappa)$ is given by
    \beq
        \begin{split}
            ||\psi_{\alpha}-\psi_{\alpha}(\kappa)||&=||\Pi-\Pi'||_{1}\\
            &=2\sqrt{1-\tr(\Pi'\Pi)}
        \\&=2||\Pi-\Pi'||\\&\leqslant\frac{16}{\kappa\epsilon_{\alpha}}(\frac{4\kappa^{2}}{3\Delta}+\delta E_{L})
        \end{split}
    \eeq
    where $||A||_{1}:=\sup_{||B||=1}|\tr(AB)|$ is the trace norm.
    
    By Lemma \ref{lemma:weak*_lsc},
    \beq
     ||\psi_{\alpha}^{\infty}-\psi_{\alpha}^{\infty}(\kappa)||\leqslant\lowlim_{L\to\infty}||\psi_{\alpha}-\psi_{\alpha}(\kappa)||\leqslant\frac{64\kappa}{3\Delta\epsilon_{\alpha}}
    \eeq
    where $\psi^{\infty}_{\alpha}$ and $\psi^{\infty}_{\alpha}(\kappa)$ are weak-$*$ limits of $\psi_{\alpha}$ and $\psi_{\alpha}(\kappa)$ respectively.

    We note that $\psi_{\alpha}(\kappa)$ is a unique gapped ground state of an admissible Hamiltonian $H+\kappa O_{\alpha}$, so  by Theorem \ref{thm:limit_ad_gs} its weak-$*$ limit $\psi^{\infty}_{\alpha}(\kappa)$ is a locally unique gapped ground state, and hence pure by Theorem \ref{Thm:pure_gs}. By Corollary 4.8 of Ref.~\cite{Kadison1982limits}, $\psi^{\infty}_{\alpha}=\lim_{\kappa\to0^{+}}\psi^{\infty}_{\alpha}(\kappa)$ is also pure and equivalent to $\psi_{\alpha}^{\infty}(\kappa)$ for small $\kappa$.
\end{proof}

\subsection{Symmetry-enforced long-range entanglement in finite-size systems} \label{subsec: thermodynamic limit of SRE}

In this appendix, we establish Theorem \ref{thm:infinite_SRE_main} and Theorem \ref{thm:finite_SRE_main} on symmetry-enforced long-range entanglement. We start with the infinite-size version and turn to the finite-size version next.

In the thermodynamic limit, short-range entangled states are defined in Ref.~\cite{Kapustin2020invertible} via locally-generated automorphisms (LGA), which will be reviewed below. In order to discuss LGAs, we need to define the notion of almost local operators introduced in Ref.~\cite{Kapustin2020invertible}, which is obtained by completing $\A^{\ell}$ with the following family of semi-norms\footnote{A semi-norm satisfies all axioms of a norm except non-degeneracy. Non-degeneracy of a norm $||\cdot||$ means $||x||=0\Rightarrow x=0$.},
\beq
||A||_{n,j}:=||A||+\sup_{r\in\R_{>0}}(1+r)^{n}\inf_{a\in\A_{B(j,r)}}||a-A||,\quad A\in\A^l
\eeq
where $n\in\mathbb{N}$ and $B(j,r)\subseteq\Lambda$ means the ball centered at site $j\in\Lambda$ with radius $r$. We denote this completion as $\A^{al}$ and call its elements almost-local operators.
It is shown in Ref.~\cite{Kapustin2022Noether} that the completion $\A^{al}$ does not depend on the choice of $j$ above.
Roughly speaking, elements in $\A^{al}$ are quasi-local operators with $O(r^{-\infty})$ tails when approaching infinity. This feature makes it particularly useful to describe gapped quantum phases. 

Let us denote by $\cF_{\infty}^{+}$ all monotonically decreasing non-negative functions on $\R_{\geqslant 0}$, such that
\beq
\lim_{r\to\infty}f(r)=O(r^{-\infty}),\quad\forall f\in\cF_{\infty}^{+}
\eeq
Given $A\in\A^{al}$ and $f\in\cF_{\infty}^{+}$,  we say $A$ is $f$-localized at site $j\in\Lambda$ if 
\beq
\inf_{a\in B(j,r)}||a-A||<f(r)
\eeq

An almost local Hamiltonian \cite{Kapustin2020invertible,Kapustin2022Noether} is defined to be the following formal combination,
\beq
H=\sum_{j\in\Lambda}h_{j}
\eeq
where each $h_{j}\in\A^{al}$ with $||h_{j}||<C$ for some constant $C$. In general, $H\not\in\A^{al}$ but it makes sense as a derivation on $\A^{al}$, \ie $\delta_{H}(A)=\sum_{j}[h_{j},A]\in\A^{al}$ for any $A\in\A^{al}$. We say this Hamiltonian is uniformly almost local (UAL) if there is a function $f\in\cF_{\infty}^{+}$ such that each $h_{j}$ is $f$-localized at $j\in\Lambda$. It is shown in appendix C of Ref.~\cite{Kapustin2022Noether} that every UAL Hamiltonian has the following equivalent decomposition,
\beq
H=\sum_{|Y|<\infty}H^{Y}
\eeq
where $Y\subseteq\Lambda$ and $H^{Y}\in \A^{\ell}_{Y}$ satisfies the following condition
\beq\label{eq:canonical_decomposition}
\la a^{*}H^{Y}\ra_{\infty}=0,\quad\forall a\in\A_{Z}, Z \subsetneqq Y
\eeq
where $\la\cdot\ra_{\infty}$  is the tracial state (Eq.~\eqref{eq:tracial_state}). Eq.~\eqref{eq:canonical_decomposition} essentially means $H^{Y}$ is supported on $Y$ but not on any proper subset of $Y$. The UAL condition then implies that there exists $b\in\cF_{\infty}^{+}$ such that $||H^{Y}||<b(\diam(Y))$. The relation between $b$ and $f$ is given by (see Eq.~(122) of Ref.~\cite{Kapustin2022Noether})
\beq
b(r)=2^{2d+1}f(\max\{0,\frac{r}{2\sqrt{d}}-2\})
\eeq
where $d$ is the spatial dimension ($d=1$ for us). In this case, we say $H$ is $b$-localized. Similarly, if a Hamiltonian is time-dependent, \ie $H=H(t)=\sum_{|Y|<\infty}H^{Y}(t)$, then we say it is UAL if $||H^{Y}(t)||<b(\diam(Y))$ for some $b\in\cF_{\infty}^{+}$ for any $Y$ and $t$.

It is also shown in Ref.~\cite{Kapustin2020Hall} (see Lemma A.2 therein) that every UAL Hamiltonian can be exponentiated to be a strongly continuous family $\alpha^{t}$ of $*$-automorphism of $\A^{al}$, \ie it satisfies
\beq
\frac{\dd }{\dd t}\alpha^{t}(A)=i\alpha^{t}([H(t),A])
\eeq
where $H(t)$ is a UAL Hamiltonian and $\alpha_{t=0}=\id$.

A $*$-automorphism $\alpha\in\Aut(\A^{al})$ is called a locally generated automorphism (LGA) if $\alpha=\alpha^{t=1}$ for some family of $*$-automorphisms generated by UAL Hamiltonians. In particular, from the Lieb-Robinson bound (see Ref.~\cite{Nachtergaele_2006}),
an LGA is an LPA with $O(r^{-\infty})$ tail and the tail function only depends on the tail of its generating UAL Hamiltonian.

With the notion of LGA, short-range entangled (SRE) states in infinite-size systems are then defined in Ref.~\cite{Kapustin2020invertible} as follows.
\begin{definition}\label{def:SRE}
    A state $\psi$ is called short-range entangled (SRE) if 
    \beq
    \psi=\Omega\circ\alpha
    \eeq
    where $\Omega$ is a pure product state while $\alpha$ is an LGA.
\end{definition}

Because $\Omega$ is pure, $\psi$ is also pure{\footnote{To see it, suppose $\psi$ is mixed. By Definition \ref{def:alternative_pure}, there is a positive linear functional $\rho$ such that $\psi-\rho=\Omega\circ\alpha-\rho$ is still positive and $\rho\neq \lambda\psi=\lambda\Omega\circ\alpha$ with $\lambda\in[0, 1]$. This contradicts the assumption that $\Omega$ is pure, since there would be a positive linear functional $\rho\circ\alpha^{-1}$ such that $\Omega-\rho\circ\alpha^{-1}$ is positive, and $\rho\circ\alpha^{-1}\neq\lambda\Omega$ with $\lambda\in[0, 1]$.}} (and hence a type-I factor). It is proved in Ref.~\cite{Kapustin2020invertible} that all SRE states satisfy the area law of entanglement entropy (see Lemma 4.2 in Ref. \cite{Kapustin2020invertible}) and therefore splits by Lemma \ref{lemma:split again}. Thus, by Lemma \ref{lemma:Kapustin_Sopenko}, we have
\begin{corollary}\label{coro:SRE}
    An SRE state $\psi$ cannot be symmetric under an anomalous symmetry action.
\end{corollary}
We believe this corollary is known to authors of Ref.~\cite{kapustin2024anomalous}, although it is not explicitly mentioned in their paper.

Corollary \ref{coro:SRE} establishes the symmetry-enforced long-range entanglement in infinite-size systems. However, it is also necessary to discuss the finite-size analog of Corollary \ref{coro:SRE}, since all physical systems have a finite size in practice. It is important to recognize that in the finite-size case, short-range entanglement is not a property of a single state, but a property of a sequence of states.

To define finite-size SRE states, we first extend the notions of UAL Hamiltonians and LGA to sequences of finite-size systems. For a sequence of finite-size systems defined in the beginning of Appendix~\ref{sec: thermodynamic limit}, we first define a sequence of UAL Hamiltonians as follows.
\begin{definition}
    Let $H_{L}\in\A_{\Gamma_{L}}$ has the following decomposition,
    \beq
    H_{L}=\sum_{Y\subseteq \Gamma_{L}}H^{Y}_{L}
    \eeq
    We say $\{H_{L}\}_{L=1,2,\dots}$ is a sequence of UAL Hamiltonians if there exists $b\in\cF_{\infty}^{+}$, such that
    \beq\label{eq:UAL}
    ||H_{L}^{Y}||<b(\diam(Y)),\quad\forall L, Y\subseteq \Gamma_{L}
    \eeq
    If $H_{L}$ is time-dependent, \ie $H_{L}=H_{L}(t)$, it is called UAL if it satisfies Eq.~\eqref{eq:UAL} for all $t$.
\end{definition}
Similarly, the finite-size version of LGA can be obtained by exponentiating $H_{L}$ above.
\begin{definition}
    A sequence $\alpha_{L}\in\Aut(\A_{\Gamma_{L}})$ is called a sequence of (finite-size) LGA, if $\alpha_{L}$ is the time evolution generated by some UAL $H_{L}(t)$ within $t\in[0,1]$.
\end{definition}
We therefore arrive at the following definition of a sequence of SRE states in finite-size systems.
\begin{definition}\label{def:sequence_SRE}
    Let $\{\psi_{L}\}_{L=1,2,\dots}$ be a sequence of finite-size states. It is an SRE sequence if there exists a sequence of finite-size LGA $\beta_{L}$ such that
    \beq
    \psi_{L}=\Omega\circ\beta_{L}
    \eeq
    and $\beta_{L}$ is called a disentangler of $\psi_{L}$.
\end{definition}

\begin{remark}
    It is important to note that we do \textbf{NOT} require the convergence of $\psi_{L}$ or $\beta_{L}$.
\end{remark}

Again, we set $\alpha_{L}:G\to \Aut(\A^{\ell}_{\Gamma_{n}})$ to be the symmetry action in finite-size systems, which strongly converges to some $\alpha:G\to\G^{lp}$. Then, our main theorem is the following.
\begin{theorem}\label{thm:finite_SRE}
    If $\lim_{L\to\infty}\alpha_{L}=\alpha$ is an anomalous symmetry action and $\{\psi_{L}\}$ is an SRE sequence, then the symmetry condition $\psi_{L}\circ\alpha_{L}(g)=\psi_{L},\,\forall g\in G$ can be satisfied for only finitely many $L$.
\end{theorem}
\begin{proof}
    Suppose the symmetry condition is true for infinitely many $L$, then Theorem \ref{thm:BAB} implies there is a weak-$*$ convergent subsequence. We focus on such a subsequence and rename it as $\{\psi_{L}\}$, and write $\lim_{L\to\infty}\psi_{L}=\psi$. Below we show that $\psi$ is clustering (hence it is a factor state by Theorem \ref{thm:clustering}) and satisfies area law of entanglement entropy. We then conclude from Proposition~\ref{prop:factor_entropy} that $\psi$ is of type-I. Moreover, $\psi$ splits due to Lemma \ref{lemma:split again}.

    On the other hand, $\psi_{L}\circ\alpha_{L}(g)=\psi_{L}$ for any $L$, so by lemma~\ref{lemma:composition}, we have $\psi\circ\alpha(g)=\psi$. It follows from Lemma \ref{lemma:typeI_KS} that $\alpha:G\to\G^{lp}$ has vanishing anomaly index, a contradiction.

    Now we first show that $\psi$ satisfies the entanglement area law. For a fixed, finite and connected region $\Gamma\subseteq\Lambda$, let us assume $L$ is large enough such that $\Gamma\subseteq\Gamma_{L}$. It is shown in Lemma A.4 of Ref.~\cite{Kapustin2020Hall} that $\alpha_{L}$ has the following approximate factorization:
    \beq
    \alpha_{L}=\alpha_{\Gamma}\Ad_{U}\alpha_{\Gamma^{c}}
    \eeq
    where $\Gamma^{c}:=\Gamma_{L}\setminus\Gamma$, $\alpha_{\Gamma}$ acts trivially on $\Gamma^{c}$, $\alpha_{\Gamma^{c}}$ acts trivially on $\Gamma$ and $U$ is an almost local unitary which is $h$-localized on $\partial\Gamma$ for some $h\in\cF_{\infty}^{+}$. Importantly, it is shown in Appendix A of Ref.~\cite{kapustin2024anomalous} (see the proof of Proposition A.1 and Lemma A.3 therein) that the tail function $h(r)$ depends only on the tail of the generating UAL Hamiltonians and is hence independent of $L$.
    
    Note $\alpha_{\Gamma}$ and $\alpha_{\Gamma^{c}}$ do not change the entropy on $\Gamma$, so we have
    \beq
    S(\psi_{L}|_{\Gamma})=S(\Omega\circ\Ad_{U}|_{\Gamma})
    \eeq
    We can bound the right hand side as follows. Note that for a site $j\in\Lambda$ and $A\in\A_{j}$, we have
    \beq
    |\Omega\circ\Ad_{U}(A)-\Omega(A)|<h(d(j,\partial\Gamma))||A||
    \eeq
    where $d(j,\partial\Gamma)$ means the distance from $j$ to $\partial\Gamma$. By the Fannes-Audenaert inequality (Theorem 5.39 of Ref.~\cite{Watrous_2018}),
    \beq
    S(\Omega\circ\Ad_{U}|_{j})=|S(\Omega\circ\Ad_{U}|_{j})-S(\Omega|_{j})|<a(d(j,\partial\Gamma))
    \eeq
    where $a\in\cF_{\infty}^{+}$ depends only on the tail function $h$ and the dimension of local Hilbert space.
    By the subaddtivity of entanglement entropy, we have
    \beq
    \begin{split}
        S(\Omega\circ\Ad_{U}|_{\Gamma})&\leqslant \sum_{j\in\Gamma}S(\Omega\circ\Ad_{U}|_{j})\\
        &\leqslant \sum_{j\in\Gamma}a(d(j,\partial\Gamma))\\
        &\leqslant \sum_{j\in\Gamma}\sum_{p\in\partial\Gamma}a(d(j,p))\\
        &\leqslant\sum_{p\in\partial\Gamma}\sum_{r=0}^{\infty}a(r)
        =C|\partial\Gamma|
    \end{split}
    \eeq
    where we have used $C:=\sum_{r=0}^{\infty}a(r)<\infty$ since $a\in\cF_{\infty}^{+}$ and $C$ does not depend on $L$. This implies that $\{\psi_{L}\}$ satisfies the area law in the sense of Definition~\ref{def:finite_area_law}. By Lemma \ref{lemma:finite_area_law}, $\psi$ satisfies the area law as well.

    Next, we show that $\psi$ is clustering. Let $A,B$ be local operators supported on finite subsets $X,Y$ respectively, and we assume $L$ is large enough such that $X,Y\subseteq\Gamma_{L}$.
    We set $l:=d(X,Y)$. According to the Lieb-Robinson bound \cite{Kapustin2022Noether,Ranard_2022} and the assumption that $\alpha_{L}$ is an LGA, there exists $A'\in\A_{B(X,l/3)}$, $B'\in\A_{B(Y,l/3)}$ and $f\in\cF_{\infty}^{+}$ such that
    \beq
    \begin{split}
        ||\alpha_{L}(A)-A'||<f(\frac{l}{3})||A||\\
        ||\alpha_{L}(B)-B'||<f(\frac{l}{3})||B||
    \end{split}
    \eeq
    Note that $B(X,\frac{l}{3})\cap B(Y,\frac{l}{3})=\emptyset$, so $\Omega(A'B')=\Omega(A')\Omega(B')$. Therefore,
    \beq
    \begin{split}
        |\psi_{L}(AB)-\psi_{L}(A)\psi_{L}(B)|&=|\Omega(\alpha_{L}(A)\alpha_{L}(B))-\Omega(\alpha_{L}(A))\Omega(\alpha_{L}(B))|\\
        &= |\Omega(\alpha_{L}(A)(\alpha_{L}(B)-B'))+\Omega((\alpha_{L}(A)-A')B')\\
        &+\Omega(A'B')-\Omega(A')\Omega(\alpha_{L}(B))+\Omega(A')\Omega(\alpha_{L}(B))-\Omega(\alpha_{L}(A))\Omega(\alpha_{L}(B))|\\
        &\leqslant||A||\cdot ||\alpha_{L}(B)-B'||+||\alpha_{L}(A)-A'||\cdot ||B'||\\&+||A'||\cdot||B'-\alpha_{L}(B)||+||B||\cdot||A'-\alpha_{L}(A)||\\
        &\leqslant2f(\frac{l}{3})(2+f(\frac{l}{3}))||A||\cdot ||B||
    \end{split}
    \eeq
    Note the right hand side does not depend on $L$, thus we can take $L\to\infty$ and use the definition of weak-$*$ limit. We end up with
    \beq
    |\psi(AB)-\psi(A)\psi(B)|\leqslant 2f(\frac{l}{3})(2+f(\frac{l}{3}))||A||\cdot||B||,\quad l:=d(\supp(A),\supp(B))
    \eeq
    This shows $\psi$ is clustering and completes the proof.
\end{proof}
\begin{remark}
    We remark that it is currently not known if our definition on SRE sequences ensures its weak-$*$ limit (if exists) is an SRE state in the sense of Definition~\ref{def:SRE}, since the convergence of states does not imply the convergence of disentanglers (\ie $\beta_L$ in Definition \ref{def:sequence_SRE}) even if we restrict ourselves to a subsequence.
\end{remark}

\subsection{Periodic boundary conditions}\label{sec:periodic}

Although the theory of thermodynamic limit is easier to establish for open boundary conditions, finite systems with periodic boundary conditions are also relevant and sometimes more interesting, especially when one considers anomalous symmetries. For example, translations cannot even be defined under open boundary conditions.

An intuitive approach to generalize our discussion to periodic chains is to cut the periodic chain somewhere and then one can embed this chain to the infinite lattice $\Lambda$ isometrically. However, special care is needed when doing so, since some originally local operators become non-local when the periodic chain is cut. Below we discuss this procedure in detail. In particular, we will extend Theorem \ref{thm:finite_SRE} to systems under periodic boundary condition.

Let $\hat{\Gamma}_{L}$ be a periodic chain with $L$ sites. We label the sites by integers in $\{0,1,\dots,L-1\}$ with the identification $n\sim n+L$ for any integer $n$, and we denote their equivalence class by $\bar{n}$. The metric on $\hat{\Gamma}_{L}$ is given by
\beq
\hat{d}(\bar{m},\bar{n}):=\min_{a\in\bar{m},b\in\bar{n}}\,|a-b|
\eeq
Similarly, let $Y\subseteq\hat{\Gamma}_{L}$, and we define its diameter as
\beq
\widehat{\diam}(Y):=\max_{p,q\in Y}\hat{d}(p,q)
\eeq
With this notation, $\widehat{\diam}(\hat{\Gamma}_{L})=\lfloor\frac{L}{2}\rfloor$. One can define a UAL Hamiltonian as 
\beq
\hat{H}_{L}=\sum_{Y\subseteq\hat{\Gamma}_{L}}\hat{H}^{Y}_{L}
\eeq
with $||\hat{H}^{Y}||<b(\widehat{\diam}(Y))$ for some $b\in\cF_{\infty}^{+}$. Similarly, one can define LGAs and SRE states on periodic chains, by following Appendix \ref{subsec: thermodynamic limit of SRE}. By the Lieb-Robinson bound, the LGA $\beta_{L}$ generated by $\hat{H}_{L}$ is an LPA with tail $f\in\cF_{\infty}^{+}$ which only depends on $b$. Following the proof of Theorem \ref{thm:finite_SRE}, one can show that if $\psi_{L}$ is an SRE whose disentangler has tail $f\in\cF_{\infty}^{+}$, then
\begin{enumerate}
    \item For a finite, connected subset $\hat{\Gamma}\subseteq\hat{\Gamma}_{L}$, $S(\psi_{L}|_{\Gamma})<\const$ where the constant depends on the tail $f$ and the dimension of local Hilbert space only.
    \item For local operators $A,B$ with disjoint supports, we have
    \beq\label{eq:finite_clustering}
    |\psi_{L}(AB)-\psi_{L}(A)\psi_{L}(B)|<g(\hat{d}(\supp(A),\supp(B))||A||\cdot||B||
    \eeq
    where $g\in\cF^{+}_{\infty}$ depends on $f$ only.
\end{enumerate}

However, on a closed chain, it is not clear how to embed $\hat{\Gamma}_{L}\hookrightarrow\hat{\Gamma}_{L'}$ when $L'>L$ or how to view $\A_{\hat{\Gamma}_{L}}$ as a subalgebra of $\A^{ql}$. This makes the thermodynamic limit less obvious.
To address this, let us choose a bijective map $f_{L}:\hat{\Gamma}_{L}\to\Gamma_{L}\subseteq\Lambda$, where $\{\Gamma_{L}\}$ is an increasing and exhausting sequence of finite connected subsets in $\Lambda$. Although $f_{L}$ cannot be an isometry on $\hat{\Gamma}_{L}$, we do require the following,
\begin{enumerate}
    \item $\partial\Gamma_{L}=\{f_{L}(\bar{0}),f_{L}(\overline{L-1})\}$.
    \item $f_{L}$ is an isometry on $\hat{\Gamma}\subseteq\hat{\Gamma}_{L}$ if $\hat{\Gamma}$ is concatenate, “sits in the bulk of $\hat{\Gamma}_{L}$”, \ie $\bar{0},\overline{L-1}\not\in \hat{\Gamma}$ and $\hat{\Gamma}$ is “small enough”, \ie $\widehat{\diam}(\hat{\Gamma})<\lfloor\frac{L}{2}\rfloor$. We call such a subet $\hat{\Gamma}$ a good subset.
\end{enumerate}
We write $\hat{\Gamma}_{L}\hookrightarrow\hat{\Gamma}_{L'}$ if $\Gamma_{L}\subseteq\Gamma_{L'}$ and identify $\A_{\bar{n}}$ with $\A_{f(\bar{n})}$. In this way, one is able to identify $\A_{\hat{\Gamma}_{L}}$ as a $*$-subalgebra of $\A^{ql}$. Fixing $\Omega$ to be a pure product state of $\A^{ql}$ and letting $\psi_{L}$ be a quantum state of $\A_{\hat{\Gamma}_{L}}$, we extend it to be a state of $\A^{ql}$ as follows,
\beq
\tilde{\psi}_{L}:=\psi_{L}\otimes \Omega|_{\Gamma_{L}^{c}}
\eeq
where $\Gamma_{L}^{c}$ is the complement of $\Gamma_{L}$ in $\Lambda$. We call the sequence $\{\tilde{\psi}_{L}\}$ an SRE sequence iff $\psi_{L}$ has an LGA disentangler with uniform tail function $f\in\cF_{\infty}^{+}$.
Below we will drop the $\Omega$-dependence and not distinguish $\tilde{\psi}_{L}$ and $\psi_{L}$.

For a fixed, finite and connected $\Gamma\subseteq\Lambda$, there exists $L_{0}$ such that $f_{L}^{-1}(\Gamma)$ is a good subset in $\hat{\Gamma}_{L}$ for any $L>L_{0}$. In this case, if an SRE sequence $\psi_{L}$ is weak-$*$ convergent (which is always possible by passing to a subsequence) with limit $\psi$, then
\begin{enumerate}
    \item $S(\psi_{L}|_{\Gamma})<\const$, where $\const$ depends on the tail $f\in\cF_{\infty}^{+}$ of disentanglers and local Hilbert space dimension only. Therefore, by Lemma \ref{lemma:finite_area_law}, $\psi$ satisfies the area law as well.
    \item For any local operator $A,B\in\A^{l}$ on $\Lambda$ with disjoint supports and $L$ large enough such that $f_{L}^{-1}(\supp(A)\cup\supp(B))$ is contained in a good subset, we obtain from Eq.~\eqref{eq:finite_clustering}
    \beq
    |\psi_{L}(AB)-\psi_{L}(A)\psi_{L}(B)|<g(d(\supp(A),\supp(B)))||A||\cdot||B||
    \eeq
    where $g\in\cF_{\infty}^{+}$ depends on $f$ only. By taking $L\to\infty$ and the definition of weak-$*$ convergence, this implies
    \beq
    |\psi(AB)-\psi(A)\psi(B)|<g(d(\supp(A),\supp(B)))||A||\cdot||B||
    \eeq
    In particular, $\psi$ is clustering.
\end{enumerate}
Thus, we establish the analogue of Theorem \ref{thm:finite_SRE} under periodic boundary conditions. The thermodynamic limit of unique gapped ground state can also be achieved similarly for periodic boundary conditions.

\newpage
\bibliography{lib.bib}

\printunsrtglossary[type=symbols, style=long, title=List of symbols]

\end{document}